\documentclass[a4paper,10pt]{article}

\usepackage[vlined,boxed]{algorithm2e}
\usepackage{thmtools}
\usepackage{amsmath}
\usepackage{amsfonts}
\usepackage{epsfig}
\usepackage{graphics}
\usepackage{color}
\usepackage{comment}
\usepackage{paralist}
\usepackage{mathtools}
\usepackage{multirow}
\usepackage{subfigure}
\usepackage{eepic}
\usepackage{stmaryrd}
\usepackage{textcomp}
\usepackage{balance}
\usepackage[inline]{enumitem}
\usepackage{ulem}
\usepackage{bigdelim}
\usepackage{makecell}
\usepackage{tikz}
\usetikzlibrary{arrows}

\usepackage[utf8]{inputenc}
\usepackage{amssymb}
\usepackage{amsmath}
\usepackage{amsthm}
\usepackage[margin=3cm]{geometry}
\usepackage{comment}
\usepackage{graphicx}
\usepackage{booktabs}

\theoremstyle{definition}
\newtheorem{definition}{Definition}
\newtheorem{theorem}{Theorem}

\newtheorem{lemma}{Lemma}

\newtheorem{proposition}{Proposition}
\newtheorem{corollary}{Corollary}
\newtheorem{example}{Example}

\title{Exact and Approximate Counting of Database Repairs}

\author{
	Marco Calautti$^1$,
	Ester Livshits$^2$,
	Andreas Pieris$^{2,3}$,
	Markus Schneider$^3$\\[2mm]
	$^1$Department of Computer Science, University of Milan\\
	$^2$School of Informatics, University of Edinburgh\\
	$^3$Department of Computer Science, University of Cyprus
}

\date{}

\begin{document}

\newcommand{\OMIT}[1]{}
\newcommand{\mi}[1]{\mathit{#1}}
\newcommand{\ins}[1]{\mathbf{#1}}
\newcommand{\adom}[1]{\mathsf{dom}(#1)}
\renewcommand{\paragraph}[1]{\textbf{#1}}
\newcommand{\ra}{\rightarrow}
\newcommand{\fr}[1]{\mathsf{fr}(#1)}
\newcommand{\dep}{\Sigma}
\newcommand{\sch}[1]{\mathsf{sch}(#1)}
\newcommand{\ar}[1]{\mathsf{ar}(#1)}
\newcommand{\body}[1]{\mathsf{body}(#1)}
\newcommand{\head}[1]{\mathsf{head}(#1)}
\newcommand{\guard}[1]{\mathsf{guard}(#1)}
\newcommand{\class}[1]{\mathbb{#1}}
\newcommand{\pos}[2]{\mathsf{pos}(#1,#2)}
\newcommand{\app}[2]{\langle #1,#2 \rangle}
\newcommand{\crel}[1]{\prec_{#1}}

\newcommand{\ccrel}[1]{\prec_{#1}^+}

\newcommand{\tcrel}[1]{\prec_{#1}^{\star}}
\newcommand{\rctaa}{\class{CT}_{\forall \forall}^{\mathsf{res}}}
\newcommand{\rctaapr}{\mathsf{CT}_{\forall \forall}^{\mathsf{res}}}
\newcommand{\rctae}{\class{CT}_{\forall \exists}^{\mathsf{res}}}
\newcommand{\rctaepr}{\mathsf{CT}_{\forall \exists}^{\mathsf{res}}}
\newcommand{\base}[1]{\mathsf{base}(#1)}
\newcommand{\eqt}[1]{\mathsf{eqtype}(#1)}
\newcommand{\result}[1]{\mathsf{result}(#1)}
\newcommand{\chase}[2]{\mathsf{ochase}(#1,#2)}
\newcommand{\pred}[1]{\mathsf{pr}(#1)}
\newcommand{\origin}[1]{\mathsf{org}(#1)}
\newcommand{\eq}[1]{\mathsf{eq}(#1)}
\newcommand{\dept}[1]{\mathsf{depth}(#1)}
\newcommand{\comp}[2]{\mathsf{comp}_{#2}(#1)}

\def\val#1{\mathtt{#1}}
\def\attr#1{\mathrm{#1}}
\def\rel#1{\mathit{#1}}

\newcommand{\rep}[2]{\mathsf{rep}_{#2}(#1)}
\newcommand{\repp}[2]{\mathsf{rep}_{#2}\left(#1\right)}
\newcommand{\rfreq}[2]{\mathsf{rfreq}_{#2}(#1)}
\newcommand{\homs}[3]{\mathsf{hom}_{#2,#3}(#1)}
\newcommand{\prob}[1]{\mathsf{#1}}
\newcommand{\key}[1]{\mathsf{key}(#1)}
\newcommand{\keyval}[2]{\mathsf{key}_{#1}(#2)}
\newcommand{\block}[2]{\mathsf{block}_{#2}(#1)}
\newcommand{\sblock}[2]{\mathsf{sblock}_{#2}(#1)}

\newcommand{\rt}[1]{\mathsf{root}(#1)}
\newcommand{\child}[1]{\mathsf{child}(#1)}

\newcommand{\var}[1]{\mathsf{var}(#1)}
\newcommand{\const}[1]{\mathsf{const}(#1)}
\newcommand{\pvar}[2]{\mathsf{plhsvar}_{#2}(#1)}
\newcommand{\prhsvar}[2]{\mathsf{prhsvar}_{#2}(#1)}

\newcommand{\att}[1]{\mathsf{att}(#1)}
\newcommand{\card}[1]{\sharp #1}

\newcommand{\pr}{\mathsf{Pr}}
\newcommand{\prsp}{\mathsf{PS}}

\newcommand{\sign}[1]{\mathsf{sign}(#1)}
\newcommand{\litval}[2]{\mathsf{lval}_{#2}(#1)}
\newcommand{\angletup}[1]{\langle #1 \rangle}


\def\qed{\hfill{\qedboxempty}      
  \ifdim\lastskip<\medskipamount \removelastskip\penalty55\medskip\fi}

\def\qedboxempty{\vbox{\hrule\hbox{\vrule\kern3pt
                 \vbox{\kern3pt\kern3pt}\kern3pt\vrule}\hrule}}

\def\qedfull{\hfill{\qedboxfull}   
  \ifdim\lastskip<\medskipamount \removelastskip\penalty55\medskip\fi}

\def\qedboxfull{\vrule height 4pt width 4pt depth 0pt}

\newcommand{\markfull}{\qedboxfull}
\newcommand{\markempty}{\qed}

\def\lhs{\mathsf{lhs}}
\def\rhs{\mathsf{rhs}}
\def\prim{\mathsf{prim}}
\def\Eval{\mathsf{Eval}}
\def\IsSafe{\mathsf{IsSafe}}
\def\true{\mathsf{true}}
\def\false{\mathsf{false}}
\def\set#1{\mathord{\{#1\}}}
\def\lrhd{\triangleright}
\newcommand*{\translrhd}{\ensuremath{\mathop{\overset{+}{\lrhd}}}}
\def\e#1{\emph{#1}}

\def\dl{\mathrel{{:}{\text{-}}}}

\def\rev#1{{\color{black}#1}}

\maketitle

\begin{abstract}
Database integrity constraints allow us to specify semantic properties that should be satisfied by all databases of a certain relational schema. However, real-life databases are often inconsistent, i.e., do not conform to their specifications in the form of integrity constraints. This leads to the problem of querying inconsistent databases, which aims at a mechanism that allows us to compute trustable answers to user queries, despite the inconsistency of the underlying database.
Consistent query answering (CQA) is such a mechanism that goes back in the late 1990s. Its key elements are the notion of (database) repair, that is, a consistent database whose difference with the original inconsistent database is minimal according to some measure, and the notion of certain answers, that is, answers that are certainly true in all repairs.
Although the notion of certain answers is conceptually meaningful, it is also very rigid as it may completely discard a candidate answer only because it is not entailed by few repairs. A more refined measure of certainty is the so-called relative frequency of a candidate answer, that is, the percentage of repairs that entail it. Of course, to compute the relative frequency of a candidate answer, we need a way to compute the number of repairs that entail it, and the total number of repairs. This brings us to the central problem of counting database repairs (with or without a query in place), which is the main theme of the present work. A key task in this context is to establish a complete complexity classification (a.k.a. dichotomy). 
When there is no query in place, this has been recently achieved for the central class of functional dependencies (FDs) via an \text{\rm FP}/$\sharp$\text{\rm P}-complete dichotomy. With a query in place, an analogous classification exists for primary keys and self-join-free conjunctive queries (CQs), but nothing is known once we go beyond primary keys. We lift the latter result to arbitrary FDs and self-join-free CQs.
Another important task in this context is whenever the counting problem in question is intractable, to classify it as approximable, i.e., the target value can be efficiently approximated with error guarantees via a fully polynomial-time randomized approximation scheme (FPRAS), or as inapproximable. Although for primary keys and CQs (even with self-joins) the problem is always approximable, we prove that this is not the case for FDs (with or without a CQ in place).
We show, however, that the class of FDs with a so-called left-hand side chain forms an island of approximability. These results, apart from being interesting in their own right, are crucial steps towards a complete classification of approximate counting of repairs in the case of FDs.
\end{abstract}

\section{Introduction}\label{sec:introduction}

In a relational database system, it is possible to specify semantic properties that should be satisfied by all databases of a certain relational schema such as ``every social security number is associated with at most one person''. Such properties are crucial in the development of transparent and usable database schemas for complex applications, as well as for optimizing the evaluation of queries~\cite{AbHV95}.
Database semantic properties are typically specified via integrity constraints. For example, the above property about persons is a simple key constraint, which essentially states that in the relation that stores persons, the attribute ``social security number'' functionally determines all the other attributes of that relation.
Unfortunately, real-life databases are often {\em inconsistent}, i.e., do not conform to their specifications in the form of integrity constraints. The reason behind this is that data is not perfect and clean; it may come, for example, from several conflicting sources~\cite{Bertossi11}.
Data cleaning attempts to fix this problem by resolving the inconsistencies~\cite{FaGe12}. However, this may lead to the loss of valuable information. Furthermore, there are  scenarios like virtual data integration, where the data stays at the autonomous data sources, in which there is no way to modify the data without having ownership of the sources. This leads to the problem of querying inconsistent databases. In other words, we need a mechanism that allows us to compute meaningful and trustable answers to user queries, despite the fact that the underlying database is inconsistent.

\subsection{Consistent Query Answering and Database Repairs}

Consistent query answering (CQA) is an elegant framework introduced in the late 1990s by Arenas, Bertossi, and Chomicki~\cite{ArBC99} towards a mechanism for querying inconsistent databases that allows us to deliver meaningful answers to queries that can still be obtained from inconsistent data. The key elements underlying the CQA approach are:
\begin{enumerate} 
	\item the notion of {\em (database) repair} of an inconsistent database $D$, that is, a consistent database whose difference with $D$ is somehow minimal, and
	\item the notion of query answering based on {\em certain answers}, that is, answers that are entailed by every database repair. 
\end{enumerate}
A simple example, taken from~\cite{CaCP19}, that illustrates the above notions follows:

\begin{example}\label{exa:cqa}
	Consider the relational schema consisting of a single relation name 
	\[
	\rel{Employee}(\attr{id}, \attr{name}, \attr{dept})
	\]
	that comes with the constraint that the attribute $\attr{id}$ functionally determines $\attr{name}$ and $\attr{dept}$. This means that the attribute $\attr{id}$ is the key of $\rel{Employee}$. Consider also the simple database $D$
	\begin{figure}[h]
		\centering
		\begin{tabular}{c|c|c}
			\multicolumn{3}{c}{$\rel{Employee}$}\\\hline\hline
			$\attr{id}$ & $\attr{name}$ & $\attr{dept}$\\\hline
			$1$ & $\val{Bob}$ & $\val{HR}$\\
			$1$ & $\val{Bob}$ & $\val{IT}$\\
			$2$ & $\val{Alice}$ & $\val{IT}$\\
			$2$ & $\val{Tim}$ & $\val{IT}$\\
			\hline
		\end{tabular}
	\end{figure}

	\noindent It is easy to see that $D$ is inconsistent since we are uncertain about Bob's department, and the name of the employee with id $2$. To devise a repair, we need to keep one tuple from each conflicting pair, which leads to a maximal subset of $D$ that is consistent. Thus, we get the four repairs depicted below.
		\begin{figure}[h!]
			\centering
			\begin{tabular}{c|c|c}
				\multicolumn{3}{c}{$\rel{Repair~1}$}\\\hline\hline
				$\attr{id}$ & $\attr{name}$ & $\attr{dept}$\\\hline
				$1$ & $\val{Bob}$ & $\val{HR}$\\
				$2$ & $\val{Alice}$ & $\val{IT}$\\
	\hline
\end{tabular}
			\hspace{0.5cm}
			\begin{tabular}{c|c|c}
	\multicolumn{3}{c}{$\rel{Repair~2}$}\\\hline\hline
	$\attr{id}$ & $\attr{name}$ & $\attr{dept}$\\\hline
	$1$ & $\val{Bob}$ & $\val{HR}$\\
	$2$ & $\val{Tim}$ & $\val{IT}$\\
	\hline
\end{tabular}
			\hspace{0.5cm}
\begin{tabular}{c|c|c}
	\multicolumn{3}{c}{$\rel{Repair~3}$}\\\hline\hline
	$\attr{id}$ & $\attr{name}$ & $\attr{dept}$\\\hline
	$1$ & $\val{Bob}$ & $\val{IT}$\\
	$2$ & $\val{Alice}$ & $\val{IT}$\\
	\hline
\end{tabular}
			\hspace{0.5cm}
\begin{tabular}{c|c|c}
	\multicolumn{3}{c}{$\rel{Repair~4}$}\\\hline\hline
	$\attr{id}$ & $\attr{name}$ & $\attr{dept}$\\\hline
	$1$ & $\val{Bob}$ & $\val{IT}$\\
	$2$ & $\val{Tim}$ & $\val{IT}$\\
	\hline
\end{tabular}
		\end{figure}
	\noindent Observe now that the (Boolean) query that asks whether employees $1$ and $2$ work in the same department is true only in two out of four repairs, that is, $\rel{Repair~3}$ and $\rel{Repair~4}$. Therefore, since the query is false in at least one repair, its final answer over the inconsistent database $D$ is false. This is a conceptually meaningful answer as we are not entirely certain that employees $1$ and $2$ work in the same department, and thus, it will be misleading to evaluate the query to true. \hfill\markfull
\end{example}

Example~\ref{exa:cqa}, despite its simplicity, illustrates one of the limitations of the CQA approach. The notion of certain answers only says that a candidate answer (i.e. a tuple of database values) is entailed by all repairs, or is not entailed by some repair. But, as discussed in~\cite{CaLP18}, the former is too strict, while the latter is not very useful in a practical context. Instead, we would like to know how often a tuple is an answer, that is, its {\em relative frequency}, or, in other words, the percentage of repairs that entail that tuple. For instance, in Example~\ref{exa:cqa}, the relative frequency of the empty tuple, which corresponds to true and is the only candidate answer as the query is Boolean, is $\frac{1}{2}$ since, out
of four repairs in total, only two of those entail the query. Of course, to compute the relative frequency of a tuple, we need a way to compute (i) the number of repairs that entail a tuple (the numerator), and (ii)
the total number of repairs (the denominator). 

\rev{Note that another fundamental task where the number of repairs plays a crucial role is that of measuring the inconsistency of a database, a topic that has received significant attention over the years; see, e.g.,~\cite{DBLP:conf/lpnmr/Bertossi19,
		DBLP:journals/jiis/GrantH06,
		DBLP:conf/ecsqaru/GrantH11,
		DBLP:journals/ijar/GrantH17,
		DBLP:journals/ai/HunterK10,
		DBLP:conf/ijcai/KoniecznyLM03,
		DBLP:journals/lmcs/LivshitsK22,
		DBLP:journals/ai/ParisiG23,
		DBLP:journals/ki/Thimm17}.
Inconsistency measures serve as tools for assessing the reliability of a database and the amount of effort required for cleaning it. One particularly well-explored metric is based on the number repairs~\cite{DBLP:conf/ecsqaru/GrantH11,
	DBLP:journals/ijar/GrantH17,
	DBLP:journals/lmcs/LivshitsK22,
	DBLP:journals/ai/ParisiG23,
	DBLP:journals/ki/Thimm17}. In this context, the task of counting repairs is that of estimating the extent of inconsistency in the database.}

\subsection{Counting Database Repairs}

The problem of counting database repairs is indeed a central one in the context of querying inconsistent data, and has been extensively studied since the advent of the CQA approach. The two key counting problems of main concern are the following, with $\dep$ being a set of integrity constraints and $Q$ a query (note that, for the sake of clarity, we base our discussion on Boolean queries, but it can be naturally extended to non-Boolean queries):

\begin{center}
	\fbox{\begin{tabular}{c}
			$\sharp \prob{Repairs}(\dep)$: For a database $D$, compute the number of repairs of $D$ w.r.t.~$\dep$
			\\[0.3cm]
			$\sharp \prob{Repairs}(\dep,Q)$: For a database $D$, compute the number of repairs of $D$ w.r.t.~$\dep$ that entail $Q$
	\end{tabular}}
\end{center}

\noindent A significant amount of research work was dedicated to the task of pinpointing the complexity of the above problems, which is also the main concern of this work. It has been observed that both problems can be tractable or intractable depending on the syntactic shape of $\dep$ and $Q$. This leads to the natural question whether we can establish a complete complexity classification, i.e., for every $\dep$ and $Q$, classify $\sharp \prob{Repairs}(\dep)$ and $\sharp \prob{Repairs}(\dep,Q)$ as tractable, i.e., in \text{\rm FP} (the counting analogue of \textsc{PTime}), or intractable, i.e., $\sharp \text{\rm P}$-complete (with $\sharp \text{\rm P}$ being the counting analogue of \text{\rm NP}), by simply inspecting $\dep$ and $Q$.
Such a classification, apart from being very interesting from a theoretical point of view, can be also very useful in a more practical context as it can tell us precisely when the set of constraints and query at hand render the problem of interest tractable, which in turn guides the choice of the counting algorithm to be applied.

It is well-known that establishing such complexity classifications (a.k.a.~dichotomies) is a highly non-trivial task. Nevertheless, despite the underlying technical difficulties, such classifications have been established for important classes of integrity constraints. In particular, there are interesting classifications for the central class of {\em functional dependencies} (FDs), which generalize key constraints and allow us to express useful properties of the form ``some attributes functionally determine some other attributes''; the standard syntax of an FD is $R : X \ra Y$, where $R$ is a relation name and $X,Y$ are sets of attributes of $R$, stating that the attributes in $X$ functionally determine the attributes in $Y$. Here is a comprehensive summary of what we currently know:

\begin{itemize}
\item The work~\cite{LiKW21} has recently established a complexity classification for the problem $\sharp \prob{Repairs}(\dep)$ assuming that $\dep$ is a set of FDs. More precisely, it has been shown that whenever $\dep$ has a so-called left-hand side (LHS, for short) chain (up to equivalence), $\sharp \prob{Repairs}(\dep)$ is in \text{\rm FP}; otherwise, it is $\sharp\text{\rm P}$-complete. We also know that checking whether $\dep$ has an LHS chain (up to equivalence) is feasible in polynomial time. Let us recall that a set $\dep$ of FDs has an LHS chain if, for every two FDs $R : X_1 \ra Y_1$ and $R : X_2 \ra Y_2$ of $\dep$, $X_1 \subseteq X_2$ or $X_2 \subseteq X_1$.

\item When it comes to $\sharp \prob{Repairs}(\dep,Q)$, such a classification has been established in~\cite{MaWi13} assuming that $\dep$ is a set of {\em primary keys}, i.e., at most one key constraint per relation name, and $Q$ is a {\em self-join-free conjunctive query} (SJFCQ), i.e., a CQ that cannot mention a relation name more than once. In particular, it has been shown that whenever $Q$ is a so-called safe query w.r.t.~$\dep$, where safety is a technical notion introduced in~\cite{MaWi13} that can be checked in polynomial time, $\sharp \prob{Repairs}(\dep,Q)$ is in \text{\rm FP}; otherwise, it is $\sharp\text{\rm P}$-complete. 
An analogous classification for arbitrary CQs with self-joins was established by the same authors in~\cite{MaWi14} under the assumption that the primary keys are simple, i.e., they consist of a single attribute.
\end{itemize}

\subsection{Main Research Questions}

From the above summary, it is clear that for the central class of FDs, which is the main concern of this work, the problem $\sharp \prob{Repairs}(\dep)$ is well-understood.
However, for the problem $\sharp \prob{Repairs}(\dep,Q)$, once we go beyond primary keys we know very little concerning the existence of a complete complexity classification as the ones described above. In particular, the dichotomy result of~\cite{MaWi13} does not apply when we consider arbitrary FDs. This brings us to the following question:

\medskip

\noindent {\em \textbf{Research Question 1:} Can we lift the dichotomy result for primary keys and SJFCQs to the more general case of functional dependencies?}

\medskip

Let us stress that the above question is deliberately stated for self-join-free CQs and not for arbitrary CQs. The reason is because providing a complete classification for arbitrary CQs with self-joins is an overly ambitious task that is currently out of reach. Indeed, it is well-known that self-joins in CQs cause several complications, and this is why the question whether such a classification exists even for primary keys and CQs with self-joins remains a challenging open problem ten years after the classification result for primary keys and CQs without self-joins was established.

Another key task is to classify the problems $\sharp \prob{Repairs}(\dep)$ and $\sharp \prob{Repairs}(\dep,Q)$, for a set $\dep$ of FDs and a CQ $Q$ (with or without self-joins), as approximable, that is, the target value can be efficiently approximated with error guarantees via a fully polynomial-time randomized approximation scheme (FPRAS), or as inapproximable. This leads to the following question:

\medskip

\noindent {\em \textbf{Research Question 2:} For a set $\dep$ of FDs and a CQ $Q$ (with or without self-joins), can we determine whether $\sharp \prob{Repairs}(\dep)$ and $\sharp \prob{Repairs}(\dep,Q)$ admit an FPRAS by inspecting $\dep$ and $Q$?
}

\medskip

Note that whenever $\sharp \prob{Repairs}(\dep)$ or $\sharp \prob{Repairs}(\dep,Q)$ is tractable, then it is trivially approximable. Therefore, the above question for $\sharp \prob{Repairs}(\dep)$ trivializes whenever $\dep$ has an LHS chain (up to equivalence). However, for FDs without an LHS chain (up to equivalence) this is not true since there exists a set $\dep$ such that $\sharp \prob{Repairs}(\dep)$ is not approximable; the latter is actually a result of the present work.
Concerning $\sharp \prob{Repairs}(\dep,Q)$, we only know that if $\dep$ consists of primary keys and $Q$ is a CQ (even with self-joins), $\sharp \prob{Repairs}(\dep,Q)$ is always approximable; this is implicit in~\cite{DaSu07}, and it has been made explicit in~\cite{CaCP19}.
Once we go beyond primary keys, the above question for $\sharp \prob{Repairs}(\dep,Q)$ is a non-trivial one since, depending on the syntactic shape of $\dep$ and $Q$, $\sharp \prob{Repairs}(\dep,Q)$ can be approximable or not; the latter is also a result of the present work.

\subsection{Summary of Contributions}

Our goal is to provide answers to the questions discussed above. For Question~1, we provide a definitive answer. For Question~2, we establish results that form crucial steps towards a definitive answer. As we explain, an answer to Question~2 will resolve a challenging graph-theoretic open problem. Our contributions can be summarized as follows:

\begin{itemize}
\item Concerning Question 1, we lift the dichotomy of~\cite{MaWi13} for primary keys and SJFCQs to the general case of FDs (Theorem~\ref{the:fds-dichotomy}). To this end, we build on the dichotomy for the problem of counting repairs (without a query) from~\cite{LiKW21}, which allows us to concentrate on FDs with an LHS chain (up to equivalence) since for all the other cases we can inherit the $\sharp \text{P}$-hardness.
Therefore, our main technical task was actually to lift the dichotomy for primary keys and SJFCQs from~\cite{MaWi13} to the case of FDs with an LHS chain (up to equivalence). Although the proof of this result borrows several ideas from the proof of~\cite{MaWi13}, the task of lifting the result to FDs with an LHS chain (up to equivalence) is a non-trivial one. This is due to the significantly more complex structure of database repairs under FDs with an LHS chain compared to those under primary keys. \rev{The tractable side of the dichotomy is proved by introducing a polynomial-time algorithm. For the intractable side, we first introduce certain rewrite rules for pairs consisting of a set $\dep$ of FDs and an SJFCQ $Q$. A crucial property of those rules is that whenever $(\dep,Q)$ can be rewritten into $(\dep',Q')$, there is a polynomial-time Turing reduction from $\sharp \prob{Repairs}(\dep',Q')$ to $\sharp \prob{Repairs}(\dep,Q)$. We then show that every pair $(\dep,Q)$ can be rewritten, by following a sequence of rewrite rules, into a pair $(\dep',Q')$ that we call final. Intuitively, a final pair cannot be rewritten further without breaking its non-safety. For final pairs, we establish $\sharp\text{\rm P}$-hardness via a reduction from the primary key case, and the $\sharp\text{\rm P}$-hardness for all the non-final pairs is obtained from the sequence of rewritings; further details are discussed in Section~\ref{sec:lhs-chain-fds}.}

\item Concerning Question 2, although we do not establish a complete classification, we provide results that, apart from being interesting in their own right, are crucial steps towards a complete classification. We first reveal the difficulty underlying a proper approximability/inapproximability 
dichotomy by discussing that such a result will resolve the challenging open problem of whether counting maximal matchings in a bipartite graph is approximable. We then show that, for every set $\dep$ of FDs with an LHS chain (up to equivalence) and a CQ $Q$ (even with self-joins), $\sharp \prob{Repairs}(\dep,Q)$ admits an FPRAS (item (1) of Theorem~\ref{the:apx-main-result}). 
On the other hand, we show that there exists a very simple set $\dep$ of FDs such that $\sharp \prob{Repairs}(\dep)$ does not admit an FPRAS (item (2) of Theorem~\ref{the:apx-main-result}), which in turn allows us to show that, for every SJFCQ $Q$, $\sharp \prob{Repairs}(\dep,Q)$ does not admit an FPRAS (item (3) of Theorem~\ref{the:apx-main-result}). As usual, these inapproximability results hold under a standard complexity-theoretic assumption. The inapproximability of $\sharp \prob{Repairs}(\dep)$  exploits the interesting technique of gap amplification, used to prove inapproximability of optimization problems~\cite{Has01}.
\end{itemize}

\medskip

\noindent
\paragraph{Roadmap.} In Section~\ref{sec:preliminaries}, we recall the basics on relational databases, FDs and CQs. The problems of interest are formally defined in Section~\ref{sec:problem-definition}. In Section~\ref{sec:lhs-chain-fds}, we introduce the class of FDs with a left-hand side chain that plays a crucial role in our analysis. A definitive answer to Question 1 stated above is provided in Section~\ref{sec:exact-counting}, whereas a partial answer to Question 2 is given in Section~\ref{sec:apx-counting}. We conclude in Section~\ref{sec:future-work} by briefly discussing our results and giving directions for future research.
For the sake of readability, several technical details and proofs are given in a clearly marked appendix.
\section{Preliminaries}\label{sec:preliminaries}

\begin{figure}[t]
\centering
\begin{tabular}{c||c|c|c|c|c}
\multicolumn{6}{c}{$\rel{Schedule}$}\\\hline\hline
 & $\attr{train}$ & $\attr{origin}$ & $\attr{destination}$ & $\attr{time}$ & $\attr{duration}$\\\hline
$f_1$ & $\val{16}$ & $\val{NYP}$ & $\val{BBY}$ & $\val{1030}$ & $\val{315}$\\
$f_2$ & $\val{16}$ & $\val{NYP}$ & $\val{PVD}$ & $\val{1030}$ & $\val{250}$\\
$f_3$ & $\val{16}$ & $\val{PHL}$ & $\val{WIL}$ & $\val{1030}$ & $\val{20}$\\
$f_4$ & $\val{16}$ & $\val{PHL}$ & $\val{BAL}$ & $\val{1030}$ & $\val{70}$\\
$f_5$ & $\val{16}$ & $\val{PHL}$ & $\val{WAS}$ & $\val{1030}$ & $\val{120}$\\
$f_6$ & $\val{16}$ & $\val{BBY}$ & $\val{PHL}$ & $\val{1030}$ & $\val{260}$\\
$f_7$ & $\val{16}$ & $\val{BBY}$ & $\val{NYP}$ & $\val{1030}$ & $\val{260}$\\
$f_8$ & $\val{16}$ & $\val{BBY}$ & $\val{WAS}$ & $\val{1030}$ & $\val{420}$\\
$f_9$ & $\val{16}$ & $\val{WAS}$ & $\val{PVD}$ & $\val{1030}$ & $\val{390}$\\\hline
\end{tabular}
\hspace{1cm}
\begin{tabular}{c||c|c}
\multicolumn{3}{c}{$\rel{Station}$}\\\hline\hline
& $\attr{name}$ & $\attr{city}$\\\hline
$g_1$ & $\val{NYP}$ & $\val{New York}$\\
$g_2$ & $\val{PHL}$ & $\val{Philadelphia}$\\
$g_3$ & $\val{BBY}$ & $\val{Boston}$\\
$g_4$ & $\val{WAS}$ & $\val{Washington}$\\
$g_5$ & $\val{PVD}$ & $\val{Providence}$\\
$g_6$ & $\val{WIL}$ & $\val{Wilmington}$\\
$g_7$ & $\val{BAL}$ & $\val{Baltimore}$\\
\hline
\end{tabular}
\caption{\label{fig:trains} An inconsistent database of a railroad company.}
\end{figure}

We recall the basics on relational databases, functional dependencies, and conjunctive queries. We consider the disjoint countably infinite sets $\ins{C}$ and $\ins{V}$ of {\em constants} and {\em variables}, respectively. For an integer $n > 0$, we write $[n]$ for the set $\{1,\ldots,n\}$.

\medskip

\noindent\paragraph{Relational Databases.}
A {\em (relational) schema} $\ins{S}$ is a finite set of relation names with associated arity; we write $R/n$ to denote that $R$ has arity $n > 0$. Each relation name $R/n$ is associated with a tuple of distinct attribute names $(A_1,\ldots,A_n)$; we write $\att{R}$ for the set of attributes $\{A_1,\ldots,A_n\}$. A {\em position} of $\ins{S}$ is a pair $(R,A)$, where $R \in \ins{S}$ and $A \in \att{R}$, that essentially identifies the attribute $A$ of the relation name $R$.
A {\em fact} over $\ins{S}$ is an expression of the form $R(c_1,\ldots,c_n)$, where $R/n \in \ins{S}$, and $c_i \in \ins{C}$ for each $i \in [n]$; we may say $R$-fact to indicate that the relation name is $R$. A {\em database} $D$ over $\ins{S}$ is a finite set of facts over $\ins{S}$. We write $D_{R}$, for $R \in \ins{S}$, for the database $\{f \in D \mid f \text{ is an } R\text{-fact}\}$. The {\em active domain} of $D$, denoted $\adom{D}$, is the set of constants occurring in $D$.
For a fact $f = R(c_1,\ldots,c_n)$, with $(A_1,\ldots,A_n)$ being the tuple of attribute names of $R$, we write $f[A_i]$ for $c_i$, that is, the value of the attribute $A_i$ in the fact $f$.
We say that a constant $c$ occurs in $f$ at position $(R,A)$ if $f[A] = c$.

\begin{example}\label{example:trains}
 Figure~\ref{fig:trains} depicts a
  database, partially taken from~\cite{LiKi21}, over a schema with two relation names: $\rel{Schedule}$ (storing a train schedule) and $\rel{Station}$ (storing the locations of train stations). The relation name $\rel{Schedule}$ is associated with the tuple $(\attr{train},\attr{origin},\attr{destination},\attr{time},\attr{duration})$ of attribute names,
and the relation name $\rel{Station}$ is associated with $(\attr{name},\attr{city})$. The fact
\[
  f_1\ =\ \rel{Schedule}(\val{16}, \val{NYP}, \val{BBY}, \val{1030}, \val{315})
\]
  states that train number 16 is scheduled to depart from the New York Penn (NYP) Station at 10:30, and arrive at the Boston Back Bay (BBY) Station 315 minutes later. Hence, for example, we have that $f_1[\attr{origin}]=\val{NYP}$. The first and the third fact of $\rel{Station}$ 
  \[
  g_1\ =\ \rel{Station}(\val{NYP}, \val{New York}) \qquad \textrm{and} \qquad   g_3\ =\ \rel{Station}(\val{BBY}, \val{Boston})
  \]
  state that the NYP station is located in New York and the BBY station is located in Boston. \hfill\markfull 
\end{example}

\medskip

\noindent
\paragraph{Functional Dependencies.}
A {\em functional dependency} (FD) $\phi$ over a schema $\ins{S}$ is an expression
\[
R\ :\ X \ra Y
\]
where $R/n \in \ins{S}$ and $X,Y \subseteq \att{R}$; \rev{we may call $\phi$ an $R$-FD to indicate that it is defined over the relation name $R$}. We call $\phi$ a {\em key} if $X \cup Y = \att{R}$. When $X$ or $Y$ are singletons, we may avoid the curly brackets, and simply write the attribute name.
Given a set $\dep$ of FDs over $\ins{S}$, we write $\dep_R$, for $R \in \ins{S}$, for the set $\{\phi \in \dep \mid \phi \text{ is an } R\text{-FD}\}$. 
We call $\dep$ a set of {\em primary keys} if it consists only of keys, and $|\dep_R| \leq 1$ for each $R \in \ins{S}$.
A database $D$ over $\ins{S}$ satisfies an FD $\phi$ over $\ins{S}$, denoted $D \models \phi$, if, for every two $R$-facts $f,g \in D$ the following holds: $f[A]=g[A]$ for every $A \in X$ implies $f[B]=g[B]$ for every $B \in Y$. In simple words, $D$ satisfies $\phi$ if every two $R$-facts $f,g$ of $D$ either agree on the values of all the attributes of $X \cup Y$, or disagree on the value of at least one attribute of $X$.
We say that $D$ is {\em consistent} w.r.t.~a set $\dep$ of FDs, written $D \models \dep$, if $D \models \phi$ for every $\phi \in \dep$; otherwise, $D$ is {\em inconsistent} w.r.t.~$\dep$.
Two sets of FDs $\dep$ and $\dep'$ are {\em equivalent} if, for every database $D$, $D \models \dep$ iff $D \models \dep'$.

\begin{example}\label{example:train_fds}
 Consider again the database depicted in Figure~\ref{fig:trains} and the FDs  
 \[
  \rel{Schedule}\ :\ \{\attr{train}, \attr{time}\} \rightarrow \attr{origin} \quad\quad\quad
  \rel{Schedule}\ :\ \{\attr{train}, \attr{time}, \attr{duration}\} \rightarrow \attr{destination}
\]
  stating that the train number and departure time determine the origin station, and the train number, departure time, and duration determine the destination station. We also consider the FD
  \[
  \rel{Station}\ :\ \attr{name}\rightarrow \attr{city}
  \]
  which is actually a single (primary) key over $\rel{Station}$. Note that there are no violations of this key in the database; however, the first and third fact of  $\rel{Schedule}$, for example, jointly violate the first FD given above. Therefore, the database is inconsistent. \hfill\markfull
\end{example}

\medskip

\noindent
\paragraph{Conjunctive Queries.}
A {\em (relational) atom} $\alpha$ over a schema $\ins{S}$ is an expression $R(t_1,\ldots,t_n)$, where $R/n \in \ins{S}$, and $t_i \in \ins{C} \cup \ins{V}$ for each $i \in [n]$; we may say $R$-atom to indicate that the relation name is $R$. Note that a fact is an atom without variables. As for facts, we may say that \rev{a term $t\in \ins{C} \cup \ins{V}$} occurs in $\alpha$ at a certain position $(R,A)$ with $A \in \att{R}$, and we may write $\alpha[A]$ for the term $t$ occurring in $\alpha$ at position $(R,A)$.
A {\em conjunctive query} (CQ) $Q$ over $\ins{S}$ is an expression
\[
\textrm{Ans}(\bar x)\ \dl\ R_1(\bar y_1), \ldots, R_n(\bar y_n)
\]
where each $R_i(\bar y_i)$, for $ i \in [n]$, is an atom over $\ins{S}$, $\bar x$ are the {\em answer variables} of $Q$, and each variable in $\bar x$ is mentioned in $\bar y_i$ for some $i \in [n]$. We may write $Q(\bar x)$ to indicate that $\bar x$ are the answer variables of $Q$. When $\bar x$ is empty, $Q$ is called {\em Boolean}. 
A subclass of CQs that will play a crucial role in our analysis is that of self-join-free CQs. Formally, a {\em self-join-free} CQ (SJFCQ) over a schema $\ins{S}$ is a CQ over $\ins{S}$ that mentions every relation name of $\ins{S}$ at most once. We may treat a CQ as the set of atoms in the right-hand side of the $\dl$ symbol. Let $\var{Q}$ and $\const{Q}$ be the set of variables and constants in $Q$, respectively. The semantics of CQs is given via homomorphisms.
A {\em homomorphism} from a CQ $Q$ as the one above to a database $D$ is a function $h : \var{Q} \cup \const{Q} \ra \adom{D}$, which is the identity over $\ins{C}$, such that $R_i(h(\bar y_i)) \in D$ for $i \in [n]$.
A tuple $\bar t \in \adom{D}^{|\bar x|}$ is an {\em answer to $Q$ over $D$} if there is a homomorphism $h$ from $Q$ to $D$ with $h(\bar x) = \bar t$. Let $Q(D)$ be the answers to $Q$ over $D$. For Boolean CQs, we write $D \models Q$, and say that $D$ {\em entails} $Q$, if $() \in Q(D)$.

\begin{example}\label{example:train_queries}
 We consider the following CQs over the schema of the database of Figure~\ref{fig:trains}:
 \[
 Q_1\ =\ \textrm{Ans}(x,y)\ \dl\ \rel{Schedule}(x,\val{BBY},z,y,w),\rel{Station}(z,\val{Washington})
 \] 
 that asks for all the train numbers and departure times of trains that depart from the BBY station and arrive to Washington, and
 \[
 Q_2\ =\ \textrm{Ans}(x,y,w)\ \dl\ \rel{Schedule}(x,\val{BBY},z,y,w),\rel{Station}(z,\val{Washington})
 \]
 that asks for all the train numbers, departure times, and ride durations of such trains. Note that both queries are self-join free. 
 Note that the tuple 
 $(\val{16},\val{1030})$ is an answer to $Q_1$ over the database, as witnessed by the homomorphism $h$ with $h(x)=\val{16}$, $h(y)=\val{1030}$, $h(z)=\val{WAS}$, and $h(w)=\val{420}$. The same homomorphism, \rev{this time with respect to the variables of $Q_2$}, witnesses that $(\val{16},\val{1030},\val{420})$ is an answer to $Q_2$ over the database. \hfill\markfull
\end{example}

\medskip
\noindent
\paragraph{Database Repairs.}
\rev{Informally, a repair of an inconsistent database $D$ w.r.t.~a set of integrity constraints is a consistent database $D'$ that minimally deviates from $D$. Here,
 minimality refers to the symmetric difference between $D$ and $D'$. In the
case of anti-monotonic constraints (e.g.,~FDs), where consistency
cannot be violated by removing facts, a repair is a maximal consistent subset (i.e.,~subset repair~\cite{DBLP:conf/icdt/AfratiK09}).}
Formally, given a database $D$ and a set $\dep$ of FDs, both over a schema $\ins{S}$, a {\em repair} of $D$ w.r.t.~$\dep$ is a database $D' \subseteq D$ such that
\begin{enumerate}
	\item $D' \models \dep$, and
	\item for every $D'' \supsetneq D'$, $D'' \not\models \dep$.
\end{enumerate}
In simple words, a repair of $D$ w.r.t.~$\dep$ is a maximal subset of $D$ that is consistent w.r.t.~$\dep$. Note that if $D$ is consistent w.r.t.~$\dep$, then there is only one repair, that is, $D$ itself. Let $\rep{D}{\dep}$ be the set of repairs of $D$ w.r.t~$\dep$.
Given a CQ $Q(\bar x)$ over $\ins{S}$, and a tuple $\bar t \in \adom{D}^{|\bar x|}$, we write $\rep{D}{\dep,Q,\bar t}$ for $\{D' \in \rep{D}{\dep} \mid \bar t \in Q(D')\}$, that is, the set of repairs $D'$ of $D$ w.r.t.~$\dep$ such that $\bar t$ is an answer to $Q$ over $D'$. 
If $Q$ is Boolean, and thus, the only possible answer is the empty tuple $()$, we write $\rep{D}{\dep,Q}$ instead of $\rep{D}{\dep,Q,()}$.
For brevity, we write $\card{\rep{D}{\dep}}$ and $\card{\rep{D}{\dep,Q,\bar t}}$ for the cardinality of $\rep{D}{\dep}$ and $\rep{D}{\dep,Q,\bar t}$, respectively.

\begin{example}\label{example:train_num_repairs}
 It is easy to verify that the database $D$ of Figure~\ref{fig:trains} has five repairs w.r.t.~the set $\dep$ consisting of the FDs given in  Example~\ref{example:train_fds}: 
 \begin{enumerate}
     \item $\set{f_1,f_2,g_1,g_2,g_3,g_4,g_5,g_6,g_7}$,
     \item $\set{f_3,f_4,f_5,g_1,g_2,g_3,g_4,g_5,g_6,g_7}$,
     \item $\set{f_6,f_8,g_1,g_2,g_3,g_4,g_5,g_6,g_7}$,
     \item $\set{f_7,f_8,g_1,g_2,g_3,g_4,g_5,g_6,g_7}$,
     \item $\set{f_9,g_1,g_2,g_3,g_4,g_5,g_6,g_7}$.
 \end{enumerate}

\smallskip

\noindent Therefore, it holds that $\card{\rep{D}{\dep}}=5$. Now, consider the query $Q_1$ of Example~\ref{example:train_queries} and the answer $(\val{16},\val{1030})$. Only the third and fourth repairs belong to $\rep{D}{\dep,Q_1,(\val{16},\val{1030})}$, as these are the only repairs that contain the fact $f_8$ that mentions a train departing from the BBY station and arriving at the WAS station in Washington. Hence, we have that $\card{\rep{D}{\dep,Q_1,(\val{16},\val{1030})}}=2$. The same holds for the query $Q_2$ and the answer $(\val{16},\val{1030},\val{420})$; that is, $\card{\rep{D}{\dep,Q_2,(\val{16},\val{1030},\val{420})}}=2$. \hfill\markfull
\end{example}

A useful observation is that $\card{\rep{D}{\dep}} = \card{\rep{D'}{\dep,Q,\bar t}}$, for some easily computable database $D'$ and tuple $\bar t$, providing that $Q$ is self-join-free.

\def\lemmarepairreduction{
	Consider a database $D$, a set $\dep$ of FDs, and an SJFCQ $Q(\bar x)$. We can compute in polynomial time a database $D'$ and a tuple $\bar t \in \adom{D'}^{|\bar x|}$ such that $\card{\rep{D}{\dep}} = \card{\rep{D'}{\dep,Q,\bar t}}$.
}
\begin{lemma}\label{lem:cook-reduction}
\lemmarepairreduction
\end{lemma}

\begin{proof}
	We assume, w.l.o.g., that $\const{Q}\cap\adom{D}=\emptyset$; indeed, if this is not the case, then we can simply rename all the constants in $D$ with fresh ones not occurring in $Q$ without affecting the number of repairs of $D$ w.r.t.~$\dep$.
	Let $D_Q$ be the database obtained by replacing each variable $x$ in $Q$ with a fresh constant $c_x \in \ins{C}$ not occurring in $\adom{D}$. Since $\adom{D} \cap \adom{D_Q} = \emptyset$, and $D_Q$ is consistent w.r.t.~$\dep$ since it contains at most one $R$-fact for each relation name $R$ occurring in $Q$, we get that $\card{\rep{D}{\dep}}=\card{\rep{D\cup D_Q}{\dep}}$. Moreover, for every repair $D'$ of $D \cup D_Q$ w.r.t.~$\dep$, $c(\bar x) \in Q(D')$ with $c(\bar x)$ being the tuple of constants obtained by replacing each variable $x$ in $\bar x$ with $c_x$. Therefore, $\card{\rep{D}{\dep}} = \card{\rep{D \cup D_Q}{\dep,Q,c(\bar x)}}$. It is clear that $D_Q$ and $c(\bar x)$ can be constructed in polynomial time, and the claim follows with $D' = D \cup D_Q$ and $\bar t = c(\bar x)$.
\end{proof}
\section{Problems of Interest}\label{sec:problem-definition}

The two counting problems of interest, already discussed in Section~\ref{sec:introduction}, are formally defined as follows. Fix a set $\dep$ of FDs and a CQ $Q(\bar x)$. We consider the problems

\begin{center}
	\fbox{\begin{tabular}{ll}
			{\small PROBLEM} : & $\sharp \prob{Repairs}(\dep)$
			\\
			{\small INPUT} : & A database $D$
			\\
			{\small OUTPUT} : &  $\card{\rep{D}{\dep}}$
	\end{tabular}}
\end{center}

\begin{center}
	\fbox{\begin{tabular}{ll}
			{\small PROBLEM} : & $\sharp \prob{Repairs}(\dep,Q(\bar x))$
			\\
			{\small INPUT} : & A database $D$ and a tuple $\bar t \in \adom{D}^{|\bar x|}$
			\\
			{\small OUTPUT} : &  $\card{\rep{D}{\dep,Q,\bar t}}$
	\end{tabular}}
\end{center}

\medskip
\noindent
\paragraph{Exact Counting of Database Repairs.} Recall that one of our goals, towards a definitive answer to Research Question 1, is to provide a complete complexity classification for $\sharp \prob{Repairs}(\dep,Q(\bar x))$. In other words, for each set $\dep$ of FDs and CQ $Q(\bar x)$, the goal is to classify it as tractable (i.e., place it in FP) or as intractable (i.e., show that it is $\sharp$P-complete). Recall that FP and $\sharp$P are classes of {\em counting functions} $\{0,1\}^* \ra \mathbb{N}$. In particular, FP (the counting analog of \textsc{PTime}) is defined as
\[
\text{\rm FP}\ =\ \{f \mid f \text{ is a polynomial-time computable counting function}\}
\]
whereas $\sharp$P (the counting analog of NP) is defined as
\[
\sharp \text{\rm P}\ =\ \{\mathsf{accept}_M \mid M \text{ is a polynomial-time non-deterministic Turing machine}\}
\]
where $\mathsf{accept}_M$ is the counting function $\{0,1\}^* \ra \mathbb{N}$ that assigns to each $x \in \{0,1\}^*$ the number of accepting computations of $M$ on input $x$.
Typically, hardness results for $\sharp$P rely on polynomial-time Turing reductions (a.k.a.~Cook reductions). In particular, given two counting functions $f,g : \{0,1\}^* \ra \mathbb{N}$, we say that $f$ is {\em Cook reducible} to $g$ if there exists a polynomial-time deterministic transducer $M$, with access to an oracle for $g$, such that, for every $x \in \{0,1\}^*$, $f(x) = M(x)$. 
As discussed in Section~\ref{sec:introduction}, Maslowski and Wijsen~\cite{MaWi13} established the following classification:

\begin{theorem}[\cite{MaWi13}]\label{the:pk-dichotomy}
	For a set $\dep$ of primary keys and an SJFCQ $Q$, the following statements hold: 
	\begin{enumerate}
		\item $\sharp \prob{Repairs}(\dep,Q)$ is either in \text{\rm FP} or $\sharp$\text{\rm P}-complete. 
		\item We can decide in polynomial time in $||\dep|| + ||Q||$ whether $\sharp \prob{Repairs}(\dep,Q)$ is in \text{\rm FP} or $\sharp$\text{\rm P}-complete.\footnote{As usual, $||o||$  denotes the size of the encoding of a syntactic object $o$.}
	\end{enumerate}
\end{theorem}

\noindent As we shall see in Section~\ref{sec:exact-counting}, one of the main results of the present work is a generalization of the above result to arbitrary functional dependencies.

\medskip
\noindent
\paragraph{Approximate Counting of Database Repairs.}  Another key task is whenever $\sharp \prob{Repairs}(\dep)$ or $\sharp \prob{Repairs}(\dep,Q)$ is $\sharp$\text{\rm P}-complete to classify it as approximable, i.e., the target value can be efficiently approximated with error guarantees via a {\em fully polynomial-time randomized approximation scheme} (FPRAS, for short), or as inapproximable.
Formally, an FPRAS for $\sharp \prob{Repairs}(\dep)$ is a randomized algorithm $\mathsf{A}$ that takes as input a database $D$, $\epsilon > 0$, and $0 < \delta < 1$, runs in polynomial time in $||D||$, $1/\epsilon$ and $\log(1/\delta)$, and produces a random variable $\mathsf{A}(D,\epsilon,\delta)$ such that
\[
\pr\left(|\mathsf{A}(D,\epsilon,\delta) - \card{\rep{D}{\dep}}|\ \leq\ \epsilon \cdot \card{\rep{D}{\dep}}\right)\ \geq\ 1-\delta.
\]
Analogously, an FPRAS for $\sharp \prob{Repairs}(\dep,Q(\bar x))$ is a randomized algorithm $\mathsf{A}$ that takes as input a database $D$, a tuple $\bar t \in \adom{D}^{|\bar x|}$, $\epsilon > 0$, and $0 < \delta < 1$, runs in polynomial time in $||D||$, $||\bar t||$, $1/\epsilon$ and $\log(1/\delta)$, and produces a random variable $\mathsf{A}(D,\bar t,\epsilon,\delta)$ such that
\[
\pr\left(|\mathsf{A}(D,\bar t,\epsilon,\delta) - \card{\rep{D}{\dep,Q,\bar t}}|\ \leq\ \epsilon \cdot \card{\rep{D}{\dep,Q,\bar t}}\right)\ \geq\ 1-\delta.
\]
The only known non-trivial result concerning approximate counting is that $\sharp \prob{Repairs}(\dep,Q(\bar x))$ always admits an FPRAS whenever $\dep$ is a set of primary keys; this is implicit in~\cite{DaSu07}, and it has been made explicit in~\cite{CaCP19}.
The picture concerning approximate counting for arbitrary FDs is rather unexplored. We make crucial steps of independent interest towards a complete classification of approximate counting for FDs.
We show that (under a reasonable complexity-theoretic assumption) the existence of an FPRAS for both problems is not guaranteed whenever we focus on FDs. On the other hand, it is guaranteed for FDs with a so-called left-hand side chain (up to equivalence); the class of FDs with left-hand side chain is discussed in Section~\ref{sec:lhs-chain-fds}.

\medskip
\noindent
\paragraph{Boolean vs. Non-Boolean CQs.}
For technical clarity, both exact and approximate counting are studied for Boolean CQs, but all the results can be easily generalized to non-Boolean CQs. 
Indeed, given a set $\dep$ of FDs and a CQ $Q(\bar x)$, $\card{\rep{D}{\dep,Q,\bar t}}$, for some database $D$ and tuple $\bar t \in \adom{D}^{|\bar x|}$, coincides with $\card{\rep{D}{\dep,Q(\bar t)}}$, where $Q(\bar t)$ is the Boolean CQ obtained after instantiating $\bar x$ with $\bar t$.
Therefore, in the rest of the paper, by CQ or SJFCQ we silently refer to a Boolean query.

\begin{example}\label{example:train_bool_queries}
    Consider again the CQs $Q_1$ and $Q_2$ given in Example~\ref{example:train_queries}. Consider also the tuples $(\val{16},\val{1030})$ and $(\val{16},\val{1030},\val{420})$. As mentioned in Example~\ref{example:train_num_repairs}, we have that $\card{\rep{D}{\dep,Q_1,(\val{16},\val{1030})}}=2$ and $\card{\rep{D}{\dep,Q_2,(\val{16},\val{1030},\val{420})}}=2$ for the database $D$ of Figure~\ref{fig:trains} and the FD set $\dep$ of Example~\ref{example:train_fds}. By instantiating the output variables of $Q_1$ with the values of $(\val{16},\val{1030})$ (i.e., replacing the variable $x$ with the constant $\val{16}$ and the variable $y$ with the constant $\val{1030}$), we obtain the Boolean CQ
    \[
    Q_3\ =\ \textrm{Ans}()\ \dl\ \rel{Schedule}(\val{16},\val{BBY},z,\val{1030},w),\rel{Station}(z,\val{Washington}).
    \]
    It is easy to see that $\card{\rep{D}{\dep,Q_1,(\val{16},\val{1030})}}=\card{\rep{D}{\dep,Q_3}}$. Similarly, from the CQ $Q_2$ and the tuple $(\val{16},\val{1030},\val{420})$ we obtain the Boolean CQ 
    \[
    Q_4\ =\ \textrm{Ans}()\ \dl\ \rel{Schedule}(\val{16},\val{BBY},z,\val{1030},\val{420}),\rel{Station}(z,\val{Washington})
    \]
    by replacing the variable $x$ with the constant $\val{16}$, the variable $y$ with the constant $\val{1030}$, and the variable $w$ with the constant $\val{420}$. We then have that $\card{\rep{D}{\dep,Q_2,(\val{16},\val{1030},\val{420})}}=\card{\rep{D}{\dep,Q_4}}$. Henceforth, in our examples, we will exploit the Boolean CQs $Q_3$ and $Q_4$. \hfill\markfull
\end{example}
\section{LHS Chain Functional Dependencies}\label{sec:lhs-chain-fds}

We recall the class of FDs with a left-hand side chain together with some basic notions that will be useful for obtaining our main results.

\begin{definition}[\textbf{LHS Chain FDs},~\cite{LiKW21}]\label{def:lhs-fds}
	Consider a set $\dep$ of FDs over a schema $\ins{S}$ and a relation name $R \in \ins{S}$. We say that $\dep_R$ has a {\em left-hand side chain} (LHS chain, for short) if, for every two FDs $R : X_1 \ra Y_1$ and $R : X_2 \ra Y_2$ from $\dep_R$, $X_1 \subseteq X_2$ or $X_2 \subseteq X_1$. 
	We further say that $\dep$ has an LHS chain if, for every $R \in \ins{S}$, $\dep_R$ has an LHS chain. \hfill\markfull
\end{definition}

Note that the existence of an LHS chain means that the FDs of $\dep_R$ can be arranged in a sequence $R : X_1 \ra Y_1,\ldots,R : X_n \ra Y_n$ such that $X_1 \subseteq X_2 \subseteq \cdots \subseteq X_n$; it is easy to verify that this sequence is not always unique. We call such a sequence an LHS chain of $\dep_R$.

\begin{example}
The set of FDs consisting of those given in Example~\ref{example:train_fds} has an LHS chain since $\{\attr{train},\attr{time}\}\subseteq\{\attr{train},\attr{time},\attr{duration}\}$ for the FDs over the relation name $\rel{Schedule}$, and there is a single FD over the relation name $\rel{Station}$ (hence, this FD trivially forms an LHS chain). Moreover,
\[
\rel{Schedule}\ :\ \{\attr{train}, \attr{time}\} \rightarrow \attr{origin},\,\,\,
\rel{Schedule}\ :\ \{\attr{train},\attr{time},\attr{duration}\} \rightarrow \attr{destination}
\]
is an LHS chain of $\dep_{\rel{Schedule}}$. If we had, in addition, the FD
\[
\rel{Station}\ :\ \attr{city}\rightarrow \attr{name}
\] 
that essentially states that every city has a single train station, then the FD set $\dep_{\rel{Station}}$ would not have an LHS chain since $\{\attr{name}\}\not\subseteq \{\attr{city}\}$ and $\{\attr{city}\}\not\subseteq \{\attr{name}\}$. \hfill\markfull
\end{example}

\medskip

\noindent
\paragraph{Building Database Repairs.}
When we focus on sets of FDs with an LHS chain, there is a convenient way to build database repairs, which is in turn the building block for showing that repairs can be counted in polynomial time in the size of the database~\cite{LiKW21}.
Given a database $D$, and an FD $\phi$ of the form $R : X \ra Y$, a {\em block} (resp., {\em subblock}) of $D$ w.r.t.~$\phi$ is a maximal subset $D'$ of $D_R$ such that $f,g \in D'$ implies $f[A] = g[A]$ for every $A \in X$ (resp., $A \in X \cup Y$). In simple words, a block (resp., subblock) of $D$ w.r.t.~$\phi$ collects all the $R$-facts of $D$ that agree on the attributes of $X$ (resp., $X \cup Y$). We write $\block{D}{\phi}$ and $\sblock{D}{\phi}$ for all the blocks and subblocks, respectively, of $D$ w.r.t.~$\phi$.
By exploiting the existence of an LHS chain, we can arrange blocks and sublocks in a rooted tree~\cite{LiKi21}, which will then guide the construction of database repairs.

\begin{definition}[\textbf{Blocktree}]\label{def:blocktree}
	Consider a database $D$ over $\ins{S}$, and a set $\dep$ of FDs with an LHS chain over $\ins{S}$, and let $\Lambda = \phi_1,\ldots,\phi_n$ be an LHS chain of $\dep_R$ for some $R \in \ins{S}$. The {\em $(R,\Lambda)$-blocktree} of $D$ w.r.t.~$\dep$ is a labeled rooted tree $T= (V,E,\lambda)$ of height $2n$ (with the root being at level $0$), with $\lambda$ being a function that assigns subsets of $D_R$ to the nodes of $V$, such that the following hold:
	\begin{itemize}
		\item If $v$ is the root node, then $\lambda(v) = D_R$.
		\item If $v_1,\ldots,v_k$, for $k > 0$, are the nodes at level $2i+1$ for some $i \in \{0,\ldots,n-1\}$ with parent $u$, then $|\block{\lambda(u)}{\phi_{i+1}}| = k$ and $\block{\lambda(u)}{\phi_{i+1}} = \{\lambda(v_1),\ldots,\lambda(v_k)\}$.
		\item If $v_1,\ldots,v_k$, for $k > 0$, are the nodes at level $2i+2$ for some $i \in \{0,\ldots,n-1\}$ with parent $u$, then $|\sblock{\lambda(u)}{\phi_{i+1}}| = k$ and $\sblock{\lambda(u)}{\phi_{i+1}} = \{\lambda(v_1),\ldots,\lambda(v_k)\}$. \hfill\markfull
	\end{itemize}
\end{definition}

In simple words, each node of $T$ is associated with a subset of $D_R$ in a way that the root gets the database $D_R$ itself, a node that occurs at an odd level $2i+1$, for $i \in \{0,\ldots,n-1\}$, gets a block of the database assigned to its parent w.r.t.~$\phi_{i+1}$, and a node that occurs at an even level $2i+2$, for $i \in \{0,\ldots,n-1\}$, gets a subblock of the database assigned to its parent w.r.t.~$\phi_{i+1}$.

It is straightforward to verify that, for every node $v$ of an $(R,\Lambda)$-blocktree $T= (V,E,\lambda)$ of a database $D$ w.r.t.~a set $\dep$ of FDs with an LHS chain, and for every two distinct children $u_1$ and $u_2$ of $v$, the following statements hold: (i) if $v$ appears at an even level of $T$ (including level $0$), then $\{f,g\} \models \dep_R$ for every $f \in \lambda(u_1)$ and $g \in \lambda(u_2)$, and (ii) if $v$ appears at an odd level of $T$, then $\{f,g\} \not\models \dep_R$ for every $f \in \lambda(u_1)$ and $g \in \lambda(u_2)$.
This easy observation leads to a simple recursive procedure, which operates on $T$, for constructing repairs of $D_R$ w.r.t.~$\dep_R$: 
\begin{enumerate}
	\item Select the root node.
	\item For every selected node at an even level, select {\em all} of its children.
	\item For every selected node at an odd level, select {\em one} of its children.
	\item Output the union of the labels of the selected leaf nodes.
\end{enumerate}
It should be clear that the obtained database is a repair of $D_R$ w.r.t.~$\dep_R$.
This is essentially the building block underlying the recursive procedure obtained from~\cite{LiKW21} for counting repairs in polynomial time in the size of the database:

\rev{
\begin{proposition}[\cite{LiKW21}]\label{pro:counting-repairs-lhs-fp}
	Let $\dep$ be a set of FDs with an LHS chain (up to equivalence). Given a database $D$, $\card{\rep{D}{\dep}}$ is computable in polynomial time, and thus, $\sharp \prob{Repairs}(\dep)$ is in \text{\rm FP}.
\end{proposition}}

\rev{Let us clarify that blocktrees have not been used in~\cite{LiKW21}. Instead, a connection is established with the class of $P_4$-free graphs, which are also known as cographs. These graphs are characterized by the absence of a path of length four as an induced subgraph. This connection is made through the notion of conflict graph; the conflict graph of a database $D$ with respect to a set $\dep$ of FDs is a graph that has a node for each fact in the database, and there is an edge between two nodes if the corresponding facts jointly violate an FD of $\dep$. In~\cite{LiKW21}, it was shown that a set of FDs has an LHS chain (up to equivalence) if and only if the conflict graph of \e{every} database $D$ w.r.t.~$\dep$ is $P_4$-free. The notion of blocktree introduced above essentially corresponds to the cotree of the associated conflict graph; thus, these two representations of sets of FDs with an LHS chain are equivalent, and we consider blocktrees for technical convenience.}

 \begin{figure}
     \centering
     \begin{tikzpicture}
[
->,>=stealth',
    level 1/.style = {sibling distance = 10cm},
    level 2/.style = {sibling distance = 4.5cm},
    level 3/.style = {sibling distance = 1.8cm},
    level 4/.style = {sibling distance = 1.3cm}
]
    \node [draw]{$r$}
    child {node [draw] {$v_1$}
    child {node [draw] {$v_2$}
    child {node [draw] {$v_5$}
    child {node [draw] {$v_{12}$}
    edge from parent node [align=left,right,font=\scriptsize] (d){$\val{BBY}$}}
    edge from parent node [align=left,right,font=\scriptsize] (c){$\val{315}$}}
    child {node [draw] {$v_6$}
    child {node [draw] {$v_{13}$}
    edge from parent node [align=left,right,font=\scriptsize] {$\val{PVD}$}}
    edge from parent node [align=left,right,font=\scriptsize] {$\val{250}$}}
    edge from parent node [align=left,right,font=\scriptsize,xshift=1em] (b){$\val{NYP}$}}   
    child {node [draw] {$v_3$}
    child {node [draw] {$v_7$}
    child {node [draw] {$v_{14}$}
    edge from parent node [align=left,right,font=\scriptsize] {$\val{WIL}$}}
    edge from parent node [align=left,right,font=\scriptsize] {$\val{20}$}}
    child {node [draw] {$v_8$}
    child {node [draw] {$v_{15}$}
    edge from parent node [align=left,right,font=\scriptsize] {$\val{BAL}$}}
    edge from parent node [align=left,right,font=\scriptsize] {$\val{70}$}}
    child {node [draw] {$v_9$}
    child {node [draw] {$v_{16}$}
    edge from parent node [align=left,right,font=\scriptsize] {$\val{WAS}$}}
    edge from parent node [align=left,right,font=\scriptsize] {$\val{120}$}}
    edge from parent node [align=left,right,font=\scriptsize] {$\val{PHL}$}}
    child {node [draw] {$v_4$}
    child {node [draw] {$v_{10}$}
    child {node [draw] {$v_{17}$}
    edge from parent node [align=left,right,font=\scriptsize] {$\val{PHL}$}}
    child {node [draw] {$v_{18}$}
    edge from parent node [align=left,right,font=\scriptsize] {$\val{NYP}$}}
    edge from parent node [align=left,right,font=\scriptsize] {$\val{260}$}}
    child {node [draw] {$v_{11}$}
    child {node [draw] {$v_{19}$}
    edge from parent node [align=left,right,font=\scriptsize] {$\val{WAS}$}}
    edge from parent node [align=left,right,font=\scriptsize] {$\val{420}$}}
    edge from parent node [align=left,right,font=\scriptsize,xshift=1em] {$\val{BBY}$}}     
    edge from parent node [align=left,right,font=\scriptsize] (a){$\val{16}$\\$\val{1030}$}};
    \path (a) ++(-3.1in,0)coordinate(a0)  node [right,align=left,font=\footnotesize] {$\attr{train}$\\$\attr{time}$};
    \node at (b -| a0)(b0) {} (b0)++(0,0)  node [right,align=left,font=\footnotesize] {$\attr{origin}$};
    \node at (c -| a0)(c0) {} (c0)++(0,0)  node [right,align=left,font=\footnotesize] {$\attr{duration}$};
    \node at (d -| a0)(d0) {} (d0)++(0,0)  node [right,align=left,font=\footnotesize] {$\attr{destination}$};
\end{tikzpicture}
     \caption{The $(\rel{Schedule},\Lambda)$-blocktree of the database of Figure~\ref{fig:trains} w.r.t.~the FD set of Example~\ref{example:train_fds}. 
     }
     \label{fig:blocktree}
 \end{figure}

\begin{example}
Figure~\ref{fig:blocktree} depicts the $(\rel{Schedule},\Lambda)$-blocktree of the database $D$ of Figure~\ref{fig:trains} w.r.t.~the FD set $\dep$ of Example~\ref{example:train_fds}. Recall that the LHS chain $\Lambda$ of $\dep_{\rel{Schedule}}$ is
\[
\rel{Schedule}\ :\ \{\attr{train}, \attr{time}\} \rightarrow \attr{origin},\,\,\,
\rel{Schedule}\ :\ \{\attr{train},\attr{time},\attr{duration}\} \rightarrow \attr{destination}.
\]
The root $r$ is associated with the entire database $D_{\rel{Schedule}}$ (that is, $\lambda(r)=D_{\rel{Schedule}}$). Since all the facts of $D_{\rel{Schedule}}$ agree on the values of the $\attr{train}$ and $\attr{time}$ attributes, there is a single block w.r.t.~the first FD in the chain that contains all the facts; hence, $\lambda(v_1)=D_{\rel{Schedule}}$ as well. Then, the node $v_4$, for example, is associated with all the facts $f$ with $f[\attr{origin}]=\val{BBY}$; hence, $\lambda(v_4)=\{f_6,f_7,f_8\}$. This set of facts is a subblock of $D_{\rel{Schedule}}$ w.r.t.~the FD $\rel{Schedule} : \{\attr{train}, \attr{time}\} \rightarrow \{\attr{origin}\}$. The node $v_{10}$ collects all the facts $f\in \lambda(v_4)$ with $f[\attr{duration}]=\val{260}$; thus, $\lambda(v_4)=\{f_6,f_7\}$. Here, $\{f_6,f_7\}$ is a block of $\lambda(v_4)$ w.r.t.~the FD $\rel{Schedule} : \{\attr{train},\attr{time},\attr{duration}\} \rightarrow \{\attr{destination}\}$ as these facts agree on the values of the attributes $\attr{train},\attr{time}$, and $\attr{duration}$, while the fact $f_8$ has a different value for the attribute $\attr{duration}$ and belongs to a different block. \hfill\markfull
\end{example}


\noindent
\paragraph{The Case of Primary Keys.} Before we proceed further, let us focus for a moment on the simpler case of primary keys, and briefly comment on how database repairs are constructed via the machinery described above. Given a database $D$ and a set $\dep$ of primary keys, by definition, $\dep_R$ consists of a single 
FD $\phi_R$ of the form $R : X \ra Y$ with $X \cup Y = \att{R}$.
This means that the $(R,\Lambda)$-blocktree $T= (V,E,\lambda)$ of $D$ w.r.t.~$\dep_R$, where $\Lambda$ is the trivial LHS chain of $\dep_R$ mentioning only $\phi_R$, is of height two, and, for each leaf node $v$, $\lambda(v)$ contains a single $R$-fact. This implies that a repair of $D_R$ w.r.t.~$\dep_R$ can be constructed by simply selecting, for each node at level one of $T$, one of its children (which are leaf nodes), and then collect the $R$-facts that label the selected leaf nodes.
This illustrates that the underlying structure of database repairs in the case of primary keys is significantly simpler compared to that of repairs in the case of FDs with an LHS chain.

\rev{Let us further note that the case of primary keys is closely related to the block-independent-disjoint probabilistic databases~\cite{DBLP:journals/cacm/DalviRS09,DBLP:journals/jcss/DalviRS11}, where the repairs coincide with the possible worlds of the probabilistic database (see~\cite{MaWi13} for further details). However, as discussed in~\cite{MaWi13}, these are different models that require different techniques. Moreover, it remains unclear how the block-independent-disjoint probabilistic model can be adapted to align with the case of an LHS chain.}
\section{Exact Counting}\label{sec:exact-counting}

We now concentrate on exact counting and show that Theorem~\ref{the:pk-dichotomy} can be extended to arbitrary FDs. In particular, we establish  the following classification:

\begin{theorem}\label{the:fds-dichotomy}
	For a set $\dep$ of FDs and an SJFCQ $Q$, the following statements hold:
	\begin{enumerate}
		\item $\sharp \prob{Repairs}(\dep,Q)$ is either in \text{\rm FP} or $\sharp$\text{\rm P}-complete. 
		\item We can decide in polynomial time in $||\dep|| + ||Q||$ whether $\sharp \prob{Repairs}(\dep,Q)$ is in \text{\rm FP} or $\sharp$\text{\rm P}-complete.
	\end{enumerate}
\end{theorem}

We start by observing that, for every set $\dep$ of FDs and an SJFCQ $Q$, $\sharp \prob{Repairs}(\dep,Q)$ is in $\sharp$\text{\rm P}: guess a subset $D'$ of the given database $D$, and verify that $D'$ is a repair of $D$ w.r.t.~$\dep$ that entails $Q$. Since both the guess and the verify steps are feasible in polynomial time in $||D||$, we conclude that $\sharp \prob{Repairs}(\dep,Q)$ is in $\sharp$\text{\rm P}. 
We further observe that the non-existence of an LHS chain (up to equivalence) is a preliminary boundary for the hard side of the target dichotomy, no matter how the CQ looks like. We know from~\cite{LiKW21} that, for a set $\dep$ of FDs, $\sharp \prob{Repairs}(\dep)$ is $\sharp$P-hard if $\dep$ does not have an LHS chain (up to equivalence). From the above discussion and Lemma~\ref{lem:cook-reduction}, we get that:

\begin{proposition}\label{pro:no-lhs-hard}
	For a set $\dep$ of FDs without an LHS chain (up to equivalence) and an SJFCQ $Q$, it holds that $\sharp \prob{Repairs}(\dep,Q)$ is $\sharp$\text{\rm P}-complete.
\end{proposition}

From Proposition~\ref{pro:no-lhs-hard}, and the fact that checking whether a set $\dep$ of FDs has an LHS chain (up to equivalence) is feasible in polynomial time~\cite{LiKW21}, we conclude that to obtain Theorem~\ref{the:fds-dichotomy} for FDs, it suffices to provide an analogous result for FDs with an LHS chain (up to equivalence). 
To this end, following what was done for primary keys in~\cite{MaWi13}, we are going to concentrate on the problem of computing the relative frequency of the query. The {\em relative frequency} of a CQ $Q$ w.r.t.~a database $D$ and a set $\dep$ of FDs is the ratio that computes the percentage of repairs that entail it, i.e.,
\[
\rfreq{Q}{D,\dep}\ =\ \frac{\card{\rep{D}{\dep,Q}}}{\card{\rep{D}{\dep}}}.
\]
We then consider the following problem for a set $\dep$ of FDs and a CQ $Q$:

\medskip

\begin{center}
	\fbox{\begin{tabular}{ll}
			{\small PROBLEM} : & $\prob{RelFreq}(\dep,Q)$
			\\
			{\small INPUT} : & A database $D$
			\\
			{\small OUTPUT} : &  $\rfreq{Q}{D,\dep}$
	\end{tabular}}
\end{center}

\medskip

\noindent We can indeed focus on $\prob{RelFreq}(\dep,Q)$ since we know from Proposition~\ref{pro:counting-repairs-lhs-fp} that the problem of computing $\card{\rep{D}{\dep}}$ is feasible in polynomial time in $||D||$ whenever $\dep$ has an LHS chain (up to equivalence), which implies that $\sharp \prob{Repairs}(\dep,Q)$ and $\prob{RelFreq}(\dep,Q)$ have the same complexity, that is, both are either in \text{\rm FP} or $\sharp \text{P}$-hard. Hence, to obtain Theorem~\ref{the:fds-dichotomy} it suffices to show the following:

\def\thmdichotomyrfreq{
For a set $\dep$ of FDs with an LHS chain (up to equivalence) and an SJFCQ $Q$, the following statements hold:
	\begin{enumerate}
		\item $\prob{RelFreq}(\dep,Q)$ is either in \text{\rm FP} or $\sharp \text{\rm P}$-hard.
		\item We can decide in polynomial time in $||\dep|| + ||Q||$ whether $\prob{RelFreq}(\dep,Q)$ is in \text{\rm FP} or $\sharp \text{\rm P}$-hard.
	\end{enumerate}
}

\begin{theorem}\label{the:fds-dichotomy-rfreq}
\thmdichotomyrfreq
\end{theorem}

To establish the above dichotomy result, we first introduce the central notion of safety for an SJFCQ $Q$ w.r.t.~a set $\dep$ of FDs with an LHS chain, inspired by a similar notion for the case of primary keys from~\cite{MaWi13}, and show that the problem of deciding whether $Q$ is safe w.r.t.~$\dep$ is feasible in polynomial time in $||\dep|| + ||Q||$ (see Theorem~\ref{the:deciding-safety} below).
We then proceed to show that safety exhausts the tractable side of the dichotomy (item (1)) stated in Theorem~\ref{the:fds-dichotomy-rfreq}. 
In particular, we show that if $Q$ is safe w.r.t.~$\dep$, then $\prob{RelFreq}(\dep,Q)$ is in \text{\rm FP} (see Theorem~\ref{the:safe-tractability} below), and  if $Q$ is not safe w.r.t.~$\dep$, then $\prob{RelFreq}(\dep,Q)$ is $\sharp \text{\rm P}$-hard (see Theorem~\ref{the:non-safe-hardness} below).
The rest of this section is devoted to introducing the notion of safety for an SJFCQ w.r.t.~a set of FDs, and establishing Theorems~\ref{the:deciding-safety}, \ref{the:safe-tractability}, and~\ref{the:non-safe-hardness}, which will imply Theorem~\ref{the:fds-dichotomy-rfreq}.
We start by introducing some auxiliary notions.

\subsection{Auxiliary Notions}

\noindent
\paragraph{Canonical Covers.}
\rev{A set $\dep$ of FDs is called {\em canonical} if (i) it is minimal in the sense that it does not contain redundant FDs (i.e.,~FDs $\phi$ such that every database $D$ that satisfies $\dep\setminus\{\phi\}$ also satisfies $\phi$) and its FDs do not mention redundant attributes (i.e.,~there is no FD $R:X\rightarrow Y$ with some $B\in X$ such that every database $D$ that satisfies $\dep$ also satisfies $R:(X\setminus\{B\})\rightarrow Y$), and (ii) every FD of $\dep$ is of the form $R:X\ra A$ for some attribute $A\in \att{R}$ (i.e.,~it has a single attribute on its right-hand side); for further details we refer the reader to~\cite{Maier80}. }
A {\em canonical cover} of a set $\dep$ of FDs is a canonical set of FDs that is equivalent to $\dep$. Here is a simple example that illustrates this technical notion.

\begin{example}
		Let $\dep_1$ be the set of FDs consisting of
		\begin{eqnarray*}
		\phi_1 &=& \rel{Schedule}\ :\ \{\attr{train},\attr{time}\} \rightarrow \attr{origin}\\
		\phi_2 &=& \rel{Schedule}\ :\ \{\attr{train},\attr{time},\attr{duration}\} \rightarrow \attr{destination}\\
		\phi_3 &=& \rel{Schedule}\ :\ \{\attr{train},\attr{time},\attr{duration}\} \rightarrow \attr{origin}.
		\end{eqnarray*}
		Observe that $\phi_3$ is redundant since every database $D$ with $D\models \dep_1\setminus\{\phi_3\}$ satisfies $\phi_1$; hence, it clearly also satisfies $\phi_3$.
		Now, let $\dep_2$ be the set of FDs $\{\phi_1,\phi_4\}$, where
		\[
		\phi_4\ =\ \rel{Schedule}\ :\ \{\attr{train},\attr{time},\attr{duration},\attr{origin}\} \rightarrow \attr{destination}.
		\]
		Here, the FD $\phi_4$ has a redundant attribute (that is, $\attr{origin}$) since the FD $\phi_1$ implies that every two facts $f_1,f_2$ that agree on the values of the attributes $\attr{train}$ and $\attr{time}$ also agree on the value of $\attr{origin}$. Hence, if two facts $f_1,f_2$ jointly satisfy the FDs of $\dep_2$, then they also satisfy the FD $\phi_2$. In particular, due to $\phi_1$, it cannot be the case that $f_1$ and $f_2$ disagree on the value of $\attr{origin}$ (hence, jointly satisfy $\phi_4$) while they agree on the values of $\attr{train}$ and $\attr{time}$ (hence, might violate the FD $\phi_2$). We conclude that for every database $D$ such that $D\models\dep_2$ it holds that $D\models \{\phi_2\}$. 
		From the above discussion, we get that the set of FDs $\dep$ of Example~\ref{example:train_fds} is a canonical cover of $\dep_1$ and $\dep_2$. \hfill\markfull
\end{example}

Assuming that a set $\dep$ of FDs has an LHS chain (up to equivalence), we know that it has a single canonical cover $\dep'$ that has an LHS chain~\cite{LiKW21}. Furthermore, for the canonical cover $\dep'$ of $\dep$, it can be verified that, for every relation name $R$ occurring in $\dep'$, there exists a unique sequence $R : X_1 \ra A_1,\ldots,R : X_n \ra A_n$ of the FDs of $\dep'_R$, which we call {\em the LHS chain} of $\dep'_R$, such that (i) $X_i \subsetneq X_{i+1}$ for each $i \in [n-1]$, (ii) $X_i \cap Y_j = \emptyset$ for each $i,j \in [n]$, and (iii) $Y_i \cap Y_j = \emptyset$ for each $i,j \in [n]$ with $i \neq j$.
Since a canonical cover of a set of FDs can be computed in polynomial time, to establish Theorem~\ref{the:fds-dichotomy-rfreq} it suffices to concentrate on sets of FDs with an LHS chain that are canonical. \rev{In what follows, for convenience, we group together all the FDs of a canonical FD set $\dep$ that share the same left-hand side. In other words, if $R: X \rightarrow A_1, \dots, R: X \rightarrow A_n$ are all the FDs of $\dep$ with $X$ on their left-hand side, we replace these FDs with the single FD $R: X \rightarrow \{A_1, \dots, A_n\}$. Clearly, this has no impact on the satisfaction of the FDs or the LHS chain.}

\medskip
\noindent
\paragraph{Primary FDs and Positions.}
Consider now an SJFCQ $Q$. Let $\alpha_R$ be the $R$-atom in $Q$, and let $\Lambda_R = R : X_1 \ra Y_1,\ldots,R : X_n \ra Y_n$ be the LHS chain of $\dep_R$. We call an FD $R : X_i \ra Y_i$, for some $i \in [n]$, the {\em primary FD of $\dep_R$ w.r.t.~$Q$} if $X_i \cup Y_i$ contains an attribute $A$ such that at position $(R,A)$ in $\alpha_R$ we have a variable, whereas at each position of $\{(R,B) \mid B \in X_j \cup Y_j \text{ for } j < i\}$ in $\alpha_R$ we have a constant.
Note that there is no guarantee that the primary FD of $\dep_R$ w.r.t.~$Q$ exists.
If the primary FD $R : X_i \ra Y_i$ of $\dep_R$ w.r.t.~$Q$ exists, for some $i \in [n]$, then the {\em primary-lhs positions of $\alpha_R$ (w.r.t.~$\dep$)} are the positions $\{(R,A) \mid A \in X_i\}$, while the {\em non-primary-lhs positions of $\alpha_R$ (w.r.t.~$\dep$)} are the positions $\{(R,B) \mid B \in \att{R} \setminus X_i\}$. We further call the sequence of FDs $R : X_1 \ra Y_1,\ldots,R : X_{i-1} \ra Y_{i-1}$ the {\em primary prefix of $\dep_R$ w.r.t.~$Q$}.
If the primary FD of $\dep_R$ w.r.t.~$Q$ does not exist, then, by convention, the primary-lhs positions of $\alpha_R$ are the positions $\{(R,A) \mid A \in \att{R}\}$, i.e., all the positions in $\alpha_R$ are primary-lhs, which in turn implies that $\alpha_R$ has no non-primary-lhs positions. Moreover, the primary prefix of $\dep_R$ w.r.t.~$Q$ is the sequence $\Lambda_R$ itself.
We denote by \rev{$\pvar{\alpha_R}{\dep}$} the set of variables in $\alpha_R$ at primary-lhs positions of $\alpha_R$.

\begin{example}\label{example:train_primary_fds}
Consider the FDs $\dep$ of Example~\ref{example:train_fds} and the CQs $Q_3,Q_4$ of Example~\ref{example:train_bool_queries}. Recall that 
\[
\phi_1 = \rel{Schedule} : \{\attr{train}, \attr{time}\} \rightarrow \attr{origin},\,\,\,\,\,
\phi_2 = \rel{Schedule} : \{\attr{train},\attr{time},\attr{duration}\} \rightarrow \attr{destination}
\]
is the LHS chain of $\dep_{\rel{Schedule}}$ and 
\[
\phi_3\ =\ \rel{Station}\ :\ \attr{name}\rightarrow \attr{city}
\]
is the LHS chain of $\dep_{\rel{Station}}$. In both $Q_3$ and $Q_4$, all the attributes of the FD $\phi_1$ are associated with constants ($\val{16}, \val{1030}$, and $\val{BBY}$, respectively). The attribute $\attr{destination}$, on the other hand, is associated with the variable $z$ in both queries. Hence, we have that $\phi_2$ is the primary FD of $\dep_{\rel{Schedule}}$, and $\{\phi_1\}$ is the primary prefix of $\dep_{\rel{Schedule}}$.
The primary-lhs positions
of the atom over $\rel{Schedule}$ in both queries are the first, fourth, and fifth positions, corresponding to the attributes $\attr{train},\attr{time}$, and $\attr{duration}$, respectively, that occur on the left-hand side of the primary FD.
As for the relation name $\rel{Station}$, in both queries the attribute $\attr{name}$ is associated with the variable $z$; hence, $\phi_3$ is the primary FD of $\dep_{\rel{Station}}$, and its primary prefix is empty. There is a single primary-lhs position
of the atom over $\rel{Station}$ in both queries, that is, the first position, corresponding to the attribute $\attr{name}$ that occurs on the left-hand side of the primary FD. \hfill\markfull
\end{example}

\medskip
\noindent
\paragraph{Query Complex Part.}
A variable $x \in \var{Q}$ is a {\em liaison variable (of $Q$)} if it occurs more than once in $Q$. The {\em complex part of $Q$ w.r.t.~$\dep$}, denoted $\comp{Q}{\dep}$, is the set of atoms $\alpha$ of $Q$ in which we have a constant or a liaison variable at a non-primary-lhs position $(R,A)$, with $\alpha$ being the $R$-atom of $Q$, such that, for every FD $R : X \ra Y$ in the primary prefix of $\dep_R$ w.r.t.~$Q$, $A \not\in X \cup Y$.

\begin{example}\label{example:train_complex}
Consider again the FD set $\dep$ of Example~\ref{example:train_fds} and the CQ $Q_3$ of Example~\ref{example:train_bool_queries}. The variable $z$ is a liaison variable of $Q_3$ as it occurs in both atoms of the query. As discussed in Example~\ref{example:train_primary_fds}, the primary-lhs positions of the atom $\rel{Schedule}(\val{16},\val{BBY},z,\val{1030},w)$ in $Q_3$ are the ones corresponding to the attributes $\attr{train},\attr{time}$, and $\attr{duration}$. The attribute $\attr{origin}$ occurs in the primary prefix, and so the only non-primary-lhs position associated with an attribute that does not occur in the primary prefix is the one corresponding to the attribute $\attr{destination}$. This position is associated with the liaison variable $z$; hence, this atom belongs to complex part. The only non-primary-lhs position of the atom $\rel{Station}(z,\val{Washington})$ is the one corresponding to the attribute $\attr{city}$ and it is associated with a constant; hence, this atom also belongs to the complex part of $Q_3$. We can similarly show that both atoms of $Q_4$ belong to its complex part. \hfill\markfull
\end{example}

Interestingly, whenever $\comp{Q}{\dep} = \emptyset$, we can count the number of repairs that entail $Q$ in polynomial time in the size of the database. This is shown by first observing that the number of repairs of a database $D$ w.r.t.~$\dep$ that entail $Q$ coincides with the number of repairs of the subset of $D$ obtained after removing the facts of $D$ that are somehow in a conflict with $Q$, and thus, they cannot appear in a repair that entails $Q$.
Let $R$ be a relation name in $Q$ with $\alpha_R$ being the $R$-atom of $Q$. We define the database $D_{\mathsf{conf}}^{\dep_R,Q}$ as the subset of $D$ that collects all the $R$-facts $f$ of $D$ such that, for some FD $R : X \ra Y$ in the primary prefix of $\dep_R$ w.r.t.~$Q$, $f[A] = \alpha_R[A]$ for each $A \in X$, and $f[B] \neq \alpha_R[B]$ for some $B \in Y$.
Now, assuming that $R_1,\ldots,R_m$ are the relation names in $Q$, we define the database
\[
D_{\mathsf{conf}}^{\dep,Q}\ =\ \bigcup_{R \in \{R_1,\ldots,R_m\}} D_{\mathsf{conf}}^{\dep_{R},Q}.
\]
For a database $D$ with $D \models Q$, $\comp{Q}{\dep} = \emptyset$ allows us to show that
$\card{\rep{D}{\dep,Q}} = \card{\repp{D \setminus D_{\mathsf{conf}}^{\dep,Q}}{\dep}}$.

Since $D_{\mathsf{conf}}^{\dep,Q}$ can be constructed in polynomial time in $||D||$, Proposition~\ref{pro:counting-repairs-lhs-fp} implies the following (the full proof can be found in the appendix):

\def\propnocomplex{
	\rev{Let $\dep$ be a set of FDs, and $Q$ an SJFCQ with $\comp{Q}{\dep} = \emptyset$. Given a database $D$, it holds that $\rfreq{Q}{D,\dep}$ is computable in polynomial time.}
}

\begin{proposition}\label{pro:counting-repairs-no-complex-part}
\propnocomplex
\end{proposition}

\begin{restatable}{algorithm}{AlgIsSafe}
	\LinesNumbered
	\KwIn{A set $\dep$ of FDs with an LHS chain and an SJFCQ $Q$}
	\KwOut{$\true$ if $Q$ is $\dep$-safe and $\false$ otherwise}
	\vspace{2mm}
	
	\If{$\comp{Q}{\dep} = \emptyset$}{\Return $\true$}
	
	\If{$Q = Q_1 \uplus Q_2$}{\Return{$\IsSafe(\dep,Q_1)\wedge \IsSafe(\dep,Q_2)$}}
	
	\If{$\comp{Q}{\dep} \neq \emptyset$, \text{\rm and there is a variable} $x \in \var{Q}$ \text{\rm with} $x \in \pvar{\alpha}{\dep}$ \text{\rm for every} $\alpha \in \comp{Q}{\dep}$}{
		\Return{$\IsSafe(\dep,Q_{x\mapsto c})$} for an arbitrary constant $c$}
	
	\If{\text{\rm there is} $\alpha \in \comp{Q}{\dep}$ \text{\rm of the form} $R(\bar y)$ \text{\rm such that} $\pvar{\alpha}{\dep} = \emptyset$ \text{\rm and the variable} $x$ \text{\rm occurs in} $\alpha$ \text{\rm at a position of} $\{(R,A) \mid A \in Y\}$ \text{\rm with} $R : X \ra Y$ \text{\rm being the primary FD of} $\dep_R$ \text{\rm w.r.t.} $Q$}{\Return{$\IsSafe(\dep,Q_{x\mapsto c})$} for an arbitrary constant $c$}
	
	{\Return{$\false$}}
	
	\caption{The recursive procedure $\mathsf{IsSafe}$}\label{alg:issafe}
\end{restatable}

\begin{example}\label{example:train_confdb}
Consider the database $D$ of Figure~\ref{fig:trains}, the FD set $\dep$ of Example~\ref{example:train_fds}, and the CQ 
\[
Q_5\ =\ \textrm{Ans}()\ \text{:-}\ \rel{Schedule}(\val{16},\val{BBY},z,\val{1030},w)
\]
asking whether train $16$ departs from the BBY station at 10:30. Note that the single atom of $Q_5$ also occurs in $Q_3$, and, as mentioned in Example~\ref{example:train_primary_fds}, the only non-primary-lhs position of this atom is the one associated with the variable $z$. In this case, the variable $z$ is not a liaison variable as it only occurs once in $Q_5$; hence, $\comp{Q_5}{\dep} = \emptyset$. Thus, by Proposition~\ref{pro:counting-repairs-no-complex-part}, $\rfreq{Q_5}{D,\dep}$ is computable in polynomial time. Note that every repair that entails $Q_5$ must contain at least one of the facts $f_6, f_7$, or $f_8$, as these are the only facts mentioning trains leaving the BBY station. All the other facts of $D_{\rel{Schedule}}$ are in conflict with each one of these facts due to the FD $\rel{Schedule} : \{\attr{train},\attr{time}\} \rightarrow \attr{origin}$ in the primary prefix of $\dep_{\rel{Schedule}}$. Each one of these facts disagrees with the atom of $Q_5$ on the value of the attribute $\attr{origin}$; hence, they will all be collected in the database $D_{\mathsf{conf}}^{\dep,Q_5}$. \hfill\markfull
\end{example}

\subsection{Query Safety}
We now proceed to introduce the central notion of safety for an SJFCQ $Q$ w.r.t.~a set $\dep$ of FDs with an LHS chain.
This notion essentially tells us that after simplifying $Q$ into basic subqueries by recursively applying three simplification steps, each such basic subquery has an empty complex part w.r.t.~$\dep$.
We write $Q_1 \uplus Q_2$ for the CQ consisting of the atoms of two non-empty CQs $Q_1,Q_2$ that share no atoms and variables, i.e., $Q_1 \cap Q_2 = \emptyset$ and $\var{Q_1} \cap \var{Q_2} = \emptyset$. 
We also write $Q_{x \mapsto t}$, where $x \in \var{Q}$ and \rev{$t \in \ins{C}\cup \ins{V}$}, for the CQ obtained from $Q$ after replacing $x$ with $t$.

\begin{definition}[\textbf{Query Safety}]\label{def:safety}
	Consider a set $\dep$ of FDs with an LHS chain, and an SJFCQ $Q$. We inductively define when $Q$ is {\em safe w.r.t.~$\dep$}, or simply {\em $\dep$-safe}, as follows:
	\begin{itemize}
		\item[-] If $\comp{Q}{\dep} = \emptyset$, then $Q$ is $\dep$-safe.
		\item[-] If $Q = Q_1 \uplus Q_2$ and $Q_1,Q_2$ are $\dep$-safe, then $Q$ is $\dep$-safe. \rev{(Note that such a decomposition is not necessarily unique, and we can choose an arbitrary one.)}
		\item[-] If $\comp{Q}{\dep} \neq \emptyset$, there exists a variable $x \in \var{Q}$ such that $x \in \pvar{\alpha}{\dep}$ for every $\alpha \in \comp{Q}{\dep}$, and $Q_{x \mapsto c}$ is $\dep$-safe for some arbitrary constant $c \in \ins{C}$, then $Q$ is $\dep$-safe.
		\item[-] If $\comp{Q}{\dep} \neq \emptyset$, there is an $R$-atom $\alpha \in \comp{Q}{\dep}$ such that $\pvar{\alpha}{\dep} = \emptyset$ and a variable $x$ occurs in $\alpha$ at a position of $\{(R,A) \mid A \in Y\}$ assuming that $R : X \ra Y$ is the primary FD of $\dep_R$ w.r.t.~$Q$, and $Q_{x \mapsto c}$ is $\dep$-safe for some arbitrary constant $c \in \ins{C}$, then $Q$ is $\dep$-safe. \hfill\markfull
	\end{itemize}
\end{definition}

\begin{example}
As mentioned in Example~\ref{example:train_complex}, we have that 
\[
\comp{Q_3}{\dep}\ =\ \{\rel{Schedule}(\val{16},\val{BBY},z,\val{1030},w),\rel{Station}(z,\val{Washington})\}
\]
for the FD set $\dep$ of Example~\ref{example:train_fds} and the query $Q_3$ of Example~\ref{example:train_bool_queries}. Hence, the first condition of Definition~\ref{def:safety} does not apply to $Q_3$. The second condition also does not apply to $Q_3$, as the two atoms of the query share the variable $z$. The only variable of $\rel{Schedule}(\val{16},\val{BBY},z,\val{1030},w)$ associated with a primary-lhs position is $w$ and the only variable of $\rel{Station}(z,\val{Washington})$ that occurs at a primary-lhs position is $z$; hence, the third condition does not hold. This also implies that the last condition does not hold, since both atoms use a variable at a primary-lhs position. Therefore, we conclude that the query $Q_3$ is not $\dep$-safe.

For the CQ $Q_4$ of Example~\ref{example:train_bool_queries} we have that
\[
\comp{Q_4}{\dep}\ =\ \{\rel{Schedule}(\val{16},\val{BBY},z,\val{1030},\val{420}),\rel{Station}(z,\val{Washington})\}.
\]
Here, we can apply the last condition of Definition~\ref{def:safety} because there are no variables at primary-lhs positions of $\rel{Schedule}(\val{16},\val{BBY},z,\val{1030},\val{420})$, and the variable $z$ occurs at a position corresponding to the attribute $\attr{destination}$ that occurs on the right-hand side of the primary FD of $\dep_{\rel{Schedule}}$ w.r.t.~$Q_4$. After replacing the variable $z$ with some constant, we obtain a CQ $Q_6$ with no variables. Thus, both $\dep_\rel{Schedule}$ and $\dep_\rel{Station}$ have no primary FD and no non-primary-lhs positions w.r.t.~$Q_6$. Hence, $Q_6$ has an empty complex part and so it is $\dep$-safe. We then conclude that $Q_4$ is $\dep$-safe. \hfill\markfull
\end{example}

We proceed to show that query safety can be checked in polynomial time.

\begin{theorem}\label{the:deciding-safety}
	Consider a set $\dep$ of FDs with an LHS chain and an SJFCQ $Q$. The problem of deciding whether $Q$ is $\dep$-safe is feasible in polynomial time in $||\dep||+||Q||$.
\end{theorem}

\begin{proof}
	It is easy to verify that the recursive procedure $\mathsf{IsSafe}$ (depicted in Algorithm~\ref{alg:issafe}), which accepts as input $\dep$ and $Q$, and simply implements Definition~\ref{def:safety}, it correctly determines whether $Q$ is $\dep$-safe. Thus, it remains to show that $\mathsf{IsSafe}(\dep,Q)$ runs in polynomial time in $||\dep|| + ||Q||$.
	To this end, we show via an inductive argument that $\mathsf{IsSafe}(\dep,Q)$ requires polynomially many recursive calls in $||Q||$, which, together with the fact that all the conditions of the if-then statements can be checked in polynomial time in $||\dep|| + ||Q||$, implies that $\mathsf{IsSafe}(\dep,Q)$ runs in polynomial time in $||\dep|| + ||Q||$, as needed.
	We proceed by induction on the value $n=|Q|+m$, where $|Q|$ is the number of atoms in $Q$ and $m$ is the number of variables in $Q$. We show that the total number of recursive calls in an execution of $\IsSafe(\dep,Q)$ is bounded by $2|Q|-1+m$. Note that the proof is similar to the corresponding proof for primary keys from~\cite{MaWi13}.
	
	\medskip
	
	\noindent
	\paragraph{Base Case.}
	The base case of the induction (for $n=1$) is the case where $Q$ has a single atom $\alpha$ and no variables. In this case, all the positions of $\alpha$ are primary-lhs positions and $\comp{Q}{\dep}=\emptyset$. Therefore, the algorithm is called once, and the condition in line~1 is satisfied. Here, $2|Q|-1+m=2\cdot 1-1+0=1$, and indeed there is a single call to $\IsSafe(\dep,Q)$.
	
	\medskip
	
	\noindent
	\paragraph{Inductive Step.}
	We now assume that the claim holds for $|Q|+m=n$, and we show that it also holds for $|Q|+m= n+1$. If none of the if-conditions is satisfied by $Q$, then there is a single call to the algorithm (that returns $\false$) and the claim follows. Otherwise, one of the following holds:
	\begin{enumerate}
		\item \underline{The if-condition in line~1 is satisfied.} In this case, this is again a single call to $\IsSafe(\dep,Q)$ and clearly $2|Q|-1+m>1$.
		
		\item \underline{The if-condition in line~3 is satisfied.} In this case, the number of recursive calls is $1+n_1+n_2$, where $n_1$ is the number of calls in an execution of $\IsSafe(\dep,Q_1)$ and $n_2$ is the number of calls in an execution of $\IsSafe(\dep,Q_2)$. By the inductive hypothesis, we have that $n_1\le 2|Q_1|-1+m_1$ and $n_2\le 2|Q_2|-1+m_2$, where $m_1$ and $m_2$ are the number of variables in $Q_1$ and $Q_2$, respectively. Since $Q_1$ and $Q_2$ have no shared atoms or variables, it holds that $|Q|=|Q_1|+|Q_2|$ and $m=m_1+m_2$. Hence, $1+n_1+n_2\le 1+2|Q_1|-1+m_1+2|Q_2|-1+m_2=1+2|Q|-2+m=2|Q|-1+m$.
		
		\item \underline{The condition in line~5 or the if-condition in line~7 is satisfied.} In this case, the number of recursive calls is $1+n_c$, where $n_c$ is the number of calls in an execution of $\IsSafe(\dep,Q_{x\mapsto c})$ for some arbitrary constant $c$. Since the query $Q_{x\mapsto c}$ has one less variable and the same number of atoms, by the inductive hypothesis, the number of calls in this case is $1+n_c\le 1+2|Q_{x\mapsto c}|-1+(m-1)=2|Q|-1+m$.
	\end{enumerate}
	This completes our proof.
\end{proof}

\subsection{The Tractable Side}
We are now ready to discuss the tractable side of the dichotomy stated in Theorem~\ref{the:fds-dichotomy-rfreq}. In particular, we show that safety ensures that the relative frequency can be computed in polynomial time:

\begin{theorem}\label{the:safe-tractability}
	Consider a set $\dep$ of FDs with an LHS chain, and an SJFCQ $Q$. If $Q$ is $\dep$-safe, then $\prob{RelFreq}(\dep,Q)$ is in \text{\rm FP}.
\end{theorem}

The above result relies on a recursive procedure, called $\mathsf{Eval}$, that takes as input a database $D$, a set $\dep$ of FDs with an LHS chain, and a $\dep$-safe SJFCQ $Q$, runs in polynomial time in $||D||$, and computes the value $\rfreq{Q}{D,\dep}$. We proceed to discuss the details of this procedure.

\medskip
\noindent \paragraph{Three Auxiliary Lemmas.}
$\mathsf{Eval}$ exploits three auxiliary technical lemmas, which we present next (their proofs are deferred to the appendix), that essentially tell us how the relative frequency of an SJFCQ can be computed via the relative frequency of a simpler query obtained by applying one of the three simplification steps underlying the notion of query safety.

The first one, which deals with SJFCQs that are the disjoint union of two other queries, is an easy observation directly inherited from~\cite{MaWi13} since the given proof, which we recall below for the sake of completeness, works also for arbitrary FDs.

\def\lemmaauxfirst{
	Consider a database $D$, a set $\dep$ of FDs, and an SJFCQ $Q = Q_1 \uplus Q_2$. Then 
	\[
	\rfreq{Q}{D,\dep}\ =\ \rfreq{Q_1}{D,\dep} \times \rfreq{Q_2}{D,\dep}.
	\]
}

\begin{lemma}\label{lem:aux-1}
\lemmaauxfirst
\end{lemma}

\begin{proof}
	We consider the following three subsets of $D$:
	\begin{itemize}
	\item $D_0$ collects every fact of $D$ over a relation name that does not occur in $Q$.
	\item $D_1$ collects every fact of $D$ over a relation name that occurs in $Q_1$.
	\item $D_2$ collects every fact of $D$ over a relation name that occurs in $Q_2$. 
	\end{itemize}
	Since $Q_1$ and $Q_2$ do not share any atoms or variables, it is rather straightforward to see that
	\[
	\card{\rep{D}{\dep}}\ =\ \card{\rep{D_0}{\dep}}\times \card{\rep{D_1}{\dep}}\times \card{\rep{D_2}{\dep}}
	\quad
	\]
	and
	\[
	\card{\rep{D}{\dep,Q}}\ =\ \card{\rep{D_0}{\dep}}\times \card{\rep{D_1}{\dep,Q_1}}\times \card{\rep{D_2}{\dep,Q_2}}.
	\]
	Therefore, we have that
	$\rfreq{Q}{D,\dep} = \rfreq{Q_1}{D_1,\dep}\times \rfreq{Q_2}{D_2,\dep}$, as needed.
\end{proof}

The second auxiliary lemma deals with SJFCQs with a non-empty complex part that mentions a variable that occurs in every atom of the complex part at a primary-lhs position.
To this end, given a database $D$, a set $\dep$ of FDs with an LHS chain, and an SJFCQ $Q$, we need to collect the facts of $D \setminus D_{\mathsf{conf}}^{\dep,Q}$ that do not affect the relative frequency of $Q$ w.r.t.~$D \setminus D_{\mathsf{conf}}^{\dep,Q}$ and $\dep$.
Let $R$ be a relation name in $Q$ with $\alpha_R$ being the $R$-atom of $Q$. We define $D_{\mathsf{ind}}^{\dep_R,Q}$ as the subset of $D \setminus D_{\mathsf{conf}}^{\dep_R,Q}$ that collects all the $R$-facts $f$ of $D$ such that, for some FD $R : X \ra Y$ in the primary prefix of $\dep_R$ w.r.t.~$Q$, and attribute $A \in X$, $f[A] \neq \alpha_R[A]$.
Now, assuming that $R_1,\ldots,R_m$ are the relation names in $Q$, we define the database
\[
D_{\mathsf{ind}}^{\dep,Q}\ =\ \bigcup_{R \in \{R_1,\ldots,R_m\}} D_{\mathsf{ind}}^{\dep_{R},Q}.
\]

\begin{example}
Recall that the database $D_{\mathsf{conf}}^{\dep,Q_5}$ for the FD set $\dep$ of Example~\ref{example:train_fds} and the query $Q_5$ of Example~\ref{example:train_confdb}    contains all the facts over the relation name $\rel{Schedule}$ except $f_6,f_7$, and $f_8$; that is, $D\setminus D_{\mathsf{conf}}^{\dep,Q_5}=\{f_6,f_7,f_8,g_1,\ldots,g_7\}$. Now, assume that the database contains also the fact
\[
f_{10}\ =\ \rel{Schedule}(\val{17},\val{BBY},\val{NYP},\val{1030},\val{400}),
\]
which mentions train $\val{17}$. This fact does not belong to $D_{\mathsf{conf}}^{\dep,Q_5}$ as it disagrees with the query on the constant associated with the attribute $\attr{train}$ that occurs on the left-hand side of the only FD in the primary prefix of $\dep_{\rel{Schedule}}$ w.r.t.~$Q_5$; however, due to this reason, it does belong to the database $D_{\mathsf{ind}}^{\dep,Q_5}$. Since the query collects information about train number $\val{16}$, there is no homomorphism from $Q_5$ to $D$ that maps the only atom of the query to the fact $f_{10}$ that mentions train number $\val{17}$. Moreover, the fact $f_{10}$ is not in conflict with any of the facts $f_6,f_7$, or $f_8$ that are required to entail the query, as they disagree on the value of the attribute $\attr{train}$ that occurs on the left-hand side of every FD in $\dep_{\rel{Schedule}}$ (as this FD set has an LHS chain). Therefore, \rev{as we will show next}, this fact has no impact on the relative frequency of $Q_5$ w.r.t.~$D\setminus D_{\mathsf{conf}}^{\dep,Q_5}$ and $\dep$. \hfill\markfull
\end{example}

With $D^{\dep,Q} = D \setminus \left(D_{\mathsf{conf}}^{\dep,Q} \cup D_{\mathsf{ind}}^{\dep,Q}\right)$, we show in the appendix that
\[
\rfreq{Q}{D,\dep}\ =\ \rfreq{Q}{D^{\dep,Q},\dep}\ \times\ \underbrace{\frac{\card{\repp{D \setminus D_{\mathsf{conf}}^{\dep,Q}}{\dep}}}{\card{\rep{D}{\dep}}}}_{\mathsf{R}_{D,\dep,Q}},
\]
which tells us that $\rfreq{Q}{D,\dep}$ is the product of $\rfreq{Q}{D^{\dep,Q},\dep}$ with the ratio $\mathsf{R}_{D,\dep,Q}$. \rev{It is not immediate how to compute $\rfreq{Q}{D^{\dep,Q},\dep}$ in polynomial time in $||D||$, while $\mathsf{R}_{D,\dep,Q}$ is clearly computable in polynomial time in $||D||$ due to Proposition~\ref{pro:counting-repairs-lhs-fp}.}
\rev{The next lemma is shown using the above equation by providing a way to compute $\rfreq{Q}{D^{\dep,Q},\dep}$ via the relative frequency of a simplified version of $Q$ w.r.t.~$D^{\dep,Q}$ and $\dep$.}

\def\lemmaauxsecond{
	Consider a database $D$, a set $\dep$ of FDs with an LHS chain, and an SJFCQ $Q$. If $\comp{Q}{\dep} \neq \emptyset$, and there is $x \in \var{Q}$ such that $x \in \pvar{\alpha}{\dep}$ for every $\alpha \in \comp{Q}{\dep}$, then
	\[
	\rfreq{Q}{D,\dep}\ =\ \left(1 - \prod_{c \in \adom{D}} \left(1 - 	\rfreq{Q_{x \mapsto c}}{D^{\dep,Q},\dep} \right)\right)\ \times\ \mathsf{R}_{D,\dep,Q}.
	\]
}

\begin{lemma}\label{lem:aux-2}
\lemmaauxsecond
\end{lemma}

The last auxiliary lemma considers SJFCQs with at least one atom in their complex part that does not have a variable at a primary-lhs position, but has a variable at a position associated with an attribute occurring in the right-hand side of the primary FD.

\def\lemmaauxthird{
	Consider a database $D$, a set $\dep$ of FDs with an LHS chain, and an SJFCQ $Q$. If there exists an $R$-atom $\alpha \in \comp{Q}{\dep}$ such that $\pvar{\alpha}{\dep} = \emptyset$ and a variable $x$ occurs in $\alpha$ at a position of $\{(R,A) \mid A \in Y\}$ with $R : X \ra Y$ being the primary FD of $\dep_R$ w.r.t.~$Q$, then
	\[
	\rfreq{Q}{D,\dep}\ =\ \sum_{c \in \adom{D}} \rfreq{Q_{x \mapsto c}}{D,\dep}.
	\]
}

\begin{lemma}\label{lem:aux-3}
\lemmaauxthird
\end{lemma}

\begin{algorithm}[t]
	\LinesNumbered
	\KwIn{A database $D$, a set $\dep$ of FDs with an LHS chain, and a $\dep$-safe SJFCQ $Q$}
	\KwOut{$\rfreq{Q}{D,\dep}$}
	\vspace{2mm}
	
	\If{$D \not\models Q$}{\Return{$0$}}
	\If{$\comp{Q}{\dep} = \emptyset$}{\Return{$\rfreq{Q}{D,\dep}$}}
	\If{$Q = Q_1 \uplus Q_2$}{\Return{$\mathsf{Eval}(D,\dep,Q_1) \times \mathsf{Eval}(D,\dep,Q_2)$}}
	
	\If{$\comp{Q}{\dep} \neq \emptyset$, \text{\rm and there is a variable} $x \in \var{Q}$ \text{\rm with} $x \in \pvar{\alpha}{\dep}$ \text{\rm for every} $\alpha \in \comp{Q}{\dep}$}{
		\Return{$\left(1 - \prod_{c \in \adom{D}} \left(1 - 	\mathsf{Eval}\left(D^{\dep,Q},\dep,Q_{x \mapsto c}\right) \right)\right) \times \mathsf{R}_{D,\dep,Q}$}}
	
	\If{\text{\rm there exists an} $R$\text{\rm-atom} $\alpha \in \comp{Q}{\dep}$ \text{\rm such that} $\pvar{\alpha}{\dep} = \emptyset$ \text{\rm and a variable} $x$ \text{\rm occurs in} $\alpha$ \text{\rm at a position of} $\{(R,A) \mid A \in Y\}$ \text{\rm with} $R : X \ra Y$ \text{\rm being the primary FD of} $\dep_R$ \text{\rm w.r.t.} $Q$}{\Return{$\sum_{c \in \adom{D}} \mathsf{Eval}(D,\dep,Q_{x \mapsto c})$}}
	
	\caption{The recursive procedure $\mathsf{Eval}$}\label{alg:eval}
\end{algorithm}

\noindent \paragraph{The Recursive Procedure.} We are now ready to introduced the procedure $\Eval$ that generalizes the corresponding procedure for primary keys~\cite{MaWi13}.
$\Eval$, depicted in Algorithm~\ref{alg:eval}, accepts as input a database $D$, a set $\dep$ of FDs with an LHS chain, and a $\dep$-safe SJFCQ $Q$, and performs the following steps. It first checks if $D \not\models Q$, in which case it returns $0$. It then checks if $\comp{Q}{\dep} = \emptyset$, in which case it returns $\rfreq{Q}{D,\dep}$. Then, it essentially implements the three auxiliary lemmas from above by checking which condition holds, and computing the relative frequency of the simplified queries via recursive calls.
We proceed to show that $\Eval$ is indeed correct, i.e., it computes the value $\rfreq{Q}{D,\dep}$, and runs in polynomial time in $||D||$.

\begin{proposition}\label{pro:eval-procedure}
	Consider a database $D$, a set $\dep$ of FDs with an LHS chain, and a $\dep$-safe SJFCQ $Q$. The following hold: 
	\begin{enumerate}
		\item $\Eval(D,\dep,Q) = \rfreq{Q}{D,\dep}$.
		\item $\Eval(D,\dep,Q)$ terminates in polynomial time in $||D||$.
	\end{enumerate}
\end{proposition}

\begin{proof}
	Item (1) can be easily shown using \rev{Proposition~\ref{pro:counting-repairs-no-complex-part}} and Lemmas~\ref{lem:aux-1},~\ref{lem:aux-2} and~\ref{lem:aux-3}. The interesting task is to show item (2). The proof is along the lines of the corresponding proof for primary keys~\cite{MaWi13} as the structure of the algorithms is similar. However, the computations in the algorithms are different.
	Let $|Q|$ be the number of atoms in $Q$, and $m$ be the number of variables that occur in $Q$. We prove, by induction on the value $n=|Q|+m$, that the total number of recursive calls in an execution of $\Eval(D,\dep,Q)$ is bounded by $\frac{(2|Q|-1)(N^{m+1}-1)}{N-1}$, where $N = |\adom{D}|$ (which is bounded by $||D||$).
		
	\medskip
		
	\noindent\paragraph{Base Case.}
	The base case of the induction (i.e., for $n=1$) is the case where $Q$ has a single atom and no variables. \rev{Thus, the algorithm is called once, and either the if-condition in line~1 (that is, $D\not\models Q$) or the if-condition in line~3 (that is, $Q$ has no complex part) is triggered}. Here, $\frac{(2|Q|-1)(N^{m+1}-1)}{N-1}=\frac{N-1}{N-1}=1$, and indeed there is a single call to $\Eval(D,\dep,Q)$.
		
	\medskip
		
	\noindent\paragraph{Inductive Step.}
	We now assume that the claim holds for $|Q|+m=n$, and show that it also holds for $|Q|+m=n+1$. One of the following holds:
	\begin{enumerate}
		\item \underline{The if-condition in line~1 or the if-condition in line~3 is satisfied.} In this case, there is a single call to $\Eval(D,\dep,Q)$ and clearly $\frac{(2|Q|-1)(N^{m+1}-1)}{N-1}\ge 1$.
		
		\item \underline{The if-condition in line~5 is satisfied.} In this case, the number of recursive calls is $1+n_1+n_2$, where $n_1$ is the number of recursive calls in an execution of $\Eval(D,\dep,Q_1)$ and $n_2$ is the number of calls in an execution of $\Eval(D,\dep,Q_2)$. By the inductive hypothesis, we have that $n_1\le \frac{(2|Q_1|-1)(N^{m_1+1}-1)}{N-1}$ and $n_2\le \frac{(2|Q_1|-1)(N^{m_2+1}-1)}{N-1}$, where $m_1$ and $m_2$ are the number of variables in $Q_1$ and $Q_2$, respectively. Since $Q_1$ and $Q_2$ do not share atoms or variables, we have that $|Q|=|Q_1|+|Q_2|$ and $m=m_1+m_2$. Therefore, 
		\begin{eqnarray*}
			1+n_1+n_2&\le& 1+\frac{(2|Q_1|-1)(N^{m_1+1}-1)}{N-1}+\frac{(2|Q_2|-1)(N^{m_2+1}-1)}{N-1}\\
			&\le& \frac{N-1}{N-1}      +\frac{(2|Q_1|-1)(N^{m+1}-1)+(2|Q_2|-1)(N^{m+1}-1)}{N-1}\\
			&\le& \frac{N^{m+1}-1}{N-1}+\frac{(2|Q|-2)(N^{m+1}-1)}{N-1}\\
			&=& \frac{(2|Q|-1)(N^{m+1}-1)}{N-1}.
		\end{eqnarray*}
		
		\item \underline{The if-condition in line~7 or the if-condition in line~9 is satisfied.} In this case, the number of recursive calls is $1+N\cdot n_c$, where $n_c$ is the number of calls in an execution of $\Eval(D,\dep,Q_{x\mapsto c})$ for some arbitrary constant $c\in\adom{D}$. Since the query $Q_{x\mapsto c}$ has one less variable and the same number of atoms, by the inductive hypothesis, the number of calls in this case is:
		\begin{eqnarray*}
			1+N\cdot n_c&\le& 1+N\cdot \left(\frac{(2|Q|-1)(N^{m}-1)}{N-1}\right)\\
			&=& \frac{N-1+N\cdot (2|Q|-1)(N^{m}-1)}{N-1}\\
			&\le& \frac{(2|Q|-1)(N-1)+ (2|Q|-1)(N^{m+1}-N)}{N-1}\\
			&\rev{=}& \frac{(2|Q|-1)(N^{m+1}-1)}{N-1}.
		\end{eqnarray*}
		\end{enumerate}
		We conclude that the number of recursive calls is polynomial in $||D||$. An easy observation is that all the if-conditions used in $\mathsf{Eval}$ can be tested in polynomial time in $||D||$. \rev{Moreover, Proposition~\ref{pro:counting-repairs-no-complex-part} implies that if $\comp{Q}{\dep} = \emptyset$, then $\rfreq{Q}{D,\dep}$ can be computed in polynomial time in $||D||$.
		Finally, since the FD set $\dep$ has an LHS chain, by Proposition~\ref{pro:counting-repairs-lhs-fp}, the number of repairs of a given database $D$ w.r.t.~$\dep$ can be computed in polynomial time in $||D||$; hence, $\mathsf{R}_{D,\dep,Q}$ can be computed in polynomial time in $||D||$. Therefore, $\Eval(D,\dep,Q)$ runs in polynomial in $||D||$.}
\end{proof}

It is clear that Theorem~\ref{the:safe-tractability} immediately follows from Proposition~\ref{pro:eval-procedure}.

\subsection{The Intractable Side}

We now proceed to discuss the intractable side of the dichotomy stated in Theorem~\ref{the:fds-dichotomy-rfreq}. In particular, we show that non-safety implies that computing the relative frequency is an intractable problem:

\begin{theorem}\label{the:non-safe-hardness}
	Consider a set $\dep$ of FDs with an LHS chain, and an SJFCQ $Q$. If $Q$ is not $\dep$-safe, then $\prob{RelFreq}(\dep,Q)$ is $\sharp$\text{\rm P}-hard.
\end{theorem}

We know from~\cite{MaWi13} that, for every set $\dep$ of primary keys, and an SJFCQ $Q$ that is not $\dep$-safe, $\prob{RelFreq}(\dep,Q)$ is $\sharp \text{\rm P}$-hard. Hence, to establish Theorem~\ref{the:non-safe-hardness}, it suffices to show the following:

\def\thmreductionpks{
	Consider a set $\dep$ of FDs with an LHS chain, and an SJFCQ $Q$ that is not $\dep$-safe. There exists a set $\dep_{PK}$ of primary keys, and an SJFCQ $Q_{PK}$ that is not $\dep_{PK}$-safe such that $\prob{RelFreq}(\dep_{PK},Q_{PK})$ is Cook reducible to $\prob{RelFreq}(\dep,Q)$.
}

\begin{theorem}\label{the:hard-side}
	\thmreductionpks
\end{theorem}

The rest of this section is devoted to discussing the rather involved proof of Theorem~\ref{the:hard-side}. This proof consists of the following three main steps:
\begin{enumerate}
	\item We first define certain rewrite rules for pairs consisting of a set $\dep$ of FDs with an LHS chain and an SJFCQ $Q$. We then use those rules to define when a pair $(\dep,Q)$ is final, namely $Q$ is not $\dep$-safe and $(\dep,Q)$ cannot be rewritten into some other pair $(\dep',Q')$ with $Q'$ being an SJFCQ that is not $\dep'$-safe. We then observe that from each pair $(\dep,Q)$ where $Q$ is not $\dep$-safe, we can reach a pair $(\dep',Q')$ that is final via a sequence of rewritings. Note that the rewrite rules in question are based on the corresponding ones for primary keys from~\cite{MaWi13}.
		
	\item We then show that whenever a pair $(\dep,Q)$ can be rewritten into a pair $(\dep',Q')$, then there is a polynomial-time \rev{Cook reduction} from $\prob{RelFreq}(\dep',Q')$ to $\prob{RelFreq}(\dep,Q)$.
	
	\item Finally, we show that, for every set $\dep$ of FDs with an LHS chain and an SJFCQ $Q$ such that $(\dep,Q)$ is final, there exists a set $\dep_{PK}$ of primary keys and a query $Q_{PK}$ that is not $\dep_{PK}$-safe such that $\prob{RelFreq}(\dep_{PK},Q_{PK})$ is Cook reducible to $\prob{RelFreq}(\dep,Q)$.
\end{enumerate}

Having the above notions and results in place, we can then establish Theorem~\ref{the:hard-side} as follows. According to step (1), from $(\dep,Q)$ we can reach a pair $(\dep',Q')$ that is final via a sequence of rewritings. Thus, due to step (2) and the fact that Cook reductions compose, we get that $\prob{RelFreq}(\dep',Q')$ is Cook reducible to $\prob{RelFreq}(\dep,Q)$. Then, according to step (3), there exists a set $\dep_{PK}$ of primary keys and a query $Q_{PK}$ that is not $\dep_{PK}$-safe such that $\prob{RelFreq}(\dep_{PK},Q_{PK})$ is Cook reducible to $\prob{RelFreq}(\dep',Q')$. Thus, by composing reductions, we get that $\prob{RelFreq}(\dep_{PK},Q_{PK})$ is Cook reducible to $\prob{RelFreq}(\dep,Q)$, as needed. We proceed to give details for each of the above three steps.

\renewcommand\theadalign{tl}
\renewcommand\cellalign{tl}

\begin{figure*}
	\begin{tabular}{ l l l l }
		R1: & $(\dep,Q)$ & $\lrhd \, (\dep, Q_{x\mapsto c})$ & \makecell{if there is an atom $\alpha \in Q$ with \\ $x\in\pvar{\alpha}{\dep}$} \\\\
		R2: & $(\dep, Q_1\cup Q_2)$ & $\lrhd \, (\dep,Q_1)$ & \makecell{if $Q_1\neq \emptyset$, $Q_2 \neq \emptyset$ and\\ $\var{Q_1}\cap\var{Q_2}=\emptyset$}\\\\
		R3: & $(\dep,Q)$ & $\lrhd \, (\dep,Q_{y\mapsto x})$ & \makecell{if there is an atom $\alpha \in Q$ with\\ $x,y\in\pvar{\alpha}{\dep}$} \\\\
		R4: & $(\dep,Q\cup\set{\alpha})$ & $\lrhd \, (\dep,Q)$ & if $(\var{Q}\cap\var{\alpha})\subseteq\pvar{\alpha}{\dep}$ \\\\ 
		R5: & $(\dep,Q\cup\set{\alpha_1})$ & $\lrhd \, ((\dep\setminus\dep_R)\cup\dep_{R'},Q\cup\set{\alpha_2})$ & \makecell{if $\alpha_1>_\dep\alpha_2$ and the relation name of\\ $\alpha_2$ does not occur in $Q$} \\\\ 
		R6: & $(\dep,Q)$ & $\lrhd \, (\dep,Q_{x\mapsto c})$ & \makecell{if there is an atom $\alpha \in Q$ with \\ $\pvar{\alpha}{\dep}=\emptyset$ and $x\in\prhsvar{\alpha}{\dep}$,\\ and one of the variables of $\var{\alpha}$\\ is a liaison variable of $Q$} \\\\
		R7: & $(\dep,Q)$ & $\lrhd \, (\dep,Q_{y\mapsto x})$ & \makecell{if there is an atom $\alpha\in Q$ with\\ $x,y\in\var{\alpha}$, none of $x,y$ is an\\ orphan variable, and $x$ or $y$ (or both)\\ is a  liaison variable of $Q$ that occurs\\ at some  non-primary-lhs \\ position of $\alpha$} \\\\
		R8: & $(\dep,Q\cup\set{R(\bar z,x,\bar y)})$ & $\lrhd \, ((\dep\setminus\dep_R)\cup\dep_{R'},Q\cup\set{R'(\bar z,\bar y)})$ & \makecell{if $x$ is an orphan variable of \\ $Q\cup\set{R(\bar z,x,\bar y)}$, and $R'$ does not \\ occur in $Q$}\\\\
		R9: & $(\dep,Q)$ & $\lrhd \, (\dep,Q')$ & \makecell{if a constant $c$ occurs in $Q$, \\some constant $c'\neq c$ occurs in $Q$, \\ and $Q'$ is obtained from $Q$ by\\  replacing all constants with $c$} \\\\
		R10: & $(\dep,Q\cup\set{R(\bar z,x,\bar y)})$ & $\lrhd \, (\dep,Q\cup\set{R(\bar z,c,\bar y)})$ & \makecell{if $x$ occurs at a non-primary-lhs\\ position and has another\\ occurrence at a primary-lhs\\ position of $R(\bar z,x,\bar y)$}\\
	\end{tabular}
	\caption{Rewrite rules for pairs consisting of an FD set $\dep$ with an LHS chain and an SJFCQs $Q$.}\label{fig:rules}
\end{figure*}

\medskip 
\noindent
\paragraph{Rewrite Rules and Final Pairs.}
Maslowski and Wijsen defined ten rewrite rules for CQs under primary keys~\cite{MaWi13}. In these rules, only the CQ is rewritten as the primary keys are implicitly encoded in the query. We generalize these rewrite rules for FDs with an LHS chain by rewriting pairs consisting of a set $\dep$ of FDs and a CQ $Q$.
Our rules are depicted in Figure~\ref{fig:rules}. We write $(\dep,Q)\lrhd (\dep',Q')$ for the fact that $(\dep',Q')$ can be obtained from $(\dep,Q)$ by applying one of the rewrite rules. We also write $(\dep,Q) \lrhd^+ (\dep',Q')$ for the fact that there is finite sequence of pairs $s = (\dep_0,Q_0),\ldots,(\dep_n,Q_n)$, where $(\dep,Q) = (\dep_0,Q_0)$ and $(\dep',Q') = (\dep_n,Q_n)$, such that $(\dep_i,Q_i) \lrhd (\dep_{i+1},Q_{i+1})$ for each $i \in \{0,\ldots,n-1\}$; we may say that $(\dep,Q) \lrhd^+ (\dep', Q')$ is witnessed by $s$, and call $s$ a rewriting sequence of $(\dep,Q)$.
Here are some further notes about the rewrite rules:
\begin{description}
\item[The Relation $>_\dep$.] It is a binary relation over atoms used in the rule R5 defined as follows:
\begin{enumerate}
	\item \rev{$R(\bar z,x,\bar y) >_\dep R'(\bar z,\bar y)$ if $x$ has two or more occurrences in $R(\bar z,x,\bar y)$}, and
	
	\item \rev{$R(\bar z,c,\bar y) >_\dep R'(\bar z,\bar y)$}.
\end{enumerate}
Note that $>_\dep$ is relative to a set $\dep$ of FDs with an LHS chain. However, since $\dep$ is always clear from the context, we will simply write $>$ without explicitly referring to a set of FDs.

\smallskip
\item[Primary-rhs Positions.] In the rule R6, we use the notation $\prhsvar{\alpha}{\dep}$ to denote the set of variables occurring in $\alpha$ at the positions $\{(R,A)\mid A\in Y_i\}$, where $R$ is the relation name of $\alpha$ and $R:X_i\rightarrow Y_i$ is the primary FD of $\dep_R$ w.r.t.~$Q$.

\smallskip 
\item[Orphan Variables.] In the rules R7 and R8, we use the notion of orphan variable. An \e{orphan} variable of a SJFCQ $Q$ (w.r.t.~a set $\dep$ of FDs) is a variable that occurs once in $Q$\rev{, }at a non-primary-lhs position (w.r.t.~$\dep$).

\smallskip
\item[Induced FD Set.] Observe that in the rules R5 and R8, we rewrite the pair $(\dep,Q)$ into a pair $(\dep',Q')$, where $Q'$ is obtained from $Q$ by replacing an atom over a relation $R$ of arity $n$ with another atom over a relation $R'$ of arity $n-1$. In this case, the rewriting induces a new FD set $\dep'$ (for all the other rules, we have $\dep'=\dep$). If we replace the atom $R(\bar z)$ in $Q$ with the atom $R'(\bar y)$ (over a fresh relation $R'$) to obtain $Q'$, then the FD set $\dep'$ is obtained from $\dep$ by removing all the $R$-FDs, and adding, for each $R$-FD $R:X\ra Y$, the FD $R':(X\setminus\set{A})\ra (Y\setminus\set{A})$, where $A$ is the attribute corresponding to the position of $R$ that was removed to obtain $R'$ (we assume that the other attribute names are unchanged in $R'$). We denote the FD set over $R'$ by $\dep_{R'}$, and so $\dep'=(\dep\setminus\dep_R)\cup\dep_{R'}$. If after removing the attribute $A$ some FD becomes trivial, i.e., has the empty set as its right-hand side, then we assume that this FD is removed from $\dep_{R'}$. Moreover, if two FDs in $\dep_{R'}$ have the same left-hand side (which can be the case if the attribute $A$ occurs on the left-hand side of some FD in $\dep_R$), then we assume that these FDs are merged into a single FD. Hence, the resulting FD set is canonical. It is easy to to see that, if the original FD set $\dep$ has an LHS chain, then $\dep'$ also has an LHS chain.
\end{description}

With the rewrite rules in place, we now generalize the definition of final queries of Maslowski and Wijsen~\cite{MaWi13} to pairs consisting of a set of FDs with an LHS chain and an SJFCQ.

\begin{definition}[\textbf{Final Pairs}]
	Consider a set $\dep$ of FDs with an LHS chain, and an SJFCQ $Q$. We say that the pair $(\dep,Q)$ is final if
	\begin{enumerate}
		\item $Q$ is not $\dep$-safe, and
		\item for every set $\dep'$ of FDs with an LHS chain, and an SJFCQ $Q'$ such that $(\dep,Q)\lrhd^+ (\dep',Q')$, the query $Q'$ is $\dep'$-safe. \hfill\markfull
	\end{enumerate}
\end{definition}

We now show that, starting from a pair of a set of FDs with an LHS chain and an SJFCQ, we can always reach a final pair by applying rewrite rules.

\begin{lemma}\label{lem:final-pairs-via-rewriting}
	Consider a set $\dep$ of FDs with an LHS chain, and an SJFCQ $Q$ such that $Q$ is not $\dep$-safe. There is a set $\hat{\dep}$ of FDs with an LHS chain, and an SJFCQ $\hat{Q}$ such that $(\hat{\dep},\hat{Q})$ is final and $(\dep,Q)\lrhd^+ (\hat{\dep},\hat{Q})$.
\end{lemma}

\begin{proof}
We first recall two simple technical properties, already proved by Maslowski and Wijsen~\cite{MaWi13} for primary keys, that can be easily extended to FDs with an LHS chain. \rev{In particular, for every set $\dep$ of FDs with an LHS chain,
\begin{itemize}
	\item it is {\em not} the case that $(\dep,Q) \lrhd^+ (\dep',Q)$ for some $\dep'$, and
	\item the set $S = \set{(\dep',Q')\mid (\dep,Q) \lrhd^+ (\dep',Q')}$ is finite.
	\end{itemize}
 }
To prove this, for a pair $(\dep',Q')\in S$, let $n,n'$ be the sum of arities of the relation names occurring in $Q$ and $Q'$, respectively. \rev{We also denote by $m$ (resp.,~$m'$) the sum over all the atoms $\alpha$ of $Q$ (resp.,~$Q'$) of the number of positions $(R,A)$ (where $R$ is the relation name of $\alpha$) such that $\alpha[A]\in \ins{V}$.
If $(\dep,Q) \lrhd (\dep',Q')$, then one of the following holds:
\begin{enumerate}
	\item $n'<n$, or
	\item $n'=n$ and $m'<m$, or
	\item $n'=n$, $m=m'$, and $\const{Q'}\subsetneq\const{Q}$,
\end{enumerate}}
Moreover, the set of FDs $\dep'$ is uniquely determined by the applied rewrite rule. This immediately implies the desired properties.
From the above discussion, we conclude that each rewriting sequence of $(\dep,Q)$ is finite, and there are finitely many distinct rewriting sequences of $(\dep,Q)$. Among all these rewriting sequences of $(\dep,Q)$ that lead to a pair $(\dep',Q')$ such that $Q'$ is not $\dep'$-safe, we consider the longest one denoted by $s = (\dep_0,Q_0),\ldots,(\dep_n,Q_n)$. It is clear that the pair $(\dep_n,Q_n)$ is final, and the claim follows with $(\hat{\dep},\hat{Q}) = (\dep_n,Q_n)$.
\end{proof}

Note that Lemma~\ref{lem:final-pairs-via-rewriting} only applies to non-safe queries as it is not possible to obtain a pair $(\dep',Q')$ such that $Q'$ is not $\dep'$-safe via a rewriting sequence of $(\dep,Q)$ for some query $Q$ that is $\dep$-safe.

\medskip 
\noindent
\paragraph{Rewrite Rules and Relative Frequency.}
We proceed to show that whenever a pair $(\dep,Q)$ can be rewritten via a rewrite rule into a pair $(\dep',Q')$, then there is a polynomial-time Turing reduction (i.e.,~Cook reduction) from $\prob{RelFreq}(\dep',Q')$ to $\prob{RelFreq}(\dep,Q)$. 

\def\lemmacombinedhelp{
\rev{For a database $D$, a set $\dep$ of FDs with an LHS chain, and an SJFCQ $Q$,
\[
\rfreq{Q}{D,\dep}\ =\ \rfreq{Q}{D\setminus (D_{\mathsf{conf}}^{\dep,Q}\cup D_{\mathsf{ind}}^{\dep,Q}),\dep}\times\frac{\card{\rep{D\setminus D_{\mathsf{conf}}^{\dep,Q}}{\dep}}}{\card{\rep{D}{\dep}}}.\]}
}

\def\lemmarewriterules{
	Consider a set $\dep$ of FDs with an LHS chain, and an SJFCQ $Q$. For every set $\dep'$ of FDs with an LHS chain, and an SJFCQ $Q'$, $(\dep,Q) \lrhd (\dep',Q')$ implies that $\prob{RelFreq}(\dep',Q')$ is Cook reducible to $\prob{RelFreq}(\dep,Q)$.
}

\begin{lemma}\label{lem:rewrite-rules-relfreq}
	\lemmarewriterules
\end{lemma}

\begin{proof}
	We prove this by considering each one of the rewrite rules given in Figure~\ref{fig:rules} separately. To this end, given a database $D$, we need to construct in polynomial time a database $E$ such that the value $\rfreq{Q'}{D,\dep'}$ can be computed from $\rfreq{Q}{E,\dep}$. 
	The proofs for the rewrite rules R1-R4, R8, and R9, which are deferred to the appendix, are similar to the corresponding proofs for primary keys~\cite{MaWi13}, but we have to take into account the FDs in the primary prefix. To this end, we use the databases $D_{\mathsf{conf}}^{\dep,Q}$ and $D_{\mathsf{ind}}^{\dep,Q}$.
	For the remaining rules (R5, R6, R7, and R10), our proofs, which we give below, are considerably more involved. We elaborate on the differences between the primary key setting and our setting in the corresponding proof.

\rev{In what follows, we assume the following order over attributes $A,A'$ of the same relation name:
\begin{itemize}
	\item[-] $A=A'$ if both $A$ and $A'$ do not occur in any FD of $\dep$, or if 
	$A$ and $A'$ occur in the same FD of $\dep$ on the same side (i.e., both on the left-hand side or both on the right-hand side), and	
	\item[-] $A>A'$ if $A'$ occurs in some FD of $\dep$ while $A$ does not occur in any FD, or if both $A$ and $A'$ occur in FDs of $\dep$, and 
	$A$ occurs later in the chain (i.e., in an FD that appears later in the chain or on the right-hand side of the FD in which $A'$ appears on the left-hand side).
\end{itemize}
Observe that an attribute might occur in more than one FD, in which case we consider the \e{first} occurrence of the attribute in the above conditions.}
	
\medskip

\rev{\noindent\paragraph{\underline{Rewrite Rule R5}: $(\dep,Q\cup\set{\alpha_1})\lrhd \, ((\dep\setminus\dep_R)\cup\dep_{R'},Q\cup\set{\alpha_2})$ if $\alpha_1>_\dep\alpha_2$ and the relation name of $\alpha_2$ does not occur in $Q$.} }
\smallskip

\noindent The proof for rule R5 in the case of primary keys is rather straightforward~\cite{MaWi13}, as all the variables or constants that occur at primary-lhs positions (resp., non-primary-lhs positions) correspond to attributes on the left-hand side (resp., right-hand side) of the key. This is not the case in our setting, because two non-primary-lhs positions might correspond to attributes that occur in two different FDs (possibly on two different sides -- one on the left-hand side of an FD and one on the right-hand side of another FD). Hence, here, when we remove an occurrence of some variable, we have to carefully choose which occurrence of this variable to remove. We proceed by considering all the possible cases according to the definition of the binary relation $>_\dep$.

\smallskip

\noindent \textbf{Case (1): $R(\bar z,x,\bar y)>R'(\bar z,\bar y)$ if $x$ has two or more occurrences in $R(\bar z,x,\bar y)$.} Let $\alpha_1=R(\bar z,x,\bar y)$ and $\alpha_2=R'(\bar z,\bar y)$. Assume that the variable $x$ occurs at positions $(R, A), (R,A')$ of $R(\bar z,x,\bar y)$. \rev{Assume also, without loss of generality, that $A\ge A'$. As mentioned above, an induced FD set might merge several FDs into a single FD; however, in the remainder of this proof for R5 we assume, for clarity of presentation, that the FDs are not merged. Clearly, this has no impact on the satisfaction of the constraints.}

Given a database $D$ over $\ins{S}'$, let $E$ be the database over $\ins{S}$ that contains all the facts of $D$ over relation names $R''$ such that $R''\neq R'$, and, for every $R'$-fact $f$, the $R$-fact $f'$ defined as follows:
\[
f'[B]\ =\ \begin{cases}
f[B] & \mbox{if $B\neq A$}\\ 
f[A'] & \mbox{otherwise.}
\end{cases}
\]
If $A$ does not occur in any FD of $\dep_{R}$, then clearly its addition does not affect the satisfaction of the constraints. Moreover, since the value of attribute $A$ in every $R$-fact of $E$ is the same as the value of attribute $A'$ that is associated with the same variable in the query, the additional attribute has no impact on the entailment of the query. Therefore, in this case, we have that
\[
\rfreq{Q\cup\{\alpha_2\}}{D,\dep'}\ =\ \rfreq{Q\cup\{\alpha_1\}}{E,\dep}.
\]

Assume now that $A$ occurs in an FD of $\dep_{R}$ and the first occurrence of $A$ is in the FD $R:X_k\ra Y_k$. Since $A\ge A'$, $A'$ occurs in an FD $R:X_\ell \ra Y_\ell$ of $\dep_{R}$, and one of the following holds:
\begin{itemize}
	\item[-] $\ell<k$,
	\item[-] $\ell=k$ and $A'\in X_\ell$, $A\in X_\ell$,
	\item[-] $\ell=k$ and $A'\in Y_\ell$, $A\in Y_\ell$,
	\item[-] $\ell=k$ and $A'\in X_\ell$, $A\in Y_\ell$.
\end{itemize}

Let $f,g$ be two $R'$-facts of $D$, and let $f',g'$ be the corresponding $R$-facts of $E$. \rev{We will show that $\{f,g\}\models \dep'_{R'}$ if and only if $\{f',g'\}\models \dep_R$.} If $f$ and $g$ violate an FD $R':X_p\rightarrow Y_p$ in $\dep'_{R'}$, then they agree on all the attributes of $X_p$ and disagree on at least one attribute $B\in Y_p$. \rev{Then:}
\begin{enumerate}
	\item If the FD $R:X_p\rightarrow Y_p$ occurs in $\dep_{R}$, then clearly $f'$ and $g'$ jointly violate this FD and $\{f',g'\}\not\models\dep_{R}$.
	
	\item If the FD $R:(X_p\cup\set{A})\rightarrow Y_p$ occurs in $\dep_{R}$ and $f'$ and $g'$ agree on the value of $A$, then $f'$ and $g'$ jointly violate this FD, and $\{f',g'\}\not\models\dep_{R}$. If $f'$ and $g'$ disagree on the value of attribute $A$, then they also disagree on the value of $A'$, which implies that $f[A']\neq g[A']$. In this case, the attribute $A'$ cannot appear on the left-hand side of an FD $R:X_\ell \ra Y_\ell$ with $\ell\le p$, as in this case, $f$ and $g$ do not agree on the left-hand side of $R':X_p\rightarrow Y_p$, which contradicts the fact that $f$ and $g$ jointly violate this FD. If $A'$ occurs on the right-hand side of $R:X_\ell \ra Y_\ell$ with $\ell< p$, then $f'$ and $g'$ agree on the left-hand side of this FD (since $X_\ell\subseteq X_p$) and disagree on its right-hand side, and $\{f',g'\}\not\models\dep_{R}$. \rev{Note that, in this case, $A\not\in X_\ell$ since $A'\in Y_\ell$ and $A\ge A'$.}
	
	\item If the FD $R:X_p\rightarrow (Y_p\cup\set{A})$ occurs in $\dep_{R}$, then clearly $f'$ and $g'$ agree on all the attributes of $X_p$ and disagree on at least one attribute of $Y_p\cup\set{A}$ (since $f$ and $g$ disagree on at least one attribute of $Y_p$), and $\{f',g'\}\not\models\dep_{R}$.
\end{enumerate}
We conclude that if $\{f,g\}\not\models\dep'_{R'}$, then $\{f',g'\}\not\models\dep_{R}$.

Next, assume that $\{f,g\}$ satisfies all the FDs of $\dep'_{R'}$. Let $R':X_p\ra Y_p$ be such an FD. \rev{Then:}
\begin{enumerate}
	\item If the FD $R:X_p\rightarrow Y_p$ occurs in $\dep_{R}$, then clearly $f'$ and $g'$ jointly satisfy this FD.
	
	\item If the FD $R:(X_p\cup\set{A})\rightarrow Y_p$ occurs in $\dep_{R}$, then again $f'$ and $g'$ jointly satisfy this FD, as they agree on the values of all the attributes in $Y_p$.
	
	\item If the FD $R:X_p\rightarrow (Y_p\cup\set{A})$ occurs in $\dep_{R}$ and $f'$ and $g'$ either disagree on some attribute of $X_p$ or agree on the value of attribute $A$, then $f'$ and $g'$ jointly satisfy this FD. If $f'$ and $g'$ agree on all the attributes of $X_p$ and disagree on the value of $A$, then they also disagree on the value of $A'$ and so $f[A']\neq g[A']$; hence, the attribute $A'$ does not appear on the left-hand side of an FD $R:X_\ell \ra Y_\ell$ with $\ell\le p$ (since $X_\ell\subseteq X_p$ and $f$ and $g$ also agree on the values of all the attributes in $X_p$). In this case, $A'$ occurs on the right-hand side of an FD $R:X_\ell \ra Y_\ell$ with $\ell< p$ or on the right-hand side of the FD $R:X_p\ra (Y_p\cup\set{A})$; thus, $f$ and $g$ agree on the left-hand side of the corresponding FD in $\dep'_{R'}$ and disagree on its right-hand side, which contradicts the fact that $f$ and $g$ satisfy all the FDs of $\dep'_{R'}$.
\end{enumerate}
If the set of FDs $\dep_{R}$ contains an additional FD $R:X_p\ra\set{A}$ that does not have a corresponding FD in $\dep'_{R'}$ (as after removing the attribute $A$ the FD becomes trivial), and $f'$ and $g'$ either disagree on the value of some attribute in $X_p$ or agree on the value of attribute $A$, then they clearly satisfy this FD. If $f'$ and $g'$ agree on the value of all the attributes in $X_p$ and $f'[A]\neq g'[A]$, then $f'[A']\neq g'[A']$ and so $f[A']\neq g[A']$. Again, it cannot be the case that $A'\in X_\ell$ for some $\ell\le p$, as $X_\ell\subseteq X_p$ and $f'$ and $g'$ agree on the value of all the attributes in $X_p$. Hence, $A'\in Y_\ell$ for some $\ell<p$. In this case, $f$ and $g$ agree on the values of all the attributes in $X_p$ (and so also in $X_\ell$) as $f[B]=f'[B]$ and $g[B]=g'[B]$ for every $B\in X_p$ (note that $A\not\in X_p$), and disagree on the value of the attribute $A'$ in $Y_\ell$; hence, $f$ and $g$ jointly violate the FD $R':X_\ell\ra Y_\ell$, which is a contradiction to the fact that $\{f,g\}\models\dep'_{R'}$.
We conclude that if $\{f,g\}\models\dep'_{R'}$, then $\{f',g'\}\models\dep_{R}$. Therefore, we have that
\[
\card{\rep{D}{\dep'}}\ =\ \card{\rep{E}{\dep}}.
\]
Moreover, since every $R$-fact in $E$ uses the same value at positions $(R,A)$ and $(R,A')$,
every $R'$-fact $f$ of $D$ is such that $f\models\alpha_2$ iff the corresponding $R$-fact $f'$ of $E$ is such that $f'\models\alpha_1$. Thus, the additional attribute has no impact on the entailment of the query, and we have that
\[
\card{\rep{D}{\dep',Q\cup\{\alpha_2\}}}\ =\ \card{\rep{E}{\dep,Q\cup\{\alpha_1\}}}.
\]
It then follows that
\[
\rfreq{Q\cup\{\alpha_2\}}{D,\dep'}\ =\ \rfreq{Q\cup\{\alpha_1\}}{E,\dep}.
\]

\smallskip

\noindent \textbf{Case (2): $R(\bar z,c,\bar y)>R'(\bar z,\bar y)$.}
Let $\alpha_1=R(\bar z,c,\bar y)$ and $\alpha_2=R'(\bar z,\bar y)$. As in the previous case, we denote by $\dep$ the set of FDs over the schema $\ins{S}$ that contains the relation name $R$, and by $\dep'$ the set of FDs over the schema $\ins{S}'$ that contains the relation name $R'$. Assume that the constant $c$ occurs at the (non-primary-lhs) position $(R,A)$ of $\alpha_1$.
Given a database $D$ over $\ins{S}'$, let $E$ be the database over $\ins{S}$ that contains all the facts of $D$ over relation names $R''$ such that $R''\neq R'$, and, for every $R'$-fact $f$ of $D$, the $R$-fact defined as follows:
\[
f'[B]\ =\ \begin{cases}
f[B] & \mbox{if $B\neq A$}\\ 
c & \mbox{otherwise}
\end{cases}
\]
where $c$ is a constant (the same constant is used for all facts).
Clearly,
\[
\card{\rep{D}{\dep'}}\ =\ \card{\rep{E}{\dep}}.
\]
This holds since all the $R$-facts of $E$ agree on the value of the additional attribute $A$; hence, this attribute does not affect the satisfaction of the FDs -- two $R'$-facts $f$ and $g$ of $D$ violate an FD of $\dep'_{R'}$ if and only if the corresponding $R$-facts $f'$ and $g'$ of $E$ violate an FD of $\dep_{R}$. Moreover, since all the $R$-facts of $E$ agree with the atom $\alpha_1$ on the constant at the position $(R,A)$, this additional attribute $A$ has no impact on the entailment of the query, and we have that
\[
\card{\rep{D}{\dep',Q\cup\{\alpha_2\}}}\ =\ \card{\rep{E}{\dep,Q\cup\{\alpha_1\}}}.
\]
We can then conclude that
\[
\rfreq{Q\cup\{\alpha_2\}}{D,\dep'}\ =\ \rfreq{Q\cup\{\alpha_1\}}{E,\dep}.
\]

\medskip

\rev{\noindent\paragraph{\underline{Rewrite Rule R6}: $(\dep,Q)\lrhd \, (\dep,Q_{x\mapsto c})$ if there is an atom $\alpha \in Q$ with $\pvar{\alpha}{\dep}=\emptyset$ and $x\in\prhsvar{\alpha}{\dep}$, and one of the variables of $\var{\alpha}$ is a liaison variable of $Q$.} }
	 \smallskip
	 
	 \noindent
  \rev{In the case of primary keys, the rewrite rule R6 applies to liaison variables of $Q$; that is, if $x$ is a liaison variable of $Q$, then we can rewrite $Q$ into $Q_{x\mapsto c}$. In our case, the requirement that one of the variables $y$ of $\alpha$ is a liaison variable of $Q$ remains, but we may consider a different variable $x$ (that we have to carefully choose) in the rewriting, and we need to add additional requirements on the positions corresponding to the attributes that occur in the primary FD of the relation name of $\alpha$. Assume $\alpha=R(\bar z)$. Assume also that the variable $x$ occurs at the position $(R,A)$ of $\alpha$, and the liaison variable $y$ of $\var{\alpha}$ occurs at the position $(R,A')$ (it can be the case that $x=y$ and so $A=A'$). We prove the following in the appendix:}
  
\begin{lemma}\label{lem:combined_help}
    \lemmacombinedhelp
\end{lemma}

Lemma~\ref{lem:combined_help} implies that:
  \rev{
  \[
\rfreq{Q_{x\mapsto c}}{D,\dep}\ =\ \rfreq{Q_{x\mapsto c}}{D\setminus (D_{\mathsf{conf}}^{\dep,Q_{x\mapsto c}}\cup D_{\mathsf{ind}}^{\dep,Q_{x\mapsto c}}),\dep}\times\frac{\card{\rep{D\setminus D_{\mathsf{conf}}^{\dep,Q_{x\mapsto c}}}{\dep}}}{\card{\rep{D}{\dep}}}.
\]
  Recall that the ratio 
  $\frac{\card{\rep{D\setminus D_{\mathsf{conf}}^{\dep,Q_{x\mapsto c}}}{\dep}}}{\card{\rep{D}{\dep}}}$
  can be computed in polynomial time due to Proposition~\ref{pro:counting-repairs-lhs-fp}. We denote $D'=D\setminus (D_{\mathsf{conf}}^{\dep,Q_{x\mapsto c}}\cup D_{\mathsf{ind}}^{\dep,Q_{x\mapsto c}})$ and show that 
  \[
	\rfreq{Q_{x\mapsto c}}{D',\dep}\ =\ \rfreq{Q}{E,\dep}.
	\]
  for some database $E$ that we define next.}

  \rev{Let $E$ be the database that contains all the facts of $D'$ over relation names $R'$ such that $R'\neq R$, and, for every $R$-fact $f$, the $R$-fact
	\[
	f'[B]\ =\ \begin{cases}
	f[B] & \mbox{if $B\neq A'$}\\ 
	f[B] & \mbox{if $B=A'$ and $f[A]=c$}\\
	\rho(f[B]) & \mbox{if $B=A'$ and $f[A]\neq c$}
	\end{cases}
	\]
	where $\rho$ is a function that maps constants $b\in\adom{D}$ to constants $\rho(b)\not\in\adom{D}$ such that $\rho(a)\neq \rho(b)$ for $a\neq b$. Recall that, for every FD $R:X\rightarrow Y$ in the primary prefix of $\dep_R$ w.r.t.~$Q_{x\mapsto c}$ and for all $R$-facts $f$ of $D'$, we have that $f[A]=\alpha[A]$ for each $A\in X\cup Y$. Hence, every pair $f,g$ of $R$-facts of $D'$ satisfies all the FDs in the primary prefix of $\dep_R$ w.r.t.~$Q_{x\mapsto c}$. We conclude that every pair $f,g$ of $R$-facts of $D'$ satisfies all the FDs in the primary prefix of $\dep_R$ w.r.t.~$Q$, as the primary prefix of $\dep_R$ w.r.t.~$Q$ is contained in the primary prefix of $\dep_R$ w.r.t.~$Q_{x\mapsto c}$. Since in $E$ we only modify values in the attribute associated with the liaison variable $y$, we also have that every pair $f',g'$ of $R$-facts of $E$ satisfies all the FDs in the primary prefix of $\dep_R$ w.r.t.~$Q$.} 

 \rev{Clearly, two $R$-facts $f$ and $g$ in $D'$ that disagree on the value of some attribute on the left-hand side of the primary FD of $\dep_R$ w.r.t.~$Q$, jointly satisfy $\dep_R$, as this attribute belongs to the left-hand side of every FD that appears later in the chain. Since $\pvar{\alpha}{\dep}=\emptyset$, we have that $\alpha[A]=a$ for some constant $a$, for every attribute $A$ that occurs on the left-hand side of this FD. Since in $E$ we do not modify the values of attributes associated with constants, we also have that the corresponding $R$-facts $f'$ and $g'$ in $E$ jointly satisfy $\dep_R$. If $f$ and $g$ agree on the values of all the attributes on the left-hand side of the primary FD of $\dep_R$ w.r.t.~$Q$, and disagree on the value of attribute $A$, then they jointly violate the primary FD of $\dep_R$ w.r.t.~$Q$, and so do the corresponding facts $f'$ and $g'$ in $E$, since we do not modify the value of attribute $A$ if $A\neq A'$, and if $A=A'$, then we have that $f[A]=g[A]$ if and only if $f'[A]=g'[A]$ as we are only renaming constants. Finally, if $f$ and $g$ agree on the values of all the attributes on the left-hand side of the primary FD of $\dep_R$ w.r.t.~$Q$, and also agree on the value of attribute $A$, it is easy to verify that $f[A']=g[A']$ if and only if $f'[A']=g'[A']$, and since we only modify the value of $A'$ in $E$, we have that $f$ and $g$ satisfy $\dep_R$ if and only if $f'$ and $g'$ satisfy $\dep_R$. We conclude that 
	\[
	\card{\rep{D'}{\dep}}\ =\ \card{\rep{E}{\dep}}.
	\]
	We now prove that a repair $J$ of $D'$ entails $Q_{x\mapsto c}$ if and only if the corresponding repair $J'$ of $E$ (that contains the same facts over the relation names $R$ such that $R'\neq R$ as $J$, and a fact $f'$ over $R$ if and only if the corresponding $R$-fact $f$ is in $J$) entails $Q$.}
	
	\rev{Let $J$ be a repair of $D'$ that entails $Q_{x\mapsto c}$. 
	Let $h$ be a homomorphism from $Q_{x\mapsto c}$ to $J$.
	It must be the case that $f=R(h(\bar z))$ uses the value $c$ at the position $(R,A)$. Therefore, the corresponding fact $f'\in J'$ is such that $f'=f$. Moreover, $J'$ also contains every fact $R'(h(\bar z'))$ for $R'\neq R$; thus, $h$ is a homomorphism from $Q_{x\mapsto c}$ to $J'$. For $h'$ such that $h'(y)=h(y)$ for all variables $y$ that occur in $Q_{x\mapsto c}$ and $h'(x)=c$, we then have that $h'$ is a homomorphism from $Q$ to $J'$, and that concludes our proof of the first direction.}
	
	\rev{Next, let $J'$ be a repair of $E$ that entails $Q$. Let $h$ be a homomorphism from $Q$ to $J'$. Since $y$ is a liaison variable, it has at least one more occurrence in $Q$, at the position $(R',A'')$ of some atom $\alpha'=R'(\bar z')$ (it can be the case that $\alpha'=\alpha$ with $A'\neq A''$). Moreover, the fact $R'(h(\bar z'))$ uses a constant $b\in\adom{D}$ in the attribute corresponding to this position, as we have only changed values in the position $(R,A')$. Therefore, the fact $R(h(\bar z))$ also uses a constant $b\in\adom{D}$ at the position $(R,A')$. The only $R$-facts $f$ of $E$ that use a constant from $\adom{D}$ at the position $(R,A')$ are those satisfying $f[A]=c$; hence, we have that $h(x)=c$, and it follows that $h$ is a homomorphism from $Q_{x\mapsto c}$ to $J$.
	We conclude that
	\[
	\card{\rep{D'}{\dep,Q_{x\mapsto c}}}\ =\ \card{\rep{E}{\dep,Q}}
	\]
	and therefore
	\[
	\rfreq{Q_{x\mapsto c}}{D',\dep}\ =\ \rfreq{Q}{E,\dep}.
	\]}

	\medskip

\rev{\noindent\paragraph{\underline{Rewrite Rule R7}: $(\dep,Q)\lrhd \, (\dep,Q_{y\mapsto x})$ if there is an atom $\alpha\in Q$ with $x,y\in\var{\alpha}$, none of $x,y$ is an orphan variable, and $x$ or $y$ (or both) is a  liaison variable of $Q$ that occurs at some  non-primary-lhs position of $\alpha$.}}
\smallskip

\noindent Similarly to the case of the rule R5, the fact that the variables $x$ and $y$ might appear in two different FDs makes the proof more involved than the corresponding proof for primary keys~\cite{MaWi13}. Moreover, we have added an additional condition on the variables $x$ and $y$ that does not appear in the corresponding rule for primary keys -- none of the variables $x,y$ is an orphan variable. In our case, this additional condition is necessary for the reduction.

Let $\alpha=R(\bar z)$ be an atom of $Q$ such that $x,y\in\var{\alpha}$. Assume that the variables $x$ and $y$ occur at positions $(R,A)$ and $(R,A')$, respectively.
As mentioned in the proof for Case 1 of R5, if one of the positions is a non-primary-lhs position and $A\ge A'$, then the position $(R,A)$ is a non-primary-lhs position. Hence, we assume, without loss of generality, that $A\ge A'$, and we know that the position $(R,A)$ is a non-primary-lhs position. Since none of $x,y$ is an orphan variable, it holds that $x$ is a liaison variable of $Q$. 

Let $E$ be the database that contains all the facts of $D$ over relation names $R'$ such that $R'\neq R$, and, for every $R$-fact $f$ of $D$, the following $R$-fact:
\[
f'[B]\ =\ \begin{cases}
f[B] & \mbox{if $B\neq A$}\\ 
c & \mbox{if $B=A$, $f[A]=f[A']$, and $f[A]=c$}\\
\rho(c) & \mbox{if $B=A$, $f[A]\neq f[A']$, and $f[A]=c$}
\end{cases}
\]
where $\rho$ is a function that maps constants $a\in\adom{D}$ to constants $\rho(a)\not\in\adom{D}$ such that $\rho(a)\neq \rho(b)$ for $a\neq b$.

If $A$ does not occur in any FD of $\dep$, then clearly its value does not affect the satisfaction of the constraints, and we have that
\[
\card{\rep{D}{\dep}}\ =\ \card{\rep{E}{\dep}}.
\]

Now, assume that the attribute $A$ does occur in some FD $R:X_k\ra Y_k$ of $\dep_{R}$. In this case, since $A\ge A'$, the attribute $A'$ also occurs in an FD $R:X_\ell \ra Y_\ell$ of $\dep_{R}$, and, as mentioned in Case 1 of the proof for rule R5, one of the following holds: 
\begin{itemize}
	\item[-] $\ell<k$, 
	\item[-] $\ell=k$ and $A'\in X_\ell$, $A\in X_\ell$, \item[-] $\ell=k$ and $A'\in Y_\ell$, $A\in Y_\ell$,
	\item[-] $\ell=k$ and $A'\in X_\ell$, $A\in Y_\ell$.
\end{itemize}

Let $f$ and $g$ be two $R$-facts of $D$, and let $f'$ and $g'$ be the corresponding $R$-facts in $E$. \rev{We will show that $\{f,g\}\models \dep_{R}$ if and only if $\{f',g'\}\models \dep_R$.} If $f$ and $g$ violate an FD $R:X_p\rightarrow Y_p$ in $\dep_{R}$, then they agree on all the attributes of $X_p$ and disagree on at least one attribute $B\in Y_p$. \rev{Then:}
\begin{enumerate}
	\item If $A\not\in (X_p\cup Y_p)$, then clearly $\set{f',g'}\not\models\dep_R$.
	
	\item If $f[A]=f[A']$ and $g[A]=g[A']$ then, by definition, $f'=f$ and $g'=g$, and $\set{f',g'}\not\models\dep_R$.
	
	\item If $f[A]\neq f[A']$ and $A\in Y_p$, then $\set{f',g'}\not\models\dep_R$ because $f',g'$ agree on the values of all the attributes in $X_p$. Moreover, since $f$ and $g$ disagree on the value of attribute $B$, the facts $f'$ and $g'$ also disagree on the value of $B$. If $B\neq A$, then this is clearly the case, whereas if $B=A$, then this holds by the definition of $\rho$.
	
	\item If $f[A]\neq f[A']$, $A\in X_p$, and $f[A']\neq g[A']$, then it cannot be the case that $A'\in X_\ell$ for some $\ell\le p$. Otherwise, since $X_\ell\subseteq X_p$, the facts $f,g$ disagree on the value of at least one attribute of $X_p$, which contradicts the fact that they jointly violate the FD $R:X_p\rightarrow Y_p$. Thus, $A'\in Y_\ell$ for some $\ell<p$. In this case, $f',g'$ jointly violate the FD $R:X_\ell\ra Y_\ell$ and $\set{f',g'}\not\models\dep_R$. Note that $A\not\in X_\ell$ because $A\ge A'$.
	
	\item If $f[A]\neq f[A']$, $A\in X_p$, and $f[A']=g[A']$, then $g[A]\neq g[A']$ because $f, g$ agree on every value of every attribute in $X_p$, including the value of attribute $A$. Hence, if $f[A]\neq f[A']$, then $g[A]\neq g[A']$, and by the definition of $\rho$, the fact that $f[A]=g[A]$ implies that $f'[A]=g'[A]$. Therefore, $f',g'$ agree on the values of all the attributes in $X_p$ and disagree on the value of attribute $B\in Y_p$, and thus $\set{f',g'}\not\models\dep_R$.
\end{enumerate}
We conclude that if $\{f,g\}\not\models\dep_{R}$, then $\{f',g'\}\not\models\dep_{R}$.

Next, assume that $\{f,g\}$ satisfies all the FDs of $\dep_{R}$. Let $R:X_p\ra Y_p$ be such an FD. \rev{Then:}
\begin{enumerate}
	\item If $A\not\in (X_p\cup Y_p)$, then clearly $f'$ and $g'$ also satisfy this FD.
	
	\item If $f[A]=f[A']$ and $g[A]=g[A']$ then, by definition, $f'=f$ and $g'=g$, and again $f'$ and $g'$ jointly satisfy this FD.
	
	\item If $f[A]\neq f[A']$, $A\in X_p$, and $f[B]\neq g[B]$ for some $B\in X_p$, then $f'$ and $g'$ jointly satisfy this FD, as they also disagree on the value of attribute $B$. This clearly holds if $B\neq A$, and if $B=A$, then $f[A]\neq g[A]$, which implies that $f'[A]\neq g'[A]$ according to the definition of $\rho$.
	
	\item If $f[A]\neq f[A']$, $A\in X_p$, and $f[B]=g[B]$ for every $B\in X_p$, then $f'$ and $g'$ jointly satisfy this FD, as they agree on the values of all the attributes in $Y_p$.
	
	\item If $f[A]\neq f[A']$, $A\in Y_p$, and $f[B]\neq g[B]$ for some $B\in X_p$, then $f'$ and $g'$ also disagree on the value of attribute $B\in X_p$ and satisfy this FD.
	
	\item If $f[A]\neq f[A']$, $A\in Y_p$, and $f[B]= g[B]$ for every $B\in X_p$, then $f[A']=g[A']$ as $A'$ is either in $X_p$ (in which case clearly $f[A']=g[A']$) or in $Y_\ell$ for some $\ell\le p$ (in which case $f[A']=g[A']$, as otherwise, $f$ and $g$ jointly violate the FD $R:X_\ell\rightarrow Y_\ell$). Since $f$ and $g$ agree on the value of all the attributes in $X_p$, they also agree on the value of all the attributes in $Y_p$, including $A$. Hence, $f[A]=g[A]$ and $f[A']=g[A']$, which implies that $f'[A]=g'[A]$. Thus, $f'$ and $g'$ agree on all the values of the attributes in $X_p\cup Y_p$, and satisfy this FD.
\end{enumerate}
We conclude that if $\{f,g\}\models\dep_{R}$, then $\{f',g'\}\models\dep_{R}$. Therefore,
\[
\card{\rep{D}{\dep}}\ =\ \card{\rep{E}{\dep}}.
\]
We now prove that a repair $J$ of $D$ entails $Q_{y\mapsto x}$ iff the corresponding repair $J'$ of $E$ entails $Q$.

Let $J$ be a repair of $D$ that entails $Q_{y\mapsto x}$.
Let $h$ be a homomorphism from $Q_{y\mapsto x}$ to $J$. It must be the case that $f=R(h(\bar z))$ uses the same value for attributes $A$ and $A'$, as both are associated with the same variable in $Q_{y\mapsto x}$. Therefore, the corresponding fact $f'\in J'$ is such that $f'=f$. Moreover, $J'$ also contains every fact $R'(h(\bar z'))$ for $R'\neq R$; thus, $h$ is also a homomorphism from $Q_{y\mapsto x}$ to $J'$. For $h'$ that satisfies $h'(z)=h(z)$ for every variable $z$ that occurs in $Q_{y\mapsto x}$ and $h'(y)=h(x)$, we then have that $h'$ is a homomorphism from $Q$ to $J'$, and that concludes our proof of the first direction.

Next, let $J'$ be a repair of $E$ that entails $Q$. Let $h$ be a homomorphism from $Q$ to $J'$. Since $x$ is a liaison variable it has at least one more occurrence in $Q$, at position $(R',A'')$ of some atom $\alpha'=R'(\bar z')$ (it can be the case that $\alpha'=\alpha$ with $A\neq A''\neq A'$). Moreover, the fact $R'(h(\bar z'))$ uses a constant $b\in\adom{D}$ in the attribute $A''$. Therefore, the fact $R(h(\bar z))$ also uses a constant $b\in\adom{D}$ at the position $(R,A)$. The only 
$R$-facts in $E$ that use constants from $\adom{D}$ at position $(R,A)$ are those that use the same value for attributes $A$ and $A'$; hence, we have that $h(x)=h(y)$, and it follows that $h$ is a homomorphism from $Q_{y\mapsto x}$ to $J$.
We conclude that
\[
\card{\rep{D}{\dep,Q_{y\mapsto x}}}\ =\ \card{\rep{E}{\dep,Q}}
\]
and thus
\[
\rfreq{Q_{y\mapsto x}}{D,\dep}\ =\ \rfreq{Q}{E,\dep}.
\]

\medskip

\rev{\noindent\paragraph{\underline{Rewrite Rule R10}: $(\dep,Q\cup\set{R(\bar z,x,\bar y)})\lrhd \, (\dep,Q\cup\set{R(\bar z,c,\bar y)})$ if $x$ occurs at a non-primary-lhs position and has another occurrence at a primary-lhs position of $R(\bar z,x,\bar y)$.} }
\smallskip

\noindent This rule is a generalization of the corresponding rule for primary keys~\cite{MaWi13}. While the rule R10 for primary keys considers binary relation names where the same variable occurs at both positions, our rule considers arbitrary relation names where the variable $x$ occurs both at a primary-lhs position and a non-primary-lhs position.

Suppose that the variable $x$ occurs at the (non-primary-lhs) position $(R,A)$ of $R(\bar z,x,\bar y)$. Assume also that $x$ occurs at the position $(R,A')$, which is a primary-lhs position. Since $R(\bar z,x,\bar y)$ has non-primary-lhs positions, the set of FDs $\dep_R$ has a primary FD. Moreover, since $(R,A')$ is a primary-lhs position, $A'$ occurs on the left-hand side of this primary FD. Now, let $E$ be the database that contains all the $R'$-facts of $D$ with $R'\neq R$, and, for every $R$-fact $f$ of $D$, the $R$-fact
\[
f'[B]\ =\ \begin{cases}
f[B] & \mbox{if $B\neq A$}\\ 
f[A'] & \mbox{if $B=A$ and $f[A]=c$}\\ 
\rho(f[A]) & \mbox{if $B=A$ and $f[A]\neq c$}\\ 
\end{cases}
\]
where $\rho$ is a function that maps constants $b\in\adom{D}$ to constants $\rho(b)\not\in\adom{D}$ such that $\rho(b)\neq \rho(c)$ for $b\neq c$. 

\rev{Let $f$ and $g$ be two $R$-facts in $D'$, and let $f'$ and $g'$ be the corresponding $R$-facts in $E$. If $f,g$ jointly violate (resp.,~satisfy) an FD $R:X_p\rightarrow Y_p$ in $\dep_{R}$, and $A\not\in (X_p\cup Y_p)$, then clearly $f',g'$ also jointly violate (resp.~satisfy) this FD as $f'[B]=f[B]$ and $g'[B]=g[B]$ for every attribute $B\neq A$. If $A\in X_p$ or $A\in Y_p$, then the FD $R:X_p\rightarrow Y_p$ is not in the primary prefix of $\dep_R$, and $A'\in X_p$. In this case, if $f[A']\neq g[A']$, then  clearly $f'[A']\neq g'[A']$; hence, $f,g$ jointly satisfy $R:X_p\rightarrow Y_p$ and the same holds for $f',g'$. 
If $f[A']=g[A']$, by the definition of $\rho$, it holds that $f'[A]=g'[A]$ if and only if $f[A]=g[A]$; therefore, $f,g$ jointly violate $R:X_p\rightarrow Y_p$ if and only if $f',g'$ jointly violate $R:X_p\rightarrow Y_p$.} We conclude that $\{f,g\}\models\dep_{R}$ if and only if $\{f',g'\}\models\dep_{R}$. Hence,
\[
\card{\rep{D'}{\dep}}\ =\ \card{\rep{E}{\dep}}.
\]

We now prove that a repair $J$ of $D$ entails $Q\cup\{R(\bar z,c,\bar y)\}$ if and only if the corresponding repair $J'$ of $E$ (that contains the same facts over the relations $R'\neq R$ as $J$, and a fact $f'$ over $R$ if and only if the corresponding fact $f$ is in $J$) entails $Q\cup\{R(\bar z,x,\bar y)\}$.

Let $J$ be a repair of $D$ that entails $Q\cup\{R(\bar z,c,\bar y)\}$. Let $h$ be a homomorphism from $Q\cup\{R(\bar z,c,\bar y)\}$ to $J$. It must be the case that $f=R(h(\bar z,c,\bar y))$ uses the value $c$ at the position $(R,A)$. Thus, the corresponding fact $f'\in J'$ is such that $f'[A]=f'[A']$. Moreover, $J'$ also contains every fact $R'(h(\bar z'))$ for other atoms $R'(\bar z')\in Q$; thus, $h$ is a homomorphism from $Q\cup\{R(\bar z,x,\bar y)\}$ to $J'$.

Next, let $J'$ be a repair of $E$ that entails $Q\cup\{R(\bar z,x,\bar y)\}$. Let $h$ be a homomorphism from $Q\cup\{R(\bar z,x,\bar y)\}$ to $J$. Clearly, $f'=R(h(\bar z,x,\bar y))$ is a fact of $E$ that satisfies $f'[A]=f'[A']$. The only $R$-facts $f'$ in $E$ that satisfy $f'[A]=f'[A']$ are those for which the corresponding fact $f$ in $D$ satisfies $f[A]=c$. Therefore, $h$ is a homomorphism from $Q\cup\{R(\bar z,c,\bar y)\}$ to $J$.
We conclude that
\[
\card{\rep{D}{\dep,Q\cup\{R(\bar z,c,\bar y)\}}}\ =\ \card{\rep{E}{\dep,Q\cup\{R(\bar z,x,\bar y)\}}}
\]
and therefore
\[
\rfreq{D}{\dep,Q\cup\{R(\bar z,c,\bar y)\}}\ =\ \rfreq{E}{\dep,Q\cup\{R(\bar z,x,\bar y)\}}.
\]

\smallskip

This completes the proof of Lemma~\ref{lem:rewrite-rules-relfreq} for rules R5, R7, and R10; the proofs for all the other rules can be found in the appendix.
\end{proof}

\noindent
\paragraph{Final Pairs are Hard.} We now show that, for every set $\dep$ of FDs with an LHS chain and an SJFCQ $Q$ such that $(\dep,Q)$ is final, there exists a set $\dep_{PK}$ of primary keys and a query $Q_{PK}$ that is not $\dep_{PK}$-safe such that $\prob{RelFreq}(\dep_{PK},Q_{PK})$ is Cook reducible to $\prob{RelFreq}(\dep,Q)$. 
To this end, we first present four technical lemmas (Lemmas~\ref{lemma:12} - \ref{lemma:final-orphan}) that essentially tell us the following: given a set $\dep$ of FDs with an LHS chain, and an SJFCQ $Q$ such that $(\dep,Q)$ is final, for every relation name $R$ that occurs in $Q$, the set $\dep_R$ of FDs either consists of a single key, or is a set of FDs of one of six specific forms; the proofs are deferred to the appendix. \rev{In what follows, for every $R$-atom $\alpha_R$ of $Q$, assuming that $R:X_i\rightarrow Y_i$ is the primary FD of $\dep_R$ w.r.t.~$Q$, we refer to the positions $\{(R,A)\mid A\in Y_i\}$ as the \e{primary-rhs} positions of $\alpha_R$. We refer to the rest of the positions of $\alpha_R$ as \e{non-primary-rhs} positions, and to the positions outside $\{(R,A)\mid A\in (X_i\cup Y_i)\}$ as \e{non-primary-FD} positions.}
In the statements of Lemmas~\ref{lemma:12} - \ref{lemma:final-orphan}, for brevity, let $\dep$ be a set of FDs with an LHS chain, and $Q$ be an SJFCQ such that $(\dep,Q)$ is final.


\def\lemmasinglekey{
	If the $R$-atom of $Q$ occurs in $Q \setminus\comp{Q}{\dep}$, then $\dep_R$ consists of a single key.
}

\begin{lemma}\label{lemma:12}
	\lemmasinglekey
\end{lemma}

\def\lemmafinalliaison{
		If the $R$-atom of $Q$ has a liaison variable $x$ at a primary-rhs position $(R,A)$, then
	\begin{enumerate}
		\item $\dep_R$ has a single key, or
		\item $\dep_R$ has a single FD $R:X\ra\set{A}$, and there is only one non-primary-FD position $(R,A')$ that is also associated with the variable $x$.
	\end{enumerate}
}

\begin{lemma}\label{lemma:final-liaison}
	\lemmafinalliaison
\end{lemma}

\def\lemmafinalconst{
	If the $R$-atom $\alpha_R$ of $Q$ has a constant $c$ at a primary-rhs position $(R,A)$, then
\begin{enumerate}
	\item $\dep_R$ has a single key, or
	\item $\dep_R$ has a single FD $R:X\ra\set{A}$, and there is only one non-primary-FD position $(R,A')$ that is associated with a liaison variable $x$.
\end{enumerate}
}

\begin{lemma}\label{lemma:final-const}
	\lemmafinalconst
\end{lemma}

\def\lemmafinalorphan{
		If the $R$-atom $\alpha$ of $Q$ has an orphan variable $x$ at a primary-rhs position $(R,A)$, then
	\begin{enumerate}
		\item $\dep_R$ has a single key, or
		\item $\dep_R$ has a single FD $R:X\ra\set{A}$, and there is only one non-primary-FD position $(R,A')$ that is associated with a liaison variable $y$, or
		\item $\dep_R$ has a single FD $R:X\ra\set{A}$, and there is only one non-primary-FD position $(R,A')$ that is associated with a constant $c$, or
		\item $\dep_R$ has a single FD $R:X\ra\set{A}$, and there are only two non-primary-FD positions $(R,A')$ and $(R,A'')$ that are associated with the same liaison variable $y$ that occurs only in $\alpha$, or
		\item $\dep_R$ has two FDs $R:X\ra\set{A}$ and $R:(X\cup \set{A'})\ra\set{A''}$ such that the first is the primary FD, there are no non-primary-lhs positions besides $(R,A)$, $(R,A')$, $(R,A'')$ (where $(R,A')$ and $(R,A'')$ are also non-primary-rhs positions), and the positions $(R,A')$ and $(R,A'')$ are associated with the same liaison variable $y$ that occurs only in the atom $\alpha$.
	\end{enumerate}
}

\begin{lemma}\label{lemma:final-orphan}
	\lemmafinalorphan
\end{lemma}

Having the six specific forms of sets of FDs provided by Lemmas~\ref{lemma:12} - \ref{lemma:final-orphan}, we present a series of six technical lemmas (Lemmas~\ref{lemma:red_first} - \ref{lemma:red_last}), one for each specific form, that will allow us to devise the desired reduction; the proofs of those lemmas are deferred to the appendix. In the statements of Lemmas~\ref{lemma:red_first} - \ref{lemma:red_last}, we assume that $\dep$ is a set of FDs with an LHS chain, and $Q$ is an SJFCQ that is not $\dep$-safe. We further assume that $n > 0$ atoms of $Q$ are {\em not} associated with a single key in $\dep$, and we refer to such atoms as \e{non-single-key} atoms w.r.t.~$\dep$.

\def\lemmaredfirst{
	    If for the $R$-atom of $Q$ for some relation name $R$ it holds that
	\begin{enumerate}
		\item $\dep_R$ has a single FD $R:X\ra\set{A}$,
		\item there is only one non-primary-FD position $(R,A')$, and
		\item the positions $(R,A)$ and $(R,A')$ are both associated with a liaison variable $x$,
	\end{enumerate}
	then $\prob{RelFreq}(\dep',Q')$ is Cook reducible to $\prob{RelFreq}(\dep,Q)$ for some $\dep'$ and an SJFCQ $Q'$ that is not $\dep'$-safe, where $n-1$ atoms of $Q'$ are non-single-key atoms w.r.t.~$\dep'$.
}

\begin{lemma}\label{lemma:red_first} 
	\lemmaredfirst
\end{lemma}

\def\lemmaredsecond{
	 If for the $R$-atom of $Q$ for some relation name $R$ it holds that
	\begin{enumerate}
		\item $\dep_R$ has a single FD $R:X\ra\set{A}$,
		\item there is only one non-primary-FD position $(R,A')$, and
		\item the position $(R,A)$ is associated with a constant $c$ and $(R,A')$ with a liaison variable $x$,
	\end{enumerate}
	then $\prob{RelFreq}(\dep',Q')$ is Cook reducible to $\prob{RelFreq}(\dep,Q)$ for some $\dep'$ and an SJFCQ $Q'$ that is not $\dep'$-safe, where $n-1$ atoms of $Q'$ are non-single-key atoms w.r.t.~$\dep'$.
}
\begin{lemma}\label{lemma:red_second} 
	\lemmaredsecond
\end{lemma}

\def\lemmaredthird{
If for the $R$-atom of $Q$ for some relation name $R$ it holds that
\begin{enumerate}
	\item $\dep_R$ has a single FD $R:X\ra\set{A}$,
	\item there is only one non-primary-FD position $(R,A')$, and
	\item $(R,A)$ is associated with an orphan variable $x$ and $(R,A')$ with a liaison variable $y$,
\end{enumerate}
then $\prob{RelFreq}(\dep',Q')$ is Cook reducible to $\prob{RelFreq}(\dep,Q)$ for some $\dep'$ and an SJFCQ $Q'$ that is not $\dep'$-safe, where $n-1$ atoms of $Q'$ are non-single-key atoms w.r.t.~$\dep'$.
}

\begin{lemma}\label{lemma:red_third}
	\lemmaredthird
\end{lemma}

\def\lemmaredfourth{
If for the $R$-atom of $Q$ for some relation name $R$ it holds that
\begin{enumerate}
	\item $\dep_R$ has a single FD $R:X\ra\set{A}$,
	\item there is only non-primary-FD position $(R,A')$, and
	\item the position $(R,A)$ is associated with an orphan variable $x$ and $(R,A')$ with a constant $c$,
\end{enumerate}
then $\prob{RelFreq}(\dep',Q')$ is Cook reducible to $\prob{RelFreq}(\dep,Q)$ for some $\dep'$ and an SJFCQ $Q'$ that is not $\dep'$-safe, where $n-1$ atoms of $Q'$ are non-single-key atoms w.r.t.~$\dep'$.
}

\begin{lemma}\label{lemma:red_fourth}
	\lemmaredfourth
\end{lemma}

\def\lemmaredfifth{
If for the $R$-atom of $Q$ for some relation name $R$ it holds that
\begin{enumerate}
	\item $\dep_R$ has a single FD $R:X\ra\set{A}$,
	\item there are only two non-primary-FD positions $(R,A')$ and $(R,A'')$, and
	\item the position $(R,A)$ is associated with an orphan variable $x$ and $(R,A'),(R,A'')$ with the same liaison variable $y$ that occurs only in $\alpha$,
\end{enumerate}
then $\prob{RelFreq}(\dep',Q')$ is Cook reducible to $\prob{RelFreq}(\dep,Q)$ for some $\dep'$ and an SJFCQ $Q'$ that is not $\dep'$-safe, where $n-1$ atoms of $Q'$ are non-single-key atoms w.r.t.~$\dep'$.
}

\begin{lemma}\label{lemma:red_fifth} 
	\lemmaredfifth
\end{lemma}

\def\lemmaredlast{
If for the $R$-atom $\alpha$ of $Q$ for some relation name $R$ it holds that
\begin{enumerate}
	\item $\dep_R$ has two FDs $R:X\ra\set{A}$, $R:(X\cup \set{A'})\ra\set{A''}$ with the first being the primary FD,
	\item  there are no non-primary-lhs positions besides $(R,A)$, $(R,A')$, $(R,A'')$ (where $(R,A')$ and $(R,A'')$ are also non-primary-rhs positions), and
	\item the position $(R,A)$ is associated with an orphan variable $x$ and $(R,A'),(R,A'')$ with the same liaison variable $y$ that occurs only in the atom $\alpha$,
\end{enumerate}
then $\prob{RelFreq}(\dep',Q')$ is Cook reducible to $\prob{RelFreq}(\dep,Q)$ for some $\dep'$ and an SJFCQ $Q'$ that is not $\dep'$-safe, where $n-1$ atoms of $Q'$ are non-single-key atoms w.r.t.~$\dep'$.
}

\begin{lemma}\label{lemma:red_last} 
	\lemmaredlast
\end{lemma}

Having the above technical lemmas in place, we are now ready to establish the desired result concerning final pairs consisting of a set of FDs and an SJFCQ.

\begin{lemma}\label{lemma:reduction-final}
	Consider a set $\dep$ of FDs with an LHS chain, and an SJFCQ $Q$ such that $(\dep,Q)$ is final. There exists a set $\dep_{PK}$ of primary keys, and an SJFCQ $Q_{PK}$ that is not $\dep_{PK}$-safe such that $\prob{RelFreq}(\dep_{PK},Q_{PK})$ is Cook reducible to $\prob{RelFreq}(\dep,Q)$.
\end{lemma}

\begin{proof}
	The proof is by a simple induction on the number $n \geq 0$ of non-single-key atoms of $Q$.

	\medskip
	\noindent\paragraph{Base Case.}	For $n=0$, we have that every relation name that occurs in $Q$ is associated with a single key, in which case the claim follows with $\dep_{PK}=\dep$ and $Q_{PK}=Q$.
	
	\medskip
	\noindent\paragraph{Inductive Step.} For $n>0$, we know that for every non-single-key $R$-atom of $Q$ w.r.t.~$\dep$, the set $\dep_R$ is of one of six forms (Lemmas~\ref{lemma:12} - \ref{lemma:final-orphan}), and once we apply the reduction of the corresponding lemma (among Lemmas~\ref{lemma:red_first}-\ref{lemma:red_last}), we obtain a query $Q'$ that has $n-1$ non-single-key atoms w.r.t.~the FD set $\dep'$. By the inductive hypothesis, there is a set $\dep_{PK}$ of primary keys, and an SJFCQ $Q_{PK}$ that is not $\dep_{PK}$-safe such that $\prob{RelFreq}(\dep_{PK},Q_{PK})$ is Cook reducible to $\prob{RelFreq}(\dep',Q')$. The composition of the reduction from $\prob{RelFreq}(\dep_{PK},Q_{PK})$ to $\prob{RelFreq}(\dep',Q')$ and the reduction from $\prob{RelFreq}(\dep',Q')$ to $\prob{RelFreq}(\dep,Q)$ is a Cook reduction from $\prob{RelFreq}(\dep_{PK},Q_{PK})$ to $\prob{RelFreq}(\dep,Q)$, and the claim follows.
	It is important to say that in the queries $Q'$ and sets $\dep'$ of FDs obtained in Lemmas~\ref{lemma:red_first} -\ref{lemma:red_last}, all the non-single-key atoms remain in one of these six specific forms. That is, every liaison (resp., orphan) variable remains a liaison (resp., an orphan) variable.
\end{proof}

Having Lemmas~\ref{lem:final-pairs-via-rewriting}, \ref{lem:rewrite-rules-relfreq} and \ref{lemma:reduction-final} in place, it is easy to establish Theorem~\ref{the:hard-side}.
Consider a set $\dep$ of FDs with an LHS chain, and an SJFCQ $Q$ such that $Q$ is not $\dep$-safe. By Lemma~\ref{lem:final-pairs-via-rewriting}, from $(\dep,Q)$ we can reach a pair $(\dep',Q')$ that is final via a rewriting sequence. Thus, by Lemma~\ref{lem:rewrite-rules-relfreq} and the fact that Cook reductions compose, we get that $\prob{RelFreq}(\dep',Q')$ is Cook reducible to $\prob{RelFreq}(\dep,Q)$. Since, by Lemma~\ref{lemma:reduction-final}, there exists a set $\dep_{PK}$ of primary keys and a query $Q_{PK}$ that is not $\dep_{PK}$-safe such that $\prob{RelFreq}(\dep_{PK},Q_{PK})$ is Cook reducible to $\prob{RelFreq}(\dep',Q')$, by composing reductions, we get that $\prob{RelFreq}(\dep_{PK},Q_{PK})$ is Cook reducible to $\prob{RelFreq}(\dep,Q)$, and the claim follows.
\section{Approximate Counting}\label{sec:apx-counting}

We now focus on approximate counting. As already discussed in Section~\ref{sec:problem-definition}, the ultimate goal is whenever $\sharp \prob{Repairs}(\dep)$ or $\sharp \prob{Repairs}(\dep,Q)$ is intractable, where $\dep$ is a set of FDs and $Q$ an SJFCQ, to classify it as approximable (i.e., it admits an FPRAS), or as inapproximable (i.e., it does not admit an FPRAS, possibly under some widely accepted complexity assumptions).
However, as we explain below, obtaining an approximability/inapproximability dichotomy is a highly non-trivial task that remains open. We provide results that form crucial steps towards such a dichotomy. We show that FDs with an LHS chain (up to equivalence) form an island of approximability. On the other hand, unlike primary keys, there are simple cases for which the existence of an FPRAS is not guaranteed, unless $\text{\rm NP} \subseteq \text{\rm BPP}$; \text{\rm BPP} is the class of decision problems that are efficiently solvable via a randomized algorithm with a bounded two-sided error~\cite{ArBa09}.

\def\thmapxcount{
	The following hold:
	\begin{enumerate}
		\item Consider a set $\dep$ of FDs with an LHS chain (up to equivalence), and a CQ $Q$. $\sharp \prob{Repairs}(\dep,Q)$ admits an FPRAS.
		\item 	Let $\dep = \{R : A_1 \ra A_2, R: A_3 \ra A_4\}$, where $(A_1,A_2,A_3,A_4)$ is the tuple of attribute names of $R/4$. $\sharp \prob{Repairs}(\dep)$ does not admit an FPRAS, unless $\text{\rm NP} \subseteq \text{\rm BPP}$.
		\item Let $\dep = \{R : A_1 \ra A_2, R: A_3 \ra A_4\}$, where $(A_1,A_2,A_3,A_4)$ is the tuple of attribute names of $R/4$. For every SJFCQ $Q$, $\sharp \prob{Repairs}(\dep,Q)$ does not admit an FPRAS, unless $\text{\rm NP} \subseteq \text{\rm BPP}$.
	\end{enumerate}
}

\begin{theorem}\label{the:apx-main-result}
	\thmapxcount
\end{theorem}

Notice that item (1) holds for arbitrary CQs (even with self-joins); in fact, we can easily extend item (1) to unions of CQs. 
Let us also observe that, due to Lemma~\ref{lem:cook-reduction}, to establish item (3) it suffices to establish item (2).
The rest of the section is devoted to discussing the proofs of items (1) and (2) of Theorem~\ref{the:apx-main-result}. But let us first briefly discuss the difficulty underlying a pure dichotomy result. In our discussion, we focus on the problem $\sharp \prob{Repairs}(\dep,Q)$, where $\dep$ is a set of FDs and $Q$ a SJFCQ, but it will be clear that the same argument applies also for the problem $\sharp \prob{Repairs}(\dep)$.

\medskip
\noindent
\paragraph{The Difficulty Underlying a Dichotomy.}
Let $\sharp\prob{MaxMatch}$ be the counting problem that accepts as input a bipartite graph $G$, and asks for the number of maximal matchings in $G$. It is known that $\sharp\prob{MaxMatch}$ is $\sharp\text{\rm P}$-complete~\cite{Vadhan01}, whereas the question whether it admits an FPRAS is a challenging open problem~\cite{JiRa18}.
Due to item (1) of Theorem~\ref{the:apx-main-result}, to get a dichotomy result it suffices to classify, for every set $\dep$ of FDs {\em without} an LHS chain (up to equivalence), and an SJFCQ $Q$, $\sharp \prob{Repairs}(\dep,Q)$ as approximable or inapproximable. \rev{We proceed to show the following result, which essentially states that an approximability/inapproximability dichotomy for $\sharp \prob{Repairs}(\dep,Q)$ will resolve the open question whether $\sharp\prob{MaxMatch}$ admits an FPRAS.}

\rev{\begin{proposition}\label{prop:maxmatch_fpras}
    The following are equivalent.
    \begin{enumerate}
        \item There exists a set $\dep$ of FDs without an LHS chain (up to equivalence), and an SJFCQ $Q$ such that $\sharp \prob{Repairs}(\dep,Q)$ admits an FPRAS.
        \item $\sharp\mathsf{MaxMatch}$ admits an FPRAS.
    \end{enumerate}
\end{proposition}}

\begin{proof}
\rev{
$(1) \Rightarrow (2)$. Assume that there is a set $\dep$ of FDs without an LHS chain (up to equivalence), and an SJFCQ $Q$ such that $\sharp \prob{Repairs}(\dep,Q)$ admits an FPRAS. We know from~\cite{LiKW21} that $\sharp \prob{Repairs}(\dep)$ is $\sharp$P-hard, which is shown via a polynomial-time parsimonious reduction from $\sharp\mathsf{MaxMatch}$. 
We further know, by Lemma~\ref{lem:cook-reduction}, that there is a polynomial-time parsimonious reduction from $\sharp \prob{Repairs}(\dep)$ to $\sharp \prob{Repairs}(\dep,Q)$. Therefore, by composing reductions, we get a polynomial-time parsimonious reduction from $\sharp\mathsf{MaxMatch}$ to $\sharp \prob{Repairs}(\dep,Q)$. But since, by hypothesis, $\sharp \prob{Repairs}(\dep,Q)$ admits an FPRAS, we conclude that $\sharp\mathsf{MaxMatch}$ also admits an FPRAS.}

\rev{
$(2) \Rightarrow (1)$. Assume now that, for every set $\dep$ of FDs without an LHS chain (up to equivalence), and an SJFCQ $Q$, $\sharp \prob{Repairs}(\dep,Q)$ does not admit an FPRAS. Therefore, this will be also the case for $\dep^\star = \{R: A \ra B, R: B \ra A\}$ with $(A,B)$ being the tuple of attribute names of $R/2$, and every SJFCQ $Q^\star$ that is trivially entailed by every repair. It is implicit, however, in~\cite{LiKW21} that there exists a polynomial-time parsimonious reduction from $\prob{Repairs}(\dep^\star,Q^\star)$ to $\sharp\mathsf{MaxMatch}$, which in turn implies that $\sharp\mathsf{MaxMatch}$ does not admit an FPRAS.}
\end{proof}

\subsection{Approximability Result}

We proceed with item (1) of Theorem~\ref{the:apx-main-result}: given a set $\dep$ of FDs with an LHS chain (up to equivalence), and a CQ $Q$, $\sharp \prob{Repairs}(\dep,Q)$ admits an FPRAS. This is done by exploiting classical results from~\cite{Karp1989}.

\medskip
\noindent
\paragraph{FPRAS via Union of Sets.}
We observe that $\sharp \prob{Repairs}(\dep,Q)$ is an instantiation of the so-called {\em union of sets} problem~\cite{Karp1989}, which takes as input a succinct encoding of $n \geq 1$ sets $S_1,\ldots,S_n$, and asks for $| \bigcup_{i \in [n]} S_i|$. \rev{Let us first recall the formal definition of the union of sets problem; as usual, for a set $S$, we write $\mathcal{P}(S)$ for the powerset of $S$.}
\rev{Consider some alphabet $\Gamma$ (i.e., a finite set of symbols). A \emph{union of sets specification} is a pair $\mathcal{U} = (f,\mu)$, where $f$ is a function $\Gamma^* \rightarrow \mathbb{N}$, and $\mu$ is a function $\Gamma^* \times \mathbb{N} \rightarrow \mathcal{P}(\Gamma^*)$. Roughly speaking, using a string $x \in \Gamma^*$ we can succinctly encode $n = f(x)$ sets $S_1,\ldots,S_n$, and each set $S_i$ coincides with $\mu(x,i)$.
}
\rev{
We are now ready to recall the union of sets problem; in what follows, let $\mathcal{U} = (f,\mu)$ be a union of sets specification:}

\rev{
\begin{center}
	\fbox{\begin{tabular}{ll}
			{\small PROBLEM} : & $\prob{UnionOfSets}(\mathcal{U})$
			\\
			{\small INPUT} : & A string $x \in \Gamma^*$
			\\
			{\small OUTPUT} : &  $\left| \bigcup_{i \in [f(x)]} \mu(x,i) \right|$
	\end{tabular}}
\end{center}
}
\medskip

\rev{Interestingly, for certain union of sets specifications, the above problem admits an FPRAS. Such union of sets specifications are the so-called tractable ones, which we now recall. 
The union of sets specification $\mathcal{U} = (f,\mu)$ is \emph{tractable} if the following conditions on $f$ and $\mu$ hold:
\begin{itemize}
	\item[-] $f$ is polynomial-time computable, and there is a polynomial function $\mathsf{pol} : \mathbb{N} \rightarrow \mathbb{N}$ such that $f(x) \le \mathsf{pol}(|x|)$, for each $x \in \Gamma^*$, and 
	\item[-] for each $x \in \Gamma^*$ and $i \in [f(x)]$, (1) $|\mu(x,i)|$ can be computed in polynomial time in $|x|$, (2) an element from $\mu(x,i)$ can be sampled uniformly at random in polynomial time in $|x|$, and (3) checking whether $y \in \mu(x,i)$, for a given $y \in \Gamma^*$, can be done in polynomial time in $|x|$.
\end{itemize}
The following is well-known approximability result:}
\def\thmunionofsetsfpras{
	For a tractable union of sets specification $\mathcal{U}$, $\prob{UnionOfSets}(\mathcal{U})$ admits an FPRAS.
}

\rev{
\begin{theorem}[\cite{Karp1989}]\label{thm:union-of-sets-fpras}
	\thmunionofsetsfpras
\end{theorem}
}

\rev{
Hence, to prove that for every set $\dep$ of FDs with an LHS chain (up to equivalence) and CQ $Q$, $\sharp \prob{Repairs}(\dep,Q)$ admits an FPRAS, it suffices to show that for every set $\dep$ of FDs with an LHS chain (up to equivalence) and CQ $Q$, there exists a tractable union of sets specification $\mathcal{U}_{\dep,Q}=(f,\mu)$ such that, for every database $D$, with $n = f(D)$, $\card{\rep{D}{\dep,Q}} = |\bigcup_{i \in [n]} \mu(D,i)|$.
}

\rev{
Fix a set $\dep$ of FDs with an LHS chain (up to equivalence) and a CQ $Q$. Now, consider the set} $\homs{Q}{D}{\dep}$  of all homomorphic images of $Q$ in $D$ that are consistent w.r.t.~$\dep$. Formally, $\homs{Q}{D}{\dep}$ is the set of sets of facts
\[
\{ h(Q) \mid h \text{ is a homomorphism from } Q \text{ to } D \text{ such that } h(Q) \models \dep\}.
\]
Clearly, if $\homs{Q}{D}{\dep}$ is empty, then $\card{\rep{D}{\dep,Q}} = 0$. On the other hand, if $\homs{Q}{D}{\dep}$ is non-empty, then $\card{\rep{D}{\dep,Q}} = \left| \bigcup_{i \in [n]} \rep{D}{\dep,H_i}\right|$ assuming that $\homs{Q}{D}{\dep} = \{H_1,\ldots,H_n\}$; by abuse of terminology, we treat each set $H_i$ as a CQ consisting of only facts.
It is not difficult to show that checking whether $\homs{Q}{D}{\dep}$ is empty or not is feasible in polynomial time in $||D||$ \rev{(recall that, in data complexity, $||Q||$ is an integer constant)}.
Therefore, \rev{w.l.o.g., we focus on the problem of computing } $|\bigcup_{i \in [n]} \rep{D}{\dep,H_i}|$ when $\homs{Q}{D}{\dep} \neq \emptyset$. \rev{The definition of the union of sets specification $\mathcal{U}_{\dep,Q}=(f,\mu)$, in this case, is straightforward. In particular, $\mathcal{U}_{\dep,Q}=(f,\mu)$ is such that, for every database $D$, assuming $\homs{Q}{D}{\dep}=\{H_1,\ldots,H_n\}$, $f(D)=n$, and $\mu(D,i) = \rep{D}{\dep,H_i}$, for each $i \in [n]$. Clearly, $\card{\rep{D}{\dep,Q}} = \left| \bigcup_{i \in [n]} \rep{D}{\dep,H_i}\right| = \left| \bigcup_{i \in [n]} \mu(D,i)\right|$. Moreover, it is easy to see that $f$ is computable in polynomial time and that $f(D)$ is bounded by a polynomial.}
From the above discussion, and from \rev{Theorem~\ref{thm:union-of-sets-fpras}}, we conclude that showing the existence of an FPRAS for $\sharp \prob{Repairs}(\dep,Q)$ boils down to showing the following result:

\def\prounionofsets{
	Consider a set $\dep$ of FDs with an LHS chain (up to equivalence), and a CQ $Q$. Given a database $D$ with $\homs{Q}{D}{\dep} \neq \emptyset$, the following hold for each $H \in \homs{Q}{D}{\dep}$:
\begin{enumerate}
	\item We can compute $\card{\rep{D}{\dep,H}}$ in polynomial time in $||D||$.
	\item We can sample elements of $\rep{D}{\dep,H}$ uniformly at random in polynomial time in $||D||$.
	\item Given a database $D'$, we can check whether $D' \in \rep{D}{\dep,H}$ in polynomial time in $||D||$.
\end{enumerate}
}

\begin{proposition}\label{pro:union-of-sets-properties}
	\prounionofsets
\end{proposition}

\textit{Proof (Sketch).} We prove the rather easy items (1) and (3), whereas for item (2) we sketch the high-level idea and the full proof is deferred to  the appendix.
	\begin{description}
		\item[Item (1).] We observe that $H$ is a set of facts such that $H \models \dep$ and $H \subseteq D$. Hence, $\rep{D}{\dep,H} \neq \emptyset$, and each of its repairs contains $H$; recall that we may treat $H$ as a CQ.
		Thus, $\rep{D}{\dep,H}$ is precisely the set of repairs of $\rep{D}{\dep}$ that do not contain any fact $f \in D$ such that $H \cup \{f\} \not \models \dep$, i.e., $\rep{D}{\dep,H} = \rep{D \setminus D'}{\dep}$, where $D' = \{f \in D \mid H \cup \{f\} \not \models \dep\}$. The database $D'$ is clearly computable in polynomial time in $||D||$, and thus, by Proposition~\ref{pro:counting-repairs-lhs-fp}, we can compute $\card{\rep{D}{\dep,H}} = \card{\rep{D \setminus D'}{\dep}}$ in polynomial time in $||D||$, as needed.
		
		\item[Item (2).] We first show that we can sample elements of $\rep{D}{\dep}$ uniformly at random in polynomial time in $||D||$.
		In particular, we show that there is a randomized algorithm $\mathsf{Sample}$ that takes as input $D$ and $\dep$, runs in polynomial time in $||D||$, and produces a random variable $\mathsf{Sample}(D,\dep)$ such that 
		$
		\Pr( \mathsf{Sample}(D,\dep) = D' ) = \frac{1}{\card{\rep{D}{\dep}}},
		$
		for every $D' \in \rep{D}{\dep}$.
		This exploits the fact that $\card{\rep{D}{\dep}}$ can be computed in polynomial time in $||D||$, and that for each relation name $R$, the set $\rep{D_R}{\dep_R}$ has a convenient recursive definition based on the $(R,\Lambda)$-blocktree of $D$ w.r.t.~$\dep$, where $\Lambda$ is the LHS chain of $\dep_R$; we can assume that $\dep$ is canonical, and thus, $\Lambda$ is well-defined.
		%
		Now, by using the same idea as for item (1), we convert the sampler for $\rep{D}{\dep}$ into a sampler for $\rep{D}{\dep,H}$. The details are in the appendix.
		
		\item[Item (3).] This item holds trivially since checking whether $D' \in \rep{D}{\dep,H}$ boils down to checking whether $D' \in \rep{D}{\dep}$ and $H \subseteq D'$. \qed
	\end{description}

\rev{As previously discussed, our blocktrees align with the cotrees of cographs. Consequently, the proof of point (2) of Proposition~\ref{pro:union-of-sets-properties} establishes that we can efficiently sample, uniformly at random, maximal independent sets from a cograph.}

\medskip
\noindent \paragraph{Monte Carlo Sampling.}
At this point, one may wonder whether Monte Carlo sampling would lead to an FPRAS for $\sharp \prob{Repairs}(\dep,Q)$. In other words, given a database $D$, $\epsilon > 0$, and $0 < \delta < 1$, sample uniformly at random, by using the sampler for $\rep{D}{\dep}$ discussed above, a certain number $N$ of repairs from $\rep{D}{\dep}$, and output the number $\frac{S}{N} \cdot \card{\rep{D}{\dep}}$, where $S$ is the number of sampled repairs that entail $Q$. For this approach to lead to an FPRAS, the number of samples $N$ must be large enough to provide the desired error and probabilistic guarantees, but should be also bounded by a polynomial of $||D||$ in order to ensure that the running time is polynomial.
To this end, we need a polynomial $\mathsf{pol}$ such that, for every database $D$, $\rfreq{Q}{D,\dep} \ge \frac{1}{\mathsf{pol}(||D||)}$ whenever $\rfreq{Q}{D,\dep} > 0$. However, we can show that this is not always the case: 

\begin{theorem}\label{the:rfreq-bound}
	There exists a set of FDs $\dep$ with an LHS chain, an atomic CQ $Q$, and a family of databases $\{D_n\}_{n > 0}$ with $|D_n| = 2n + 1$ such that $\rfreq{Q}{D_n,\dep} = \frac{1}{2^n+1}$.
\end{theorem}

\begin{proof}
	Let $R$ be a relation name with attributes $(A_1,A_2,A_3,A_4)$. We define the set of FDs
	\[
	\dep\ =\ \{R: A_1 \ra A_2, R: \{A_1, A_3\} \ra A_4\},
	\]
	which clearly has an LHS chain, and the atomic Boolean query
	\[
	Q\ =\ \textrm{Ans}() \ \text{:-}\  \exists x \exists y \exists z\, R(x,x,y,z).
	\]
	For each $n>0$, we define $D_n$ as the following database:
	\[
	\{R(a,a,a,a)\} \cup \{R(a,b,c_i,d_1),R(a,b,c_i,d_2)\}_{i \in [n]},
	\]
	where $a$, $b$, $d_1$, $d_2$, and $c_i$, for $i \in [n]$, are distinct constants. Clearly, the database $D_n$ contains $2n+1$ facts. Moreover, it is not difficult to verify that the set of repairs $\rep{D_n}{\dep}$ contains only the database $D' = \{R(a,a,a,a)\}$, and a database of the form
	\[
	\{R(a,b,c_i,d_{\mu(i)})\}_{i \in [n]},
	\]
	for each function $\mu : [n] \rightarrow \{1,2\}$. Hence, $\card{\rep{D_n}{\dep}} = 2^n+1$. Moreover, $D'$ is the only repair in $\rep{D_n}{\dep}$ such that $D' \models Q$. Therefore, $\rfreq{Q}{D_n,\dep} = \frac{1}{1 + 2^n}$, and the claim follows.
\end{proof}

\subsection{Inapproximability Result}

We now discuss the proof of item (2) of Theorem~\ref{the:apx-main-result}. Recall that by proving item (2), we immediately get item (3) of Theorem~\ref{the:apx-main-result} due to Lemma~\ref{lem:cook-reduction}.
The proof employs the technique of {\em gap amplification} used to prove inapproximability of optimization problems. The idea is to consider an \text{\rm NP}-complete decision problem $\Pi$, and define a polynomial-time computable function $\mathsf{db}$ that maps instances of $\Pi$ to databases such that the following holds: for any \emph{yes} instance $I_Y$ of $\Pi$, and any \emph{no} instance $I_N$ of $\Pi$, the ``gap'' between the numbers $\card{\rep{\mathsf{db}(I_Y)}}{\dep}$ and $\card{\rep{\mathsf{db}(I_N)}}{\dep}$ is so large that an FPRAS for $\sharp \prob{Repairs}(\dep)$ could be used, by employing a small enough error $\epsilon$, to distinguish between \emph{yes} and \emph{no} instances of $\Pi$ with high probability, and thus, placing the \text{\rm NP}-complete problem $\Pi$ in \text{\rm BPP}.

\medskip
\noindent
\paragraph{The GapSAT Problem.} We are going to exploit the problem $\mathsf{Gap3SAT}_\gamma$ for some fixed $\gamma \in (0,\frac{1}{8})$. This is a promise problem that takes as input a Boolean formula $\varphi$ in 3CNF with exactly three literals in each clause, where each variable appears at most once in each clause, and asks whether $\varphi$ is satisfiable.
The promise is that if $\varphi$ is unsatisfiable, then every truth assignment makes at most $\frac{7}{8}+\gamma$ of the clauses of $\varphi$ true.
We know that, for $\gamma \in (0,\frac{1}{8})$, $\mathsf{Gap3SAT}_\gamma$ is \text{\rm NP}-complete. This is a consequence of H{\aa}stad's 3-bit PCP Theorem, and it is explicitly stated in~\cite[Theorem 6.5]{Has01}.

\medskip
\noindent\paragraph{From 3CNF Formulae to Databases.}
Our goal is to show, via the technique of gap amplification, that the existence of an FPRAS for $\sharp \prob{Repairs}(\dep)$ implies that $\mathsf{Gap3SAT}_\gamma$ is in $\mathsf{BPP}$, for some $\gamma \in (0, \frac{1}{8})$, which in turn establishes item (2) of Theorem~\ref{the:apx-main-result}.
To this end, we first define a polynomial-time computable function $\mathsf{db}_k$, for some $k>0$, that maps 3CNF formulas to databases. The integer $k$ is what we call the ``growing factor'', which we are going to choose later.

Consider a 3CNF formula $\varphi = C_1 \wedge \cdots \wedge C_m$, where $C_i = \ell^1_i \vee \ell^2_i \vee \ell^3_i$, for each $i \in [m]$. The intention underlying the database $\mathsf{db}_k(\varphi)$ is as follows: the number of repairs of $\mathsf{db}_k(\varphi)$ w.r.t.~$\dep$ should be proportional to the number of clauses of $\phi$ that can be satisfied by a truth assignment. We would further like to control by ``how much'' the number of repairs grows for each clause that is satisfied. Let $\var{\ell^j_i}$ be the Boolean variable of $\ell^j_i$. Moreover, for a truth value $v \in \{0,1\}$, we write $\litval{\ell^j_i}{v}$ for the truth value of $\ell^j_i$, when its variable is assigned the value $v$; for example, if $v = 0$, then $\litval{\neg x}{v} = 1$.
The database $\mathsf{db}_k(\varphi)$ is defined as the union of the two databases $D^k_T$ and $D^k_C$ over $R$ shown below.
Note that, for the sake of presentation, we rename the attribute names of $R$; $A_1$ becomes \text{\rm Var}, $A_2$ becomes \text{\rm VValue}, $A_3$ becomes \text{\rm Clause}, and $A_4$ becomes \text{\rm LValue}. Under this new naming scheme, the FDs of $\dep$ are 
\[
R : \text{\rm Var} \ra \text{\rm VValue} \qquad R : \text{\rm Clause} \ra \text{\rm LValue}.
\]
The databases $D^k_T$ and $D^k_C$  follow, where $\star$ is an arbitrary constant:

\medskip

{\renewcommand{\arraystretch}{1.2}
	\begin{center}\small
		\begin{tabular}{@{}cccccl@{}}
			\toprule
			\normalsize Var & \normalsize VValue & \normalsize Clause & \normalsize LValue & \\ \midrule
			$\var{\ell^1_i}$ & 0 & $\angletup{C^j_i,\litval{\ell^1_i}{0}}$ & $\litval{\ell^1_i}{0}$ & \rdelim\}{6}{8em}[
			{
				$i \in [m]$, $j \in [k]$
			}
			]\\
			$\var{\ell^1_i}$ & 1 & $\angletup{C^j_i,\litval{\ell^1_i}{1}}$ & $\litval{\ell^1_i}{1}$ &\\
			$\var{\ell^2_i}$ & 0 & $\angletup{C^j_i,\litval{\ell^2_i}{0}}$ & $\litval{\ell^2_i}{0}$ &\\
			$\var{\ell^2_i}$ & 1 & $\angletup{C^j_i,\litval{\ell^2_i}{1}}$ & $\litval{\ell^2_i}{1}$ &\\
			$\var{\ell^3_i}$ & 0 & $\angletup{C^j_i,\litval{\ell^3_i}{0}}$ & $\litval{\ell^3_i}{0}$ &\\
			$\var{\ell^3_i}$ & 1 & $\angletup{C^j_i,\litval{\ell^3_i}{1}}$ & $\litval{\ell^3_i}{1}$ &\\ \bottomrule	
			
			\multicolumn{5}{c}{\normalsize \textbf{The Database} $D^k_T$}\\\ \\
			
			\toprule
			\normalsize Var & \normalsize VValue & \normalsize Clause & \normalsize LValue & \\ \midrule 
			$\star$ & $\star$ & $\angletup{C^j_i,1}$ & 0 & \rdelim\}{1}{8em}[
			{
				$i \in [m]$, $j \in [k]$
			}
			]\\[1pt] \bottomrule
			\multicolumn{5}{c}{\normalsize \textbf{The Database} $D^k_C$} \\
		\end{tabular}
	\end{center}
}

Intuitively, for each clause $C_i$ of $\varphi$, the database $D^k_T$ stores $k$ copies $C^1_i,\ldots,C^k_i$ of $C_i$. Each copy is stored via 6 tuples, where, for each literal of $C_i$, its variable is stored twice, specifying one of two possibles values it can get assigned (i.e., $0$ or $1$). Moreover, the actual truth value that the literal has w.r.t.\ to the particular value given to its variable is also stored, together with the clause copy. This same literal value is replicated under the attribute \text{\rm LValue}; its purpose is explained below. The set of 6 tuples in $D^k_T$ mentioning a particular copy $C^j_i$ of a clause $C_i$ is called a \emph{$C^j_i$-gadget}.

Regarding $D^k_C$, it stores one tuple for each copy $C^j_i$ of a clause $C_i$. Each tuple in $D^k_C$ containing $C^j_i$ causes a conflict when a truth assignment of $\varphi$ ``chooses'' to assign $1$ to at least one literal of the $C^j_i$-gadget. Due to this, the number of repairs of $\mathsf{db}_k(\varphi)$ w.r.t.~$\dep$ is proportional to the number of clauses satisfied by the truth assignment, and the growing factor is controlled by $k$.

\medskip
\noindent\paragraph{Inapproximabilty via Gap Amplification.}
We now discuss that there are $\gamma \in (0,\frac{1}{8})$ and $k > 0$ such that, for every two 3CNF formulae $\varphi_Y$ and $\varphi_N$ that are {\em yes} and {\em no} instances of $\mathsf{Gap3SAT}_\gamma$, respectively, the ``gap'' between the numbers $\card{\rep{\mathsf{db}_k(\varphi_Y)}}{\dep}$ and $\card{\rep{\mathsf{db}_k(\varphi_N)}}{\dep}$ is indeed large enough that allows us to place $\mathsf{Gap3SAT}_\gamma$ in \text{\rm BPP} by exploiting an FPRAS for $\sharp \prob{Repairs}(\dep)$. This relies on a crucial property that intuitively can be described as follows.

Consider a 3CNF formula $\varphi$. The FD $R: \text{Var} \ra \text{VValue}$ of $\dep$ simply states that a Boolean variable of $\varphi$ must be assigned only one of the two values $0$ or $1$. If only this FD was considered, we would have an one-to-one correspondence between the truth assignments of $\varphi$ and the repairs of $\mathsf{db}_k(\varphi)$ w.r.t.~$\dep$. The FD $R : \text{\rm Clause} \ra \text{\rm LValue}$ forces a truth assignment $\tau$ to induce an exponential number of repairs in the number of clauses that are satisfied by $\tau$ due to the inconsistencies between the tuples of $D^k_T$ and $D^k_C$ w.r.t.~the FD $R : \text{\rm Clause} \ra \text{\rm LValue}$. 
This property is formalized by the next technical lemma; its proof can be found in the appendix. For a truth assignment $\tau$ of $\varphi$, we denote by $c_\tau$ the number of clauses of $\varphi$ that are true w.r.t. $\tau$.

\def\lemmapping{
	Consider a 3CNF formula $\varphi$. There exists a function $\mathsf{Map}$ from truth assignments of $\varphi$ to sets of databases such that the following hold:
\begin{enumerate}
	\item For every truth assignment $\tau$ of $\varphi$, $\mathsf{Map}(\tau) \subseteq \rep{D^k_\varphi}{\dep}$.
	\item For every $D' \in \rep{D^k_\varphi}{\dep}$, there is a truth assignment $\tau$ of $\varphi$ such that $D' \in \mathsf{Map}(\tau)$.
	\item For every truth assignment $\tau$ of $\varphi$, $|\mathsf{Map}(\tau)| = 2^{k \cdot c_\tau}$.
\end{enumerate}
}

\begin{lemma}\label{lem:mapping}
	\lemmapping
\end{lemma}

Consequently, each truth assignment $\tau$ of $\varphi$ corresponds to $2^{k \cdot c_\tau}$ repairs of $D^k_\varphi$, and every repair has a corresponding truth assignment. With Lemma~\ref{lem:mapping} in place, it is easy to establish the following key technical result.
We write $\varphi \in \mathsf{Gap3SAT}_\gamma$ (resp., $\varphi \not \in \mathsf{Gap3SAT}_\gamma$) to indicate that $\varphi$ is a \emph{yes} (resp., \emph{no}) instance of $\mathsf{Gap3SAT}_\gamma$.

\def\lembounds{
	Consider a 3CNF formula $\varphi$ with $n>0$ variables and $m>0$ clauses. For every $\gamma  \in (0,\frac{1}{8})$ and $k>0$:
\begin{enumerate}
	\item If $\varphi \in \mathsf{Gap3SAT}_\gamma$, then $\card{\rep{\mathsf{db}_k(\varphi)}{\dep}} \ge 2^{k \cdot m}$.
	\item If $\varphi \not \in \mathsf{Gap3SAT}_\gamma$, then $\card{\rep{\mathsf{db}_k(\varphi)}{\dep}} \le 2^n \cdot 2^{(\frac{7}{8} + \gamma) \cdot k \cdot m}$.
\end{enumerate}
}

\begin{lemma}[\textbf{Bounds}]\label{lem:gap}
	\lembounds
\end{lemma}

\begin{proof}
		We first establish item (1). Assume that $\varphi \in \mathsf{Gap3SAT}_\gamma$, and let $\tau$ be a truth assignment that satisfies $\varphi$. By item (1) of Lemma~\ref{lem:mapping}, $|\mathsf{Map}(\tau)| \le \card{\rep{\mathsf{db}_k(\varphi)}{\dep}}$, and since $\tau$ satisfies all $m$ clauses of $\varphi$, by item (3) of Lemma~\ref{lem:mapping}, $|\mathsf{Map}(\tau)| = 2^{k \cdot m}$. Hence, $\card{\rep{\mathsf{db}_k(\varphi)}{\dep}} \ge 2^{k \cdot m}$.
		
		We now establish item (2). Assume that $\varphi \not\in \mathsf{Gap3SAT}_\gamma$. By items (1) and (2) of Lemma~\ref{lem:mapping}, $\rep{\mathsf{db}_k(\varphi)}{\dep} = \bigcup\limits_\tau \mathsf{Map}(\tau)$. Hence, by item (3) of Lemma~\ref{lem:mapping}, it holds that
		\[
		\card{\rep{\mathsf{db}_k(\varphi)}{\dep}}\ \le\ \sum\limits_\tau 2^{k \cdot c_\tau}.
		\]
		Since $\varphi$ is a no instance, for any truth assignment $\tau$, $c_\tau \le (\frac{7}{8} + \gamma) \cdot m$. Consequently, we that
		\[
		\card{\rep{\mathsf{db}_k(\varphi)}{\dep}}\ \le\ \sum\limits_\tau 2^{k \cdot (\frac{7}{8} + \gamma) \cdot m} = 2^n \cdot 2^{(\frac{7}{8} + \gamma) \cdot k \cdot m},
		\]
		and the claim follows.
\end{proof}

We now show that the``gap'' between the bounds from Lemma~\ref{lem:gap} can be sufficiently large.

\def\lemgapratio{
	Consider a 3CNF formula with $n>0$ variables and $m>0$ clauses. There are $\gamma \in (0,\frac{1}{8})$, $\epsilon \in (0,1)$, and an integer $k>0$ such that $(\frac{7}{8} + \gamma) \cdot k$ is an integer, and
	\[
	\dfrac{2^{k \cdot m}}{2^n \cdot 2^{(\frac{7}{8} + \gamma) \cdot k \cdot m}} > \dfrac{1+\epsilon}{1-\epsilon}.
	\]
}

\begin{lemma}[\textbf{Gap}]\label{lem:gap-ratio}
	\lemgapratio
\end{lemma}

\begin{proof}
	The claim holds, e.g., for $\gamma = \frac{1}{16}$, $\epsilon = \frac{1}{3}$, and $k = 112$; note that $0 < n \le 3 \cdot m$.
\end{proof}

We now have all the ingredients for concluding the proof of item (2) of Theorem~\ref{the:apx-main-result}.
Let $\gamma$, $\epsilon$ and $k$ be the numbers provided by Lemma~\ref{lem:gap-ratio}.
Assume $\mathsf{A}$ is an FPRAS for $\sharp \prob{Repairs}(\dep)$. We can place $\mathsf{Gap3SAT}_\gamma$ in $\mathsf{BPP}$ via the randomized algorithm $\mathsf{BPPAlgo}_{\gamma,\epsilon,k}$ that takes as input a 3CNF formula $\varphi$ with $n > 0$ variables and $m > 0$ clauses, and does the following:
\begin{itemize}
	\item[-] if $\mathsf{A}(\mathsf{db}_k(\varphi),\epsilon,\frac{1}{3}) > (1+\epsilon) \cdot 2^n \cdot 2^{(\frac{7}{8}+ \gamma) \cdot k \cdot m}$, then \textbf{Accept}; \item[-] otherwise, \textbf{Reject}.
\end{itemize}
It is clear that $\mathsf{BPPAlgo}_{\gamma,\epsilon,k}(\varphi)$ runs in polynomial time in $||\varphi||$ since both the database $\mathsf{db}_k(\varphi)$ and $2^n \cdot 2^{(\frac{7}{8}+ \gamma) \cdot k \cdot m}$ can be computed in polynomial time in $||\varphi||$; recall that, by Lemma~\ref{lem:gap-ratio}, $(\frac{7}{8} + \gamma) \cdot k$ is an integer.
It remains to prove that the algorithm is correct, for which we exploit Lemmas~\ref{lem:gap} and~\ref{lem:gap-ratio}.
We only consider the case when $\varphi \in \mathsf{Gap3SAT}_\gamma$, as the case $\varphi \not \in \mathsf{Gap3SAT}_\gamma$ is similar. The probability that $\mathsf{BPPAlgo}_{\gamma,\epsilon,k}$ is wrong, i.e., \emph{rejects} $\varphi$ is
\[
p\ =\ \Pr\left(\mathsf{A}(\mathsf{db}_k(\varphi),\epsilon,\frac{1}{3}) \le (1 + \epsilon) \cdot 2^n \cdot 2^{(\frac{7}{8}+ \gamma) \cdot k \cdot m}\right).
\]
We can conclude, by combining Lemma~\ref{lem:gap-ratio} and Lemma~\ref{lem:gap}, that \[
(1+\epsilon) \cdot 2^n \cdot 2^{(\frac{7}{8} + \gamma) \cdot k \cdot m}\ <\ (1-\epsilon) \cdot \card{\rep{\mathsf{db}_k(\varphi)}{\dep}},
\]
and therefore,
\[
p\ \le\
\Pr\left(\mathsf{A}(\mathsf{db}_k(\varphi),\epsilon,\frac{1}{3}) < (1 - \epsilon) \cdot \card{\rep{\mathsf{db}_k(\varphi)}{\dep}} \right).
\]
Recall that $\Pr(E_1) \le \Pr(E_1 \cup E_2)$, for any two events $E_1,E_2$. Thus, $p \le \Pr(E_1 \cup E_2)$, where $E_1$ is the event $\mathsf{A}(\mathsf{db}_k(\varphi),\epsilon,\frac{1}{3}) < (1 - \epsilon) \cdot \card{\rep{\mathsf{db}_k(\varphi)}{\dep}}$, and $E_2$ is the event $\mathsf{A}(\mathsf{db}_k(\varphi),\epsilon,\frac{1}{3}) > (1 + \epsilon) \cdot \card{\rep{\mathsf{db}_k(\varphi)}{\dep}}$.
Thus, $\Pr(E_1 \cup E_2)$ is the probability that the event
\[
(1-\epsilon) \cdot \card{\rep{\mathsf{db}_k(\varphi)}{\dep}} \le \mathsf{A}(\mathsf{db}_k(\varphi),\epsilon,\frac{1}{3})\ \le\ (1+\epsilon) \cdot \card{\rep{\mathsf{db}_k(\varphi)}{\dep}}
\]
does not hold. Therefore, by definition of $\mathsf{A}$, $\mathsf{BPPAlgo}_{\gamma,\epsilon,k}$ rejects $\varphi$ with probability at most $\frac{1}{3}$, i.e., accepts with probability at least $\frac{2}{3}$, and the claim follows.

\subsection{The Relative Frequency Problem}

Before concluding this section, let us discuss the relative frequency problem, which is interesting in its own right.
In Section~\ref{sec:exact-counting}, for establishing Theorem~\ref{the:fds-dichotomy}, we focused on the problem of computing the relative frequency of a CQ w.r.t.~a database and a set of FDs.
Recall that the relative frequency of a CQ $Q$ w.r.t.~a database $D$ and a set $\dep$ of FDs is the ratio that computes the percentage of repairs that entail it, that is, the ratio
\[
\rfreq{Q}{D,\dep}\ =\ \frac{\card{\rep{D}{\dep,Q}}}{\card{\rep{D}{\dep}}},
\]
and the relative frequency problem for a set $\dep$ of FDs and a CQ $Q$ is defined as follows:

\medskip

\begin{center}
	\fbox{\begin{tabular}{ll}
			{\small PROBLEM} : & $\prob{RelFreq}(\dep,Q)$
			\\
			{\small INPUT} : & A database $D$
			\\
			{\small OUTPUT} : &  $\rfreq{Q}{D,\dep}$
	\end{tabular}}
\end{center}

\medskip

\noindent Unlike Theorem~\ref{the:fds-dichotomy}, Theorem~\ref{the:apx-main-result} (the main result of Section~\ref{sec:apx-counting}) was shown by directly considering the problem $\sharp \prob{Repairs}(\dep,Q)$, without relying on the relative frequency problem.
An interesting question that comes up is whether the relative frequency problem is approximable or not; the notion of FPRAS for $\prob{RelFreq}(\dep,Q)$ is defined in the obvious way.
Interestingly, by exploiting Theorem~\ref{the:apx-main-result}, we can show the following result concerning the approximability of the relative frequency problem:

\def\thmapxrelfreq{
	The following hold:
	\begin{enumerate}
		\item Consider a set $\dep$ of FDs with an LHS chain (up to equivalence), and a CQ $Q$. $\prob{RelFreq}(\dep,Q)$ admits an FPRAS.
		
		\item There exists a set $\dep$ of FDs and an SJFCQ $Q$ such that $\prob{RelFreq}(\dep,Q)$ does not admit an FPRAS, unless $\text{\rm NP} \subseteq \text{\rm BPP}$.
	\end{enumerate}
}

\begin{theorem}\label{the:apx-relfreq}
	\thmapxrelfreq
\end{theorem}

\textit{Proof (Sketch).}  Item (1) is a consequence of item (1) of Theorem~\ref{the:apx-main-result}, which states that $\prob{Repairs}(\dep,Q)$ admits an FPRAS, and Proposition~\ref{pro:counting-repairs-lhs-fp}, which states that $\card{\rep{D}{\dep}}$ is computable in polynomial time in $||D||$. Indeed, since the problem of computing the numerator admits an FPRAS, and the problem of computing the denominator can be solved in polynomial time, it is easy to convert the FPRAS for computing the numerator (i.e., $\prob{Repairs}(\dep,Q)$) to an FPRAS for $\prob{RelFreq}(\dep,Q)$.

We now present the main ideas underlying the proof of item (2); the details can be found in the appendix. By item (2) of Theorem~\ref{the:apx-main-result},  $\sharp \prob{Repairs}(\hat{\dep})$ with $\hat{\dep} = \{R : A_1 \ra A_2, R: A_3 \ra A_4\}$, where $(A_1,A_2,A_3,A_4)$ is the tuple of attribute names of $R/4$, does not admit an FPRAS (unless $\text{\rm NP} \subseteq \text{\rm BPP}$). Therefore, to establish item (2), it suffices to provide an approximation-preserving reduction from $\sharp \prob{Repairs}(\hat{\dep})$ to $\prob{RelFreq}(\dep,Q)$ for some set $\dep$ of FDs and an SJFCQ $Q$. The latter means that $\dep$ and $Q$ should be devised in such a way that $\prob{RelFreq}(\dep,Q)$ admits an FPRAS implies $\sharp \prob{Repairs}(\hat{\dep})$ admits an FPRAS.
To this end, we are going to define $\dep$ and $Q$ over $\{R'/6\}$ in a way that the following holds: for every database $D$ over $\{R\}$, we can construct in polynomial time in $||D||$ a database $D'$ over $\{R'/6\}$ such that the following equation (which we dubbed $\blacklozenge$) holds:
\[
\rfreq{Q}{D',\dep}\ =\ \frac{1}{1 + \card{\rep{D}{\hat{\dep}}}}.
\]
Then, by exploiting the equation $\blacklozenge$, the fact that $D'$ can be constructed in polynomial time, and the assumption that $\prob{RelFreq}(\dep,Q)$ admits in FPRAS, we are going to devise an FPRAS for the problem $\sharp \prob{Repairs}(\hat{\dep})$, which will complete the proof of item (2).

Assume that $(A,B,A_1,A_2,A_3,A_4)$ is the tuple of attribute names of $R'/6$. Let $\dep$ be the set
\[
\{R' : A_1 \ra A_2,\,\, R' : A_3 \ra A_4,\,\, R' : A \ra B\}.
\]
Now, the SJFCQ $Q$ is defined as the Boolean query
\[
\textrm{Ans}\ \text{:-}\ R'(x,x,x,x,x,x).
\]
In simple words, $Q$ asks whether there exists a fact such that all the attributes have the same value.
We proceed to show that $\dep$ and $Q$ enjoy the desired property.

Consider an arbitrary database $D$ over $\{R\}$. We define the database $D'$ over $\{R'\}$ as 
\[
\{R'(a,b,a_1,a_2,a_3,a_4) \mid R(a_1,a_2,a_3,a_4) \in D\}\ \cup\ \{R'(a,a,a,a,a,a)\}
\]
with $a,b$ being constants not in $\adom{D}$.
It is clear that $D'$ can be constructed in polynomial time in $||D||$. Moreover, it is easy to verify that, by construction, $\card{\rep{D'}{\dep}} = 1 + \card{\rep{D}{\hat{\dep}}}$, and only one of the repairs of $D'$ w.r.t.~$\dep$ entails the Boolean CQ $Q$, that is, the repair consisting only of the fact $R'(a,a,a,a,a,a)$. Therefore, the desired equation $\blacklozenge$ holds.

We proceed to devise an FPRAS for $\sharp \prob{Repairs}(\hat{\dep})$ by exploiting the equation $\blacklozenge$, the fact that $D'$ can be constructed in polynomial time, and the FPRAS $\mathsf{A}'$ for $\prob{RelFreq}(\dep,Q)$, which we assume it exists.
Given a database $D$ over $\{R\}$, $\epsilon > 0$, and $0 < \delta < 1$, we define $\mathsf{A}$ as the procedure:
\begin{itemize}
	\item[-] compute the database $D'$ from $D$;
	\item[-] let $\epsilon' = \frac{\epsilon}{2+\epsilon}$;
	\item[-] let $r = \max\left\{\frac{1-\epsilon'}{2 \cdot (1+2^{|D|})}, \mathsf{A}'(D',\epsilon',\delta)\right\}$;
	\item[-] output $\frac{1}{r} - 1$.
\end{itemize}
It clear that the above procedure runs in polynomial time in $||D||$, $1/\epsilon$, and $\log(1/\delta)$.
It can be also shown that $\mathsf{A}$ correctly approximates the integer $\card{\rep{D}{\hat{\dep}}}$, which in turn implies that $\mathsf{A}$ is an FPRAS for $\sharp \prob{Repairs}(\hat{\dep})$, as needed.
The latter heavily exploits the fact that $\mathsf{A}'$ is an FPRAS for $\prob{RelFreq}(\dep,Q)$, and the equation $\blacklozenge$ shown above; the full proof can be found in the appendix. \qed

\OMIT{
We show that if an FPRAS exists for $\sharp \prob{Repairs}(\dep_{H})$, then $\mathsf{gap3SAT}_\gamma$ is in $\mathsf{BPP}$, for some $\gamma \in (0, \frac{1}{8})$.
In what follows, fix a 3CNF formula $\varphi = C_1 \wedge \cdots \wedge C_m$, with Boolean variables $x_1, \ldots, x_n$, and let $C_i = \ell^1_i \vee \ell^2_i \vee \ell^3_i$, for each $i \in [m]$. We write $\varphi \in \mathsf{gap3SAT}_\gamma$ (resp., $\varphi \not \in \mathsf{gap3SAT}_\gamma$) to specify that $\varphi$ is a \emph{yes} (resp., \emph{no}) instance of $\mathsf{gap3SAT}_\gamma$. We let $\var{\ell^j_i}$ be the Boolean variable of literal $\ell^j_i$.  Moreover, for a truth value $v \in \{0,1\}$, we let $\litval{\ell^j_i}{v}$ be the truth value of $\ell^j_i$, when its \emph{variable} is assigned the truth value $v$. For example, for the literal $\neg x$, and $v = 0$, $\litval{\neg x}{v} = 1$.

\medskip
\paragraph{The database.} We construct a database from $\varphi$ such that the more clauses of $\varphi$ can be satisfied by a truth assignment, the more its repairs. Moreover, we would like to control by "how much" the number of repairs grows, for each clause that is satisfied.

We consider an integer $k > 0$ as the "growing factor", which we are going to choose later, and define $D^k_\varphi$ as the union of two databases $D^k_T$ and $D^k_C$ over relation $R$, as shown below. 
{\renewcommand{\arraystretch}{1.2}
\begin{center}\small
\begin{tabular}{@{}cccccl@{}}
	\toprule
	\normalsize Var & \normalsize Value & \normalsize Clause & \normalsize D & \\ \midrule
	$\var{\ell^1_i}$ & 0 & $\angletup{C^j_i,\litval{\ell^1_i}{0}}$ & $\litval{\ell^1_i}{0}$ & \rdelim\}{6}{8em}[
	{
		$i \in [m]$, $j \in [k]$
	}
]\\
	$\var{\ell^1_i}$ & 1 & $\angletup{C^j_i,\litval{\ell^1_i}{1}}$ & $\litval{\ell^1_i}{1}$ &\\
	$\var{\ell^2_i}$ & 0 & $\angletup{C^j_i,\litval{\ell^2_i}{0}}$ & $\litval{\ell^2_i}{0}$ &\\
	$\var{\ell^2_i}$ & 1 & $\angletup{C^j_i,\litval{\ell^2_i}{1}}$ & $\litval{\ell^2_i}{1}$ &\\
	$\var{\ell^3_i}$ & 0 & $\angletup{C^j_i,\litval{\ell^3_i}{0}}$ & $\litval{\ell^3_i}{0}$ &\\
	$\var{\ell^3_i}$ & 1 & $\angletup{C^j_i,\litval{\ell^3_i}{1}}$ & $\litval{\ell^3_i}{1}$ &\\ \bottomrule	
	
	\multicolumn{5}{c}{\normalsize \textbf{The database} $D^k_T$}\\\ \\
	
	\toprule
	\normalsize Var & \normalsize Value & \normalsize Clause & \normalsize D & \\ \midrule 
	$\star$ & $\star$ & $\angletup{C^j_i,1}$ & 0 & \rdelim\}{1}{8em}[
	{
		$i \in [m]$, $j \in [k]$
	}
	]\\[1pt] \bottomrule
	\multicolumn{5}{c}{\normalsize \textbf{The database} $D^k_C$} \\
\end{tabular}
\end{center}
}
To make the discussion more clear, we give names to the attributes $A, B, C, D$ of the relation $R$; $\star$ is some arbitrary constant.

Intuitively, for each clause $C_i$ of the formula $\varphi$, the database $D^k_T$ stores $k$ copies $C^1_i,\ldots,C^k_i$ of clause $C_i$. Each copy is stored in $D^k_T$ via 6 tuples, where for each literal of $C_i$, its variable is stored twice, specifying one of two possibles values it can get assigned (i.e., 0 or 1). Moreover, the actual truth value that the literal has w.r.t.\ to the particular value given to its variable is also stored, together with the clause copy. This same literal value is replicated on column D, and its meaning will be explained next. The set of 6 tuples in $D^k_T$ mentioning a particular copy $C^j_i$ of a clause $C_i$ is called a \emph{$C^j_i$-gadget}.

Regarding the database $D^k_C$, it stores one tuple for each copy $C^j_i$ of a clause $C_i$. The goal of each tuple in $D^k_C$ containing $C^j_i$ is to cause a conflict whenever a truth assignment of $\varphi$ "chooses" to assign $1$ to at least one literal of the $C^j_i$-gadget. Thanks to this, the more clauses are satisfied by the truth assignment, the more repairs $D^k_\varphi$ will have. Moreover, the larger $k$, the higher the growing factor.

\medskip
\paragraph{The set of FDs.} The set of functional dependencies $\dep_H$ under the current naming scheme is as follows:
\[
\begin{array}{ll}
	\phi_1 = & R: \text{Var} \ra \text{Value}\\
	\phi_2 = & R: \text{Clause}\ra \text{D}.
\end{array}
\]
The first FD $\phi_1$ simply states that a Boolean variable must be assigned only one of the two values $0$ or $1$. If only this FD was considered, we would have a one-to-one correspondence between truth assignments and repairs of $D^k_\varphi$. Introducing the
FD $\phi_2$, we make a truth assignment $\tau$ actually induce an exponential number of repairs in the number of clauses that are satisfied by $\tau$, thanks to the inconsistencies between the tuples of $D^k_T$ and $D^k_C$ w.r.t.\ $\phi_2$.

\medskip
This property is stated in the following crucial lemma, which we prove in the appendix; for a truth assignment $\tau$, we use $c_\tau$ to denote the number of clauses of $\varphi$ that are true w.r.t.\ $\tau$.

\begin{lemma}\label{lem:mapping}
	There exists a mapping $\mathsf{Map}$ from truth assignments of $\varphi$ to sets of databases, such that the following hold:
	\begin{itemize}
		\item For every truth assignment $\tau$, $\mathsf{Map}(\tau) \subseteq \rep{D^k_\varphi}{\dep_H}$.
		\item For every $D' \in \rep{D^k_\varphi}{\dep}$, there is $\tau$ such that $D' \in \mathsf{Map}(\tau)$.
		\item For every truth assignment $\tau$, $|\mathsf{Map}(\tau)| = 2^{k \cdot c_\tau}$.
	\end{itemize}
\end{lemma}

Hence, each truth assignment $\tau$ of $\varphi$ corresponds to $2^{k \cdot c_\tau}$ repairs of $D^k_\varphi$, and every repair has a corresponding truth assignment. With the above lemma in place, the following key property of the database $D^k_\varphi$ is straightforward.

\begin{lemma}[Bounds]\label{lem:gap}
	For every integer $k>0$, and every rational $\gamma  \in (0,\frac{1}{8})$, the following hold:
	\begin{itemize}
		\item If $\varphi \in \mathsf{gap3SAT}_\gamma$, then $|\rep{D^k_\varphi}{\dep_H}| \ge 2^{k \cdot m}$.
		\item If $\varphi \not \in \mathsf{gap3SAT}_\gamma$, then $|\rep{D^k_\varphi}{\dep_H}| \le 2^n \cdot 2^{(\frac{7}{8} + \gamma) \cdot k \cdot m}$.
	\end{itemize}
\end{lemma}

%

Finally, we can prove that the "gap" between the bounds of the above lemma is sufficiently large.

\begin{lemma}[Gap]\label{lem:gap-ratio}
	There exist rationals $\gamma \in (0,\frac{1}{8})$ and $\epsilon \in (0,1)$, and an integer $k>0$ such that
	$$ \dfrac{2^{k \cdot m}}{2^n \cdot 2^{(\frac{7}{8} + \gamma) \cdot k \cdot m}} > \dfrac{1+\epsilon}{1-\epsilon},$$
	and such that $(\frac{7}{8} + \gamma) \cdot k$ is an integer.
\end{lemma}
\begin{proof}
	It is not difficult to verify that the claim holds, e.g., for $\gamma = \frac{1}{16}$, $\epsilon = \frac{1}{3}$ and $k = 112$; recall that $0 < n \le 3 \cdot m$.
\end{proof}

We can now finally combine the above lemmas to prove our main result. Let $\gamma$, $\epsilon$ and $k$ be the numbers from Lemma~\ref{lem:gap-ratio}.

Assume $\mathcal{A}$ is an FPRAS for $\sharp \prob{Repairs}(\dep_{H})$. We now show that $\mathsf{gap3SAT}_\gamma$ is in $\mathsf{BPP}$. The algorithm, which we dub $\mathsf{BPPAlgo}_{\gamma,\epsilon,k}$ is as follows; $\varphi$ is the input 3CNF formula of the algorithm.
\begin{enumerate}
	\item Construct the database $D^k_\varphi$.
	\item If $\mathcal{A}(D^k_\varphi,\epsilon,\frac{1}{3}) > (1+\epsilon) \cdot 2^n \cdot 2^{(\frac{7}{8}+ \gamma) \cdot k \cdot m}$, then \textbf{Accept}
	\item Else \textbf{Reject}.
\end{enumerate}
The above algorithm clearly runs in polynomial time w.r.t.\ $||\varphi||$, in particular, $2^n \cdot 2^{(\frac{7}{8}+ \gamma) \cdot k \cdot m}$ can be computed in polynomial time.\footnote{Recall that for the numbers $\gamma, \epsilon, k$ of Lemma~\ref{lem:gap-ratio}, $(\frac{7}{8} + \gamma) \cdot k$ is an integer.}

Regarding the correctness, we consider the case when $\varphi \in \mathsf{gap3SAT}_\gamma$, as the case $\varphi \not \in \mathsf{gap3SAT}_\gamma$ is similar.
Assume  $\varphi \in \mathsf{gap3SAT}_\gamma$. The probability that $\mathsf{BPPAlgo}_{\gamma,\epsilon,k}$ \emph{rejects} $\varphi$ is:
$$ p = \Pr\left(\mathcal{A}(D^k_\varphi,\epsilon,\frac{1}{3}) \le (1 + \epsilon) \cdot 2^n \cdot 2^{(\frac{7}{8}+ \gamma) \cdot k \cdot m}\right).$$
We conclude, by combining Lemma~\ref{lem:gap-ratio} and Lemma~\ref{lem:gap}, that $(1+\epsilon) \cdot 2^n \cdot 2^{(\frac{7}{8} + \gamma) \cdot k \cdot m} < (1-\epsilon) \cdot |\rep{D^k_\varphi}{\dep_H}|$.
Hence,
$p \le
\Pr\left(\mathcal{A}(D^k_\varphi,\epsilon,\frac{1}{3}) < (1 - \epsilon) \cdot |\rep{D^k_\varphi}{\dep_H}| \right)$.
Recalling that $\Pr(E_1) \le \Pr(E_1 \cup E_2)$, for any two events $E_1,E_2$, we obtain $p \le \Pr(E_1 \cup E_2)$, where $E_1$ is the event $\mathcal{A}(D^k_\varphi,\epsilon,\frac{1}{3}) < (1 - \epsilon) \cdot |\rep{D^k_\varphi}{\dep_H}|$ and $E_2$ is the event $\mathcal{A}(D^k_\varphi,\epsilon,\frac{1}{3}) > (1 + \epsilon) \cdot |\rep{D^k_\varphi}{\dep_H}|$.
Thus, $\Pr(E_1 \cup E_2)$ is precisely the probability that the event
$$(1-\epsilon) \cdot |\rep{D^k_\varphi}{\dep_H}| \le \mathcal{A}(D^k_\varphi,\epsilon,\frac{1}{3}) \le (1+\epsilon) \cdot |\rep{D^k_\varphi}{\dep_H}|$$
does not hold. Thus, by definition of $\mathcal{A}$, $\mathsf{BPPAlgo}_{\gamma,\epsilon,k}$ rejects $\varphi$ with probability at most $\frac{1}{3}$, i.e., accepts with probability at least $\frac{2}{3}$.

\OMIT{
It is now easy to reduce $\sharp \prob{Repairs}(\dep_{H})$ to $\sharp \prob{Repairs}(\dep_{H},Q)$, for every SJFCQ $Q$, by exploiting Lemma~\ref{lem:cook-reduction}:

\begin{theorem}\label{the:no-fpras-fds-cq}
	Consider $R/4 \in \ins{S}$, with attributes $(A,B,C,D)$, the set $\dep_H = \{R: A \ra B, R: C \ra D\}$ and an SJFCQ $Q$. $\sharp \prob{Repairs}(\dep_{H},Q)$ admits no FPRAS, unless NP $\subseteq$ BPP.
\end{theorem}
}
}
\section{Conclusions and Future Work}\label{sec:future-work}

This work focused on the exact and approximate counting of database repairs.
Concerning the exact version of the problem, we have lifted the \text{\rm FP}/$\sharp$\text{\rm P}-complete dichotomy in the case of primary keys and self-join-free CQs from~\cite{MaWi13} to the more general case of FDs and self-join-free CQs (Theorem~\ref{the:fds-dichotomy}).
Concerning the approximate version of the problem, although we have not provided a complete approximability/inapproximability classification, we established that (i) the problem admits an FPRAS in the case of FDs with an LHS chain (up to equivalence) and CQs (even with self-joins), but (ii) it does not admit an FPRAS (under a standard complexity assumption) in the case of arbitrary FDs; note that the latter holds for {\em every} self-join-free CQ (Theorem~\ref{the:apx-main-result}). These are crucial steps towards a complete approximability/inapproximability classification. Recall that establishing such a classification is highly non-trivial as it will resolve the challenging open problem of whether counting maximal matchings in a bipartite graph admits an FPRAS. 
Note that on our way to understand the problem of counting database repairs, we also established analogous result for the relative frequency problem, which is interesting in its own right.

The research problems in the context of  counting repairs that remain open are the following: 
\begin{enumerate}
	\item Lift the dichotomy of Theorem~\ref{the:fds-dichotomy} to arbitrary CQs with self-joins.
	\item 
Establish an approximability/inapproximability dichotomy for FDs and CQs (with or without self-joins).  
\end{enumerate}
These are highly non-trivial problems that deserve our attention.

\medskip


\bibliographystyle{plain}
\bibliography{references}

\newpage
\appendix

\section*{Appendix}

We provide full proofs and additional explanations for the technical results given in the main body of the paper that need to be proven. For the sake of readability, we repeat the statements.

\newenvironment{repeatresult}[2]
{\vskip0.5em\par\textbf{#1} #2.\em}
{\vskip1em}
\newenvironment{replemma}[1]{\begin{repeatresult}{Lemma}{#1}}{\end{repeatresult}}
\newenvironment{repproposition}[1]{\begin{repeatresult}{Proposition}{#1}}{\end{repeatresult}}
\newenvironment{reptheorem}[1]{\begin{repeatresult}{Theorem}{#1}}{\end{repeatresult}}
\section{Proofs of Section~\ref{sec:exact-counting}}\label{sec:appendix-exact}

In this section, we provide the missing proofs of Section~\ref{sec:exact-counting}. To this end, we first establish some auxiliary technical lemmas that will be used throughout the proofs.
Recall that we consider canonical FD sets $\dep$ that have a single LHS chain $R : X_1 \ra Y_1,\ldots,R : X_n \ra Y_n$ for every relation name $R$ of the schema such that:
\begin{enumerate}
    \item $X_i \subsetneq X_{i+1}$ for each $i \in [n]$,
    \item $X_i \cap Y_j = \emptyset$ for each $i,j \in [n]$, and
    \item $Y_i \cap Y_j = \emptyset$ for each $i,j \in [n]$ with $i \neq j$.
\end{enumerate}
We use the following notation. Let $\dep$ be an FD set with an LHS chain and let $Q$ be an SJFCQ. For every atom $R(\bar z)$ of $Q$, and for every FD $R:X_{j}\ra Y_{j}$ in the LHS chain $R:X_1\ra Y_1,\dots, R:X_n\ra Y_n$ of $\dep_R$, we denote by $I_{R,j}^\lhs$ the positions corresponding to the attributes of $X_j$, i.e., $I_{R,j}^\lhs=\{(R,A)\mid A\in X_j\}$, and by $I_{R,j}^\rhs$ the positions that correspond to the attributes of $Y_j$, i.e., $I_{R,j}^\rhs=\{(R,A)\mid A\in Y_j\}$. We also denote by $R:X_{i_R}\ra Y_{i_R}$ the primary FD of $\dep_R$ w.r.t.~$Q$.

\subsection{Auxiliary Technical Lemmas }\label{sec:appendix-exact-aux}
We first show that the facts of $D_{\mathsf{conf}}^{\dep,Q}$ have no impact on the number of repairs that entail $Q$. Recall that for a database $D$, a set $\dep$ of FDs with an LHS chain, and an SJFCQ $Q$, the database $D_{\mathsf{conf}}^{\dep,Q}$ contains all the $R$-facts $f$ of $D$ such that for some FD $R:X_j\ra Y_j$ in the primary prefix of $\dep_R$:
\begin{enumerate}
    \item $f[A]=\alpha[A]$ for every $A\in X_j$,
    \item $f[A]\neq\alpha[A]$ for some $A\in Y_j$,
\end{enumerate}
where $\alpha$ is the $R$-atom of $Q$.

\begin{lemma}\label{lemma:help}
Consider a database $D$, a set $\dep$ of FDs with an LHS chain, and an SJFCQ $Q$. Then,
\[
\card{\rep{D}{\dep,Q}}\ =\ \card{\rep{D\setminus D_{\mathsf{conf}}^{\dep,Q}}{\dep,Q}}.
\]
\end{lemma}

\begin{proof}
If $D\not\models Q$, then $(D\setminus D_{\mathsf{conf}}^{\dep,Q})\not\models Q$, and therefore, $\card{\rep{D}{\dep,Q}}=\card{\rep{D\setminus D_{\mathsf{conf}}^{\dep,Q}}{\dep,Q}}=0$. Hence, from now on, we focus on the case where $D\models Q$. We show that a database $E$ is a repair of $D$ w.r.t.~$\dep$ that entails $Q$ if and only if $E$ is a repair of $D\setminus D_{\mathsf{conf}}^{\dep,Q}$ w.r.t.~$\dep$ that entails $Q$.

($\Rightarrow$) Let $E$ be a repair of $D$ w.r.t.~$\dep$ that entails $Q$. Thus, there exists a homomorphism $h$ from $Q$ to $D$ such that $R(h(\bar z)) \in E$ for every atom $R(\bar z)$ of $Q$. Let $f=R(h(\bar z))$ for some atom $R(\bar z)$ of $Q$. Clearly, $f$ agrees with $R(\bar z)$ on the constants in the positions of $I_{R,j}^\lhs$ and $I_{R,j}^\rhs$ for every FD $R:X_j\ra Y_j$ in the primary prefix of $\dep_R$. On the other hand, for every $R$-fact $g\in D_{\mathsf{conf}}^{\dep,Q}$, there exists an FD $R:X_j\ra Y_j$ in the primary prefix of $\dep_R$ such that $g$ agrees with $R(\bar z)$ on the constants in the positions of $I_{R,j}^\lhs$, but disagrees with $R(\bar z)$ on a constant in some position of $I_{R,j}^\rhs$. Therefore, $f$ and $g$ agree on all the values of the attributes in $X_j$ but disagree on the value of at least one attribute in $Y_j$. We conclude that $\{f,g\}\not\models\dep_R$. Since $D_{\mathsf{conf}}^{\dep,Q}$ contains only facts over relations that occur in $Q$ and the above argument holds for each such relation, it holds that $E\cap D_{\mathsf{conf}}^{\dep,Q}=\emptyset$, and $E$ is also a consistent subset of $D\setminus D_{\mathsf{conf}}^{\dep,Q}$. Since every fact of $D\setminus D_{\mathsf{conf}}^{\dep,Q}$ also occurs in $D$, it is rather straightforward that if $E$ is maximal in $D$ then it is also maximal in $D\setminus D_{\mathsf{conf}}^{\dep,Q}$. Thus, $E$ is a repair of $D\setminus D_{\mathsf{conf}}^{\dep,Q}$ that entails $Q$.

($\Leftarrow$) Conversely, let $E$ be a repair of $D\setminus D_{\mathsf{conf}}^{\dep,Q}$ w.r.t.~$\dep$ that entails $Q$. Clearly, $E$ is also a consistent subset of $D$ that entails $Q$; hence, it is only left to show that it is maximal. If we could add a fact from $D\setminus D_{\mathsf{conf}}^{\dep,Q}$ to $E$ without violating consistency, we would get a contradiction to the fact that $E$ is a repair of $D\setminus D_{\mathsf{conf}}^{\dep,Q}$. Moreover, the same argument used in the proof of the first direction shows that we cannot add a fact of $D_{\mathsf{conf}}^{\dep,Q}$ to $E$ without violating consistency; since $E$ entails $Q$, there exists a homomorphism $h$ from $Q$ to $D$ such that $R(h(\bar z))\subseteq E$ for every atom $R(\bar z)$ of $Q$, and every $R$-fact $g\in D_{\mathsf{conf}}^{\dep,Q}$ violates $\dep_R$ jointly with the fact $R(h(\bar z))$. Hence, $E$ is a maximal consistent subset of $D$.
\end{proof}

We next show that the facts of $D_{\mathsf{ind}}^{\dep,Q}$ have no impact on the relative frequency of $Q$ w.r.t.~$D\setminus D_{\mathsf{conf}}^{\dep,Q}$ and $\dep$.  Recall that for a database $D$, a set $\dep$ of FDs with an LHS chain, and an SJFCQ $Q$, the database $D_{\mathsf{ind}}^{\dep,Q}$ contains all the $R$-facts $f$ of $D\setminus D_{\mathsf{conf}}^{\dep,Q}$ such that for some FD $R:X_j\ra Y_j$ in the primary prefix of $\dep_R$, $f[A]\neq\alpha[A]$ for an attribute $A\in X_j$,
where $\alpha$ is the $R$-atom of $Q$.

\begin{lemma}\label{lemma:help2}
Consider a database $D$, a set $\dep$ of FDs with an LHS chain, and an SJFCQ $Q$. Then,
\[
\rfreq{Q}{D\setminus D_{\mathsf{conf}}^{\dep,Q},\dep}\ =\ \rfreq{Q}{D\setminus (D_{\mathsf{conf}}^{\dep,Q}\cup D_{\mathsf{ind}}^{\dep,Q}),\dep}.
\]
\end{lemma}

\begin{proof}
Since every $R$-fact $f$ of $D_{\mathsf{ind}}^{\dep,Q}$ disagrees with the atom $R(\bar z)$ of $Q$ on at least one constant (that occurs in a position of $I_{R,j}^\lhs$ for some FD $R:X_j\ra Y_j$ in the primary prefix of $\dep_R$), there is no homomorphism $h$ from $Q$ to $D\setminus D_{\mathsf{conf}}^{\dep,Q}$ such that $f=R(h(\bar z))$. We now show that no two $R$-facts $g\in (D\setminus (D_{\mathsf{conf}}^{\dep,Q}\cup D_{\mathsf{ind}}^{\dep,Q}))$ and $f
\in D_{\mathsf{ind}}^{\dep,Q}$ jointly violate $\dep_R$. Assume, towards a contradiction, that $\{f,g\}\not\models\dep_R$.
Since $f\in D_{\mathsf{ind}}^{\dep,Q}$, for some FD $R:X_j\ra Y_j$ in the primary prefix of $\dep_R$, it holds that $f$ disagrees with $R(\bar z)$ on a constant at a position of $I_{R,j}^\lhs$. Since $g\not\in (D_{\mathsf{conf}}^{\dep,Q}\cup D_{\mathsf{ind}}^{\dep,Q})$, the fact $g$ agrees with the atom $R(\bar z)$ on all the constants at the positions of $I_{R,j}^\lhs$. Therefore, $f$ and $g$ disagree on the value of at least one attribute of $X_j$, and for every $k\ge j$, it holds that $\{f,g\}\models R:X_k\ra Y_k$, as $X_j\subseteq X_k$. Therefore, the only possible case is that $\{f,g\}\not\models R:X_k\ra Y_k$ for some $k<j$. Note that the FD $R:X_k\ra Y_k$ also belongs to the primary prefix of $\dep_R$.

As said above, $g$ agrees with the atom $R(\bar z)$ on all the constants at the positions of $I_{R,j}^\lhs$. Since $X_k\subset X_j$, $g$ agrees with $R(\bar z)$ on all the constants at the positions of $I_{R,k}^\lhs$. Since $g\not\in D_{\mathsf{conf}}^{\dep,Q}$, we get that $g$ also agrees with $R(\bar z)$ on all the constants at the positions of $I_{R,k}^\rhs$. Hence, $f$ agrees with $R(\bar z)$ on all the constants at the positions of $I_{R,k}^\lhs$, but disagrees with $R(\bar z)$ on a constant at a position of $I_{R,k}^\rhs$, which implies that $f\in D_{\mathsf{conf}}^{\dep,Q}$. This is a contradiction since $f\in D_{\mathsf{ind}}^{\dep,Q}$ and $D_{\mathsf{conf}}^{\dep,Q}\cap D_{\mathsf{ind}}^{\dep,Q}=\emptyset$.

We have shown that for every atom $R(\bar z)$ of $Q$ and every homomorphism $h$ from $Q$ to $D\setminus D_{\mathsf{conf}}^{\dep,Q}$, the fact $R(h(\bar z))$ belongs to $D\setminus (D_{\mathsf{conf}}^{\dep,Q}\cup D_{\mathsf{ind}}^{\dep,Q})$. Moreover, no fact of $D_{\mathsf{ind}}^{\dep,Q}$ is in conflict with a fact of $D\setminus (D_{\mathsf{conf}}^{\dep,Q}\cup D_{\mathsf{ind}}^{\dep,Q})$. Therefore, it holds that:
\[
\card{\rep{D\setminus D_{\mathsf{conf}}^{\dep,Q}}{\dep,Q}}\ =\ \card{\rep{D\setminus (D_{\mathsf{conf}}^{\dep,Q}\cup D_{\mathsf{ind}}^{\dep,Q})}{\dep,Q}}\times \card{\rep{D_{\mathsf{ind}}^{\dep,Q}}{\dep}}
\]
and
\[
\card{\rep{D\setminus D_{\mathsf{conf}}^{\dep,Q}}{\dep}}\ =\ \card{\rep{D\setminus (D_{\mathsf{conf}}^{\dep,Q}\cup D_{\mathsf{ind}}^{\dep,Q})}{\dep}}\times \card{\rep{D_{\mathsf{ind}}^{\dep,Q}}{\dep}}.
\]
We can then conclude that
\[
\rfreq{Q}{D\setminus D_{\mathsf{conf}}^{\dep,Q},\dep}\ =\ \rfreq{Q}{D\setminus (D_{\mathsf{conf}}^{\dep,Q}\cup D_{\mathsf{ind}}^{\dep,Q}),\dep},
\]
and the claim follows.
\end{proof}

The following is an immediate corollary of Lemmas~\ref{lemma:help} and~\ref{lemma:help2}.

\begin{replemma}{\ref{lem:combined_help}}
\lemmacombinedhelp
\end{replemma}

The next lemma generalizes a property of primary keys: two $R$-facts that disagree on the value of an attribute in the left-hand side of the key of $R$ cannot jointly violate the constraints.

\begin{lemma}\label{lemma:basic}
Consider a database $D$, a set $\dep$ of FDs with an LHS chain, and an SJFCQ $Q$. For every relation name $R$ that occurs in $Q$, if two $R$-facts $f,g\in D\setminus (D_{\mathsf{conf}}^{\dep,Q}\cup D_{\mathsf{ind}}^{\dep,Q})$ are such that $f[A]\neq g[A]$ for some primary-lhs position $(R,A)$, then $\set{f,g}\models\dep_R$.
\end{lemma}

\begin{proof}
By the definition of $D_{\mathsf{conf}}^{\dep,Q}$ and $D_{\mathsf{ind}}^{\dep,Q}$, the $R$-facts $f$ and $g$ agree with the atom $R(\bar z)$ of $Q$ on all the constants that occur at the positions of $I_{R,j}^\lhs\cup I_{R,j}^\rhs$ for every FD $R:X_j\ra Y_j$ in the primary prefix of $\dep_R$.
Hence, $f$ and $g$ do not jointly violate any of these FDs. If $\dep_R$ has no primary FD, these are the only FDs in $\dep_R$, and the claim follows.
If $\dep_R$ does have a primary FD, for every FD $R:X_k\ra Y_k$ outside the primary prefix, i.e., every FD that occurs after the primary FD $R:X_{i_R}\ra Y_{i_R}$, $f$ and $g$ disagree on the value of at least one attribute of $X_k$ since $X_{i_R}\subseteq X_k$ and $f$ and $g$ disagree on the value of some primary-lhs position, and thus, $\set{f,g}\models\dep_R$.
\end{proof}

\rev{We finish the section with the following lemma.}

\begin{lemma}\label{lemma:no-complex-no-conflicts}
 \rev{Consider a database $D$, a set $\dep$ of FDs with an LHS chain, and an SJFCQ $Q$ with $\comp{Q}{\dep} = \emptyset$. If $D\models Q$, then
\[
\card{\rep{D\setminus D_{\mathsf{conf}}^{\dep,Q}}{\dep,Q}}=\card{\rep{D\setminus D_{\mathsf{conf}}^{\dep,Q}}{\dep}}.
\]}   
\end{lemma}

\begin{proof}
\rev{We will show that a database $E$ is a repair of $D\setminus D_{\mathsf{conf}}^{\dep,Q}$ w.r.t.~$\dep$ that entails $Q$ if and only if $E$ is a repair of $D\setminus D_{\mathsf{conf}}^{\dep,Q}$ w.r.t.~$\dep$. The ``only if'' direction is straightforward; hence, we focus on the other direction.
Let $E$ be a repair of $D\setminus D_{\mathsf{conf}}^{\dep,Q}$ w.r.t.~$\dep$. By contradiction, assume that $E\not\models Q$. Since $D\models Q$, there is a homomorphism $h$ from $Q$ to $D$. Let $R(\bar z)$ be an atom of $Q$. By definition, $R(\bar z)$ associates a constant with every position of $I_{R,j}^\lhs\cup I_{R,j}^\rhs$ for the FDs $R:X_j\ra Y_j$ in the primary prefix of $\dep_R$. Therefore, the fact
$R(h(\bar z))$ of $D$ uses these constants for the corresponding attributes. Since every fact of $D_{\mathsf{conf}}^{\dep,Q}$ disagrees with $R(\bar z)$ on at least one of those constants, $R(h(\bar z))\in (D\setminus D_{\mathsf{conf}}^{\dep,Q})$, and $h$ is a homomorphism from $Q$ to $D\setminus D_{\mathsf{conf}}^{\dep,Q}$.}

 \rev{
Now, assume that the atoms of $Q$ are $R_1(\bar z_1),\dots, R_n(\bar z_n)$ and let $h$ be a homomorphism from $Q$ to $D\setminus D_{\mathsf{conf}}^{\dep,Q}$ such that $|\set{R_t(\bar z_t)\mid R_t(h(\bar z_t))\in E}|=k$ for some $k$, and there is no other homomorhism $h'$ from $Q$ to $D\setminus D_{\mathsf{conf}}^{\dep,Q}$ such that $|\set{R_t(\bar z_t)\mid R_t(h'(\bar z_t))\in E}|>k$. Note that $k<n$ due to our assumption that $E\not\models Q$. Let $R_t(\bar z_t)$ be an atom of $Q$ such that $R_t(h(\bar z_t))\not\in E$. We denote $f=R_t(h(\bar z_t))$. Since we know that $f\in (D\setminus D_{\mathsf{conf}}^{\dep,Q})$ and $f\not\in E$, and \e{$E$ is a maximal consistent subset of $D\setminus D_{\mathsf{conf}}^{\dep,Q}$}, there exists a fact $g\in E$ such that $\{f,g\}\not\models\dep_{R_t}$.}

\rev{
For every $R_t:X_j\ra Y_j$ in the primary prefix of $\dep_{R_t}$, we know that $f$ agrees with the atom $R_t(\bar z_t)$ of $Q$ on the constants associated with the attributes of $X_j\cup Y_j$. Since $g\not\in D_{\mathsf{conf}}^{\dep,Q}$, it is either the case that $g$ also agrees with $R_t(\bar z_t)$ on these constants, or $g$ disagrees with $R_t(\bar z_t)$ on some value in an attribute of $X_j$. In both cases, we have that $\{f,g\}\models (R_t:X_j\ra Y_j)$. Note that if $R_t$ has no primary FD, we can conclude at this point that there is no $g\in E$ such that $\{f,g\}\not\models\dep_{R_t}$ (because all the FDs of $\dep_{R_t}$ belong to the primary prefix), and we immediately get a contradiction to the fact that $E$ is a repair (hence, maximal). Thus, in the rest of the proof, we consider the case where $R_t$ has a primary FD. In this case, $\{f,g\}\not\models (R_t:X_k\ra Y_k)$ for some FD outside the primary prefix, and since $X_{i_{R_t}}\subseteq X_k$ for each such FD, we conclude that $f$ and $g$ agree on the values of all the attributes in $X_{i_{R_t}}$. Note that in this case, $f$ and $g$ also agree on the values of all the attributes of $X_j$ for each FD $R_t:X_j\ra Y_j$ in the primary prefix of $\dep_{R_t}$ since $X_j\subset X_{i_{R_t}}$. Therefore, $g$ agrees with the atom $R_t(\bar z_t)$ on the constants associated with the attributes of $X_j\cup Y_j$ for each such FD. }

\rev{
Let $h' : \var{Q} \cup \const{Q} \ra \adom{D}$ be a function such that $h'(c)=c$ for every constant $c$ and $h'(x)=h(x)$ for every variable $x$ except those that occur at non-primary-lhs positions of $R_t(\bar z_t)$. For every variable $y$ that occurs at a non-primary-lhs position $(R_t,A)$ of $R_t(\bar z_t)$, we define $h'(y)=g[A]$. Since $\comp{Q}{\dep}=\emptyset$, each such variable occurs only once in the query. Therefore, for every $\ell\neq t$, it holds that $R_\ell(h'(\bar z_\ell))=R_\ell(h(\bar z_\ell))$ and $R_\ell(h'(\bar z_\ell))\in (D\setminus D_{\mathsf{conf}}^{\dep,Q})$. Moreover, $R_t(h'(\bar z_t))\in (D\setminus D_{\mathsf{conf}}^{\dep,Q})$ because $R_t(h'(\bar z_t))=g$; thus, $h'$ is a homomorphism from $Q$ to $D\setminus D_{\mathsf{conf}}^{\dep,Q}$. Clearly, we have that $R_\ell(h'(\bar z_\ell))\in E$ whenever $R_\ell(h(\bar z_\ell))\in E$, and it also holds that $R_t(h'(\bar z_t))\in E$ because $g\in E$. Therefore, we obtained a homomorphism from $Q$ to $D\setminus D_{\mathsf{conf}}^{\dep,Q}$ such that $|\set{R_t(\bar z_t)\mid R_t(h'(\bar z_t))\in E}|= k+1>k$, which is a contradiction to our assumption.}
\end{proof}

\subsection{Proof of Proposition~\ref{pro:counting-repairs-no-complex-part}}

\begin{repproposition}{\ref{pro:counting-repairs-no-complex-part}}
\propnocomplex
\end{repproposition}
\begin{proof}
Since $\dep$ has an LHS chain, $\card{\rep{D}{\dep}}$ can be computed in polynomial time in $||D||$, and the computation of $\rfreq{Q}{D,\dep}$ boils down to the computation of $\card{\rep{D}{\dep,Q}}$. Hence, it remains to show that $\card{\rep{D}{\dep,Q}}$ can be computed in polynomial time if $\comp{Q}{\dep} = \emptyset$. If $D\not\models Q$, then clearly no repair of $D$ entails $Q$, and $\card{\rep{D}{\dep,Q}}=0$. Therefore, in the rest of the proof, we focus on the case where $D\models Q$. Our proof for this case is by a reduction to the problem of counting repairs w.r.t.~$\dep$. \rev{In particular, 
Lemma~\ref{lemma:help} implies that \[\card{\rep{D}{\dep,Q}}=\card{\rep{D\setminus D_{\mathsf{conf}}^{\dep,Q}}{\dep,Q}}.\] Moreover, Lemma~\ref{lemma:no-complex-no-conflicts} implies that 
\[
\card{\rep{D\setminus D_{\mathsf{conf}}^{\dep,Q}}{\dep,Q}}=\card{\rep{D\setminus D_{\mathsf{conf}}^{\dep,Q}}{\dep}}.
\]
Hence, we conclude that 
\[
\card{\rep{D}{\dep,Q}}\ =\ \card{\rep{D\setminus D_{\mathsf{conf}}^{\dep,Q}}{\dep}}.
\]
The claim then follows from Proposition~\ref{pro:counting-repairs-lhs-fp}. }
\end{proof}

Observe that Proposition~\ref{pro:counting-repairs-no-complex-part} also captures the case where the atoms of $Q$ use only constants (and no variables), since, by definition, all the positions are primary-lhs positions, and $Q$ has no complex part. In the corresponding algorithm for primary keys, the case of a single atom with no variables is considered separately~\cite{MaWi13}.

\subsection{Proof of Lemma~\ref{lem:aux-2}}

Let us first recall that 
\[ 
D^{\dep,Q}\ =\ D\setminus(D_{\mathsf{conf}}^{\dep,Q}\cup D_{\mathsf{ind}}^{\dep,Q}) \qquad \text{and} \qquad \mathsf{R}_{D,\dep,Q}=\frac{\card{\repp{D \setminus D_{\mathsf{conf}}^{\dep,Q}}{\dep}}}{\card{\rep{D}{\dep}}}.
\]

\begin{replemma}{\ref{lem:aux-2}}
\lemmaauxsecond
\end{replemma}
\begin{proof}
We prove that
\[
\rfreq{Q}{D^{\dep,Q},\dep}\ =\ 1 - \prod_{c \in \adom{D}} \left(1 - 	\rfreq{Q_{x\mapsto c}}{D^{\dep,Q},\dep} \right)
\]
and then the desired result follows from Lemma~\ref{lem:combined_help}.

Let $Q_1=Q\setminus\comp{Q}{\dep}$, and let $E_1$ be the subset of $D^{\dep,Q}$ that contains all the facts over the relations of $Q_1$. Let $E_2$ be the subset of $D^{\dep,Q}$ that contains all the facts over the relations that do not occur in $Q$. Finally, let $E_3$ be the subset of $D^{\dep,Q}$ that contains all the facts over the relations that occur in $\comp{Q}{\dep}$. We denote these relations by $R_1,\dots,R_k$. For simplicity, we assume that the variable $x$ occurs at the left-most position of each one of the atoms $R_1(\bar z_1),\dots,R_k(\bar z_k)$ of $Q$, and that this position is a primary-lhs position. (Recall that $x \in \pvar{R_t(\bar z_t)}{\dep}$ for every $t \in [k]$.)

For every $c\in\adom{D}$, let $E_3^{c}$ be the subset of $E_3$ that contains all the facts over the relations $R_1,\dots,R_k$ that use the value $c$ at the left-most position. Note that two $R_t$-facts $f\in E_3^{c}$ and $g\in E_3^{b}$ for $c\neq b$ cannot jointly violate $\dep_{R_t}$. This is because $f,g\in D \setminus \left(D_{\mathsf{conf}}^{\dep,Q} \cup D_{\mathsf{ind}}^{\dep,Q}\right)$ and they use a different constant at some primary-lhs position; hence, Lemma~\ref{lemma:basic} implies that $\set{f,g}\models\dep_{R_t}$.

Now, consider the query $Q_1$. Since it holds that $\comp{Q_1}{\dep}=\emptyset$ and $D^{\dep,Q}$ contains no facts of $D_{\mathsf{conf}}^{\dep,Q}$, \rev{Lemma~\ref{lemma:no-complex-no-conflicts}} implies that all the repairs of $E_1$ entail $Q_1$. Moreover, if $J_1,\dots,J_m$ are the repairs of $E_1$ and $K$ is a repair of $E_3$, it holds that if $J_\ell\cup K$ entails $Q$ for some $\ell\in\set{1,\dots,m}$, then $J_\ell\cup K$ entails $Q$ for all $\ell\in\set{1,\dots,m}$. This holds since $\comp{Q_1}{\dep}=\emptyset$; thus, the only variables of $Q_1$ that occur more than once in $Q$ appear at primary-lhs positions of the atoms in $Q_1$, and the same holds for the constants that occur in $Q_1$. Moreover, Lemma~\ref{lemma:basic} implies that for every relation name $R$ in $Q$, two $R$-facts in $D^{\dep,Q}$ are in conflict only if they agree on all the attributes corresponding to primary-lhs positions of the atom $R(\bar z)$ of $Q$. Hence, if there exists a homomorphism $h$ from $Q$ to $J_\ell\cup K$ for some $\ell\in\set{1,\dots,m}$, then for every $\ell\in\set{1,\dots,m}$, there exists a homomorphism $h'$ from $Q$ to $J_\ell\cup K$ (that only differs from $h$ on the values of the variables that occur once in $Q_1$). Therefore, we first limit our discussion to the repair $J_1$.

In what follows, for a database $K$, we denote by $\card{\rep{K}{\dep,\neg Q}}$ the number of repairs of $K$ w.r.t.~$\dep$ that do not entail $Q$ (i.e., $\card{\rep{K}{\dep,\neg Q}}=\card{\rep{K}{\dep}}-\card{\rep{K}{\dep, Q}}$), and by $\rfreq{\neg Q}{K,\dep}$ the ratio $\frac{\card{\rep{K}{\dep,\neg Q}}}{\card{\rep{K}{\dep}}}$. For every $c\in\adom{D}$, let
$m_c = \card{\rep{E_3^c\cup J_1}{\dep,\neg Q}}$ and $n_c = \card{\rep{E_3^c}{\dep}}.$ Since all the facts of $E_3^c$ use the constant $c$ at their left-most position, associated with the variable $x$ in the corresponding atom, we have that:
\[
\card{\rep{E_3^c\cup J_1}{\dep,\neg Q}}\ =\ \card{\rep{E_3^c\cup J_1}{\dep,\neg Q_{x\mapsto c}}}.
\]
Clearly, it holds that:
\[
D^{\dep,Q}\ =\ E_1\cup E_2\cup\left(\bigcup_{c\in\adom{D}}E_3^c\right).
\]
Therefore, if $\card{\rep{E_1}{\dep}}=n_1$ and $\card{\rep{E_2}{\dep}}=n_2$, then we have that:
\[
\card{\rep{D^{\dep,Q}}{\dep}}\ =\ n_1\times n_2\times \Pi_{c\in\adom{D}}n_c
\]
since, as explained above, there are no conflicts among facts from $E_3^c$ and $E_3^b$ for $c\neq b$. Moreover,
\[
\card{\rep{D^{\dep,Q}}{\dep,\neg Q_{x\mapsto c}}}\ =\ m_c\times n_1\times n_2\times \Pi_{b\in\adom{D},b\neq c}n_b.
\]
This holds since the repairs of $E_2$ do not affect the entailment of the query $Q_{x\mapsto c}$ (as $E_2$ contains only facts over relation names that do not occur in $Q$), and so do not the repairs of $E_3^b$ for $b\neq c$. Moreover, as aforementioned, if $K\cup J_1$ entails $Q_{x\mapsto c}$ for some repair $K$ of $E_3^c$, then $K\cup J_\ell$ also entails $Q_{x\mapsto c}$ for all $\ell\in\set{2,\dots,m}$.

By employing similar arguments we can show that
\[
\card{\rep{D^{\dep,Q}}{\dep,\neg Q}}\ =\ n_1\times n_2\times \Pi_{c\in\adom{D}}m_c.
\]
Thus, we have that 
\[
\rfreq{\neg Q_{x\mapsto c}}{D^{\dep,Q},\dep}\ =\ \frac{m_c\times n_1\times n_2\times \Pi_{b\in\adom{D},b\neq c}n_b}{n_1\times n_2\times \Pi_{b\in\adom{D}}n_b}=\frac{m_c}{n_c}
\]
and
\[
\rfreq{\neg Q}{D^{\dep,Q},\dep}\ =\ \frac{n_1\times n_2\times \Pi_{c\in\adom{D}}m_c}{n_1\times n_2\times \Pi_{c\in\adom{D}}n_c}=\Pi_{c\in\adom{D}}\frac{m_c}{n_c}.
\]
By combining these two equations we get that
\[
\rfreq{\neg Q}{D^{\dep,Q},\dep}\ =\ \Pi_{c\in\adom{D}}\rfreq{\neg Q_{x\mapsto c}}{D^{\dep,Q},\dep}.
\]
Now,
\begin{eqnarray*}
\rfreq{Q}{D^{\dep,Q},\dep}&=&\frac{\card{\rep{D^{\dep,Q}}{\dep, Q}}}{\card{\rep{D^{\dep,Q}}{\dep}}}\\
&=& \frac{\card{\rep{D^{\dep,Q}}{\dep}}-\card{\rep{D^{\dep,Q}}{\dep, \neg Q}}}{\card{\rep{D^{\dep,Q}}{\dep}}}\\
&=& 1-\frac{\card{\rep{D^{\dep,Q}}{\dep,\neg Q}}}{\card{\rep{D^{\dep,Q}}{\dep}}}\\
&=&1-\rfreq{\neg Q}{D^{\dep,Q},\dep}\\
&=& 1-\Pi_{c\in\adom{D}}\rfreq{\neg Q_{x\mapsto c}}{D^{\dep,Q},\dep}\\
&=& 1-\Pi_{c\in\adom{D}}\left(1-\rfreq{Q_{x\mapsto c}}{D^{\dep,Q},\dep}\right).
\end{eqnarray*}
This concludes our proof.
\end{proof}

\subsection{Proof of Lemma~\ref{lem:aux-3}}
Recall that this lemma deals with SJFCQs with a non-empty complex part that contains at least one atom $R(\bar z)$ such that: \e{(1)} $\pvar{R(\bar z)}{\dep}=\emptyset$, and \e{(2)} $R(\bar z)$ has a variable $x$ at a position of $I_{R,i_R}^\rhs$.
Note that we have added here a condition on the variable $x$ that does not occur in the corresponding lemma of Maslowski and Wijsen~\cite{MaWi13}---the variable $x$ is associated with an attribute on the right-hand side of the primary FD. In the case of primary keys, every position of an $R$-atom $\alpha$ in $Q$ is associated with either an attribute that occurs on the left-hand side of the key of $R$ or an attribute that occurs on the right-hand side of the key. In our setting, a position of $\alpha$ may be associated with an attribute that does not occur at all in the primary FD; hence, this additional condition is needed to establish the correctness of the lemma in question.

\begin{replemma}{\ref{lem:aux-3}}
\lemmaauxthird
\end{replemma}
\begin{proof}
Assume that $\alpha=R(\bar z)$. It clearly holds that
\[
\rep{D}{\dep, Q}\ =\ \bigcup_{c\in\adom{D}}\rep{D}{\dep, Q_{x\mapsto c}}.
\]
We now prove that for all $a\neq b$ we have that
\[
\rep{D}{\dep, Q_{x\mapsto a}}\cap \rep{D}{\dep, Q_{x\mapsto b}}\ =\ \emptyset.
\]
Assume, towards a contradiction, that there is some repair $J$ of $D$ such that $J\in (\rep{D}{\dep, Q_{x\mapsto a}}\cap \rep{D}{\dep, Q_{x\mapsto b}})$ for some $a\neq b$. Let $h$ be a homomorphism from $Q_{x\mapsto a}$ to $J$ and let $h'$ be a homomorphism from $Q_{x\mapsto b}$ to $J$.
Let $f=R(h(\bar z))$ and $g=R(h'(\bar z))$.

Since $h$ maps the atoms of $Q_{x\mapsto a}$ to facts of $J$, the fact $f$ uses the constant $a$ at each position of $\alpha$ that uses the variable $x$. Similarly, the fact $g$ uses the constant $b$ at each position that uses the variable $x$. Moreover, the facts $f$ and $g$ agree with the atom $\alpha$ on all the constants that occur in this atom. Since for each $A\in X$ (where $X$ is the left-hand side of the primary FD of $\dep_R$) the position $(R,A)$ is a primary-lhs position and all these positions are associated with constants in $\alpha$, the facts $f$ and $g$ agree on the values of all the attributes in $X$, but disagree on the value of at least one attribute in $Y$ (the attribute corresponding to the position of $x$ in $\alpha$). We conclude that $\{f,g\}\not\models\dep_R$, which is a contradiction to the fact that $J$ is a repair.
Hence, it holds that
\[
\card{\rep{D}{\dep, Q}}\ =\ \sum_{c\in\adom{D}}\card{\rep{D}{\dep, Q_{x\mapsto c}}}
\]
and
\[
\rfreq{Q}{D,\dep}\ =\frac{\rep{D}{\dep, Q}}{\rep{D}{\dep}}\ =\ \sum_{c\in\adom{D}}\frac{\rep{D}{\dep, Q_{x\mapsto c}}}{\rep{D}{\dep}}=\sum_{c\in\adom{D}}\rfreq{Q_{x\mapsto c}}{D,\dep},
\]
and the claim follows.
\end{proof}

\subsection{Proof of Lemma~\ref{lem:rewrite-rules-relfreq}}

We first prove two auxiliary lemmas that will be used in the proof of Lemma~\ref{lem:rewrite-rules-relfreq}. Similar lemmas have been used in the proof of hardness for primary keys, but here we have to take into account not only the primary FDs, but all the FDs of $\dep$. To this end, we use $D_{\mathsf{conf}}^{\dep,Q}$ and $D_{\mathsf{ind}}^{\dep,Q}$. 

\begin{lemma}\label{lemma:hard_help1}
	Consider a database $D$, a set $\dep$ of FDs with an LHS chain, and an SJFCQ $Q$. For an atom $R(\bar z)$ of $Q$, and a set $E$ of $R$-facts such that (1) $E\not\models R(\bar z)$, and (2) for all $R$-facts $f\in E$ and $g\in D$, \rev{there is a primary-lhs position $(R,A)$ such that $f[A]\neq g[A]$}, it holds that
	\begin{equation*}
	\rfreq{Q}{D\cup E,\dep}\ =\ \rfreq{Q}{D,\dep}\times\frac{\card{\rep{D}{\dep}}}{\card{\rep{D\setminus D_{\mathsf{conf}}^{\dep,Q}}{\dep}}}
	\times \frac{\card{\rep{(D\cup E)\setminus (D_{\mathsf{conf}}^{\dep,Q}\cup E_{\mathsf{conf}}^{\dep,Q})}{\dep}}}{\card{\rep{D\cup E}{\dep}}}.
	\end{equation*}
\end{lemma}

\begin{proof}
	Let $D^{\dep,Q}=D\setminus (D_{\mathsf{conf}}^{\dep,Q}\cup D_{\mathsf{ind}}^{\dep,Q})$ and $E^{\dep,Q}=E\setminus (E_{\mathsf{conf}}^{\dep,Q}\cup E_{\mathsf{ind}}^{\dep,Q})$.
	Lemma~\ref{lemma:basic} implies that there are no conflicts among $R$-facts $f\in D^{\dep,Q}$ and $g\in E^{\dep,Q}$, and therefore:
	$$\card{\rep{D^{\dep,Q}\cup E^{\dep,Q}}{\dep}}\ =\ \card{\rep{D^{\dep,Q}}{\dep}}\times \card{\rep{E^{\dep,Q}}{\dep}}.$$
	Since $E\not\models R(\bar z)$ (and so $E^{\dep,Q}\not\models R(\bar z)$ as well), it also holds that:
	$$\card{\rep{D^{\dep,Q}\cup E^{\dep,Q}}{\dep,Q}}\ =\ \card{\rep{D^{\dep,Q}}{\dep,Q}}\times \card{\rep{E^{\dep,Q}}{\dep}}.$$
	We can then conclude that:
	$$\rfreq{Q}{D^{\dep,Q}\cup E^{\dep,Q},\dep}\ =\ \rfreq{Q}{D^{\dep,Q},\dep}$$
	Now, Lemma~\ref{lem:combined_help} implies that:
	\[
	\rfreq{Q}{D,\dep}\ =\ \rfreq{Q}{D^{\dep,Q},\dep}\times\frac{\card{\rep{D\setminus D_{\mathsf{conf}}^{\dep,Q}}{\dep}}}{\card{\rep{D}{\dep}}}
	\]
	and
	\begin{equation*}
	\rfreq{Q}{D\cup E,\dep} = \rfreq{Q}{D^{\dep,Q}\cup E^{\dep,Q},\dep}
	\times\frac{\card{\rep{(D\cup E)\setminus (D_{\mathsf{conf}}^{\dep,Q}\cup E_{\mathsf{conf}}^{\dep,Q})}{\dep}}}{\card{\rep{D\cup E}{\dep}}}.
	\end{equation*}
	By combining these three results we obtain that:
	\begin{align*}
	\rfreq{Q}{D\cup E,\dep}&=\rfreq{Q}{D^{\dep,Q}\cup E^{\dep,Q},\dep}\times\frac{\card{\rep{(D\cup E)\setminus (D_{\mathsf{conf}}^{\dep,Q}\cup E_{\mathsf{conf}}^{\dep,Q})}{\dep}}}{\card{\rep{D\cup E}{\dep}}}\\
	&=\rfreq{Q}{D^{\dep,Q},\dep}\times\frac{\card{\rep{(D\cup E)\setminus (D_{\mathsf{conf}}^{\dep,Q}\cup E_{\mathsf{conf}}^{\dep,Q})}{\dep}}}{\card{\rep{D\cup E}{\dep}}}\\
	&=\rfreq{Q}{D,\dep}\times \frac{\card{\rep{D}{\dep}}}{\card{\rep{D\setminus D_{\mathsf{conf}}^{\dep,Q}}{\dep}}}\times\frac{\card{\rep{(D\cup E)\setminus (D_{\mathsf{conf}}^{\dep,Q}\cup E_{\mathsf{conf}}^{\dep,Q})}{\dep}}}{\card{\rep{D\cup E}{\dep}}}
	\end{align*}
	This concludes our proof.
\end{proof}

The next technical lemma is rather straightforward.

\begin{lemma}\label{lemma:hard_help2}
	Consider a database $D$, a set $\dep$ of FDs with an LHS chain, and an SJFCQ $Q$. For $E\subseteq D$ that contains all the facts of $D$ over relations that occur in $Q$ we have that:
	$$\rfreq{Q}{D,\dep}\ =\ \rfreq{Q}{E,\dep}$$
\end{lemma}

\begin{proof}
	Clearly, it holds that:
	\[
	\card{\rep{D}{\dep}}\ =\ \card{\rep{E}{\dep}}\times \card{\rep{D\setminus E}{\dep}}
	\]
	Moreover, since the relations of the facts in $D\setminus E$ do not occur in $Q$, we have that:
	\[
	\card{\rep{D}{\dep,Q}}\ =\ \card{\rep{E}{\dep,Q}}\times \card{\rep{D\setminus E}{\dep}}
	\]
	We conclude that:
	\[
	\rfreq{Q}{D,\dep}\ =\ \rfreq{Q}{E,\dep},
	\]
	and the claim follows.
\end{proof}

We are now ready to give the proof of the lemma in question, which we recall here:

\begin{replemma}{\ref{lem:rewrite-rules-relfreq}}
	\lemmarewriterules
\end{replemma}

\begin{proof}
	We consider each one of the rewrite rules given in Figure~\ref{fig:rules} separately. To this end, given a database $D$, we need to construct in polynomial time a database $E$ such that the value $\rfreq{Q'}{D,\dep'}$ can be computed from $\rfreq{Q}{E,\dep}$. 
	The proofs for the rewrite rules R5, R7, and R10, have been already given in the main body of the paper. We now proceed to give the proofs for the remaining rewrite rules, which are similar to the corresponding proofs for primary keys~\cite{MaWi13}, but we have to take into account the FDs in the primary prefix. To this end, we use $D_{\mathsf{conf}}^{\dep,Q}$ and $D_{\mathsf{ind}}^{\dep,Q}$.
	
	\medskip
	
\rev{\noindent\paragraph{\underline{Rewrite Rule R1}: $(\dep,Q)\lrhd \, (\dep, Q_{x\mapsto c})$ if there is an atom $\alpha \in Q$ with $x\in\pvar{\alpha}{\dep}.$}}
	\smallskip
	
	\noindent Let $\alpha=R(\bar z)$ be an atom of $Q$, and let $x\in\pvar{\alpha}{\dep}$. Assume, without loss of generality, that $x$ occurs at the first position of $\alpha$. Let $D'\subseteq D$ be the set of $R$-facts that use a constant $b\neq c$ at the first position. Let $E=D\setminus D'$.
	Lemma~\ref{lemma:hard_help1} implies that
	\[
	\rfreq{Q_{x\mapsto c}}{D,\dep}\ =\ \rfreq{Q_{x\mapsto c}}{E,\dep}\times\frac{\card{\rep{E}{\dep}}}{\card{\rep{E\setminus E_{\mathsf{conf}}^{\dep,Q_{x\mapsto c}}}{\dep}}}\times \frac{\card{\rep{D\setminus D_{\mathsf{conf}}^{\dep,Q_{x\mapsto c}}}{\dep}}}{\card{\rep{D}{\dep}}}.
	\]
	Note that $\dep_R$ may have one primary FD w.r.t.~$Q$ and another primary FD w.r.t.~$Q_{x\mapsto c}$. This is the case, for example, if the position $(R,A)$ of the variable $x$ corresponds to the only attribute of the primary FD that is associated with a variable. In this case, the primary FD of $\dep_R$ w.r.t.~$Q_{x\mapsto c}$ might be another FD that occurs later in the LHS chain of $\dep_R$ or $\dep_R$ might have no primary FD w.r.t.~$Q_{x\mapsto c}$. However, if $(R,A)$ is a primary-lhs position w.r.t.~$Q$, then it is also a primary-lhs position w.r.t.~$Q_{x\mapsto c}$.
	Since every $R$-fact of $E$ uses the value $c$ at the first position, we get that
	\[
	\rfreq{Q_{x\mapsto c}}{E,\dep}\ =\ \rfreq{Q}{E,\dep}.
	\]
	Therefore, we can conclude that
	\[
	\rfreq{Q_{x\mapsto c}}{D,\dep}\ =\ \rfreq{Q}{E,\dep}\times\frac{\card{\rep{E}{\dep}}}{\card{\rep{E\setminus E_{\mathsf{conf}}^{\dep,Q_{x\mapsto c}}}{\dep}}}\times \frac{\card{\rep{D\setminus D_{\mathsf{conf}}^{\dep,Q_{x\mapsto c}}}{\dep}}}{\card{\rep{D}{\dep}}}.
	\]
 \rev{(The ratios in the above equation can be computed in polynomial time due to Proposition~\ref{pro:counting-repairs-lhs-fp}.)}
	
	\medskip
	
\rev{\noindent\paragraph{\underline{Rewrite Rule R2}: $(\dep, Q_1\cup Q_2)\lrhd \, (\dep,Q_1)$ if $Q_1\neq \emptyset$, $Q_2 \neq \emptyset$ and $\var{Q_1}\cap\var{Q_2}=\emptyset$.} }
	\smallskip
	
	\noindent Let $D'\subseteq D$ be the subset of $D$ that contains every fact over a relation name $R$ that occurs in $Q_1$. Let $D''$ be a set of facts over the relation names that occur in $Q_2$ that is obtained via some function $h$ from the variables of $Q_2$ to constants (that is, for every atom $R(\bar z)$ of $Q_2$, the database $D''$ will contain the $R$-fact $R(h(\bar z))$). Clearly, $D''$ is a consistent database (as it contains one fact per relation name) that entails $Q_2$. Let $E = D'\cup D''$. Lemma~\ref{lemma:hard_help2} implies the following:
	\begin{eqnarray*}
		&\rfreq{Q_1}{D,\dep}\ =\ \rfreq{Q_1}{D',\dep}&\\
		&\rfreq{Q_1}{E,\dep}\ =\ \rfreq{Q_1}{D',\dep}&\\
		&\rfreq{Q_2}{E,\dep}\ =\ \rfreq{Q_2}{D'',\dep}\ =\ 1.&
	\end{eqnarray*}
	By Lemma~\ref{lem:aux-1}, we get that
	\[
	\rfreq{Q}{E,\dep}\ =\ \rfreq{Q_1}{E,\dep}\times \rfreq{Q_2}{E,\dep}=\rfreq{Q_1}{D',\dep}\times \rfreq{Q_2}{D'',\dep}
	\]
	By combining these above results, we conclude that
	\[
	\rfreq{Q_1}{D,\dep}\ =\ \rfreq{Q_1}{D',\dep}\ =\ \rfreq{Q_1}{D',\dep}\times \rfreq{Q_2}{D'',\dep}\ =\ \rfreq{Q}{E,\dep}.
	\]
	
	\medskip

\rev{\noindent\paragraph{\underline{Rewrite Rule R3}: $(\dep,Q)\lrhd \, (\dep,Q_{y\mapsto x})$ if there is an atom $\alpha \in Q$ with $x,y\in\pvar{\alpha}{\dep}$.} }
	\smallskip

	\noindent Let $\alpha=R(\bar z)$ be an atom of $Q$, and let $x,y\in\pvar{\alpha}{\dep}$. Assume, without loss of generality, that $x$ and $y$ occur at the first and second positions of $\alpha$, respectively. Let $D'\subseteq D$ be the set of $R$-facts that use different constants at the first and second positions (i.e., the facts of the form $R(a,b,\bar w)$ for $a\neq b$). Let $E=D\setminus D'$. Lemma~\ref{lemma:hard_help1} implies that
	\[
	\rfreq{Q_{y\mapsto x}}{D,\dep}\ =\ \rfreq{Q_{y\mapsto x}}{E,\dep}\times\frac{\card{\rep{E}{\dep}}}{\card{\rep{E\setminus E_{\mathsf{conf}}^{\dep,Q_{y\mapsto x}}}{\dep}}}\times \frac{\card{\rep{D\setminus D_{\mathsf{conf}}^{\dep,Q_{y\mapsto x}}}{\dep}}}{\card{\rep{D}{\dep}}}.
	\]
	(Note that, since every $R$-fact of $D\setminus E$ uses different constants for the first two positions, and every $R$-fact of $E$ uses the same constant for both positions, for all $R$-facts $f\in E$ and $g\in D\setminus E$, the facts $f$ and $g$ disagree on at least one primary-lhs position.)
	Since all the $R$-facts of $E$ are of the form $R(a,a,\bar w)$ for some constant $a$, we have that
	\[
	\rfreq{Q_{y\mapsto x}}{E,\dep}\ =\ \rfreq{Q}{E,\dep}.
	\]
	Therefore, we can conclude that
	\[
	\rfreq{Q_{y\mapsto x}}{D,\dep}\ =\ \rfreq{Q}{E,\dep}\times\frac{\card{\rep{E}{\dep}}}{\card{\rep{E\setminus E_{\mathsf{conf}}^{\dep,Q_{y\mapsto x}}}{\dep}}}\times \frac{\card{\rep{D\setminus D_{\mathsf{conf}}^{\dep,Q_{y\mapsto x}}}{\dep}}}{\card{\rep{D}{\dep}}}.
	\]
	
	 \medskip
	 
\rev{	 \noindent\paragraph{\underline{Rewrite Rule R4}: $(\dep,Q\cup\set{\alpha})\lrhd \, (\dep,Q)$ if $(\var{Q}\cap\var{\alpha})\subseteq\pvar{\alpha}{\dep}$.} }
	 \smallskip
	 
	 \noindent
	\rev{ Assume $\alpha=R(\bar z)$.} Let $D' \subseteq D$ be the database that contains all the facts of $D$ over the relation names that occur in $Q$. If $\pvar{\alpha}{\dep} = \emptyset$, then $\var{Q}\cap\var{\alpha}=\emptyset$, and the rule R2 implies
	\[
	\rfreq{Q}{D,\dep}\ =\ \rfreq{Q\cup\{\alpha\}}{E,\dep},
	\]
	where $E=D'\cup D''$, and $D''$ has a single fact that is obtained via a function $h$ from the variables of $\alpha$ to constants. If $\pvar{\alpha}{\dep}\neq\emptyset$, let $K$ be the database that contains, for every function $h$ from $\pvar{\alpha}{\dep}$ to $\adom{D}$, an $R$-fact $f$ such that:
	\[
	f[A]=\begin{cases}
	h(x) & \rev{\mbox{if $\alpha[A]=x$ and $x\in \pvar{\alpha}{\dep}$}}\\
	c & \mbox{if $\alpha[A]=c$ (where $c$ is a constant)} \\ 
	\odot & \mbox{otherwise.}
	\end{cases}
	\]
	Note that $K$ is a consistent database since all the facts of $K$ agree on the values of all attributes in $X_j\cup Y_j$ for the FDs $R:X_j\ra Y_j$ in the primary prefix of $\dep_R$ (as all these attributes are associated with constants in $\alpha$). Moreover, no two facts of $K$ agree on the values of all primary-lhs positions; hence, no two facts violate an FD that occurs after the primary FD in the LHS chain of $\dep_R$. In addition, for every $f\in K$, we clearly have that $f\models \alpha$. Let $E=D'\cup K$.
	Lemma~\ref{lemma:hard_help2} implies that
	\[
	\rfreq{Q}{E,\dep}\ =\ \rfreq{Q}{D',\dep}
	\]
	and
	\[
	\rfreq{Q}{D',\dep}\ =\ \rfreq{Q}{D,\dep}.
	\]
	By the definition of the database $K$, we have that
	$$\rfreq{Q}{E,\dep}\ =\ \rfreq{Q\cup\{\alpha\}}{E,\dep}.$$
	This is because the only shared variables among $Q$ and $\alpha$ occur in $\pvar{\alpha}{\dep}$, and $K$ contains a corresponding fact for every possible function from $\pvar{\alpha}{\dep}$ to $\adom{D}$. Moreover, the database $K$ is consistent; hence, every repair of $E$ contains all the facts of $K$.
	From the above results,
	\[
	\rfreq{Q}{D,\dep}\ =\ \rfreq{Q\cup\{\alpha\}}{E,\dep}.
	\]
	
	 	\medskip
	
\rev{\noindent\paragraph{\underline{Rewrite Rule R8}: $(\dep,Q\cup\set{R(\bar z,x,\bar y)})\lrhd \, ((\dep\setminus\dep_R)\cup\dep_{R'},Q\cup\set{R'(\bar z,\bar y)})$ if $x$ is an orphan variable of $Q\cup\set{R(\bar z,x,\bar y)}$, and $R'$ does not occur in $Q$.} }
	\smallskip
	
	\noindent
	Assume that the variable $x$ occurs at the position $(R,A)$ of $R(\bar z,x,\bar y)$. As in rule R5, we denote by $\dep$ the set of FDs over the schema $\ins{S}$ that contains the relation $R(\bar z,x,\bar y)$, and by $\dep'$ the set of FDs over the schema $\ins{S}'$ that contains the relation $R'(\bar z,\bar y)$. We have that $\dep_{R''}=\dep'_{R''}$ for every relation name of $\ins{S} \cap \ins{S}'$, and $\dep'_{R'}$ is obtained from $\dep_{R}$ by removing the attribute $A$ from every FD of $\dep_{R}$, and then removing trivial FDs and merging FDs with the same left-hand side.
	Given a database $D$ over $\ins{S}$, let $E$ be the database that contains all the facts of $D$ over relation names $R''$ such that $R''\neq R'$, and, for every fact $f$ over $R'$, the following $R$-fact
	\[
	f'[B]\ =\ \begin{cases}
	f[B] & \mbox{if $B\neq A$}\\
	\odot & \mbox{otherwise.}
	\end{cases}
	\]
	Clearly, we have that
	\[
	\card{\rep{D}{\dep'}}\ =\ \card{\rep{E}{\dep}}.
	\]
	This holds since all the $R$-facts of $E$ agree on the value of the additional attribute $A$; hence, this attribute does not affect the satisfaction of the FDs. Moreover, since the variable $x$ occurs only once in the query, the value of this variable has no impact on the entailment of the query, and thus
	\[
	\card{\rep{D}{\dep',Q\cup\{R'(\bar z,\bar y)\}}}\ =\ \card{\rep{E}{\dep,Q\cup\{R(\bar z,x,\bar y)\}}}.
	\]
	We conclude that
	\[
	\rfreq{D}{\dep', Q\cup\{R'(\bar z,\bar y)\}}\ =\ \rfreq{E}{\dep,Q\cup\set{R(\bar z,x,\bar y)}}.
	\]

	\medskip
	
\rev{	\noindent\paragraph{\underline{Rewrite Rule R9}: $(\dep,Q)\lrhd \, (\dep,Q')$ if a constant $c$ occurs in $Q$, some constant $c'\neq c$ occurs in $Q$, and $Q'$ is obtained from $Q$ by replacing all constants with $c$.} }
	\smallskip
	
	\noindent The proof for this rule is straightforward, as we are only renaming constants.
\end{proof}

\subsection{Proofs of Lemmas~\ref{lemma:12} -- \ref{lemma:final-orphan}}
We start by proving a series of basic properties of final pairs (Lemmas~\ref{lemma:final-two-atoms} -- \ref{lemma:final-orphan-liaison-and-constant}) that are crucial for showing Lemmas~\ref{lemma:12} -- \ref{lemma:final-orphan}. Note that similar properties have been proved for primary keys~\cite{MaWi13}. \rev{Recall that a pair $(\dep,Q)$, where $\dep$ is an FD set with an LHS chain and $Q$ is an SJFCQ, is final if: (1) $Q$ is not $\dep$-safe, and (2) for every set $\dep'$ of FDs with an LHS chain, and an SJFCQ $Q'$ such that $(\dep,Q)\lrhd^+ (\dep',Q')$, the query $Q'$ is $\dep'$-safe. As the proofs rely on the conditions of the procedure $\mathsf{IsSafe}$, we give the algorithm again below for the sake of readability.}

 \AlgIsSafe*

 The first property applies to any unsafe SJFCQ.

\begin{lemma}\label{lemma:final-two-atoms}
Consider a set $\dep$ of FDs with an LHS chain, and an SJFCQ $Q$ that is not $\dep$-safe. It holds that $\comp{Q}{\dep}$ has at least two atoms.
\end{lemma}

\begin{proof}
We give a proof by contradiction. Assume that $\comp{Q}{\dep}$ is empty or consists of a single atom.
If $\comp{Q}{\dep}=\emptyset$, then $Q$ is clearly $\dep$-safe, which contradicts our hypothesis that $Q$ is not $\dep$-safe.
If $Q$ is such that $\comp{Q}{\dep}=\set{\alpha}$, we proceed to show, by induction on the number of variables occurring in $\alpha$, that it is $\dep$-safe, which again contradicts our hypothesis that $Q$ is not $\dep$-safe.
Note that, since $\alpha\in \comp{Q}{\dep}$, it cannot be the case that $\alpha$ has no variables; otherwise, all the positions of $\alpha$ are primary-lhs positions by definition.

\medskip

\noindent
\paragraph{Base Case.}
If $\alpha$ mentions a single variable $x$ at a primary-lhs position, then the condition of line~5 is satisfied by $Q$, and since $Q_{x\mapsto c}$ is $\dep$-safe (as it has no complex part), $Q$ is $\dep$-safe as well. If $\pvar{\alpha}{\dep}=\emptyset$ and $Q$ uses a single variable $x$ at a non-primary-lhs position, then the condition of line~7 is satisfied by $Q$, and, again, since $Q_{x\mapsto c}$ is $\dep$-safe, $Q$ is $\dep$-safe.

\medskip

\noindent
\paragraph{Inductive Step.}
We now prove that if $\alpha$ uses $n+1$ variables, then $Q$ is $\dep$-safe. We proceed by considering the following two cases:
\begin{enumerate}
    \item $\pvar{\alpha}{\dep}\neq\emptyset$. In this case, the condition of line~5 of $\IsSafe$ is satisfied. By the inductive hypothesis, the query $Q_{x\mapsto c}$ that contains $n$ variables is $\dep$-safe; hence, $Q$ is $\dep$-safe as well.
    \item $\pvar{\alpha}{\dep}=\emptyset$. In this case, the condition of line~7 of $\IsSafe$ is satisfied. By the inductive hypothesis, the query $Q_{x\mapsto c}$ that contains $n$ variables is $\dep$-safe; hence, $Q$ is $\dep$-safe as well.
\end{enumerate}
We can therefore conclude that $\comp{Q}{\dep}$ contains at least two atoms.
\end{proof}

We next prove that the queries of final pairs falsify all the conditions of $\IsSafe$. This is a crucial property that will be heavily used in our proofs.

\begin{lemma}\label{lemma:final-no-condition}
Consider a set $\dep$ of FDs with an LHS chain, and an SJFCQ $Q$ such that $(\dep,Q)$ is final. It holds that $Q$ falsifies all the conditions of $\IsSafe$.
\end{lemma}

\begin{proof}
Assume, towards a contradiction, that $Q$ satisfies one of the conditions of the algorithm.
\begin{enumerate}
    \item If $Q$ satisfies the condition of line~1, then $Q$ is $\dep$-safe.
    
    \item If $Q$ satisfies the condition of line~3, then $Q$ is $\dep$-safe if both $Q_1$ and $Q_2$ are $\dep$-safe. By the rewrite rule R2, we have that $(\dep,Q) \lrhd (\dep,Q_1)$ and $(\dep,Q) \lrhd (\dep,Q_2)$, and since $(\dep,Q)$ is final, we get that both $Q_1$ and $Q_2$ are $\dep$-safe.
    
    \item If $Q$ satisfies the condition of line~5, then $Q$ is $\dep$-safe if $Q_{x\mapsto c}$ is $\dep$-safe  for an arbitrary constant $c$. By the rule R1, $(\dep,Q) \lrhd (\dep,Q_{x\mapsto c})$, and since $(\dep,Q)$ is final, $Q_{x\mapsto c}$ is $\dep$-safe.
    
    \item If $Q$ satisfies the condition of line~7, then there is an $R$-atom $\alpha$ in $Q$ with $\pvar{\alpha}{\dep}=\emptyset$, and a variable $x$ that occurs at a position of $\{(R,A)\mid A\in Y\}$, with $R:X\rightarrow Y$ being the primary FD of $\dep_R$ w.r.t.~$Q$.
    If there exists a liaison variable $y\in\var{\alpha}$, then by the rewrite rule R6, we have that $(\dep,Q)\lrhd (\dep,Q_{x\mapsto c})$. Since $(\dep,Q)$ is final, $Q_{x\mapsto c}$ is $\dep$-safe, and according to the condition of line~7, $Q$ is also $\dep$-safe.
    If there is no liaison variable $y\in\var{\alpha}$, Lemma~\ref{lemma:final-two-atoms} implies that $Q\setminus\set{\alpha}$ is not empty, and since $\pvar{\alpha}{\dep}=\emptyset$ and there are no liaison variables in $\alpha$, $\var{Q\setminus\set{\alpha}}\cap\var{\alpha}=\emptyset$. By the rewrite rule R2, we have that $(\dep,Q)\lrhd (\dep,Q\setminus\set{\alpha})$ and $(\dep,Q)\lrhd (\dep,\set{\alpha})$. Since $(\dep,Q)$ is $\dep$-final, both $Q\setminus\set{\alpha}$ and $\set{\alpha}$ are $\dep$-safe; hence, $Q$ is $\dep$-safe according to the condition of line~3.
\end{enumerate}
Since in all the cases we obtain a contradiction to the fact that $(\dep,Q)$ is final, we conclude that $Q$ is not $\dep$-safe, and the claim follows.
\end{proof}

For the next lemma, recall that an orphan variable of $Q$ is a variable that occurs once in $Q$, at a non-primary-lhs position.

\begin{lemma}\label{lemma:no-empty-pvars}
Consider a set $\dep$ of FDs with an LHS chain, and a an SJFCQ $Q$ such that $(\dep,Q)$ is final. For every atom $\alpha\in Q$, $\pvar{\alpha}{\dep}\neq\emptyset$.
\end{lemma}

\begin{proof}
If $\alpha\not\in\comp{Q}{\dep}$ and $\pvar{\alpha}{\dep}=\emptyset$, then the condition of line~3 of $\IsSafe$ is satisfied by $Q$, since a variable at a non-primary-lhs position of an atom outside the complex part is an orphan variable. Note that, by Lemma~\ref{lemma:final-two-atoms}, $Q\setminus\set{\alpha} \neq \emptyset$.
If $\alpha\in\comp{Q}{\dep}$ and $\pvar{\alpha}{\dep}=\emptyset$, then $Q$ satisfies the condition of line~7 of $\IsSafe$. \rev{This holds since every atom in $\comp{Q}{\dep}$ has a primary FD (as it has a non-primary-lhs position), and if all the primary-lhs positions of the atom are associated with constants, then at least one of the primary-rhs positions must be associated with a variable. } In both cases, we get a contradiction to Lemma~\ref{lemma:final-no-condition}.
\end{proof}

We now use the above lemma to prove the following technical property of final pairs. Recall that, for every relation name $R$, we denote by $R:X_{i_R}\ra Y_{i_R}$ the primary FD of $\dep_R$, and by $I_{R,i_R}^\lhs$ (resp.,~$I_{R,i_R}^\rhs$) the positions corresponding to the attributes of $X_{i_R}$ (resp.,~$Y_{i_R}$).

\begin{lemma}\label{lemma:orphan}
Consider a set $\dep$ of FDs with an LHS chain, and an SJFCQ $Q$ such that $(\dep,Q)$ is final. There is no $R$-atom of $Q$ that uses an orphan variable $x$ at a position $(R,A)\not\in I_{R,i_R}^\rhs$.
\end{lemma}

\begin{proof}
Assume, towards a contradiction, that $Q$ contains an atom $\alpha=R(\bar y,x,\bar z)$ such that $x$ is an orphan variable that occurs at position $(R,A)\not\in I_{R,i_R}^\rhs$. Let $Q'$ be the query obtained from $Q$ by replacing the atom $\alpha$ with an atom $\alpha'=R'(\bar y,\bar z)$. Let $\dep'$ be the FD set that is obtained from $\dep$ by removing all the FDs of $\dep_R$ and adding, for each such FD $R:X\ra Y$, the FD $R':(X\setminus\{A\})\ra (Y\setminus\{A\})$ to $\dep'$ (unless it is trivial). Note that it cannot be the case that $A\in X_{i_R}$ or $A\in Y_{i_R}$; hence, the FD $R':X_{i_R}\ra Y_{i_R}$ is the primary FD of $\dep'_{R'}$. 
Also note that since the variable $x$ occurs only once in $Q$, we have that for every atom $\beta\in (Q\cap Q')$, $\beta\in\comp{Q}{\dep}$ if and only if $\beta\in\comp{Q'}{\dep'}$. Moreover, it holds that $\alpha\in\comp{Q}{\dep}$ if and only if $\alpha'\in\comp{Q'}{\dep'}$, as we have only removed an orphan variable from $\alpha$ to obtain $\alpha'$.

By the rule R8, $(\dep,Q) \lrhd (\dep',Q')$ for some $\dep'$. Since $(\dep,Q)$ is final, the query $Q'$ is $\dep'$-safe, and one of the conditions of $\IsSafe$ is satisfied by $Q'$. We proceed by considering all possible cases:
\begin{enumerate}
    \item As explained above, we have that $\comp{Q}{\dep}=\emptyset$ if and only if $\comp{Q'}{\dep'}=\emptyset$. Lemma~\ref{lemma:final-no-condition} implies that the condition of line~1 is falsified by $Q$; hence, $\comp{Q}{\dep}\neq\emptyset$, and this condition is falsified by $Q'$ as well.
    
    \item If the condition of line~3 is satisfied by $Q'$, the same condition is satisfied by $Q$, since the only additional variable in $Q$ is an orphan variable. Thus, if $Q'=Q_1'\cup Q_2'$ with $Q_1'\neq\emptyset\neq Q_2'$ and $\var{Q_1'}\cap \var{Q_2'}=\emptyset$, it also holds that $Q=Q_1\cup Q_2$ with $Q_1\neq\emptyset\neq Q_2$ and $\var{Q_1}\cap \var{Q_2}=\emptyset$, where $Q_1=Q_1'$ and $Q_2=Q_2'\setminus\set{\alpha'}\cup\set{\alpha}$, assuming, without loss of generality, that $\alpha'\in Q_2'$. This cannot be the case due to Lemma~\ref{lemma:final-no-condition}; thus, the condition of line~3 is falsified by $Q'$.
    
    \item If the condition of line~5 is satisfied by $Q'$, then there is a variable $z$ that occurs in $\pvar{\beta}{\dep}$ for every atom $\beta\in\comp{Q'}{\dep'}$. As said above, $\comp{Q}{\dep}$ also contains all such atoms with $\beta\in (Q\cap Q')$, and it contains $\alpha$ if and only if $\comp{Q'}{\dep}$ contains $\alpha'$. If this is the case, then $z\in\pvar{\alpha'}{\dep'}$ implies that $z\in\pvar{\alpha}{\dep}$ (as we have not changed variables that occur at primary-lhs positions), and the same condition is satisfied by $Q$. Therefore, by Lemma~\ref{lemma:final-no-condition}, we get that the condition of line~5 is falsified by $Q'$.
    
    \item If the condition of line~7 is satisfied by $Q'$, then there exists $\beta \in \comp{Q'}{\dep'}$ such that $\pvar{\beta}{\dep'}=\emptyset$ and $z\in\var{\beta}$. If $\beta\neq\alpha'$, then clearly the same holds for $Q$. If $\beta=\alpha'$, then it also holds that $\pvar{\alpha}{\dep}=\emptyset$ (as we have not changed the primary-lhs positions of this atom), and $z\in\var{\alpha}$. Lemma~\ref{lemma:no-empty-pvars} implies that $\pvar{\alpha}{\dep}\neq\emptyset$; hence, this cannot be the case, and the condition of line~7 is falsified by $Q'$.
\end{enumerate}
We have shown that all the conditions are falsified by $Q'$, which is a contradiction to the fact that $Q'$ is $\dep'$-safe, and the claim follows.
\end{proof}

We next prove another technical property of final pairs $(\dep,Q)$ that will allow us to show that atoms outside the complex part of $Q$ are associated with a single key in $\dep$.

\begin{lemma}\label{lemma:final-no-constant-primary}
\rev{Consider a set $\dep$ of FDs with an LHS chain, and an SJFCQ $Q$ such that $(\dep,Q)$ is final. For every $R$-atom of $Q$, it holds that $\dep_R$ has no primary prefix w.r.t.~$Q$.}
\end{lemma}

\begin{proof}
\rev{Assume, towards a contradiction, that for some $R$-atom $\alpha$ of $Q$, $\dep_R$ has a primary prefix w.r.t.~$Q$. Let $R:X_i\rightarrow Y_i$ be an FD in the primary prefix of $\dep_R$ w.r.t.~$Q$. By the definition of primary prefix, we have that $\alpha[A]\in\ins{C}$ for every $A\in (X_i\cup Y_i)$; hence, at least once constant $c$ occurs in $\alpha$. Assume that $\alpha=R(\bar x,c,\bar y)$. If $\alpha[A]=c$ for some $A\in Y_i$, then $c$ occurs at a non-primary-lhs position, and we can apply case (2) of R5. If $\alpha[A]=c$ for some $A\in X_i$, then $c$ also occurs on the left-hand side of the primary FD, that is, it occurs at a primary-lhs position. In this case, Lemma~\ref{lemma:no-empty-pvars} implies that there is at least one primary-lhs position in $\alpha$ that is associated with a variable; hence, we can apply case (4) of R5. In both cases, by rule R5, we have that $(\dep,Q) \lrhd (\dep',Q\setminus\set{R(\bar x,c,\bar y)}\cup\set{R'(\bar x,\bar y)})$, where $R'$ is a new relation name that does not occur in $Q$, and $\dep'$ is the FD set induced by the rewriting.}

\rev{Let $Q'=Q\setminus\set{R(\bar x,c,\bar y)}\cup\set{R'(\bar x,\bar y)}$. 
Since $(\dep,Q)$ is final, $Q'$ is $\dep'$-safe; hence, one of the conditions of $\IsSafe$ is satisfied by $Q'$. The position of $c$ in $\alpha$ is a non-primary-rhs position, as we consider canonical FD sets, thus it cannot appear in some FD in the primary prefix and on the right-hand side of the primary FD; hence, if $R:X_{i_R}\ra Y_{i_R}$ is the primary FD of $\dep_R$, then $R':(X_{i_R}\setminus\set{A})\ra Y_{i_R}$ is the primary FD of $\dep'_{R'}$ (and if $\dep_R$ has no primary FD, then so does $\dep'_{R'}$). Moreover, we have that $\beta\in\comp{Q}{\dep}$ if and only if $\beta\in\comp{Q'}{\dep'}$ for every $\beta\in (Q\cap Q')$, and  $R(\bar x,c,\bar y)\in\comp{Q}{\dep}$ if and only if $R'(\bar x,\bar y)\in\comp{Q'}{\dep'}$. This is because constants at positions associated with attributes of the primary prefix have no impact on the complex part of the query. Hence, a condition satisfied by $Q'$ is also satisfied by $Q$, which contradicts Lemma~\ref{lemma:final-no-condition}.}
\end{proof}

At this point, we have all the ingredients for showing Lemma~\ref{lemma:12}.

\begin{replemma}{\ref{lemma:12}}
	\lemmasinglekey
\end{replemma}

\begin{proof}
	By Lemma~\ref{lemma:final-no-constant-primary}, $\dep_R$ has no primary prefix.
	Since the $R$-atom $\alpha$ of $Q$ is not in $\comp{Q}{\dep}$, there are no constants or liaison variables at non-primary-lhs positions of $\alpha$. Moreover, Lemma~\ref{lemma:orphan} implies that every atom of $Q$ uses orphan variables only at positions of $I_{R,i_R}^\rhs$.
	Since there are no variables at non-primary-lhs positions of $\alpha$, $\dep_R$ has no FDs, i.e., it contains a single trivial key $R:\att{R}\ra\emptyset$. If there are (orphan) variables at non-primary-lhs positions of $\alpha$, then they occur on the right-hand side of the primary FD of $\dep_R$, and $\dep_R$ has a single key $R:X_{i_R}\ra Y_{i_R}$.
\end{proof}

We now know that for relation names $R$ that occur in atoms outside the complex part of a final query, the set of FDs $\dep_R$ consists of a single key. We proceed to show that for atoms over a relation name $R$ in the complex part of the query, $\dep_R$ either consists of a single key, or it is a set of FDs of one of six specific forms (Lemmas~\ref{lemma:final-liaison} -- \ref{lemma:final-orphan}). To this end, we continue proving crucial technical properties of final pairs ((Lemmas~\ref{lemma:final-no-variable-prim-nonprim} -- \ref{lemma:final-orphan-liaison-and-constant})). The proofs of most properties are similar and take the following general form. We assume, towards a contradiction, that the property is not satisfied by a final pair $(\dep,Q)$, where $\dep$ is a set of FDs with an LHS chain and $Q$ an SJFCQ. We then show that $Q$ can be rewritten into another query $Q'$ using one of the rewrite rules, and then show that $Q'$ is not $\dep'$-safe, which contradicts the fact that $(\dep,Q)$ is final.

In most proofs, we use the rules R5 and R8, and the proofs are mainly based on the following observations (assuming that in the rewriting we remove an atom over a relation name $R$ and add an atom over $R'$). Whenever we rewrite a pair $(\dep,Q)$ into a pair $(\dep',Q')$ using one of these two rules, the set of FDs changes accordingly as we have already explained.  It might be the case that the set of FDs $\dep_R$ has a certain primary FD, while the set $\dep'_{R'}$ has a different primary FD or no primary FD at all. This is the case, for example, when the $R'$-atom of $Q'$ is obtained from the $R$-atom of $Q$ by removing an attribute $A$, and the primary FD of $\dep_R$ is of the form $R:X\ra\set{A}$. In this case, the corresponding FD in $\dep'_{R'}$ is trivial; hence, it is removed from $\dep'_{R'}$. This is also the case when we remove the only attribute $A$ in the primary FD that is associated with a variable in the atom (Lemma~\ref{lemma:no-empty-pvars} implies that $A$ corresponds to a primary-lhs position). Here, the FD in $\dep'_{R'}$ that corresponds to the primary FD of $\dep_R$ becomes part of the primary prefix. 

In such cases, the primary-lhs (and non-primary-lhs) positions might significantly change, which, in turn, heavily affects the complex part of the query.
However, as long as we make sure that we do not remove in the rewriting the only attribute in the right-hand side of the primary FD of $\dep_R$ or the only attribute of the primary FD that is associated with a variable, the FD of $\dep'_{R'}$ that is obtained from the primary FD of $\dep_R$ remains the primary FD, and the set of non-primary-lhs positions of $R'$ is the same as the set of non-primary-lhs positions of $R$ (except for the position corresponding to the removed attribute $A$). If we also make sure that all the liaison variables in the atoms of $Q$ remain liaison variables in the corresponding atoms of $Q'$, then both queries will have a similar complex part. This will allow us to prove that if $Q$ is not $\dep$-safe, then $Q'$ is also not $\dep'$-safe and obtain the desired contradiction.

\begin{lemma}\label{lemma:final-no-variable-prim-nonprim}
Consider a set $\dep$ of FDs with an LHS chain, and an SJFCQ $Q$ where $(\dep,Q)$ is final. No atom of $Q$ has a variable that occurs at a primary-lhs position and a non-primary-lhs position.
\end{lemma}

\begin{proof}
Assume, towards a contradiction, that $Q$ has an atom $R(\bar z,x,\bar y)$ that uses the variable $x$ at both a primary-lhs position and a non-primary-lhs position. Then, by the rewrite rule R10, we have that $(\dep,Q)\lrhd (\dep,Q')$, where $Q'=Q\setminus\set{R(\bar z,x,\bar y)}\cup\set{R(\bar z,c,\bar y)}$, i.e., $Q'$ is obtained from $Q$ by replacing the occurrence of the variable $x$ at a non-primary-lhs position with the constant $c$. Since $(\dep,Q)$ is final, $Q'$ is $\dep$-safe, and one of the conditions of $\IsSafe$ is satisfied by $Q'$.
\begin{enumerate}
    \item The condition of line~1 is not satisfied by $Q'$ since the atom $R(\bar z,c,\bar y)$ uses a constant at a non-primary-lhs position; hence, $\comp{Q'}{\dep}\neq\emptyset$.
    
    \item If $Q'$ satisfies the condition of line~3, then $Q$ also satisfies the condition of line~3, since the variable $x$ still occurs at $R(\bar z,c,\bar y)$, which contradicts Lemma~\ref{lemma:final-no-condition}.
    
    \item For every atom $\alpha\neq R(\bar z,x,\bar y)$ in $\comp{Q}{\dep}$ we have that $\alpha\in\comp{Q'}{\dep}$. This holds since the variable $x$ still occurs at $R(\bar z,c,\bar y)$; hence, if it is a liaison variable of $\alpha$ in $Q$ it is also a liaison variable of $\alpha$ in $Q'$. Moreover, $R(\bar z,x,\bar y)\in\comp{Q}{\dep}$ because it uses the liaison variable $x$ at a non-primary-lhs position, and $R(\bar z,c,\bar y)\in\comp{Q'}{\dep}$ because it uses a constant at a non-primary-lhs position. Therefore, if $Q'$ satisfies the condition of line~5, then clearly $Q$ also satisfies this condition (with the same variable that occurs at a primary-lhs position of every atom in $\comp{Q}{\dep}$), which contradicts Lemma~\ref{lemma:final-no-condition}.
    
    \item It is rather straightforward that if $Q'$ satisfies the condition of line~7, then so does $Q$, which again contradicts Lemma~\ref{lemma:final-no-condition}, as we have not changed the primary-lhs positions of any atom of $Q$, and, as said above, for every $\alpha\in\comp{Q'}{\dep}$ we have that $\alpha\in\comp{Q}{\dep}$. Note that it cannot be the case that $Q'$ satisfies this condition due to the atom $R(\bar z,c,\bar y)$ as it has the variable $x$ at a primary-lhs position.
\end{enumerate}
That concludes our proof.
\end{proof}

\begin{lemma}\label{lemma:final-no-two-constants-non-prim}
Consider a set $\dep$ of FDs with an LHS chain, and an SJFCQ $Q$ such that $(\dep,Q)$ is final. No atom of $Q$ associates two non-primary-lhs positions with constants.
\end{lemma}

\begin{proof}
Assume, towards a contradiction, that there is an atom $\alpha$ of $Q$ over a relation name $R$ that uses a constant $c$ at position $(R,A)$ and a constant $c'$ (that might be the constant $c$) at position $(R,A')$; both are non-primary-lhs positions. Lemma~\ref{lemma:final-no-constant-primary} implies that $\dep_R$ has no primary prefix w.r.t.~$Q$. Assume, without loss of generality, that $A'\ge A$. Thus, if one of these positions occurs in $I_{R,i_R}^\rhs$, then $(R,A)\in I_{R,i_R}^\rhs$. Also, assume that $\alpha=R(\bar z,c',\bar y)$. By the rewrite rule R5, we have that $(\dep,Q) \lrhd (\dep',Q')$, where $Q'=Q\setminus\set{R(\bar z,c',\bar y)}\cup\set{R'(\bar z,\bar y)}$ for a relation name $R'$ not in $Q$. Since $(\dep,Q)$ is final, $Q'$ is $\dep'$-safe, and one of the conditions of $\IsSafe$ is satisfied by $Q'$. Observe that if the primary FD of $\dep_R$ is $R:X_{i_R}\ra Y_{i_R}$, then the primary FD of $\dep'_{R'}$ is $R':X_{i_R}\ra (Y_{i_R}\setminus\set{A'})$ because $Y_{i_R}\setminus\set{A'}$ is not empty (if $A\in Y_{i_R}$, then $A\in Y_{i_R}\setminus\set{A'}$, and if $A\not\in Y_{i_R}$ then, by our assumption, $A'\not\in Y_{i_R}$ and $Y_{i_R}$ has an attribute $B$ such that $B\neq A$ and $B\neq A'$).
\begin{enumerate}
    \item The condition of line~1 is not satisfied by $Q'$ because $R'(\bar z,\bar y)$ uses the constant $c$ at a non-primary-lhs position; thus, $R'(\bar z,\bar y)\in\comp{Q'}{\dep'}$.
    
    \item Clearly, if $Q'$ satisfies the condition of line~3, then so does $Q$, as we have only removed a constant from an atom of $Q$ to obtain $Q'$, which contradicts Lemma~\ref{lemma:final-no-condition}.
    
    \item For every $\beta\in (Q\cap Q')$, we have that $\beta\in\comp{Q'}{\dep'}$ if and only if $\beta\in\comp{Q}{\dep}$, since every liaison variable of $Q'$ is also a liaison variable of $Q$ and vice versa. Moreover, $R'(\bar z,\bar y)\in\comp{Q'}{\dep'}$ and $R(\bar z,c,\bar y)\in\comp{Q}{\dep}$ since both atoms use the constant $c$ at a non-primary-lhs position. Thus, it is clear that if $Q'$ satisfies the condition of line~5 or the condition of line~7, then $Q$ satisfies the same condition, which contradicts Lemma~\ref{lemma:final-no-condition}.
\end{enumerate}
That concludes our proof.
\end{proof}

\begin{lemma}\label{lemma:final-no-two-orphan-non-prim}
Consider a set $\dep$ of FDs with an LHS chain, and an SJFCQ $Q$ such that $(\dep,Q)$ is final. No atom of $Q$ associates two non-primary-lhs positions with orphan variables.
\end{lemma}

\begin{proof}
Assume, towards a contradiction, that there is an atom $\alpha$ of $Q$ over a relation name $R$ that uses orphan variables $x$ and $w$ at positions $(R,A)$ and $(R,A')$, respectively; both are non-primary-lhs positions. Assume, without loss of generality, that $A'\ge A$; hence, if one of these positions occurs in $I_{R,i_R}^\rhs$, then $(R,A)\in I_{R,i_R}^\rhs$. Also assume that $\alpha=R(\bar z,w,\bar y)$. By rule R8, we have that $(\dep,Q) \lrhd (\dep',Q')$, where $Q'=Q\setminus\set{R(\bar z,w,\bar y)}\cup\set{R'(\bar z,\bar y)}$ for a relation name $R'$ that does not occur in $Q$. Since $(\dep,Q)$ is final, $Q'$ is $\dep'$-safe, and one of the conditions of $\IsSafe$ is satisfied by $Q'$. As in Lemma~\ref{lemma:final-no-two-constants-non-prim}, if the primary FD of $\dep_R$ is $R:X_{i_R}\ra Y_{i_R}$, then the primary FD of $\dep'_{R'}$ is $R':X_{i_R}\ra (Y_{i_R}\setminus\set{A'})$ since $Y_{i_R}\setminus\set{A'} \neq \emptyset$, and Lemma~\ref{lemma:no-empty-pvars} implies that an attribute of $X_{i_R}$ is associated with a variable in $\alpha$.
\begin{enumerate}
    \item If the condition of line~1 is satisfied by $Q'$, then it is also satisfied by $Q$ as we have only removed an orphan variable, which does not affect the complex part of the query. This is a contradiction to Lemma~\ref{lemma:final-no-condition}.
    
    \item For the same reason, if $Q'$ satisfies the condition of line~3, then so does $Q$, which contradicts Lemma~\ref{lemma:final-no-condition}.
    
    \item For every $\beta\in (Q\cap Q')$, we have that $\beta\in\comp{Q'}{\dep'}$ if and only if $\alpha\in\comp{Q}{\dep}$, since every liaison variable of $Q'$ is also a liaison variable of $Q$ and vice versa. Moreover, $R'(\bar z,\bar y)\in\comp{Q'}{\dep'}$ if and only if $R(\bar z,w,\bar y)\in\comp{Q}{\dep}$, as $w$ is an orphan variable and for all the non-primary-lhs positions of $R(\bar z,w,\bar y)$ besides $(R,A')$, the corresponding positions of $R'(\bar z,\bar y)$ are also non-primary-lhs positions. This holds because, as aforementioned, $R':X_{i_R}\ra (Y_{i_R}\setminus\set{A'})$ is the primary FD of $\dep'_{R'}$. Hence, if $Q'$ satisfies the condition of line~5 or the condition of line~7, then $Q$ satisfies the same condition, which contradicts Lemma~\ref{lemma:final-no-condition}.
\end{enumerate}
That concludes our proof.
\end{proof}

\begin{lemma}\label{lemma:final-no-two-liaison-shared}
Consider a set $\dep$ of FDs with an LHS chain, and an SJFCQ $Q$ such that $(\dep,Q)$ is final. No atom $\alpha$ of $Q$ associates two non-primary-lhs positions with the same liaison variable $x$ if $x$ also occurs in another atom $\alpha'$ of $Q$.
\end{lemma}

\begin{proof}
Assume, towards a contradiction, that there is an atom $\alpha$ of $Q$ over a relation name $R$ that uses the same liaison variable $x$ at two non-primary-lhs positions $(R,A)$ and $(R,A')$, and there is another atom $\alpha'$ that also uses the variable $x$. Assume, without loss of generality, that $A'\ge A$; thus, if one of these positions occurs in $I_{R,i_R}^\rhs$, then $(R,A)\in I_{R,i_R}^\rhs$. Also assume that $\alpha=R(\bar z,x,\bar y)$. By the rewrite rule R5, we have that $(\dep,Q) \lrhd (\dep',Q')$, where $Q'=Q\setminus\set{R(\bar z,x,\bar y)}\cup\set{R'(\bar z,\bar y)}$ for a relation name $R'$ that does not occur in $Q$ (and we have removed the occurrence of $x$ in the position $(R,A')$ to obtain $R'(\bar z,\bar y)$). Since $(\dep,Q)$ is final, $Q'$ is $\dep'$-safe, and one of the conditions of $\IsSafe$ is satisfied by $Q'$. If the primary FD of $\dep_R$ is $R:X_{i_R}\ra Y_{i_R}$, then the primary FD of $\dep'_{R'}$ is $R':X_{i_R}\ra (Y_{i_R}\setminus\set{A'})$ due to Lemma~\ref{lemma:no-empty-pvars} and the fact that $Y_{i_R}\setminus\set{A'}$ is not empty.
\begin{enumerate}
    \item The condition of line~1 is not satisfied by $Q'$ as the atom $R'(\bar z,\bar y)$ uses the liaison variable $x$ at a non-primary-lhs position (and $x$ remains a liaison variable as it has another occurrence in the atom $\alpha'$).
    
    \item Clearly, if $Q'$ satisfies the condition of line~3, then so does $Q$, as the variable $x$ still occurs in $R'(\bar z,\bar y)$, which contradicts Lemma~\ref{lemma:final-no-condition}.
    
    \item For every $\beta\in (Q\cap Q')$, we have that $\beta\in\comp{Q'}{\dep'}$ if and only if $\alpha\in\comp{Q}{\dep}$, since every liaison variable of $Q'$ is also a liaison variable of $Q$ and vice versa. Moreover, $R'(\bar z,\bar y)\in\comp{Q'}{\dep'}$ and $R(\bar z,x,\bar y)\in\comp{Q}{\dep}$, as $x$ is a liaison variable that occurs at a non-primary-lhs position in both atoms. Hence, if $Q'$ satisfies the condition of line~5 or the condition of line~7, then $Q$ satisfies the same condition, which contradicts Lemma~\ref{lemma:final-no-condition}.
\end{enumerate}
That concludes our proof.
\end{proof}

\begin{lemma}\label{lemma:final-no-two-liaison-unique}
Consider a set $\dep$ of FDs with an LHS chain, and an SJFCQ $Q$ such that $(\dep,Q)$ is final. No atom of $Q$ associates three non-primary-lhs positions with the same liaison variable $x$.
\end{lemma}

\begin{proof}
Assume, towards a contradiction, that there is an atom $\alpha$ of $Q$ over a relation name $R$ that uses the same liaison variable $x$ at three non-primary-lhs positions $(R,A)$, $(R,A')$, $(R,A'')$. Assume, without loss of generality, that $A'\ge A$; hence, if one of these positions occurs in $I_{R,i_R}^\rhs$, then $(R,A)\in I_{R,i_R}^\rhs$. Also assume that $\alpha=R(\bar z,x,\bar y)$. By the rewrite rule R5, we have that $(\dep,Q) \lrhd (\dep',Q')$, where $Q'=Q\setminus\set{R(\bar z,x,\bar y)}\cup\set{R'(\bar z,\bar y)}$ for a relation name $R'$ that does not occur in $Q$ (and we have removed the occurrence of $x$ in position $(R,A')$ to obtain $R'(\bar z,\bar y)$). Since $(\dep,Q)$ is final, $Q'$ is $\dep'$-safe, and one of the conditions of $\IsSafe$ is satisfied by $Q'$. If the primary FD of $\dep_R$ is $R:X_{i_R}\ra Y_{i_R}$, then the primary FD of $\dep'_{R'}$ is $R':X_{i_R}\ra (Y_{i_R}\setminus\set{A'})$ due to Lemma~\ref{lemma:no-empty-pvars} and since $Y_{i_R}\setminus\set{A'} \neq \emptyset$.
\begin{enumerate}
    \item The condition of line~1 is not satisfied by $Q'$ as $R'(\bar z,\bar y)$ uses the liaison variable $x$ at a non-primary-lhs position, and $x$ remains a liaison as it has two more occurrences in $\alpha$.
    \item Clearly, if $Q'$ satisfies the condition of line~3, then so does $Q$, as the variable $x$ still occurs in $R'(\bar z,\bar y)$, which contradicts Lemma~\ref{lemma:final-no-condition}.
    
    \item For every $\beta\in (Q\cap Q')$, we have that $\beta\in\comp{Q'}{\dep'}$ if and only if $\alpha\in\comp{Q}{\dep}$, since every liaison variable of $Q'$ is also a liaison variable of $Q$ and vice versa. Moreover, $R'(\bar z,\bar y)\in\comp{Q'}{\dep'}$ and $R(\bar z,x,\bar y)\in\comp{Q}{\dep}$, as $x$ is a liaison variable in both cases. Hence, if $Q'$ satisfies the condition of line~5 or the condition of line~7, then $Q$ satisfies the same condition, which contradicts Lemma~\ref{lemma:final-no-condition}.
\end{enumerate}
That concludes our proof.
\end{proof}

\begin{lemma}\label{final-no-lvar-and-const}
Consider a set $\dep$ of FDs with an LHS chain, and an SJFCQ $Q$ such that $(\dep,Q)$ is final. No atom of $Q$ has a liaison variable at a position of $I_{R,i_R}^\rhs$ and a constant at a non-primary-lhs position.
\end{lemma}

\begin{proof}
Assume, towards a contradiction, that there is an atom $\alpha$ of $Q$ over a relation name $R$ that uses a liaison variable $x$ at the position $(R,A)$ of $I_{R,i_R}^\rhs$, and a constant $c$ at a non-primary-lhs position $(R,A')$. Assume also that $\alpha=R(\bar z,c,\bar y)$. By the rewrite rule R5, we have that $(\dep,Q) \lrhd (\dep',Q')$, where $Q'=Q\setminus\set{R(\bar z,c,\bar y)}\cup\set{R'(\bar z,\bar y)}$ for a relation name $R'$ not in $Q$. Since $(\dep,Q)$ is final, $Q'$ is $\dep'$-safe, and one of the conditions of $\IsSafe$ is satisfied by $Q'$. Note that if the primary FD of $\dep_R$ is $R:X_{i_R}\ra Y_{i_R}$, then the primary FD of $\dep'_{R'}$ is $R':X_{i_R}\ra (Y_{i_R}\setminus\set{A'})$ because $Y_{i_R}\setminus\set{A'}$ contains the attribute $A$ that is associated with a variable.
\begin{enumerate}
    \item The condition of line~1 is not satisfied by $Q'$ because $R'(\bar z,\bar y)$ uses the liaison variable $x$ at a non-primary-lhs position; thus, $R'(\bar z,\bar y)\in\comp{Q'}{\dep'}$.
    
    \item Clearly, if $Q'$ satisfies the condition of line~3, then so does $Q$, as we have only removed a constant from an atom of $Q$ to obtain $Q'$, which contradicts Lemma~\ref{lemma:final-no-condition}.
    
    \item For every $\beta\in (Q\cap Q')$, we have that $\beta\in\comp{Q'}{\dep'}$ if and only if $\beta\in\comp{Q}{\dep}$, since every liaison variable of $Q'$ is also a liaison variable of $Q$ and vice versa. Moreover, $R'(\bar z,\bar y)\in\comp{Q'}{\dep'}$ and $R(\bar z,c,\bar y)\in\comp{Q}{\dep}$ since both atoms use the liaison variable $x$ at a non-primary-lhs position. Thus, if $Q'$ satisfies the condition of line~5 or the condition of line~7, then $Q$ satisfies the same condition, which contradicts Lemma~\ref{lemma:final-no-condition}.
\end{enumerate}
That concludes our proof.
\end{proof}

\begin{lemma}\label{final-no-lvar-and-orphan}
Consider a set $\dep$ of FDs with an LHS chain, and an SJFCQ $Q$ such that $(\dep,Q)$ is final. No atom of $Q$ has a liaison variable or a constant at a position of $I_{R,i_R}^\rhs$ and an orphan variable at a non-primary-lhs position.
\end{lemma}

\begin{proof}
Assume, towards a contradiction, that there is an atom $\alpha$ of $Q$ over a relation name $R$ that uses a liaison variable $x$ or a constant $c$ at the position $(R,A)$ of $I_{R,i_R}^\rhs$ and an orphan variable $w$ at a non-primary-lhs position $(R,A')$. Also, assume that $\alpha=R(\bar z,w,\bar y)$. By the rewrite rule R8, we have that $(\dep,Q) \lrhd (\dep',Q')$, where $Q'=Q\setminus\set{R(\bar z,w,\bar y)}\cup\set{R'(\bar z,\bar y)}$ for a relation name $R'$ that does not occur in $Q$. Since $(\dep,Q)$ is final, $Q'$ is $\dep'$-safe, and one of the conditions of $\IsSafe$ is satisfied by $Q'$. If the primary FD of $\dep_R$ is $R:X_{i_R}\ra Y_{i_R}$, then the primary FD of $\dep'_{R'}$ is $R':X_{i_R}\ra (Y_{i_R}\setminus\set{A'})$ because $A\in Y_{i_R}\setminus\set{A'}$, and Lemma~\ref{lemma:no-empty-pvars} implies that an attribute of $X_{i_R}$ is associated with a variable.
\begin{enumerate}
    \item The condition of line~1 is not satisfied by $Q'$ because $R'(\bar z,\bar y)$ uses the liaison variable $x$ or the constant $c$ at a non-primary-lhs position; thus, $R'(\bar z,\bar y)\in\comp{Q'}{\dep'}$.
    \item Clearly, if $Q'$ satisfies the condition of line~3, then so does $Q$, as we have only removed an orphan variable from an atom of $Q$ to obtain $Q'$, which contradicts Lemma~\ref{lemma:final-no-condition}.
    
    \item For every $\beta\in (Q\cap Q')$, we have that $\beta\in\comp{Q'}{\dep'}$ if and only if $\beta\in\comp{Q}{\dep}$, since every liaison variable of $Q'$ is also a liaison variable of $Q$ and vice versa. Moreover, $R'(\bar z,\bar y)\in\comp{Q'}{\dep'}$ and $R(\bar z,c,\bar y)\in\comp{Q}{\dep}$ since both atoms use the liaison variable $x$ or the constant $c$ at a non-primary-lhs position. Thus, if $Q'$ satisfies the condition of line~5 or the condition of line~7, then $Q$ satisfies the same condition, which contradicts Lemma~\ref{lemma:final-no-condition}.
\end{enumerate}
That concludes our proof.
\end{proof}

\begin{lemma}\label{final-no-lvar-and-lvar}
Consider a set $\dep$ of FDs with an LHS chain, and an SJFCQ $Q$ such that $(\dep,Q)$ is final. No atom $\alpha$ of $Q$ has two distinct liaison variables \rev{with at least one of them occurring at a non-primary-lhs position of $\alpha$}.
\end{lemma}

\begin{proof}
Assume, towards a contradiction, that there is an atom $\alpha=R(\bar z)$ of $Q$ that uses a liaison variable $x$ at position $(R,A)$, and a liaison variable $y\neq x$ at position $(R,A')$. Since none of the variables is an orphan variable and at least one of them occurs at a non-primary-lhs position of $\alpha$, by the rewrite rule R7 we have that $(\dep,Q) \lrhd (\dep,Q_{y\mapsto x})$. Let $Q'=Q_{y\mapsto x}$. For every atom $\beta\in Q$, we have that $\beta\in\comp{Q}{\dep}$ if and only if $\beta_{y\mapsto x}\in\comp{Q'}{\dep}$, where $\beta_{y\mapsto x}$ is obtained from $\beta$
by replacing every occurrence of the variable $y$ with the variable $x$. This is clearly the case if $\beta$ has a constant or a liaison variable $z\neq y$ at a non-primary-lhs position. This is also the case if $y$ is a liaison variable of $\beta$, since $x$ is a liaison variable of $\beta_{y\mapsto x}$.
Since $(\dep,Q)$ is final, $Q'$ is $\dep$-safe, and one of the conditions of $\IsSafe$ is satisfied by $Q'$.
\begin{enumerate}
    \item The condition of line~1 is not satisfied by $Q'$ because the variable $x$ is a liaison variable that occurs at a non-primary-lhs position of $\alpha_{y\mapsto x}$; hence, $\comp{Q'}{\dep}\neq\emptyset$.
    
    \item If $Q'$ satisfies the condition of line~3, then $Q'=Q_1'\cup Q_2'$ for $Q_1'\neq\emptyset\neq Q_2'$ and $\var{Q_1'}\cap\var{Q_2'}=\emptyset$. Assume, without loss of generality, that all the atoms that use the variable $x$ are in $Q_1'$. Then, we have that $Q=Q_1\cup Q_2$, where $Q_1$ contains every atom that uses either the variable $x$ or the variable $y$ and $Q_2=Q_2'$. Clearly, we have that     $Q_1\neq\emptyset\neq Q_2$ and $\var{Q_1}\cap\var{Q_2}=\emptyset$; thus, $Q$ satisfies the condition of line~3, which contradicts Lemma~\ref{lemma:final-no-condition}.
    
    \item As said above, we have that $\beta\in\comp{Q}{\dep}$ if and only if $\beta_{y\mapsto x}\in\comp{Q'}{\dep}$. If there is a variable $z$ that appears at a primary-lhs position of every atom in $\comp{Q'}{\dep}$, then it must be the case that $z\neq x$, since Lemma~\ref{lemma:final-no-variable-prim-nonprim} implies that none of the variables $x$ or $y$ occurs at a primary-lhs position of $\alpha$ and $\alpha_{y\mapsto x}\in\comp{Q'}{\dep}$. Therefore, we have that $z\in\pvar{\beta}{\dep}$ for every $\beta\in\comp{Q}{\dep}$ and $Q$ satisfies the condition of line~5, which contradicts Lemma~\ref{lemma:final-no-condition}.
    
    \item If the condition of line~7 is satisfied by $Q'$, then it is also satisfied by $Q$ as we have not changed constants or removed variables from primary-lhs positions, which contradicts Lemma~\ref{lemma:final-no-condition}.
\end{enumerate}
That concludes our proof.
\end{proof}

\begin{lemma}\label{lemma:final-twice-liaison}
Consider a set $\dep$ of FDs with an LHS chain, and an SJFCQ $Q$ such that $(\dep,Q)$ is final. If the $R$-atom of $Q$ uses a liaison variable $x$ or a constant $c$ at a position of $I_{R,i_R}^\rhs$, then any liaison variable has at most one occurrence at a non-primary-lhs position outside $I_{R,i_R}^\rhs$.
\end{lemma}

\begin{proof}
Assume, towards a contradiction, that there is an atom $\alpha$ of $Q$ over a relation name $R$ that uses a liaison variable $x$ or a constant $c$ at the position $(R,A)$ of $I_{R,i_R}^\rhs$, and there is a liaison variable $w$ that occurs at two distinct non-primary-lhs positions $(R,A')$ and $(R,A'')$ of $\alpha$ outside $I_{R,i_R}^\rhs$. Assume, without loss of generality, that $A'\ge A''$. Also assume that $\alpha=R(\bar z, w,\bar y)$. By the rewrite rule R5, we have that $(\dep,Q) \lrhd (\dep',Q')$, where $Q'=Q\setminus\set{R(\bar z,w,\bar y)}\cup\set{R'(\bar z,\bar y)}$ for a relation name $R'$ that does not occur in $Q$ (and we have removed the occurrence of $w$ at position $(R,A')$ to obtain $R'(\bar z,\bar y)$). Since $(\dep,Q)$ is final, $Q'$ is $\dep'$-safe, and one of the conditions of $\IsSafe$ is satisfied by $Q'$. If the primary FD of $\dep_R$ is $R:X_{i_R}\ra Y_{i_R}$, then the primary FD of $\dep'_{R'}$ is $R':X_{i_R}\ra (Y_{i_R}\setminus\set{A'})$ because $A\in Y_{i_R}\setminus\set{A'}$, and Lemma~\ref{lemma:no-empty-pvars} implies that an attribute of $X_{i_R}$ is associated with a variable.
\begin{enumerate}
    \item The condition of line~1 is not satisfied by $Q'$ because $R'(\bar z,\bar y)$ uses the liaison variable $x$ or the constant $c$ at a non-primary-lhs position; thus, $R'(\bar z,\bar y)\in\comp{Q'}{\dep'}$. Note that $x$ is a liaison variable of $R'(\bar z,\bar y)$ even if $w=x$, as in this case, there is another occurrence of the variable $x$ at position $(R,A'')$.
    
    \item Clearly, if $Q'$ satisfies the condition of line~3, then so does $Q$, as the variable $w$ still occurs in $R'(\bar z,\bar y)$, which contradicts Lemma~\ref{lemma:final-no-condition}.
    
    \item For every $\beta\in (Q\cap Q')$, we have that $\beta\in\comp{Q'}{\dep'}$ if and only if $\beta\in\comp{Q}{\dep}$, since every liaison variable of $Q'$ is also a liaison variable of $Q$ and vice versa (if the variable $w$ occurs at a non-primary-lhs position of $\beta$, then it is still a liaison variable as it has another occurrence in $R'(\bar z,\bar y)$). Moreover, $R'(\bar z,\bar y)\in\comp{Q'}{\dep'}$ and $\alpha\in\comp{Q}{\dep}$ since both atoms use the liaison variable $x$ or the constant $c$ at a non-primary-lhs position. Therefore, if $Q'$ satisfies the condition of line~5 or the condition of line~7, then $Q$ satisfies the same condition, which contradicts Lemma~\ref{lemma:final-no-condition}.
\end{enumerate}
That concludes our proof.
\end{proof}

\begin{lemma}\label{lemma:final-orphan-liaison-and-constant}
Consider a set $\dep$ of FDs with an LHS chain, and an SJFCQ $Q$ such that $(\dep,Q)$ is final. If the $R$-atom of $Q$ uses an orphan variable $x$ at a position of $I_{R,i_R}^\rhs$, then non-primary-lhs positions outside $I_{R,i_R}^\rhs$ are associated only with constants or only with liaison variables.
\end{lemma}

\begin{proof}
Assume, towards a contradiction, that there is an atom $\alpha$ of $Q$ over a relation name $R$ that uses an orphan variable $x$ at the position $(R,A)$ of $I_{R,i_R}^\rhs$, and there is another non-primary-lhs position $(R,A')$ outside $I_{R,i_R}^\rhs$ that is associated with a constant $c$ and a non-primary-lhs position $(R,A'')$ outside $I_{R,i_R}^\rhs$ that is associated with a liaison variable $y$. Assume that $\alpha=R(\bar z, c, \bar y)$. By the rewrite rule R5, $(\dep,Q) \lrhd (\dep',Q')$, where $Q'=Q\setminus\set{R(\bar z,c,\bar y)}\cup\set{R'(\bar z,\bar y)}$ for a relation name $R'$ not in $Q$. Since $(\dep,Q)$ is final, $Q'$ is $\dep'$-safe, and one of the conditions of $\IsSafe$ is satisfied by $Q'$. If the primary FD of $\dep_R$ is $R:X_{i_R}\ra Y_{i_R}$, then the primary FD of $\dep'_{R'}$ is $R':X_{i_R}\ra (Y_{i_R}\setminus\set{A'})$ because $Y_{i_R}\setminus\set{A'}$ contains the attribute $A$ that is associated with a variable in $\alpha$.
\begin{enumerate}
    \item The condition of line~1 is not satisfied by $Q'$ because $R'(\bar z,\bar y)$ uses the liaison variable $y$ at a non-primary-lhs position; thus, $R'(\bar z,\bar y)\in\comp{Q'}{\dep'}$.
    
    \item Clearly, if $Q'$ satisfies the condition of line~3, then so does $Q$, as we have only removed a constant from an atom of $Q$ to obtain $Q'$, which contradicts Lemma~\ref{lemma:final-no-condition}.
    
    \item For every $\beta\in (Q\cap Q')$, we have that $\beta\in\comp{Q'}{\dep'}$ if and only if $\beta\in\comp{Q}{\dep}$, since every liaison variable of $Q'$ is also a liaison variable of $Q$ and vice versa. Moreover, $R'(\bar z,\bar y)\in\comp{Q'}{\dep'}$ and $R(\bar z,c,\bar y)\in\comp{Q}{\dep}$ since both atoms use the liaison variable $y$ at a non-primary-lhs position. Thus, it is clear that if $Q'$ satisfies the condition of line~5 or the condition of line~7, then $Q$ satisfies the same condition, which contradicts Lemma~\ref{lemma:final-no-condition}.
\end{enumerate}
That concludes our proof.
\end{proof}

We now have all the ingredients for proving Lemmas~\ref{lemma:final-liaison} -- \ref{lemma:final-orphan}, which essentially tell us that, for a set $\dep$ of FDs with an LHS chain and an SJFCQ $Q$ such that $(\dep,Q)$ is final, all the $R$-atoms of $\comp{\dep}{Q}$ are such that $\dep_R$ either has a single key or is a set of FDs of one of six specific forms.

\begin{replemma}{\ref{lemma:final-liaison}}
	\lemmafinalliaison
\end{replemma}

\begin{proof}
	Let $\alpha$ be the $R$-atom of $Q$. Lemma~\ref{lemma:final-no-constant-primary} implies that $\dep_R$ has no primary prefix.
	Lemmas~\ref{final-no-lvar-and-const}, ~\ref{final-no-lvar-and-orphan} and~\ref{final-no-lvar-and-lvar} imply that there are no constants, orphan variables, or other liaison variables at non-primary-lhs positions of $\alpha$. Moreover, Lemma~\ref{lemma:final-twice-liaison} implies that there are at most two occurrences of the variable $x$ at non-primary-lhs positions of $\alpha$. If there is one occurrence of the variable $x$, then clearly $\dep_R$ has a single key (with one attribute on the right-hand side, i.e., the attribute $A$), and if there are two occurrences, then either $\dep_R$ has a single key (with two attributes on the right-hand side) or $\dep_R$ has a single FD and the second occurrence of $x$ is at a position corresponding to an attribute that does not occur in any FD. This is because there is only one attribute that does not occur at the primary FD (as there are only two occurrences of $x$), and if there is another FD in $\dep_R$, then this attribute occurs in its left-hand side, which means that its right-hand side is empty. In this case, the additional FD is a trivial FD that can be removed from $\dep_R$.
\end{proof}

\begin{replemma}{\ref{lemma:final-const}}
	\lemmafinalconst
\end{replemma}

\begin{proof}
	The proof is similar to the proof of Lemma~\ref{lemma:final-liaison}. Let $\alpha$ be the $R$-atom of $Q$.
	Lemma~\ref{lemma:final-no-constant-primary} implies that $\dep_R$ has no primary prefix. Lemmas~\ref{lemma:final-no-two-constants-non-prim}, ~\ref{final-no-lvar-and-orphan} and~\ref{final-no-lvar-and-lvar} imply that there are no constants, orphan variables, or two distinct liaison variables at non-primary-lhs positions of $\alpha$ (besides the position of $I_{R,i_R}^\rhs$ that is associated with the constant $c$). Moreover, Lemma~\ref{lemma:final-twice-liaison} implies that there is at most one occurrence of some liaison variable $x$ at non-primary-lhs positions of $\alpha$. If there are no occurrences of a liaison variable $x$, then clearly $\dep_R$ has a single key (with the attribute $A$ in its right-hand side). If there is one occurrence of such a variable, then Lemma~\ref{final-no-lvar-and-const} implies that it occurs in a position outside $I_{R,i_R}^\rhs$, in which case $\dep_R$ has a single FD and the occurrence of $x$ is at a position corresponding to an attribute that does not occur in any FD. As for Lemma~\ref{lemma:final-liaison}, there is no additional FD in the set, as there is only one non-primary-lhs position outside $I_{R,i_R}^\rhs$.
\end{proof}

\begin{replemma}{\ref{lemma:final-orphan}}
	\lemmafinalorphan
\end{replemma}

\begin{proof}
	Lemma~\ref{lemma:final-no-constant-primary} implies that $\dep_R$ has no primary prefix.
	Lemma~\ref{final-no-lvar-and-orphan} implies that there are no constants or liaison variables at positions of $I_{R,i_R}^\rhs$. Lemma~\ref{lemma:final-no-two-orphan-non-prim} implies that there are no other orphan variables at the positions of $I_{R,i_R}^\rhs$. Hence, there is a single position in $I_{R,i_R}^\rhs$ and this position is associated with the orphan variable $x$. Then, the primary FD of $\dep_R$ is $R:X\ra\set{A}$.
	
	Lemma~\ref{lemma:final-no-two-orphan-non-prim} implies that no non-primary-lhs positions outside $I_{R,i_R}^\rhs$ are associated with orphan variables. Lemma~\ref{lemma:final-orphan-liaison-and-constant} implies that the non-primary-lhs positions outside $I_{R,i_R}^\rhs$ are associated either with constants or liaison variables, but not both. If the non-primary-lhs positions outside $I_{R,i_R}^\rhs$ are associated with constants, then Lemma~\ref{lemma:final-no-two-constants-non-prim} implies that there is only one such position $(R,A')$. Thus, the only possible case is that this position is associated with an attribute that does not occur in any FD. There cannot be an additional FD that occurs after the primary FD in the chain as $(R,A')$ is the only position that does not correspond to an attribute of the primary FD, and as explained in the previous lemmas, at least two additional positions are required to add a non-trivial FD. Moreover, as aforementioned, the primary prefix is empty.
	
	If the non-primary-lhs positions outside $I_{R,i_R}^\rhs$ are associated with liaison variables, by Lemma~\ref{final-no-lvar-and-lvar} they are all associated with the same liaison variable $x$. Lemma~\ref{lemma:final-no-two-liaison-shared} implies that if $x$ occurs also at another atom of the query, then it has only one occurrence at a non-primary-lhs position outside $I_{R,i_R}^\rhs$, in which case the corresponding attribute does not occur in any FD (as explained above, there are not enough positions that are not associated with attributes of the primary FD for an additional FD). Lemma~\ref{lemma:final-no-two-liaison-unique} implies that if $x$ occurs only in one atom, then it has two occurrences at non-primary-lhs positions $(R,A')$ and $(R,A'')$ outside $I_{R,i_R}^\rhs$. In this case, either the two corresponding attributes do not occur in any FD or there is an FD $R:(X\cup\set{A'})\ra \set{A''}$.
\end{proof}

\subsection{Proofs of Lemmas~\ref{lemma:red_first} -- \ref{lemma:red_last}}

For a set $\dep$ of FDs with an LHS chain, and an SJFCQ $Q$ such that $(\dep,Q)$ is final, Lemmas~\ref{lemma:12} -- \ref{lemma:final-orphan} imply that for every relation name $R$ that occurs in $Q$, $\dep_R$ consists of a single key, a single FD, or two FDs. We now prove that if $Q$ has $n>0$ non-single-key atoms w.r.t.~$\dep$, then there is another set $\dep'$ of FDs and query $Q'$ such that:
\begin{itemize}
    \item $n-1$ atoms of $Q'$ are non-single-key atoms w.r.t.~$\dep'$,
    \item $Q'$ is not $\dep'$-safe, and
    \item $\prob{RelFreq}(\dep',Q')$ is Cook reducible to $\prob{RelFreq}(\dep,Q)$.
\end{itemize}
We prove this separately for each one of the six cases that we have identified in Lemmas~\ref{lemma:final-liaison}, ~\ref{lemma:final-const}, and~\ref{lemma:final-orphan}.
The proofs for the different cases are similar; hence, we provide the full proof once, and only the essential details for all the other cases. The idea is the following. We replace an atom $R(\bar z)$ of $Q$ for which $\dep_R$ does not consist of a single key with two atoms $P(\bar y)$ and $T(\bar w)$ in $Q'$. The atom $P(\bar y)$ has the primary FD of $\dep_R$ as its single key, and the atom $T(\bar w)$ is similar to the original atom, except that its associated set of FDs is empty, or equivalently, consists of a single trivial key (hence, $T(\bar w)$ does not belong to the complex part of $Q'$). Using the atom $P(\bar y)$ we ensure that there is a one-to-one correspondence between repairs w.r.t.~$\dep$ and repairs w.r.t.~$\dep'$, and using the atom $T(\bar w)$ we ensure that: \e{(1)} there is one-to-one correspondence between repairs w.r.t.~$\dep$ that entail $Q$ and repairs w.r.t.~$\dep'$ that entail $Q'$, and \e{(2)} $P(\bar y)$ is in the complex part of $Q'$ (as $R(\bar z)$ is in the complex part of $Q$). This implies that since $Q$ is not $\dep$-safe, the query $Q'$ is not $\dep'$-safe.
In what follows, we write $\ins{S}$ for the schema of $\dep$ and $Q$, and $\ins{S'}$ for the schema of $\dep'$ and $Q'$.

\begin{replemma}{\ref{lemma:red_first}}
	\lemmaredfirst
\end{replemma}

\begin{proof}
Assume, without loss of generality, that the $R$-atom of $Q$ is $R(\bar z,x,x)$. Let $Q'=Q\setminus\set{R(\bar z,x,x)}\cup\set{P(\bar z, x),T(\bar z,x,x)}$, where $P,T$ are fresh relation names not in $Q$. We denote by $\dep'$ the set of FDs obtained from $\dep$ by removing the FDs of $\dep_R$ and adding the FDs $P:(\att{P}\setminus\set{B})\ra\set{B}$ and $T:\att{T}\ra\emptyset$, where $B$ is the attribute associated with the variable $x$ in $P(\bar z, x)$.

Let $D$ be a database over the schema $\ins{S}'$. Lemma~\ref{lemma:hard_help2} implies that the only relations that affect the relative frequency of $Q'$ w.r.t.~$D$ and $\dep'$ are those that occur in $Q'$. Hence, we assume, without loss of generality, that $D$ contains only such relations. Given the database $D$, we construct a database $E$ over the schema $\ins{S}$ as follows. For every relation $R'$ that occurs in both $\ins{S}$ and $\ins{S}'$, we add all the $R'$-facts of $D$ to $E$. Moreover, for every fact $P(\bar a,b)$ in $D$ we add to $R$ the fact $R(\bar a, b, \odot)$, where $\odot$ is a constant not in $D$, if $T(\bar a,b,b)\not\in D$ and the fact $R(\bar a, b, b)$ if $T(\bar a,b,b)\in D$.

Given a subset $J$ of $D$, the corresponding subset $J'$ of $E$ is defined as follows. For every relation $R'$ that occurs in both $\ins{S}$ and $\ins{S}'$, the subsets $J$ and $J'$ contain the same $R'$-facts. Moreover, for every fact $P(\bar a,b)\in J$, the subset $J'$ contains the fact $R(\bar a,b,*)$, where $*$ is either $\odot$ or $b$. It is rather straightforward that $J$ is a repair of $D$ if and only if $J'$ is a repair of $E$. Clearly, for every relation name $R'$ that occurs in both $\ins{S}$ and $\ins{S}'$, we have that $J$ satisfies $\dep'_{R'}$ if and only if $J$ satisfies $\dep_{R'}$ as both repairs contain the same facts over these relation names and the sets of FDs are the same. Moreover, $J$ satisfies $\dep'_P$ if and only if $J'$ satisfies $\dep_R$ because for every $\bar a$, $J$ contains a single fact of the form $P(\bar a,b)$ (and satisfies $\dep'_P$) if and only if $J'$ contains a single fact of the form $P(\bar a,b,*)$ for the same $b$ (hence, it satisfies $\dep_R$). Clearly, $J$ is maximal if and only if $J'$ is maximal, since we could add another fact $R(\bar a,b,*)$ to $J'$ without violating consistency (if no fact of the form $R(\bar a,*,*)$ already occurs in $J'$) if and only if we could add the fact $R(\bar a,b)$ to $J$ without violating consistency (if no fact of the form $R(\bar a,*)$ already occurs in $J$). Note that the only FD in $\dep'_{T}$ is trivial; hence, $J$ contains all the $T$-facts of $D$. We conclude that there is one-to-one correspondence between the repairs of $D$ w.r.t.~$\dep'$ and the repairs of $E$ w.r.t.~$\dep$. Therefore,
\[
\card{\rep{D}{\dep'}}\ =\ \card{\rep{E}{\dep}}.
\]

Next, we show that a repair $J$ of $D$ entails $Q'$ if and only if the corresponding repair $J'$ of $E$ entails $Q$. Let $J$ be a repair of $D$ that entails $Q'$, and let $h$ be a homomorphism from $Q'$ to $J$. We claim that $h$ is a homomorphism from $Q$ to $J'$. Clearly, for every relation name $R'$ that occurs in both $\ins{S}$ and $\ins{S}'$, the fact that $R'(h(\bar z))\in J$, where $R'(\bar z)$ is the atom of $Q'$ over the relation name $R'$, implies that $R'(h(\bar z))\in J'$, as the two repairs contain the same facts over these relation names. Moreover, by the definition of the repair $J'$, the fact that $P(h(\bar z, x))\in J$ and $T(h(\bar z, x,x))\in J$ implies that $R(h(\bar z,x,x))\in J'$. Therefore, $h$ is indeed a homomorphism from $Q$ to $J'$.
We can similarly show that if $J'$ is a repair of $E$ that entails $Q$, then the corresponding repair $J$ of $D$ entails $Q'$. For a homomorphism $h$ from $Q$ to $J'$, we clearly have that $R'(h(\bar z))\in J$ if and only if $R'(h(\bar z))\in J'$ for relation names $R'$ that occur in both $\ins{S}$ and $\ins{S}'$. Furthermore, if $R(h(\bar z,x,x))\in J'$, then $P(h(\bar z, x))\in J$, and since $J$ contains all the $T$-facts of $D$, we also have that $T(h(\bar z, x,x))\in J$, and $h$ is a homomorphism from $Q'$ to $J$. We conclude that
\[
\card{\rep{D}{\dep',Q'}}\ =\ \card{\rep{E}{\dep,Q}}
\]
and
\[
\rfreq{Q'}{D,\dep'}\ =\ \rfreq{Q}{E,\dep}.
\]

It is only left to show that $Q'$ is not $\dep'$-safe. Assume, towards a contradiction, that $Q'$ is $\dep'$-safe, which in turn implies that it satisfies one of the conditions of $\IsSafe$.
\begin{enumerate}
    \item We have that $\comp{Q}{\dep'}\neq\emptyset$ because $P(\bar z, x)$ has the liaison variable $x$ at a non-primary-lhs position (it is a liaison variable as it also occurs in $T(\bar z,x,x)$); thus, $Q'$ does not satisfy the condition of line~1.
    
    \item If $Q'=Q_1'\cup Q_2'$ with $Q_1'\neq\emptyset\neq Q_2'$ and $\var{Q_1'}\cap \var{Q_2'}=\emptyset$, it also holds that $Q=Q_1\cup Q_2$ with $Q_1\neq\emptyset\neq Q_2$ and $\var{Q_1}\cap \var{Q_2}=\emptyset$, where $Q_1=Q_1'$ and $Q_2=Q_2'\setminus\set{P(\bar z,x),T(\bar z,x,x)}\cup\set{R(\bar z,x,x)}$, assuming, w.l.o.g, that $P(\bar z,x),T(\bar z,x,x)\in Q_2'$. In this case, $Q$ also satisfies the condition of line~3, which contradicts Lemma~\ref{lemma:final-no-condition}.
    
    \item For every atom $\beta$ that occurs in both $Q$ and $Q'$, we have that $\beta\in\comp{Q}{\dep}$ if and only if $\beta\in\comp{Q'}{\dep'}$ because every liaison variable in $Q$ is also a liaison variable in $Q'$ and vice versa. Moreover, $R(\bar z,x,x)\in \comp{Q}{\dep}$ and $P(\bar z,x)\in \comp{Q'}{\dep'}$, and since both atoms have the same variables at primary-lhs positions, it is rather clear that if $Q'$ satisfies the condition of line~5 or the condition of line~7, then so does $Q$, which contradicts Lemma~\ref{lemma:final-no-condition}.
\end{enumerate}
This concludes our proof.
\end{proof}

\begin{replemma}{\ref{lemma:red_second}}
	\lemmaredsecond
\end{replemma}

\begin{proof}
Assume, without loss of generality, that the $R$-atom of $Q$ is $R(\bar z,c,x)$. Let $Q'=Q\setminus\set{R(\bar z,c,x)}\cup\set{P(\bar z, c),T(\bar z,c,x)}$, where $P$ and $T$ are two fresh relation names not in $Q$. We denote by $\dep'$ the set of FDs obtained from $\dep$ by removing the FDs of $\dep_R$ and adding the FDs $P:(\att{P}\setminus\set{B})\ra\set{B}$ and $T:\att{T}\ra\emptyset$, where $B$ is the attribute associated with the constant $c$ in $P(\bar z, c)$.

Given a database $D$ over the schema $\ins{S}'$, we construct a database $E$ over $\ins{S}$ by adding to $E$ all the $R'$-facts of $D$ for relations $R'$ that occur in both $\ins{S}$ and $\ins{S}'$. Moreover, for every fact $P(\bar a,b)$ such that either \e{(1)} $b\neq c$ or \e{(2)} $b=c$ and there is no fact of the form $T(\bar a,c,*)$ in $D$, we add a fact $R(\bar a,b,\odot)$ to $E$ for a constant $\odot$ that does not occur in $D$. For every fact $P(\bar a,c)$ such that there is a fact of the form $T(\bar a,c,d)$ in $D$, we add the fact $R(\bar a,c,d)$ to $E$. (If there are multiple facts of the form $T(\bar a,c,*)$ for some $\bar a$ we also add multiple facts to $R$. These facts are not in a conflict, as they agree on the value of the attribute in the right-hand side of the FD of $\dep_R$.)

We can again show that $\rfreq{Q'}{D,\dep'}=\rfreq{Q}{E,\dep}$ and $Q'$ is not $\dep'$-safe. The proof is similar to that of Lemma~\ref{lemma:red_first} and we do not give it again. Observe that a homomorphism $h$ from $Q'$ to a repair $J$ of $D$ is such that $P(h(\bar z,c))\in J$ and $T(h(\bar z,c,x))\in J$, in which case, by the definition of $J'$, $R(h(\bar z,c,x))\in J'$. A homomorphism $h$ from $Q$ to a repair $J'$ of $E$ is such that $R(h(\bar z,c,x))\in J'$; hence, $P(h(\bar z,c))\in J$. Moreover, the variable $x$ has at least one more occurrence in $Q$; hence, $h(x)\in\adom{D}$ and $h(x)\neq\odot$. This implies that $T(h(\bar z,c,x))\in J$ by the definition of $J'$.
\end{proof}

\begin{replemma}{\ref{lemma:red_third}}
	\lemmaredthird
\end{replemma}

\begin{proof}
Assume, without loss of generality, that the $R$-atom of $Q$ is $R(\bar z,x,y)$. Let $Q'=Q\setminus\set{R(\bar z,x,y)}\cup\set{P(\bar z, x),T(\bar z,x,y)}$, where $P$ and $T$ are two fresh relation names not in $Q$. We denote by $\dep'$ the FD set obtained from $\dep$ by removing the FDs of $\dep_R$ and adding the FDs $P:(\att{P}\setminus\set{B})\ra\set{B}$ and $T:\att{T}\ra\emptyset$, where $B$ is the attribute associated with the variable $x$ in $P(\bar z, x)$.

Given a database $D$ over the schema $\ins{S}'$, we construct a database $E$ over $\ins{S}$ by adding to $E$ all the $R'$-facts of $D$ for the relation names $R'$ of $\ins{S} \cap \ins{S}'$. Moreover, for every fact $P(\bar a,b)$ such that there is no fact of the form $T(\bar a,b,*)$ in $D$ we add $R(\bar a,b,\odot)$ to $E$ for a constant $\odot$ that does not occur in $D$. For every fact $P(\bar a,b)$ such that there is a fact of the form $T(\bar a,b,c)$ in $D$, we add the fact $R(\bar a,b,c)$ to $E$. (If there are multiple facts of the form $T(\bar a,b,*)$ for some $\bar a$ we also add multiple facts to $R$, and these facts do not jointly violate $\dep_R$.)

We can show that $\rfreq{Q'}{D,\dep'}=\rfreq{Q}{E,\dep}$
and $Q'$ is not $\dep'$-safe. Note that, similarly to Lemma~\ref{lemma:red_second}, a homomorphism $h$ from $Q$ to a repair $J'$ of $E$ is such that $R(h(\bar z,x,y))\in J'$, and since $y$ is a liaison variable of $Q$, $h(y)\neq\odot$. This implies that $T(\bar z,x,y)\in J$ by the definition of $J$.
\end{proof}

\begin{replemma}{\ref{lemma:red_fourth}}
	\lemmaredfourth
\end{replemma}

\begin{proof}
Assume, without loss of generality, that the $R$-atom of $Q$ is $R(\bar z,x,c)$. Let $Q'=Q\setminus\set{R(\bar z,x,c)}\cup\set{P(\bar z, x),T(\bar z,x,c)}$, where $P$ and $T$ are two fresh relation names not in $Q$. We denote by $\dep'$ the FD set obtained from $\dep$ by removing the FDs of $\dep_R$ and adding the FDs $P:(\att{P}\setminus\set{B})\ra\set{B}$ and $T:\att{T}\ra\emptyset$, where $B$ is the attribute associated with the variable $x$ in $P(\bar z, x)$.

Given a database $D$ over $\ins{S}'$, we construct a database $E$ over $\ins{S}$ by adding to $E$ all the $R'$-facts of $D$ for relations $R'$ that occur in both $\ins{S}$ and $\ins{S}'$. Moreover, for every fact $P(\bar a,b)$ such that there is no fact $T(\bar a,b,c)$ in $D$ we add a fact $R(\bar a,b,\odot)$ to $E$ for a constant $\odot$ not in $D$. For every fact $P(\bar a,b)$ such that there is a fact $T(\bar a,b,c)$ in $D$, we add the fact $R(\bar a,b,c)$ to $E$. Similarly to the previous lemmas, We can show that $\rfreq{Q'}{D,\dep'}=\rfreq{Q}{E,\dep}$
and that $Q'$ is not $\dep'$-safe.
\end{proof}

\begin{replemma}{\ref{lemma:red_fifth}}
	\lemmaredfifth
\end{replemma}

\begin{proof}
Assume, without loss of generality, that $\alpha=R(\bar z,x,y,y)$.
Let $Q'=Q\setminus\set{R(\bar z,x,y,y)}\cup\set{P(\bar z, x),T(\bar z,x,y,y)}$, where $P$ and $T$ are two fresh relation names that do not occur in $Q$. We denote by $\dep'$ the FD set obtained from $\dep$ by removing the FDs of $\dep_R$ and adding the FDs $P:(\att{P}\setminus\set{B})\ra\set{B}$ and $T:\att{T}\ra\emptyset$, where $B$ is the attribute associated with the variable $x$ in $P(\bar z, x)$.

Given a database $D$ over $\ins{S}'$, we construct a database $E$ over $\ins{S}$ by adding to $E$ all the $R'$-facts of $D$ for relations $R'$ that occur in both $\ins{S}$ and $\ins{S}'$. Moreover, for every fact $P(\bar a,b)$ such that there is no fact of the form $T(\bar a,b,c,c)$, for some constant $c$, in $D$ we add a fact $R(\bar a,b,\odot,\otimes)$ to $E$ for constants $\odot\neq \otimes$ that do not occur in $D$. For every fact $P(\bar a,b)$ such that there is a fact of the form $T(\bar a,b,c,c)$ in $D$, we add the fact $R(\bar a,b,c,c)$ to $E$. (If there are multiple facts of the form $T(\bar a,b,*,*)$ for some $\bar a$ we also add multiple facts to $R$, and these facts do not jointly violate $\dep_R$.)

We can show that $\rfreq{Q'}{D,\dep'}=\rfreq{Q}{E,\dep}$
and that $Q'$ is not $\dep'$-safe.
The proof is similar to the proofs of the previous lemmas. An important observation here is that a homomorphism $h$ from $Q$ to a repair $J'$ of $E$ is such that $R(h(\bar z,x,y,y))\in J'$; hence, it cannot be the case that $R(h(\bar z,x,y,y))$ is a fact of the form $R(\bar a,b,\odot,\otimes)$, but rather a fact of the form $R(\bar a,b,c,c)$ for some constant $c$. By the definition of $J'$, $P(h(\bar z,x))\in J$ and $T(h(\bar z,x,y,y))\in J$, and $h$ is a homomorphism from $Q'$ to $J$.
\end{proof}

\begin{replemma}{\ref{lemma:red_last}}
	\lemmaredlast
\end{replemma}

\begin{proof}
Assume, without loss of generality, that $\alpha=R(\bar z,x,y,y)$.
Let $Q'=Q\setminus\set{R(\bar z,x,y,y)}\cup\set{P(\bar z, x),T(\bar z,x,y,y)}$, where $P,T$ are two fresh relation names not in $Q$. We denote by $\dep'$ the FD set obtained from $\dep$ by removing the FDs of $\dep_R$ and adding the FDs $P:(\att{P}\setminus\set{B})\ra\set{B}$ and $T:\att{T}\ra\emptyset$, where $B$ is the attribute associated with the variable $x$ in $P(\bar z, x)$.

Given a database $D$ over $\ins{S}'$, we construct a database $E$ over $\ins{S}$ by adding to $E$ all the $R'$-facts of $D$ for relations $R'$ that occur in both $\ins{S}$ and $\ins{S}'$. Moreover, for every fact $P(\bar a,b)$ such that there is no fact of the form $T(\bar a,b,c,c)$, for some constant $c$, in $D$ we add a fact $R(\bar a,b,b,\odot)$ to $E$ for a constant $\odot$ not in $D$. For every fact $P(\bar a,b)$ such that there is a fact of the form $T(\bar a,b,c,c)$ in $D$, we add the fact $R(\bar a,b,c,c)$ to $E$. If there are multiple facts of the form $T(\bar a,b,*,*)$ for some $\bar a$ we also add multiple facts to $R$. Note that there are no violations of the FD $R:(X\cup \set{A'})\ra\set{A''}$ in $E$ as no two facts of $E$ agree on the values of all the attributes in $X\cup\set{A'}$.

We can show that $\rfreq{Q'}{D,\dep'}=\rfreq{Q}{E,\dep}$
and that $Q'$ is not $\dep'$-safe.
The proof is similar to the proofs of the previous lemmas. Observe that a homomorphism $h$ from $Q$ to a repair $J'$ of $E$ is such that $R(h(\bar z,x,y,y))\in J'$; hence, it cannot be the case that $R(h(\bar z,x,y,y))$ is a fact of the form $R(\bar a,b,b,\odot)$, but rather a fact of the form $R(\bar a,b,c,c)$ for some constant $c$. By the definition of $J'$, $P(h(\bar z,x))\in J$ and $T(h(\bar z,x,y,y))\in J$, and $h$ is a homomorphism from $Q'$ to $J$.
\end{proof}
\section{Proofs of Section~\ref{sec:apx-counting}}

\subsection{Proof of Proposition~\ref{pro:union-of-sets-properties}}

\begin{repproposition}{\ref{pro:union-of-sets-properties}}
\prounionofsets
\end{repproposition}

\begin{proof}
	The proofs for items (1) and (3) have been given in the main body of the paper. We proceed to establish item (2).
	As discussed in the proof of item (1), $\rep{D}{\dep,H} = \rep{D \setminus D'}{\dep}$, where $D' = \{f \in D \mid H \cup \{f\} \not \models \dep\}$. Thus, sampling repairs from $\rep{D}{\dep,H}$ uniformly at random in polynomial time boils down to sampling repairs from $\rep{D \setminus D'}{\dep}$ uniformly at random in polynomial time. Therefore, to prove item (2), we only need to prove the following lemma:

\begin{lemma}\label{lem:lhs-sampler}
	Consider a set of FDs $\dep$ with an LHS chain (up to equivalence). There exists a randomized procedure $\mathsf{Sample}$ that accepts as input $D$ and $\dep$ such that the following hold:
	\begin{enumerate}
		\item For each $D' \in \rep{D}{\dep}$, $\Pr( \mathsf{Sample}(D,\dep) = D' ) = \frac{1}{\card{\rep{D}{\dep}}}$.
		\item $\mathsf{Sample}(D,\dep)$ runs in polynomial time in $||D||$.
	\end{enumerate}
\end{lemma}

\begin{proof}
We first need some auxiliary notions and results. Consider a database $D$, a set $\dep$ of FDs, and a partition $P = \{D_1,\ldots,D_n\}$ of $D$. We say that $P$ is {\em $\dep$-independent} (resp., {\em $\dep$-conflicting}) if for every distinct $i,j \in [n]$, $f \in D_i$ and $g \in D_j$ implies that $\{f,g\} \models \dep$ (resp., $\{f,g\} \not \models \dep$).
We first establish a useful technical result for the case where $P$ is $\dep$-independent.

\begin{lemma}\label{lem:partition-repairs-independent}
	Assume that $P$ is $\dep$-independent. For every database $D'$, the following are equivalent:
	\begin{enumerate}
		\item $D' \in \rep{D}{\dep}$.
		\item $D' = \bigcup_{i \in [n]} D'_i$, where $D'_i \in \rep{D_i}{\dep}$ for each $i \in [n]$.
	\end{enumerate}
\end{lemma}

\begin{proof}
	We first establish the direction $(1)$ implies $(2)$. Assume that $D' \in \rep{D}{\dep}$. Let $D'_i$ be the restriction of $D'$ to the facts in $D_i$, for each $i \in [n]$. We show that $D'_i \in \rep{D_i}{\dep}$, for each $i \in [n]$. Consider an arbitrary $i \in [n]$. Since $D' \models \dep$, then $D'_i \models \dep$. It remains to show that $D'_i$ is maximal. 
	Assume, towards a contradiction, that there is a fact $f \in D_i \setminus D'_i$ such that $D'_i \cup \{f\} \models \dep$. Since $P$ is a $\dep$-independent partition, no fact $g \in D'_j$, for $j \in [n]$ and $j \neq i$, is such that $\{f,g\} \not \models \dep$. Hence, $D' \cup \{f\} \ \models \dep$. This contradicts the fact that $D'$ is a repair of $\rep{D}{\dep}$, and the claim follows.
	
	We now establish $(2)$ implies $(1)$. Assume that $D' = \bigcup_{i \in [n]} D'_i$, where $D'_i \in \rep{D_i}{\dep}$ for each $i \in [n]$. Assume, towards a contradiction, that $D' \not \in \rep{D}{\dep}$. Note that $D'$ is consistent, i.e., $D' \models \dep$, since $P$ is $\dep$-independent. Hence, the fact that $D' \not \in \rep{D}{\dep}$ implies that $D'$ is not maximal, i.e.,  there is a fact $f \in D \setminus D'$ such that $D' \cup \{f\} \models \dep$.
	However, since $D'_i \subseteq D'$, for each $i \in [n]$, we conclude that $D'_i \cup \{f\} \models \dep$. The latter holds also if $i \in [n]$ is such that $f \in D_i$. This implies that $D'_i$ is not a repair of $\rep{D_i}{\dep}$, which contradicts our hypothesis.
\end{proof}

We now proceed to establish an analogous technical result for the case where $P$ is $\dep$-conflicting.

\begin{lemma}\label{lem:partition-repairs-conflicting}
	Assume that $P$ is $\dep$-conflicting. For every database $D'$, the following are equivalent:
	\begin{enumerate}
		\item $D' \in \rep{D}{\dep}$.
		\item $D' \in \rep{D_i}{\dep}$ for some $i \in [n]$.
	\end{enumerate}
\end{lemma}

\begin{proof}
	We first establish $(1)$ implies $(2)$. Assume that $D' \in \rep{D}{\dep}$. Since $P$ is a $\dep$-conflicting partition, and since $D' \models \dep$, $D'$ must necessarily contain only facts that all belong to some $D_i$, i.e., $D' \subseteq D_i$, for some $i \in [n]$. Moreover, since $D' \in \rep{D}{\dep}$, for every fact $f \in D \setminus D'$, $D' \cup \{f\} \not \models \dep$. This also holds if $f \in D_i$, and therefore $D' \in \rep{D_i}{\dep}$, as needed.
	
	We now establish $(2)$ implies $(1)$. Assume that $D' \in \rep{D_i}{\dep}$ for some $i \in [n]$. Since $D'$ is a repair of $D_i$ w.r.t.~$\dep$ and $P$ is a $\dep$-conflicting partition of $D$, for every fact $f \in D \setminus D'$ it holds that $D' \cup \{f\} \not \models \dep$. Hence, $D' \in \rep{D}{\dep}$, and the claim follows.
\end{proof}

The following is an easy consequence of Lemmas~\ref{lem:partition-repairs-independent} and~\ref{lem:partition-repairs-conflicting}:

\begin{corollary}\label{cor:partition-count}
		Consider a database $D$, a set of FDs $\dep$, and a partition $P = \{D_1,\ldots,D_n\}$ of $D$.
		\begin{enumerate}
			\item If $P$ is $\dep$-independent, then $\card{\rep{D}{\dep}} = \prod_{i \in [n]} \card{\rep{D_i}{\dep}}$.
			\item If $P$ is $\dep$-conflicting, then $\card{\rep{D}{\dep}} = \sum_{i \in [n]} \card{\rep{D_i}{\dep}}$.
		\end{enumerate}
\end{corollary}

By exploiting the above results and the notion of blocktree, we show that we can sample repairs of a single-relation database when considering sets of FDs with an LHS chain.

\begin{algorithm}[t]
	\LinesNumbered
	\KwIn{A database $D$ over $\ins{S}$, a set $\dep$ of FDs with an LHS chain over $\ins{S}$, and $R \in \ins{S}$}
	\KwOut{$D' \in \rep{D_R}{\dep_R}$ with probability $\frac{1}{\card{\rep{D_R}{\dep_R}}}$}
	\vspace{2mm}
	\SetKwProg{Fn}{Function}{:}{}
	
	Let $\Lambda = \phi_1,\ldots,\phi_n$ be an LHS chain of $\dep_R$\\
	Construct the $(R,\Lambda)$-blocktree $T= (V,E,\lambda)$ of $D$ w.r.t $\dep$\\
	Let $u_r$ be the root of $T$\\
	\Return $\mathsf{SampleFromTree}(T, u_r)$
	
	\vspace{2mm}
	\Fn{$\mathsf{SampleFromTree}(T = (V,E,\lambda), u)$}{
		\eIf{$u$ is a leaf node of $T$}{
			\textbf{return} $\lambda(u)$\;
		}{
			repair $\gets\emptyset$\;
			\For{$v\in \child{u}$}{
				Choose $v'\in \child{v}$ with probability $\frac{\card{\rep{\lambda(v')}{\dep_R}}}{\card{\rep{\lambda(v)}{\dep_R}}}$\\
				$D_{v'}$ $\gets$ $\mathsf{SampleFromTree}(T, v')$\\
				repair $\gets$ repair $\cup D_{v'}$
			}
		}
		\textbf{return} repair\;
	}
	
	\vspace{1em}

	\caption{The procedure $\mathsf{RSample}$}\label{alg:rsample}
\end{algorithm}

\begin{lemma}\label{lem:lhs-sampler-R}
	Consider a database $D$, and a set of FDs $\dep$ with an LHS chain, both over a schema $\ins{S}$. There is a randomized procedure $\mathsf{RSample}$ that takes as input $D$, $\dep$ and $R \in \ins{S}$ such that:
	\begin{enumerate}
		\item For each $D' \in \rep{D_R}{\dep_R}$, $\Pr(\mathsf{RSample}(D,\dep,R) = D') = \frac{1}{\card{\rep{D_R}{\dep_R}}}$. 
		\item $\mathsf{RSample}(D,\dep,R)$ runs in polynomial time in $||D||$.
	\end{enumerate}
\end{lemma}

\begin{proof}
	Consider the randomized procedure $\mathsf{RSample}$ shown in Algorithm~\ref{alg:rsample}. It first constructs the $(R,\Lambda)$-blocktree $T = (V,E,\lambda)$ of $D$ w.r.t.~$\dep$, for some LHS chain $\Lambda = \phi_1,\ldots,\phi_n$ of $\dep_R$. Then, it executes the auxiliary procedure $\mathsf{SampleFromTree}$, with input $T$ and its root, and returns the result of this call; here, for a node $u$ of $T$, we use $\child{u}$ to denote the set of children of $u$ in $T$.

	\medskip
	
	\noindent\paragraph{\underline{Correctness of $\mathsf{RSample}$}} 
	\smallskip
	
	\noindent We first show that $\mathsf{RSample}$ is correct, i.e., item (1) of Lemma~\ref{lem:lhs-sampler-R} holds.
	To this end, we only need to show that for every node $u$ of $T$ at an even level, $\mathsf{SampleFromTree}(T,u)$ outputs a repair $D' \in \rep{\lambda(u)}{\dep_R}$ with probability $\frac{1}{\card{\rep{\lambda(u)}{\dep_R}}}$. This, since the root $u_r$ of $T$ is at level $0$, and $\lambda(u_r) = D_R$, will immediately imply correctness.
	We proceed by induction on the level of $u$.
	
	\medskip
	\noindent \paragraph{Base Case.} Assume that $u$ is a leaf node of $T$, i.e., at level $2n$. In this case, $\mathsf{SampleFromTree}$ returns $\lambda(u)$ with probability $1$. Since $u$ is a leaf of $T$, $\lambda(u)$ is consistent w.r.t.~$\dep_R$. Hence, $\lambda(u)$ is the only repair in $\rep{\lambda(u)}{\dep_R}$, and the claim follows.
	
	\medskip
	\noindent \paragraph{Inductive Step.} Assume now that $u$ is a node of $T$ that appears at level $2i$, for some $i \in \{0,,\ldots,n-1\}$. Assuming $\child{u} = \{v_1,\ldots,v_k\}$, since $u$ occurs at an even level, we have that $\{\lambda(v_1), \ldots,\lambda(v_k)\} = \block{\lambda(u)}{\phi_{i+1}}$. That is, each child of $u$ is associated to a block of the database associated with $u$ w.r.t.~the FD $\phi_{i+1}$ of $\Lambda$.
	By the definition of block, and from the fact that $\Lambda$ is an LHS chain of $\dep_R$, we conclude that $\{\lambda(v_1), \ldots,\lambda(v_k)\} = \block{\lambda(u)}{\phi_{i+1}}$ is a $\dep_R$-independent partition of $\lambda(u)$.
	Hence, by Lemma~\ref{lem:partition-repairs-independent} and Corollary~\ref{cor:partition-count}, to prove that $\mathsf{SampleFromTree}(T,u)$ outputs a repair of $\rep{\lambda(u)}{\dep_R}$ with probability $\frac{1}{\card{\rep{\lambda(u)}{\dep_R}}}$, we need to show that, for each $v \in \child{u}$, after lines 11 and 12 of $\mathsf{SampleFromTree}$, $D_{v'} \in \rep{\lambda(v)}{\dep_R}$, and is chosen with probability $\frac{1}{\card{\rep{\lambda(v)}{\dep_R}}}$.
	To this end, let $v \in \child{u}$.
	Since $v$ is at level $2i + 1$, we have that $\child{v} = \{v'_1,\ldots,v'_k\}$, and $\{\lambda(v'_1),\ldots,\lambda(v'_k)\} = \sblock{\lambda(v)}{\phi_{i+1}}$. That is, each child of $v$ is associated to a subblock of the database associated with $v$, w.r.t.~$\phi_{i+1}$. By definition of subblock, and from the fact that $\lambda(v)$ is a block of $\lambda(u)$ w.r.t.~$\phi_{i+1}$, we conclude that $\{\lambda(v'_1),\ldots,\lambda(v'_k)\} = \sblock{\lambda(v)}{\phi_{i+1}}$ is a $\dep_R$-conflicting partition of $\lambda(v)$. Hence, by Corollary~\ref{cor:partition-count}, line 11 is choosing a child $v' \in \child{v}$ over a well-defined probability distribution. Moreover, since by inductive hypothesis $\mathsf{SampleFromTree}(T,v')$ outputs a repair of $\rep{\lambda(v')}{\dep_R}$, for some child $v' \in \child{v}$,  with probability $\frac{1}{\card{\rep{\lambda(v')}{\dep_R}}}$, by Lemma~\ref{lem:partition-repairs-conflicting}, lines 11 and 12 actually construct a repair $D_{v'}$ of $\rep{\lambda(v)}{\dep_R}$ with probability
	\[ \frac{\card{\rep{\lambda(v')}{\dep_R}}}{\card{\rep{\lambda(v)}{\dep_R}}} \cdot \frac{1}{\card{\rep{\lambda(v')}{\dep_R}}}\ =\ \frac{1}{\card{\rep{\lambda(v)}{\dep_R}}}.
	\]

	\medskip
	
	\noindent\paragraph{\underline{Running time of $\mathsf{RSample}$}} 
	\smallskip
	
	\noindent We now show that $\mathsf{RSample}(D,\dep,R)$ runs in polynomial time in $||D||$, i.e., item (2) of Lemma~\ref{lem:lhs-sampler-R}.
	Computing an LHS chain $\Lambda = \phi_1,\ldots,\phi_n$ of $\dep_R$ is independent of $D$. Moreover, constructing  the $(R,\Lambda)$-blocktree $T = (V,E,\lambda)$ of $D$ w.r.t~$\dep$ is feasible in polynomial time in $||D||$ for the following reasons: $T$ has height $2n$ and each node $u$ in $T$ has at most $|D|$ children; hence, $T$ has polynomially many nodes in $||D||$. Furthermore, given a node $v$, if $v$ is the root of $T$, constructing its associated database $\lambda(v)$ is trivial, and if $v$ is not the root, constructing $\lambda(v)$ boils down to construct a block (or a subblock) of $\lambda(u)$ w.r.t.~some FD in $\Lambda$, where $u$ is the parent of $v$. Clearly, computing a block (resp., subblock) is feasible in polynomial time in $||D||$.
	
	Regarding the procedure $\mathsf{SampleFromTree}$, this is a recursive procedure performing a number of recursive calls which is at most the number of nodes of $T$, which, as discussed above, is polynomial in $||D||$. Moreover, at each call, $\mathsf{SampleFromTree}$ either receives a leaf node $u$, and thus it only needs to output the (already computed) database $\lambda(u)$, or it needs to compute the number of repairs of (polynomially many) subsets of $D$ w.r.t.~$\dep_R$. Since $\dep_R$ has an LHS chain, the latter can be performed in polynomial time (by Proposition~\ref{pro:counting-repairs-lhs-fp}). Therefore, the call $\mathsf{SampleFromTree}(T,u_r)$ in $\mathsf{RSample}$ runs in polynomial time in $||D||$, and the claim follows.
\end{proof}

With the above lemma in place, it is now straightforward to define the procedure $\mathsf{Sample}$. Given a database $D$, and a set $\dep$ of FDs with an LHS chain (up to equivalence), both over a schema $\ins{S}$, $\mathsf{Sample}$ first constructs a canonical cover $\dep'$ of $\dep$ (which has an LHS chain), and then outputs
\[
\bigcup_{R \in \ins{S}} \mathsf{RSample}(D,\dep',R).
\]
The construction of $\dep'$ is database-independent. Since $\rep{D_R}{\dep'_R} = \rep{D_R}{\dep'}$, for each $R \in \ins{S}$, by Lemmas~\ref{lem:partition-repairs-independent} and~\ref{lem:partition-repairs-conflicting}, Corollary~\ref{cor:partition-count} and Lemma~\ref{lem:lhs-sampler-R}, for every database $D'$, $\Pr(\mathsf{Sample}(D,\dep) = D') = \frac{1}{\card{\rep{D}{\dep}}}$, and $\mathsf{Sample}(D,\dep)$ runs in polynomial time in $||D||$, as needed.
\end{proof}

With the procedure $\mathsf{Sample}$ in place, item (2) of Proposition~\ref{pro:union-of-sets-properties} follows.
\end{proof}

\subsection{Proof of Lemma~\ref{lem:mapping}}

\begin{replemma}{\ref{lem:mapping}}
	\lemmapping
\end{replemma}

\begin{proof}
	Let $\varphi = C_1 \wedge \cdots \wedge C_m$, with $C_i = \ell^1_i \vee \ell^2_i \vee \ell^3_i$ for each $i \in [m]$, and let $x_1,\ldots,x_n$ be the Boolean variables of $\varphi$.
	In what follows, we let $\phi_1 = R: \text{Var} \ra \text{VValue}$ and $\phi_2 = R : \text{\rm Clause} \ra \text{\rm LValue}$.  
	For notational convenience, we write $D^k_\varphi$ for the database $\mathsf{db}_k(\varphi)$. We start by defining the mapping $\mathsf{Map}$ that assigns a truth assignment of $\varphi$ to a set of reparis of $D^k_\varphi$. In what follows, by abuse of notation, we refer to the symbol $C^j_i$ simply as a clause, and treat it as the clause $C_i$ of $\varphi$.
	
	Let $\tau : \{x_1,\ldots,x_n\} \ra \{0,1\}$ be a truth assignment (not necessarily satisfying) of $\varphi$. By overloading the notation, $\tau$ is extended to literals and formulas in the obvious way. For a clause $C^j_i = \ell^1_i \vee \ell^2_i \vee \ell^3_i$, we denote by $D^\tau_{i,j}$ the database
	\begin{center}
		\begin{tabular}{@{}ccccc@{}}
			\toprule
			Var & VValue & Clause & LValue &\\ \midrule
			$\var{\ell^1_i}$ & $\tau(\var{\ell^1_i})$ & $\angletup{C^j_i,\tau(\ell^1_i)}$ & $\tau(\ell^1_i)$ \\
			$\var{\ell^2_i}$ & $\tau(\var{\ell^2_i})$ & $\angletup{C^j_i,\tau(\ell^2_i)}$ & $\tau(\ell^2_i)$  \\
			$\var{\ell^3_i}$ & $\tau(\var{\ell^3_i})$ & $\angletup{C^j_i,\tau(\ell^3_i)}$ & $\tau(\ell^3_i)$  \\ \bottomrule
		\end{tabular}
	\end{center}
	Roughly, $D^\tau_{i,j}$ encodes how clause $C^j_i$ evaluates w.r.t.\ $\tau$,  by specifying the truth value of each variable in $C^j_i$. This is what we meant for \emph{$\tau$ chooses values for the literals in the $C^j_i$-gadget}. For example, with assignment $\tau = \{x_1 \mapsto 1, x_2 \mapsto 0, x_3 \mapsto 0\}$ and a clause $C^j_i = x_1 \vee \neg x_2 \vee x_3$, the database $D^\tau_{i,j}$ is:
	\begin{center}
		\begin{tabular}{@{}ccccc@{}}
			\toprule
			Var & VValue & Clause & LValue\\ \midrule
			$x_1$ & 1 & $\angletup{C^j_i,1}$ & 1\\
			$x_2$ & 0 & $\angletup{C^j_i,1}$ & 1 \\
			$x_3$ & 0 & $\angletup{C^j_i,0}$ & 0 \\ \bottomrule
		\end{tabular}
	\end{center}
	Note that each $D^\tau_{i,j}$ is a consistent subset of $D^k_\varphi$.

	For each clause $C^j_i$, we denote by $D_{C^j_i}$ the database containing the following fact of $D^k_C$:
	\begin{center}
		\begin{tabular}{@{}ccccc@{}}
			\toprule
			Var & VValue & Clause & LValue\\ \midrule
			$\star$ & $\star$ & $\angletup{C^j_i,1}$ & 0 \\ \bottomrule
		\end{tabular}
	\end{center}
	
	Consider now the set $\mathcal{C}_T$ (resp., $\mathcal{C}_F$) of all clauses $C^j_i$ such that $\tau(C^j_i) = 1$ (resp., $\tau(C^j_i) = 0$), i.e., $\mathcal{C}_T$ contains all the clauses that are true w.r.t. $\tau$, and $\mathcal{C}_F$ all clauses that are false w.r.t. $\tau$. 
	Let $\mathcal{C}= \mathcal{C}_F \cup \mathcal{C}_T$.\footnote{These sets should be parametrized with $\tau$, but in order to keep the notation simple, we omit it since $\tau$ is always clear.}
	Let $\mathcal{C}'$ be \emph{any} (possibly empty, not necessarily proper) subset of $\mathcal{C}_T$, and let, for each $C^j_i \in \mathcal{C}$, $D^{\tau,-}_{i,j}$ be the database obtained from $D^\tau_{i,j}$ by removing all facts having $1$ in position $(R,\text{\rm LValue})$. We define the database
	\[
	D^{\mathcal{C}'}_\text{start}\ =\ \bigcup\limits_{C^j_i \in \mathcal{C} \setminus \mathcal{C}'} D^\tau_{i,j} \cup \bigcup\limits_{C^j_i \in \mathcal{C}'}  (D^{\tau,-}_{i,j} \cup D_{C^j_i}).
	\]
	Note that $D^{\mathcal{C}'}_\text{start}$ is consistent. Intuitively, the above database represents the assignment $\tau$ in the following way. The truth state of each clause $C^j_i$ that is \emph{not} in $\mathcal{C}'$ is stored via the database $D^\tau_{i,j}$, specifying \emph{explicitly} the truth value of each literal in $C^j_i$ (this value is stored in position $(R,\text{\rm LValue})$), while for clauses $C^j_i \in \mathcal{C}'$, this information is stored using $D^{\tau,-}_{i,j}$, i.e., by omitting from $D^\tau_{i,j}$ the facts corresponding to literals that evaluate to true, and replacing the information that $C^j_i$ is true with the singleton database $D_{C^j_i}$. Hence, the set $\mathcal{C}'$ represents a choice of which true clauses $C^j_i$ we want to store via $(D^{\tau,-}_{i,j} \cup D_{C^j_i})$ rather than via $D^\tau_{i,j}$.
	In general, $D^{\mathcal{C}'}_\text{start}$ is not a repair, as it is not maximal, but it can be ``expanded'' with additional facts from $D^k_\varphi$ in order to obtain a repair. Different ways might exist to expand it, and thus lead to different repairs containing it. We use $\mathsf{expand}(D^{\mathcal{C}'}_\text{start})$ to denote the set of all repairs of $\rep{D^k_\varphi}{\dep}$ that contain $D^{\mathcal{C}'}_\text{start}$.

	We are finally ready to define our function. We let $\mathsf{Map}$ be the function from truth assignments over the variables of $\varphi$ to sets of databases such that, for every truth assignment $\tau$,
	\[
	\mathsf{Map}(\tau)\ =\ \bigcup\limits_{\mathcal{C}' \subseteq \mathcal{C}_T} \mathsf{expand}(D^{\mathcal{C}'}_\text{start}).
	\]
	In other words, $\mathsf{Map}(\tau)$ constructs all repairs that can be obtained by expanding $D^{\mathcal{C}'}_\text{start}$ by considering all possible choices of true clauses $\mathcal{C}' \subseteq \mathcal{C}_T$.
	We proceed to show that $\mathsf{Map}$ enjoys the properties stated in the three items of Lemma~\ref{lem:mapping}, which in turn will prove Lemma~\ref{lem:mapping} itself.

	\medskip
	\textbf{(Item 1)} It follows by the definition of the function $\mathsf{Map}$.
	
	\medskip
	\textbf{(Item 2)} Let $D' \in \rep{D^k_\varphi}{\dep}$. We construct from $D'$ a truth assignment $\tau$ as follows.  For each Boolean variable $x_i$ of $\varphi$, if there is a fact $R(\bar{t}) \in D'$ of the form
	\begin{center}
		\begin{tabular}{@{}ccccc@{}}
			\toprule
			Var & VValue & Clause & LValue\\ \midrule
			$x_i$ & $\mathsf{varvalue}$ & $\angletup{C^{j}_{i},\mathsf{litvalue}}$ & $\mathsf{litvalue}$ \\ \bottomrule
		\end{tabular}
	\end{center}
	then $\tau(x_i) = \mathsf{varvalue}$. Since $D'$ is a repair, and thus satisfies $\phi_1$, there cannot be a variable $x_i$ that gets assigned two different truth values, and thus, $\tau$ is well-defined. Moreover, for every Boolean variable $x_i$, there is a tuple in $D'$ containing $x_i$ of the form above, and thus $\tau$ is defined over each variable of $\varphi$. Indeed, if this was not the case, even a tuple of the following form, for some $C^j_i$
	\begin{center}
		\begin{tabular}{@{}ccccc@{}}
			\toprule
			Var & VValue & Clause & LValue\\ \midrule
			$x_i$ & $\mathsf{varvalue}$ & $\angletup{C^{j}_{i},0}$ & 0 \\ \bottomrule
		\end{tabular}
	\end{center}
	which is in $D^k_\varphi$, does not appear in $D'$. Since $D'$ is maximal, it means that the above tuple is in a conflict with some tuple of $D'$. Since position $(R,\text{\rm LValue})$ contains a $0$ in the above tuple, it can only be in a conflict with a tuple in $D'$ w.r.t.\ $\phi_1$, and thus, $D'$ contains a tuple containing $x_i$, which contradicts the hypothesis. Hence, $\tau$ is a truth assignment for $\varphi$.
	We now claim that $D' \in \mathsf{Map}(\tau)$. In particular, we show that $D^{\mathcal{C}'}_\text{start} \subseteq D'$, for some set of clauses $\mathcal{C}' \subseteq \mathcal{C}_T$.  Let $\mathcal{C}'$ be the maximal subset of $\mathcal{C}_T$ of clauses $C^j_i$ such that $D'$ does not contain $D^\tau_{i,j}$. Recall that
	$$ D^{\mathcal{C}'}_\text{start} = \bigcup\limits_{C^j_i \in \mathcal{C} \setminus \mathcal{C}'} D^\tau_{i,j} \cup \bigcup\limits_{C^j_i \in \mathcal{C}'}  (D^{\tau,-}_{i,j} \cup D_{C^j_i}).$$
	We show that $D'$ contains both unions in the above expressions.

	To see that $D'$ contains the first union of the above expression, note that $C^j_i \in \mathcal{C} \setminus{C}'$ means that either $C^j_i \in \mathcal{C}_F$ or $C^j_i \in \mathcal{C}_T \setminus \mathcal{C}'$.
	If $C^j_i \in \mathcal{C}_T \setminus \mathcal{C}'$, then $D'$ contains $D^\tau_{i,j}$ by the definition of $\mathcal{C}'$.
	If $C^j_i \in \mathcal{C}_F$, assume, towards a contradiction, that there is a fact $R(\bar{t})$ in $D^{\tau}_{i,j}$ that is not in $D'$. Since $C^j_i \in \mathcal{C}_F$, the tuple $\bar{t}$ must have a 0 in position $(R,\text{\rm LValue})$, and being $D'$ maximal, adding $R(\bar{t})$ to $D'$ causes a violation of $\phi_1$, which means $\tau(C_i) = 1$, contradicting the hypothesis that $C_i$ is false w.r.t.\ $\tau$. Hence, $D'$ contains the first union in the definition of $D^{\mathcal{C}'}_\text{start}$.
	
	To see that $D'$ contains the second union, let $C^j_i \in \mathcal{C}'$. We show that $D'$ contains $(D^{\tau,-}_{i,j} \cup D_{C^j_i})$ by showing that $D'$ contains $D^{\tau,-}_{i,j}$ and $D_{C^j_i}$. By contradiction, assume that $D'$ does not contain $D^{\tau,-}_{i,j}$. Thus, there exists a fact $R(\bar{t}) \in D^{\tau,-}_{i,j}$ that is not in $D'$. Since $D^{\tau,-}_{i,j}$ contains only facts with 0 in position $(R,\text{\rm LValue})$, by maximality of $D'$, $R(\bar{t})$ is in a conflict with some tuple of $D'$ w.r.t.\ $\phi_1$. Recall that a fact in $D^{\tau,-}_{i,j}$ encodes some literal whose value is false w.r.t.\ $\tau$. This means that $\tau$ is assigning two different truth values to the same variable, which is not possible. Hence, $D'$ contains $D^{\tau,-}_{i,j}$. We now show that $D'$ contains $D_{C^j_i}$, for each $C^j_i \in \mathcal{C}'$. Since, for any $C^j_i \in \mathcal{C}'$, $D'$ does not contain the full set $D^{\tau}_{i,j}$, but it contains $D^{\tau,-}_{i,j}$ as we have just shown, it means $D'$ has no facts from the $C^j_i$-gadget that have a 1 in position $(R,\text{\rm LValue})$. Hence, by maximality of $D'$, $D'$ must contain $D_{C^j_i}$, for each $C^j_i \in \mathcal{C}'$.

	Summing up, $D'$ contains $D^{\mathcal{C}'}_\text{start}$, which implies that $D' \in \mathsf{expand}(D^{\mathcal{C}'}_\text{start})$. Thus, $D' \in \mathsf{Map}(\tau)$.
	
	\medskip
	\textbf{(Item 3)} Let $\tau$ be a truth assignment. We first show that for each set $\mathcal{C}' \subseteq \mathcal{C}_T$, $\mathsf{expand}(D^{\mathcal{C}'}_\text{start})$ is a singleton, that is, there exists exactly one repair $D'$ of $\rep{D^k_\varphi}{\dep}$ that contains the consistent subset $D^{\mathcal{C}'}_\text{start}$. Recall that
	\[
	D^{\mathcal{C}'}_\text{start}\ =\ \bigcup\limits_{C^j_i \in \mathcal{C} \setminus \mathcal{C}'} D^\tau_{i,j} \cup \bigcup\limits_{C^j_i \in \mathcal{C}'}  (D^{\tau,-}_{i,j} \cup D_{C^j_i}).
	\]
	Consider first the set $\mathcal{T} \subseteq D^k_\varphi$ that are not in $D^{\mathcal{C}'}_\text{start}$ and that mention a clause $C^j_i \in \mathcal{C}_F$. Note that the only facts of $\mathcal{T}$ that can be added to $D^{\mathcal{C}'}_\text{start}$ without causing inconsistencies are necessarily the ones in $\mathcal{T} \cap D^k_C$. This is because for $C^j_i \in \mathcal{C}_F$,  $D^{\mathcal{C}'}_\text{start}$ contains $D^\tau_{i,j}$, and adding further facts from the $C^j_i$-gadget will cause an inconsistency w.r.t.\ $\phi_1$. Moreover, note that no fact from $\mathcal{T} \cap D^k_C$ can be in a conflict with any fact outside $D^{\mathcal{C}'}_\text{start}$ that is consistent with $D^{\mathcal{C}'}_\text{start}$.  Thus, the facts $\mathcal{T} \cap D^k_C$ can all be added to $D^{\mathcal{C}'}_\text{start}$, and every repair $D' \in \mathsf{expand}(D^{\mathcal{C}'}_\text{start})$ must contain $\mathcal{T} \cap D^k_C$. Hence,
	\[
	D_1\ =\ D^{\mathcal{C}'}_\text{start} \cup \bigcup\limits_{C^j_i \in \mathcal{C}_F} D_{C^j_i}
	\]
	is consistent, and every repair in $\mathsf{expand}(D^{\mathcal{C}'}_\text{start})$ contains $D_1$. We now show that $|\mathsf{expand}(D_1)| = 1$.
	
	Consider the set of facts $\mathcal{T}$ of $D^k_\varphi$ that are not in $D_1$, i.e., $\mathcal{T} = D^k_\varphi \setminus D_1$. By the above discussion on $D_1$, no facts in $\mathcal{T}$ that mention a clause $C^j_i \in \mathcal{C}_F$ can be further added to $D_1$ without causing inconsistency.
	Moreover, note that also facts $R(\bar{t}) \in \mathcal{T}$ that mention a clause $C^j_i \in \mathcal{C}_T \setminus \mathcal{C}'$ cannot be added to $D_1$ without causing conflicts, because $D_1$ contains $D^{\tau}_{i,j}$, and all facts in $D^k_C$ that mention $C^j_i$ have a 0 in position $(R,\text{\rm LValue})$. Therefore, the only facts of $\mathcal{T}$ that might be added to $D_1$ without conflicts are the ones that mention a clause $C^j_i \in \mathcal{C}'$. Since $D_1$ already contains $(D^{\tau,-}_{i,j} \cup D_{C^j_i})$, for each $C^j_i \in \mathcal{C}'$, among the facts of $\mathcal{T}$ that mention some $C^j_i \in \mathcal{C}'$, only the ones from the $C^j_i$-gadget, for each $C^j_i \in \mathcal{C}'$, that have a 0 in position $(R,\text{\rm LValue})$ can be added to $D_1$ without conflicts. Let $\mathcal{T}'$ be the restriction of $\mathcal{T}$ over such facts.
	To prove that $\mathsf{expand}(D_1)$ is a singleton, it suffices to show that no two facts of $\mathcal{T}'$ are in a conflict. Indeed, if this is the case, then, the only repair containing $D_1$ is the one where every fact $R(\bar{t}) \in \mathcal{T}'$ such that $D_1 \cup \{R(\bar{t})\}$ is consistent is added to $D_1$.
	
	Assume there are facts $R(\bar{t}),R(\bar{u}) \in \mathcal{T}'$ such that $\{R(\bar{t}),R(\bar{u})\} \not \models \dep$. From the above discussion, $\bar{t}$ and $\bar{u}$ must be tuples of the form
	
	\begin{center}
		\begin{tabular}{@{}cccccc@{}}
			\toprule
			\emph{Tuple} & Var & VValue & Clause & LValue\\ \midrule
			$\bar{t}$ &  $x$ & $\mathsf{varvalue}_1$ & $\angletup{C^{j_1}_{i_1},0}$ & 0\\ 
			$\bar{u}$ &  $x$ & $\mathsf{varvalue}_2$ & $\angletup{C^{j_2}_{i_2},0}$ & 0 \\ \bottomrule
		\end{tabular}
	\end{center}
	Thus, $R(\bar t)$ and $R(\bar u)$ must necessarily violate $\phi_1 = \text{Var} \ra \text{VValue}$; hence, $\mathsf{varvalue}_1 \neq \mathsf{varvalue}_2$. Since $R(\bar{t})$ and $R(\bar{u})$ are not in $D_1$, and thus, they are not in $D^{\tau,-}_{i_1,j_1}$ and $D^{\tau,-}_{i_2,j_2}$, respectively, and both have a 0 in position $(R,\text{\rm LValue})$, it must necessarily be the case that $C^{j_1}_{i_1}$ and $C^{j_2}_{i_2}$ each have a literal, over variable $x$, that is assigned $1$ by $\tau$. But this is not possible since $\mathsf{varvalue}_1 \neq \mathsf{varvalue}_2$. Hence, $\mathsf{expand}(D_1) = \mathsf{expand}(D^{\mathcal{C}'}_\text{start})$ is a singleton.

	To see now why
	$
	|\mathsf{Map}(\tau)| = \left|\bigcup\limits_{\mathcal{C}' \subseteq \mathcal{C}_T} \mathsf{expand}(D^{\mathcal{C}'}_\text{start})\right|= 2^{k \cdot c_\tau},
	$
	it suffices to note that for each $\mathcal{C}' \subseteq \mathcal{C}_T$, \emph{the} repair $\mathsf{expand}(D^{\mathcal{C}'}_\text{start})$ contains all facts of $D^k_C$ that mention a clause $C^j_i \in (\mathcal{C}_F \cup \mathcal{C}')$, and no other facts from $D^k_C$. Hence, for any two distinct $\mathcal{C}',\mathcal{C}'' \subseteq \mathcal{C}_T$, the two repairs $\mathsf{expand}(D^{\mathcal{C}'}_\text{start})$ and $\mathsf{expand}(D^{\mathcal{C}''}_\text{start})$ are different. Recall that $\mathcal{C}_T$ collects all clauses $C^j_i$ such that $\tau(C^j_i) = 1$, and thus, $|\mathcal{C}_T| = k \cdot c_\tau$. Since $\mathsf{Map}(\tau)$ contains all repairs $\mathsf{expand}(D^{\mathcal{C}'}_\text{start})$, for each (possibly empty, not necessarily proper) subset $\mathcal{C}'$ of $\mathcal{C}_T$, there are exactly $2^{k \cdot c_\tau}$ repairs in $\mathsf{Map}(\tau)$.
\end{proof}

\subsection{Proof of Theorem~\ref{the:apx-relfreq}}

\begin{reptheorem}{\ref{the:apx-relfreq}}
	\thmapxrelfreq
\end{reptheorem}

\begin{proof}
	Item (1) has been already discussed in the main body of the paper.
	Concerning item (2), the main part of the proof has been already given in the main body of the paper. It remains to show that the randomized procedure $\mathsf{A}$ correctly approximates the integer $\card{\rep{D}{\hat{\dep}}}$. For the sake of readability, let us first recall the procedure $\mathsf{A}$.
	Given a database $D$ over $\{R\}$, $\epsilon > 0$, and $0 < \delta < 1$, we define $\mathsf{A}$ as the randomized procedure:
	\begin{itemize}
		\item[-] compute the database $D'$ from $D$;
		\item[-] let $\epsilon' = \frac{\epsilon}{2+\epsilon}$;
		\item[-] let $r = \max\left\{\frac{1-\epsilon'}{2 \cdot (1+2^{|D|})}, \mathsf{A}'(D',\epsilon',\delta)\right\}$;
		\item[-] output $\frac{1}{r} - 1$.
	\end{itemize}
	Recall that $\mathsf{A}'$ is an FPRAS for $\prob{RelFreq}(\dep,Q)$, which exists by hypothesis. Recall also the key equation $\blacklozenge$ shown in the main body of the paper, which states that
	\[
	\rfreq{Q}{D',\dep}\ =\ \frac{1}{1 + \card{\rep{D}{\hat{\dep}}}}.
	\]
	We are now ready to establish the probabilistic guarantees. By assumption, it holds that
	\[
	\pr\left((1-\epsilon')\cdot \rfreq{Q}{D',\dep} \leq \mathsf{A}'(D',\epsilon',\delta) \leq (1+\epsilon') \cdot \rfreq{Q}{D',\dep}\right)\ \ge\ 1 - \delta.
	\]
	Thus, it suffices to show that the left-hand side of the above inequality is bounded from above by
	\[
	\pr\left((1-\epsilon)\cdot \card{\rep{D}{\hat{\dep}}} \leq \mathsf{A}(D,\epsilon,\delta) \leq  (1+\epsilon) \cdot \card{\rep{D}{\hat{\dep}}}\right).
	\]
	To this end, by equation $\blacklozenge$, we get that
	\[
	\pr\left((1-\epsilon')\cdot \rfreq{Q}{D',\dep} \leq \mathsf{A'}(D',\epsilon',\delta) \leq (1+\epsilon') \cdot \rfreq{Q}{D',\dep}\right)\ =\ \pr(E),
	\]
	where $E$ is the event
	\[
	\frac{1-\epsilon'}{1+\card{\rep{D}{\hat{\dep}}}} \leq \mathsf{A'}(D',\epsilon',\delta) \leq \frac{1+\epsilon'}{1+\card{\rep{D}{\hat{\dep}}}}.
	\]
	Note that $\card{\rep{D}{\hat{\dep}}} \le 2^{|D|}$, i.e., $\card{\rep{D}{\hat{\dep}}}$ is at most the the number of all subsets of $D$. Hence,
	\[
	\frac{1-\epsilon'}{1 + \card{\rep{D}{\hat{\dep}}}}\ \ge\ \frac{1-\epsilon'}{1+2^{|D|}}.
	\]
	Thus, for $E$ to hold is necessary that the output of $\mathsf{A}'(D',\epsilon',\delta)$ is no smaller than $\frac{1-\epsilon'}{1+2^{|D|}}$. Hence, for any number $p < \frac{1-\epsilon'}{1+2^{|D|}}$, $E$ coincides with the event
	\[
	\frac{1-\epsilon'}{1+\card{\rep{D}{\hat{\dep}}}} \leq \max\left\{p, \mathsf{A'}(D',\epsilon',\delta)\right\}
	\leq \frac{1+\epsilon'}{1+\card{\rep{D}{\hat{\dep}}}}.
	\]
	Hence, with $p=\frac{1-\epsilon'}{2 \cdot (1+2^{|D|})} < \frac{1-\epsilon'}{1 + 2^{|D|}}$, we conclude that
	\[	
	\pr(E)\ =\ \pr\left(\frac{1-\epsilon'}{1+\card{\rep{D}{\hat{\dep}}}} \leq \max\left\{p, \mathsf{A'}(D',\epsilon',\delta)\right\}
	\leq \frac{1+\epsilon'}{1+\card{\rep{D}{\hat{\dep}}}}\right).
	\]
	Since the random variable $\max\{p,\mathsf{A}'(D',\epsilon',\delta)\}$ always outputs a rational strictly larger than $0$, the latter probability coincides with
	\[
	\pr\left(\frac{1+\card{\rep{D}{\hat{\dep}}}}{1+\epsilon'} \leq \frac{1}{\max\left\{p, \mathsf{A'}(D',\epsilon',\delta)\right\}}
	\leq \frac{1+\card{\rep{D}{\hat{\dep}}}}{1-\epsilon'}\right).
	\]
	For short, let $X$ be the random variable $\frac{1}{\max\{p,\mathsf{A}'(D',\epsilon',\delta)\}}$. Since $\frac{1}{1-\epsilon'} = 1 + \frac{\epsilon'}{1-\epsilon'}$ and $\frac{1}{1+\epsilon'} = 1 - \frac{\epsilon'}{1+ \epsilon'} \ge 1 - \frac{\epsilon'}{1-\epsilon'}$, the probability above is less or equal than
	\[
	\pr\left(\left(1-\frac{\epsilon'}{1-\epsilon'}\right) \cdot (1+\card{\rep{D}{\hat{\dep}}}) \leq X 
	\leq \left(1+\frac{\epsilon'}{1-\epsilon'}\right) \cdot (1+\card{\rep{D}{\hat{\dep}}})\right).
	\]
	If we subtract $1$ from all sides of the inequality, then the above probability coincides with
	\[
	\pr\left(\left(1-\frac{\epsilon'}{1-\epsilon'}\right) \cdot (1+\card{\rep{D}{\hat{\dep}}}) - 1\leq X - 1 \leq \left(1+\frac{\epsilon'}{1-\epsilon'}\right) \cdot (1+\card{\rep{D}{\hat{\dep}}}) - 1\right).
	\]
	By expanding the products in the above expression, we obtain
	\[
	\pr\left(\card{\rep{D}{\hat{\dep}}} - \frac{\epsilon'}{1-\epsilon'} -  \frac{\epsilon'}{1-\epsilon'} \cdot \card{\rep{D}{\hat{\dep}}} \leq X - 1 \leq \card{\rep{D}{\hat{\dep}}} + \frac{\epsilon'}{1-\epsilon'} + \frac{\epsilon'}{1-\epsilon'} \cdot \card{\rep{D}{\hat{\dep}}}\right).
	\]
	Finally, since $\card{\rep{D}{\hat{\dep}}} \ge 1$, we have that
	\[
	\frac{\epsilon'}{1-\epsilon'}\ \le\ \frac{\epsilon'}{1-\epsilon'} \cdot \card{\rep{D}{\hat{\dep}}}.
	\] 
	Thus, the above probability is less or equal than
	\[
	\pr\left(\card{\rep{D}{\hat{\dep}}} - 2 \cdot\frac{\epsilon'}{1-\epsilon'} \cdot \card{\rep{D}{\hat{\dep}}} \leq X - 1 \leq  \card{\rep{D}{\hat{\dep}}} + 2 \cdot \frac{\epsilon'}{1-\epsilon'} \cdot \card{\rep{D}{\hat{\dep}}}\right),
	\]
	which in turn coincides with
	\[
	\pr\left( \left(1-2 \cdot\frac{\epsilon'}{1-\epsilon'}\right) \cdot \card{\rep{D}{\hat{\dep}}} \leq X - 1 \leq  \left(1+ 2 \cdot \frac{\epsilon'}{1-\epsilon'}\right) \cdot \card{\rep{D}{\hat{\dep}}}\right).
	\]
	Recalling that $\epsilon' = \frac{\epsilon}{2+\epsilon}$, one can verify that $2 \cdot \frac{\epsilon'}{1-\epsilon'} = \epsilon$. Moreover, $X-1$ is $\mathsf{A}(D,\epsilon,\delta)$. Hence, the above probability coincides with
	\[
	\pr\left( (1-\epsilon) \cdot \card{\rep{D}{\hat{\dep}}} \leq \mathsf{A}(D,\epsilon,\delta) \leq (1+ \epsilon) \cdot \card{\rep{D}{\hat{\dep}}}\right).
	\]
	Consequently, we have shown that 
	\[
	\pr\left( (1-\epsilon) \cdot \card{\rep{D}{\hat{\dep}}} \leq \mathsf{A}(D,\epsilon,\delta) \leq (1+ \epsilon) \cdot \card{\rep{D}{\hat{\dep}}}\right)\ \geq\ 1-\delta,
	\]
	which implies that $\mathsf{A}$ is an FPRAS for $\sharp \prob{Repairs}(\hat{\dep})$, and the claim follows.
\end{proof}

\end{document}